\documentclass[10.5pt,a4paper]{article}
\usepackage[utf8]{inputenc}
\usepackage{fullpage}
\usepackage{bm}
\usepackage{amsmath,amsfonts,amssymb,graphicx,color}
\usepackage{quantikz}
\usepackage{bbm}
\usepackage{cite}
\usepackage{amsthm}
\usepackage{authblk}
\usepackage{soul}
\usepackage{titlesec}
\usepackage[absolute,overlay]{textpos}

\usepackage[unicode=true,bookmarks=true,
bookmarksnumbered=false,bookmarksopen=false,
breaklinks=false,
pdfborder={0 0 1},
backref=false,colorlinks=true]{hyperref}
\usepackage{soul}

\numberwithin{equation}{section}

\titleclass{\subsubsubsection}{straight}[\subsection]

\newcounter{subsubsubsection}[subsubsection]
\renewcommand\thesubsubsubsection{\thesubsubsection.\arabic{subsubsubsection}}

\titleformat{\subsubsubsection}
  {\normalfont\normalsize\bfseries}{\thesubsubsubsection}{1em}{}
\titlespacing*{\subsubsubsection}
{0pt}{3.25ex plus 1ex minus .2ex}{1.5ex plus .2ex}

\makeatletter
\def\toclevel@subsubsubsection{4}
\def\l@subsubsubsection{\@dottedtocline{4}{7em}{4em}}
\makeatother
\hypersetup{pdftitle={title}, pdfauthor={Javier Pagan Lacambra, Aidan Chatwin-Davies, Masazumi Honda, Philipp A. Hoehn}, citecolor=blue,linkcolor=blue,urlcolor=blue,citecolor=blue}

\newtheorem{thm}{Theorem}[section]
\newtheorem{prop}[thm]{Proposition}
\newtheorem{lem}[thm]{Lemma}

\newtheorem{example}[thm]{Example}

\newcommand{\mbf}[1]{\mathbf{#1}}
\newcommand{\mrm}[1]{\mathrm{#1}}
\newcommand{\Hil}{\mathcal{H}}
\newcommand{\Oh}{\mathcal{O}}
\newcommand{\dee}{\mathrm{d}}
\newcommand{\Span}{\mathrm{span}}

\newcommand{\be} {\begin{equation}}
\newcommand{\ee} {\end{equation}}
\newcommand{\bea} {\begin{eqnarray}}
\newcommand{\eea} {\end{eqnarray}}
\newcommand{\eq} {equation}
\newcommand{\eqa} {eqnarray}
\newcommand{\NN} {\nonumber}

\usepackage{dsfont}

\newcommand{\shortket}[1]{| #1 \rangle}

\usepackage{xcolor}

\definecolor{purple}{rgb}{0.5,0,0.5}

\definecolor{jplgreen}{rgb}{0.0,0.5,0.0}

\title{\textbf{Gauss law codes and vacuum codes from \\lattice gauge theories}}

\author[1]{Javier P. Lacambra\thanks{j.pagan@oist.jp}}
\author[2,3]{Aidan Chatwin-Davies\thanks{aidan.chatwindavies@uri.edu}}
\author[4,5]{Masazumi Honda\thanks{masazumi.honda@riken.jp}}
\author[1]{Philipp~A.~H\"ohn\thanks{philipp.hoehn@oist.jp}}
\affil[1]{\small{Okinawa Institute of Science and Technology Graduate University \protect\\ Onna, Okinawa 904-0495, Japan}\smallskip}
\affil[2]{\small{Department of Physics, University of Rhode Island \protect\\ Kingston, RI 02881, United States of America}\smallskip}
\affil[3]{\small{Brown Center for Theoretical Physics and Innovation, Brown University \protect\\ Providence, RI 02912, United States of America}\smallskip}
\affil[4]{\small{Center for Interdisciplinary Theoretical and Mathematical Sciences, RIKEN \protect\\  Wako, Saitama 351-0198, Japan}\smallskip}
\affil[5]{\small{Graduate School of Science and Engineering, Saitama University \protect\\
Sakura-ku, Saitama 338-8570, Japan}}

\date{\today}

\begin{document}
\begin{textblock*}{8cm}(12cm,1cm)
\raggedleft
RIKEN-iTHEMS-Report-26, STUPP-26-295
\end{textblock*}


\maketitle
\thispagestyle{empty}

\begin{abstract}
We develop a comprehensive framework for constructing quantum error correcting codes (QECCs) from Abelian lattice gauge theories (LGTs) using quantum reference frames (QRFs) as a unifying formalism.
We consider LGTs with arbitrary compact Abelian gauge groups 
supported on lattices in arbitrary numbers of spatial dimensions, and we work with both pure gauge theories and theories with couplings to bosonic and fermionic matter. 
The codes that we construct fall into two classes:
First, \emph{Gauss law codes} identify the code subspace with the full gauge-invariant sector of the theory.
In models with matter coupled to gauge fields, these codes inherit a natural subsystem structure in which gauge-invariant Wilson loops and dressed matter excitations factorize the code space.
Second, \emph{vacuum codes} restrict the code subspace to the matter vacuum sector within the gauge-invariant subspace, yielding codes where errors correspond to gauge-invariant charge excitations rather than to violations of the Gauss law.
Despite their distinct setup, we show that when the gauge group is finite, vacuum codes are unitarily equivalent to pure gauge theory Gauss law codes, and that when the group is continuous, this is only true upon a charge coarse-graining of the vacuum code. In all cases, QRFs provide a systematic apparatus for fully characterizing the codes' algebraic structures and correctable error sets. For clarity, we illustrate our general results in $\mathbb{Z}_2$-gauge theory, as well as in scalar and fermionic QED. 
These findings offer fundamental insights into the parallelism between quantum error correction and gauge theory and point toward practical advantages for simulating LGTs on noisy quantum devices.
\end{abstract}

\clearpage

\pagenumbering{arabic}
\setcounter{page}{1}
\tableofcontents

\section{Introduction}

At a high level, quantum error correcting codes and quantum gauge theories share many anatomical similarities.
Beginning with a collection of logical degrees of freedom, a quantum error correcting code (QECC) embeds these redundantly within a larger collection of physical degrees of freedom.
The code maps out a subspace within the Hilbert space of the physical degrees of freedom, and to decode is to strip away redundancy and read out the encoded logical state.
In a gauge theory, gauge symmetry selects a smaller, gauge-invariant subspace within the larger Hilbert space of the theory's kinematical degrees of freedom, and this redundancy can be removed by fixing a gauge.
These qualitative parallels suggest a structural correspondence between QECCs and gauge theories, which can be made precise using quantum reference frames.

As their name suggests, quantum reference frames (QRFs) constitute a formalism for a relational description of quantum theory \cite{Aharonov67, Aharonov84}; that is, a setting in which the physics of subsets of a system's degrees of freedom is described conditionally with respect to the configurations of the other parts. The precursors of
QRFs trace their origins to quantum cosmology \cite{Page-Wootters, DeWitt67,Rovelli:1990pi,Rovelli02}; there are no external frames of reference when the system in question is an entire closed universe.
Instead, internal degrees of freedom, like the value of the scale factor or the homogeneous mode of a field, must be used as the \emph{frame} with respect to which other degrees of freedom evolve.
In this sense, a QRF is just a frame that is quantum mechanical.

A more modern and nuanced view of QRFs, however, is that they provide a universal toolkit for quantization in the presence of symmetry, with frames themselves being used to (de)parametrize the consequent redundancy \cite{Bartlett:2006tzx,Giacomini_2019,delaHamette:2021oex, Hoehn:2023ehz, Vanrietvelde:2018pgb,AliAhmad:2021adn,Araujo-Regado:2025ejs,Carette:2023wpz,Castro-Ruiz:2025yvi}.
As such, QRFs naturally find application to the quantization of gauge theories \cite{Carrozza:2021gju,Araujo-Regado:2024dpr,Araujo-Regado:2025ejs, Donnelly_2016,Fewster:2025ijg,Hoehn:2025pmx,Goeller:2022rsx,DeVuyst:2024khu,Freidel:2026stu,Kabel:2024lzr}.
Intuitively, within a gauge theory, different choices of reference frames amount to different ways of parametrizing gauge redundancy, yielding different pathways to non-redundant, gauge-fixed descriptions of the system.

Returning to the task at hand, the language of QRFs naturally mediates a correspondence between gauge theory and quantum error correction.
Indeed, our earlier work \cite{Carrozza:2024smc} established a precise and general dictionary between QECCs and gauge systems that are described by QRFs. As a first concrete study, we developed a detailed correspondence between qubit Pauli stabilizer QECCs and QRFs setups.  
There, a code that encodes $k$ logical qubits within $n$ physical qubits gives rise to QRFs having the symmetry group $(\mathbb{Z}_2)^{n-k}$.
A particular choice of QRF explicitly factorizes the $n$ qubits' physical Hilbert space into a system, which is isomorphic to $k$ qubits and supports the logical data, and a frame, which absorbs the code's redundancy.
A key insight from the dictionary is that different choices of QRF are in one-to-one correspondence with different choices for the error set that the code corrects, with the action of the chosen errors localizing to the frame. As an aside, this observation afforded a novel perspective on the relation between the correctability of errors sets and redundancy that the QRF encodes.
In principle, a single code subspace can be used to correct many different sets of elementary errors; in practice, this amounts to how one attributes syndromes to errors.
Intuitively, QRFs provide a clean and systematic way to parametrize this choice at the algebraic level, which organizes the $n$-qubit Hilbert space into the $k$-qubit code subspace and its images under the action of errors. As a first step towards lattice gauge theories, we explored this correspondence in detail in surface codes \cite{Kitaev:1997wr,Bravyi:1998sy}, which admit the interpretation of a lattice gauge theory with Abelian group.

In this article, we take a step further and construct QECCs out of bona fide Abelian lattice gauge theories (LGTs) \cite{Wilson:1974sk,Kogut:1974ag,Kogut:1979wt}, illustrating the general correspondence between QECCs and gauge systems laid out in \cite{Carrozza:2024smc} in a more general setting. 
We use QRFs to elucidate the codes' logical algebras, to identify what operations are errors, and to delineate sets of errors that are correctable.
We consider LGTs defined on generic directed graphs $(\mathcal{V}, \mathcal{L})$ in an arbitrary number of spatial dimensions, consisting of a collections of vertices $\mathcal{V} = \{v\}$ and directed edges $ \mathcal{L} = \{\ell\}$, and such that the lattice is closed. 
The gauge symmetry (structure) groups $G$ that we consider are compact Lie groups, and so may be finite, for example the cyclic group $G = \mathbb{Z}_D$, or continuous, for example in the case of $G = U(1)$, or they may be direct products of these finite and continuous components.
Gauge field degrees of freedom live on the edges of the lattice, and we allow for matter on the vertices that couples to the gauge fields.
Each vertex $v$, together with its adjacent edges, supports a unitary representation of a gauge transformation, $U_v^g$ for $g \in G$, and so we denote a local gauge transformation everywhere on the lattice by $U^{\bm g} \equiv \otimes_{v \in \mathcal{V}} U_v^{g_v}$.
Letting $\Hil_\mrm{kin}$ denote the full kinematical Hilbert space of gauge field and matter states on the lattice, our starting point is to interpret the gauge-invariant subspace
\begin{equation} \label{eq:Hpn}
    \Hil_\mrm{pn} = \left\{ \ket{\Psi} \in \Hil_\mrm{kin} ~ | ~ U^{\bm{g}} \ket{\Psi} = \ket{\Psi} ~ \forall ~ \bm{g} \in \bm{G} =G^{\times N_V}\right\} ,
\end{equation}
or certain subspaces thereof, as the code subspace of a QECC ($N_V$ denotes the number of vertices in $\mathcal{V}$).
The operators $U^{\bm{g}}$ thus function as the code's stabilizers.
Given that the unitary gauge transformations are generated by a Hermitian Gauss law constraint when the full gauge-invariant Hilbert space is taken as the code space, the name ``Gauss law codes'' has been adopted for such codes \cite{eczoo_gauss_law,Rajput:2021trn,Spagnoli:2024mib,Spagnoli:2026qni,Carrozza:2024smc,Yao:2025cxs}. In this work, we will also consider nontrivial subspaces of this to define codes.

The idea of identifying the stabilizers of a QECC with the unitary gauge transformations of a LGT has of course been explored before.
Before reviewing a selection of this literature, however, let us first address \emph{why} one might be interested in aligning a code with a gauge symmetry in the first place.
From a practical perspective, quantum simulation is the motivation.

Quantum simulation is anticipated to be a critically important use for quantum computers, since simulating quantum systems on conventional computers is generically inefficient, while the task may be performed efficiently on a quantum device \cite{Feynman1982Simulating, Lloyd1996UniversalQS, Georgescu_2014, Altman:2019vbv,Daley:2022eja}
Given the pervasiveness of gauge theories in science, simulating discretized LGTs is thus of extremely broad interest.
For example, while conventional numerical approaches to lattice QCD based on classical Monte Carlo simulation in the path integral formalism are successful in explaining static properties of hadrons (e.g.,~\cite{Ratti:2018ksb, Lin:2011ti, Philipsen:2012nu, Ding:2015ona}), there are still many important problems which seem difficult to explore via the conventional approach due to the infamous sign problem (see, e.g.,~\cite{Aarts:2015tyj,Nagata:2021ugx}).\footnote{
In the conventional approach, numerical integration of path integrals in lattice field theories is usually done by the Markov-chain Monte Carlo method, which treats Boltzmann weights as probabilities.
This method becomes problematic when the integrand is non-real or non-positive, and highly oscillating. This sign problem occurs when there are topological terms, chemical potentials, or one considers real-time evolution, to name a few settings.}
To that end, digital quantum simulation has recently attracted much attention as 
a means of overcoming the sign problem (see \cite{Banuls:2019bmf,Funcke:2023jbq,Bauer:2022hpo,DiMeglio:2023nsa} for reviews). 

Quantum error correction is necessary for robust quantum computing; however, its computational resource draw is prohibitively intensive for existing and near-term hardware.
Any optimizations that reduce the resource draw are thus advantageous, including optimizations that come from adapting an error correction scheme to a specific application, such as simulating LGTs.
Gauss law codes that are aligned with the gauge symmetry being simulated have the potential for a lower computational overhead compared to general-purpose codes.
Intuitively, this is because check operations of the code can be combined with the checks required during simulation to enforce the gauge theory's Gauss law, thus reducing overhead \cite{Rajput:2021trn, Carena:2024dzu}.

Some of the earliest work in this area (itself having draw inspiration from Refs.~\cite{Stryker:2018efp,Raychowdhury:2018osk}) 
examined Gauss law codes for the gauge group $\mathbb{Z}_2$ in one and two spatial dimensions \cite{Rajput:2021trn}, with later generalizations to $\mathbb{Z}_D$ in arbitrary dimensions \cite{Spagnoli:2024mib} and with coupling to matter \cite{Spagnoli:2026qni}.
Important insights obtained in this line of work include that both logical operations and time evolution can be implemented fault-tolerantly.
Ref.~\cite{Carena:2024dzu} established thresholds for when simulating lattice $\mathbb{Z}_D$ theories using Gauss law codes (i.e., with full gauge redundancy) is preferable, in terms of overhead, to simulating a gauge-fixed instance of the theory, while Ref.~\cite{Pato:2026wow} recently highlighted limitations of the approach.
Initial steps toward non-Abelian codes have been taken in Ref.~\cite{Yao:2025cxs}, which constructed codes from $SU(2)$ on the lattice.

In more detail, we consider two types of codes, both based on the idea of aligning the code subspace with an invariant subspace.
First, we construct Gauss law codes by identifying the code subspace, $\Hil_\mrm{code}$, with the full gauge-invariant subspace $\Hil_\mrm{pn}$ defined in Eq.~\eqref{eq:Hpn}.
We develop such codes in the pure gauge sector (i.e., without matter), with $\Hil_\mrm{kin} = \Hil_\mrm{gauge}$, and also from LGTs consisting of gauge fields coupled to matter, with $\Hil_\mrm{kin} = \Hil_\mrm{gauge} \otimes \Hil_\mrm{matter}$.
In so doing, we conduct a detailed analysis for both bosonic and fermionic models of matter at the lattice's vertices.
The matter-coupled Gauss law codes inherit a natural subsystem structure from the kinematical factorization into gauge and matter degrees of freedom, which reads
\begin{equation} \label{eq:introLoopsDrMatter}
    \Hil_\mrm{code} = \Hil_\mrm{pn} \simeq \Hil_\mrm{loops} \otimes_R \Hil_\mrm{matter}^\mrm{dr} \, .
\end{equation}
$\Hil_\mrm{loops}$ contains the gauge-invariant degrees of freedom from the gauge field sector, namely, excitations of Wilson loops above the vacuum state, while $\Hil_\mrm{matter}^\mrm{dr}$ describes gauge-invariant excitations of the matter fields that have been dressed by Wilson lines.
In principle, either factor can be interpreted as the information-carrying factor of the subsystem code.

The notation $\otimes_R$ indicates that the tensor product factorization is induced by a specific choice of QRF \cite{Carrozza:2024smc}, $R$, which we construct from spanning trees on the lattice \cite{Araujo-Regado:2025ejs}.
A spanning tree is a subgraph of the lattice $(\mathcal{V},\mathcal{L})$ which consists of a subset of edges $R \subset \mathcal{L}$ and a preferred vertex $v_0 \in \mathcal{V}$, called the root of the tree, such that there is a unique path $\gamma_R[v,v_0] \equiv R_v \subset R$ from each vertex $v\in\mathcal{V}$ to the root $v_0$; see Fig.~\ref{fig:latticetree}.
The Wilson lines along these tree paths constitute group-valued frames that transform covariantly under the structure group $G$ at each vertex $v\neq v_0$, and together they constitute a lattice field of such frames for $G$. Different choices of the frame $R$ yield different representations of the code's logical algebra and action of errors. Regardless of the specific error set in question, errors in these Gauss law codes are operators that violate the combined matter and gauge field Gauss law, resulting in excursions out of $\Hil_\mrm{pn}$ and mismatches between matter charge and gauge field flux.

Second, given that Gauss law codes feature errors that correspond to gauge charge excitation and Gauss law violation at each vertex, one may wonder whether there is a gauge-invariant incarnation of these codes. The underlying idea is that the charge excitations correspond to exciting an additional, previously ``frozen'' matter species that compensates the non-invariance of the fields described by the Gauss law code. This effort culminates in what we call \emph{vacuum codes}. We show that there is indeed a sense in which these codes constitute a gauge-invariant incarnation of Gauss law codes, though a unitary equivalence generally requires a nontrivial charge coarse-graining, as we will explain.

Here, our starting point is to interpret the gauge-invariant subspace $\Hil_\mrm{pn}$ in Eq.~\eqref{eq:Hpn} as the physical space of the code, i.e.,\ the one describing the ``physical qubits'' invoked to carry the code, rather than the LGT's kinematical configuration space.
We then identify the code subspace with the matter vacuum subspace of $\Hil_\mrm{pn}$.
With respect to the kinematical factorization $\Hil_\mrm{kin} = \Hil_\mrm{gauge} \otimes \Hil_\mrm{matter}$, a gauge transformation $U^g_v$ decomposes as $U^g_v = U^g_{v,\mrm{gauge}} \otimes u^g_{v,\mrm{matter}}$.
Therefore, in terms of gauge transformations, the vacuum code subspace is given by
\begin{equation}
    \Hil_\mrm{code} = \left\{ \ket{\Psi}_\mrm{pn} \in \Hil_\mrm{pn} ~ \right| \left. ~ \mathbbm{1}_\mrm{gauge} \otimes u^{\bm g}_\mrm{matter} \ket{\Psi}_\mrm{pn} = \ket{\Psi}_\mrm{pn} ~ \forall ~ \bm{g} \in \bm{G} \right\} \, .
\end{equation}
In other words, in addition to obeying the total gauge invariance condition $U^{\bm{g}} \ket{\Psi}_\mrm{pn} = \ket{\Psi}_\mrm{pn}$, demanding triviality of codewords under isolated gauge transformations in the matter sector forces the matter into its vacuum state.
Correspondingly, if there are several species of charged matter, we demand the stronger condition $u^{\bm g}_{\mrm{matter},i}\ket{\Psi}_\mrm{pn}=\ket{\Psi}_\mrm{pn}$ for each species $i$. This leads in the general case to a more fine grained error syndrome and thereby a finer code, though we demonstrate how one can map vacuum codes into Gauss law codes.

In contrast to Gauss law codes, vacuum codes are arguably more natural from the perspective of a gauge theory in Nature because errors keep the computational state within $\Hil_\mrm{pn}$; the ``physical'' space of the code is the space of ``physical'' configurations of the gauge theory.
Errors are thus excursions into $\Hil_\mrm{pn}$ and may be interpreted as gauge-invariant excitations of charge that was previously frozen.
For such errors, measuring a syndrome amounts to measuring the amount of charge at each vertex.
Gauss law codes nevertheless remain relevant for practical purposes; when simulating a gauge theory on a quantum computer, errors can of course cause the computational state to drift outside of the gauge-invariant subspace being simulated. 

We also note that both Gauss law codes and vacuum codes bear a certain similarity to surface codes \cite{Kitaev:1997wr,Bravyi:1998sy}, which include the well-known toric code and admit a lattice gauge theory interpretation. These codes feature two types of stabilizer operators, defined on plaquettes and on vertices or stars. The latter constitute Gauss law operators, whereas the former can be understood as implementing a flatness condition for the discrete gauge connection. Furthermore, the code space corresponds precisely to the ground state space of a Hamiltonian featuring these stabilizer operators. Accordingly, this class of codes can be viewed as a hybrid of Gauss law and vacuum codes: violations of the plaquette constraints map out of the ground or vacuum state space, while preserving the Gauss law, thus taking the form of vacuum code errors, while errors violating the star operators violate Gauss's law and therefore constitute errors of a Gauss law code. In contrast to what we discuss here, surface codes are often viewed as emergent gauge theories, where the gauge symmetry is not fundamental and not present in the entire physically relevant state space, but only in the ground state space and for excitations that do not violate the star operators. Moreover, the excitations beyond the ground state do not correspond to standard matter excitations, but quasi-particles instead. Hence, while our codes share certain qualitative features with surfaces codes, they can be regarded as a generalization of these features to bona fide LGTs with standard matter inclusions. We comment further on this in the main body.

For both Gauss law codes and vacuum codes, QRFs provide a systematic way of identifying the possible correctable error sets. 
This complete and precise characterization of the inequivalent, maximal sets of correctable errors is a striking output of the QRF formalism.
A collection of errors $\mathcal{E} = \{E_\mu\}_{\mu \in \Lambda}$ is correctable when all pairs of errors satisfy the Knill-Laflamme condition \cite{Knill:1996ny},
\begin{equation} \label{eq:KLintro}
    \Pi_\mrm{code} E_\mu^\dagger E_\nu \Pi_\mrm{code} = C_{\mu\nu} \Pi_\mrm{code} \, ,
\end{equation}
where $\Pi_\mrm{code}$ is the projector onto $\Hil_\mrm{code}$ and $C_{\mu\nu}$ is a Hermitian matrix.
Two error sets are equivalent when (the code restriction of) their linear spans are the same, and an error set is maximal if it cannot be enlarged by including additional linearly independent errors while still remaining correctable.
In short, QRFs provide a clean way to solve \eqref{eq:KLintro} for error sets that are demonstrably maximal.

Here, errors are operators that violate the Gauss law, or a restriction of the Gauss law to the matter vacuum sector in the case of vacuum codes.
Error therefore result in charged defects at vertices, and the pattern of defect charges constitutes a quantitative syndrome for each error.
Accordingly, we call the mappings that we construct from errors to syndromes \emph{Gauss law maps}.
A Gauss law map is of course many-to-one---there are many inequivalent ways to create an apparent charge defect at a given vertex.
In other words, we may organize the errors that produce the same syndrome into fibers of a Gauss law map.
Modulo certain subtleties when fermionic matter couples to the gauge field, correctable errors sets are then sections of the fiber bundle, and a maximal such set includes one representative from each fiber, thus constituting precisely a global section.
Using a QRF to pass from the kinematical factorization of $\Hil_\mrm{code}$ to the frame-induced factorization \eqref{eq:introLoopsDrMatter} ultimately allows us to neatly parametrize the fibers and sections of Gauss law maps in a way that makes maximality plainly apparent.

The findings that we discuss in this article emphasize the structural insights obtained by using QRFs to construct QECCs from LGTs, extending the case studies of \cite{Carrozza:2024smc}.
Nevertheless, our findings are of practical interest for computational settings in which gauge theoretic and error structures align.
As has been previously noted \cite{Stryker:2018efp, Rajput:2021trn}, Gauss law codes are essentially classical codes that protect against ``bit flips'' that consist of dislocations of gauge field flux and matter charge.
While it is possible to upgrade Gauss law codes to genuine quantum codes \cite{Rajput:2021trn,Spagnoli:2024mib,Spagnoli:2026qni}, already a Gauss law code or a vacuum code should offer significant advantage, for example, in the presence of noise that is heavily biased toward $X$-type or $Z$-type errors \cite{Aliferis_2008,Puri:2019qkc}.
We continue with a brief discussion of practical considerations in Sec.~\ref{sec:disc}.

The organization of this article is as follows.
We begin by reviewing background material on stabilizer quantum error correction, Abelian lattice gauge theories, quantum reference frames, and the QECC/QRF correspondence of \cite{Carrozza:2024smc} in Sec.~\ref{sec:preliminaries}.
Sec.~\ref{sec:general} then develops the general theory of how to obtain a QECC from a generic, compact, Abelian LGT on a closed lattice.
We discuss how to obtain Gauss law codes and vacuum codes, and in each case explain, using QRFs, how the logical algebras and correctable error sets emerge. Crucially, we also discuss the relation between Gauss law and vacuum codes, demonstrating that the latter are unitarily equivalent to the pure gauge theory version of the former, when $G$ is finite, and that a unitary equivalence only holds when $G$ has continuous components, provided a charge coarse-graining at each vertex is carried out in the vacuum code.
In Sec.~\ref{sec:examples}, we illustrate the general theory with a succession of concrete examples, encompassing various instances of $\mathbb{Z}_2$ gauge theory, as well as scalar and fermionic QED in different dimensions.
A discussion follows in Sec.~\ref{sec:disc}, and we provide an outlook and closing remarks in Sec.~\ref{sec:conc}.

\subsubsection*{Suggested Reading Guide}

An experienced practitioner of quantum error correction, lattice gauge theory, or quantum reference frames could skip the corresponding parts of Sec.~\ref{sec:preliminaries}, although briefly skimming this section is recommended to take note of our notation and conventions.
For a minimal illustration of the technical apparatus, a reader may skip to Sec.~\ref{sec:examples}; in particular, the first three examples in Secs.~\ref{ssec:3qubit_puregauge}, \ref{ssec:3qubit_gauss}, and \ref{ssec:3qubit_vacuum} are already enough to illustrate the majority of the technical results.
For a reader who wishes to dive into the general theory in full detail, we recommend reading Secs.~\ref{sec:general} and \ref{sec:examples} in tandem.

\subsubsection*{Note Added}

While completing this work, we became aware of a parallel effort to connect lattice quantum electrodynamics with quantum error correction using quantum reference frames, and we anticipate some overlap with results in this manuscript.
This work is documented in part of a M.Sc.\ thesis by Elias Rothlin \cite{Elias} and the article \cite{Lin-Qing_and_friends}, which is submitted simultaneously to the arXiv.
We thank Elias Rothlin, Carla Ferradini and Lin-Qing Chen for coordination.
~

\section{Preliminaries} \label{sec:preliminaries}

We begin with a review of the most important background material for this article.
Sec.~\ref{ssec:stabilizer_codes} briefly reviews stabilizer quantum error correcting codes, followed by a short introduction to lattice gauge theory in Sec.~\ref{sec:LGT}.
Sec.~\ref{ssec_QRF} introduces quantum reference frames and their specific adaptations to lattice gauge theories. 
We finish with a review of our previous work on a dictionary between quantum error correcting codes and quantum reference frame setups in Sec.~\ref{ssec_QECC/QRF}, in anticipation of applying the formalism to lattice gauge theory.

\subsection{Stabilizer Codes}
\label{ssec:stabilizer_codes}

We begin by briefly reviewing the stabilizer formalism for quantum error correcting codes, primarily as a way to establish conventions and notation.
For a more complete review of quantum error correction (QEC) and stabilizer quantum error correcting codes (QECCs), see, e.g., \cite{Gottesman:1997zz,Nielsen2010,Gottesman:2009zvw}.

In essence, a QECC consists of a relation among three Hilbert spaces: a logical space, $\Hil_\mrm{log}$, a kinematical Hilbert space, $\Hil_\mrm{kin}$, and a code subspace, $\Hil_\mrm{code}$.
The first space, $\Hil_\mrm{log}$, is the space of states one wishes to manipulate and use for quantum computation.
In most QECCs, an isometry $V : \Hil_\mrm{log} \rightarrow \Hil_\mrm{kin}$ is used to embed $\Hil_\mrm{log}$ into the larger Hilbert space $\Hil_\mrm{kin}$, which represents the Hilbert space realized by the actual degrees of freedom that constitute a quantum computer.
Since the computer's physical layer implements this latter space, it is usually called the ``physical'' Hilbert space in the QEC literature.
However, so as to avoid conflicting with terminology from the QRF literature, and in anticipation of forthcoming connections, here we will borrow QRF terminology and refer to this space as the ``kinematical'' Hilbert space.
We call a quantum operation which implements $V$ an \emph{encoding map}, and a quantum operation which implements $V^\dagger$ a \emph{decoding map}.
Finally, the image of $\Hil_\mrm{log}$ within $\Hil_\mrm{kin}$ under the isometry $V$ is the code subspace, $\Hil_\mrm{code}$.

Let $\mathcal{U}(\Hil_\mrm{kin})$ denote the set of unitary operators acting within $\Hil_\mrm{kin}$.
In a stabilizer QECC, one characterizes $\Hil_\mrm{code}$ as the subspace of $\Hil_\mrm{kin}$ which is invariant under the action of an Abelian subgroup $\mathcal{G} \subseteq \mathcal{U}(\Hil_\mrm{kin})$ called the code's \emph{stabilizer group}; to wit,
\begin{equation} \label{eq:sec2Hcode}
\Hil_\mrm{code} = \{\shortket{\bar{\psi}} \in \Hil_\mrm{kin} ~ | ~ U \shortket{\bar{\psi}} = \shortket{\bar{\psi}} ~ ~ \forall ~ U \in \mathcal{G} \}.
\end{equation}

Elements of the stabilizer group $\mathcal{G}$ act trivially on the code subspace and therefore furnish representations of the logical identity operator; in other words, $V^\dagger U V = \mathbbm{1}_\mrm{log}$ and $U|_{\Hil_\mrm{code}} = \mathbbm{1}_\mrm{code}$ for each $U \in \mathcal{G}$.
Unitary operators on $\Hil_\mrm{kin}$ that preserve $\Hil_\mrm{code}$ but do not act as the identity within it implement nontrivial logical operations on the encoded degrees of freedom.
In contrast, unitary operators that map states in $\Hil_\mrm{code}$ to states with support outside of $\Hil_\mrm{code}$ are interpreted as errors.

To be a complete QEC scheme, a code must also be equipped with a recovery operation that returns corrupted states back to the code subspace.
Typically, recovery proceeds by first diagnosing which error has occurred and then applying a correction conditioned on the observed syndrome.
If the combined action of an error followed by recovery acts as the identity within $\Hil_\mrm{code}$, then error correction is said to succeed.
Otherwise, if the net effect corresponds to the application of a nontrivial logical operator to the encoded state, a logical error has occurred.

A large and practically important class of QECCs is given by qubit Pauli stabilizer codes. In these constructions, $k$ logical qubits are isometrically embedded into $n$ kinematical qubits, so that $\Hil_\mrm{log}~\cong~\Hil_\mrm{code}~\cong~(\mathbb{C}^2)^{\otimes k}$ and $\Hil_\mrm{kin}~\cong~(\mathbb{C}^2)^{\otimes n}$.
The stabilizer group $\mathcal{G}$ is taken to be an Abelian subgroup of the $n$-qubit Pauli group, $\mathcal{P}_n$, whose elements are tensor products of single-qubit Pauli operators $I, X, Y, Z$, multiplied by an overall phase $\pm 1$ or $\pm i$.
In the computational basis $\{\ket{0}, \ket{1}\}$, the single-qubit Pauli operators read
\begin{equation}
    I = \begin{pmatrix}
1 & 0 \\
0 & 1
\end{pmatrix} \quad X = \begin{pmatrix}
0 & 1 \\
1 & 0
\end{pmatrix} \quad iY = \begin{pmatrix}
0 & 1 \\
-1 & 0
\end{pmatrix} \quad Z = \begin{pmatrix}
1 & 0 \\
0 & -1
\end{pmatrix},
\end{equation}
and so
\begin{equation}
    \mathcal{P}_n = \left\{ i^\lambda P_1 \otimes \cdots \otimes P_n ~ | ~ \lambda \in \{0, 1, 2, 3\}, P_j \in \{I, X, Y, Z\}, 1 \leq j \leq n  \right\} \, .
\end{equation}
For a code that encodes $k$ logical qubits, $\mathcal{G}$ contains $2^{n-k}$ distinct elements and can be generated by any subset of $n-k$ independent elements, which are then called \emph{stabilizer generators}.

Given a string of Pauli operators $P_1 \cdots P_n$ (in which we suppress then tensor product symbol), the \emph{weight} of the operator is the number of non-identity single-qubit Pauli operators appearing in the string.
The \emph{distance}, $d$, of a qubit QECC is the weight of the smallest-weight encoded logical operator.
Altogether, the parameters of a QECC are often succinctly communicated using double brackets as $[[n,k,d]]$.
Single brackets are used analogously for classical error correcting codes for bit strings.

\emph{Error sets} are sets of operators that act on $\Hil_\mrm{kin}$ and which describe the errors that can occur during a computation.
For example, given a collection of $n$ kinematical qubits, the set
\begin{equation}
    \mathcal{E} = \{X_i ~ | ~ 1 \leq i \leq n\} \cup \{ I^{\otimes n} \}
\end{equation}
allows for an erroneous Pauli $X$ operator on any given qubit,\footnote{The notation $X_i$ means $X$ applied to the $i$\textsuperscript{th} qubit, tensored with identity operators on all other qubits, and extends in the expected way to generic multipartite Hilbert spaces.} as well as the option that no error is applied---hence the inclusion of the identity operator.
An example of an \emph{error channel} constructed from this set is
\begin{equation}
\begin{aligned}
    \mathcal{N} ~ : ~ \mathcal{B}(\Hil_\mrm{kin}) &\rightarrow \mathcal{B}(\Hil_\mrm{kin}) \\
    \rho &\mapsto (1-p) \rho + \sum_{i=1}^n \frac{p}{n} X_i \rho X_i \, ,
\end{aligned}
\end{equation}
which could describe a model in which no error occurs with probability $1-p$, and a single Pauli $X$ is applied with probability $p$.

Given a Pauli stabilizer QECC and a set of stabilizer generators, a unitary error that consists of a single string of Pauli operators will anticommute with at least one stabilizer generator.
In other words, an error $E$ and a stabilizer generator $U$ as such result in
\begin{equation}
    U E \ket{\bar{\psi}} = - E U \ket{\bar{\psi}} = -E \ket{\bar{\psi}}
\end{equation}
for any $\ket{\bar\psi} \in \Hil_\mrm{code}$.
Pauli stabilizer generators, being self-adjoint, can be measured; hence, a typical recovery scheme for a Pauli stabilizer QECC is to measure a set of stabilizer generators, yielding $n-k$ measurement outcomes of $\pm 1$'s, and to match this syndrome to an error from a specified error set.
Recovery then proceeds by applying the decoded error $E$ again since, if the decoding was correct and $E$ being a Pauli string, it follows that $E^2 = \mathbbm{1}$.

When is a set of errors correctable?
This is spelled out by the Knill-Laflamme (KL) condition \cite{Knill:1996ny}.
Let $\Hil_\mrm{code}$ be a stabilizer code subspace within a kinematical space $\Hil_\mrm{kin}$, and let $\Pi_\mrm{code}$ denote the projector onto $\Hil_\mrm{code}$.
An error set $\mathcal{E} = \{E_\mu\}_{\mu \in \Lambda}$ is correctable if and only if
\begin{equation}\label{eq_KLstabcode}
    \Pi_\mrm{code} E_\mu^\dagger E_\nu \Pi_\mrm{code} = C_{\mu\nu} \Pi_\mrm{code}
\end{equation}
for a Hermitian matrix $C_{\mu\nu}$.
In particular, note that $\Hil_\mrm{code}$ and $\Hil_\mrm{kin}$ need neither be isomorphic to collections of qubits, nor even finite-dimensional \cite{Beny_2009}.\footnote{That the (KL) condition continues to hold in infinite settings follows from the result that quantum channels on infinite-dimensional Hilbert space continue to admit an operator-sum decomposition \cite{holevo2011entropy}.} 
Furthermore, error sets need not consist of unitary operators; members of a correctable set of errors $\mathcal{E}$ could include projectors, for example.
As long as the KL condition is satisfied, then $\mathcal{E}$ is correctable and the errors act unitarily \emph{when restricted to the code subspace}.

We also adopt the following definition from \cite{Carrozza:2024smc} of equivalence of correctable error sets.
We say that two correctable error sets are \emph{equivalent} if the complex linear span of their elements, restricted to the code subspace, is the same.
In other words, if $\mathcal{E} = \{E_\mu\}_{\mu \in \Lambda}$ and $\mathcal{E}' = \{E_{\mu'}'\}_{{\mu'} \in \Lambda'}$, then $\mathcal{E}$ and $\mathcal{E}'$ are equivalent if
\begin{equation}
    \Span_\mathbb{C} \{ E_\mu \Pi_\mrm{code} \}_{\mu \in \Lambda} = \Span_\mathbb{C} \{E_{\mu'} \Pi_\mrm{code} \}_{{\mu'} \in \Lambda'} .
\end{equation}
We write $\mathcal{E}_\mu \sim \mathcal{E}'_{\mu'}$ for equivalent sets of correctable errors, and we write $[\mathcal{E}]$ to denote the equivalence class implied by $\sim$.
We say that a correctable error set $\mathcal{E}$ is \emph{maximal} if no further linearly independent errors may be appended to $\mathcal{E}$ while keeping $\mathcal{E}$ correctable.

We finish this section with two additional topics that will be relevant for the coming discussion: how qubit Pauli stabilizer codes generalize to qudits, and subsystem codes.

Pauli stabilizer codes generalize readily from strings of qubits to strings of $D$-dimensional qudits.
With respect to a computational basis $\{\ket{j} ~ | ~ 0 \leq j \leq D-1 \}$, the qudit generalizations of the Pauli $X$ and $Z$ operators are the shift and phase operators
\begin{equation}\label{eq_Pauliqudit}
    X_D \ket{j} = \ket{j + 1 \mod D} \qquad Z_D \ket{j} = \omega^j \ket{j},
\end{equation}
where $\omega = e^{2\pi i/D}$.
These operators satisfy $Z_D X_D = \omega X_D Z_D$ (cf. Eqs.~\eqref{eq:weyl}, \eqref{eq:lattice_weyl_finite}-\eqref{eq:lattice_weyl_infinite}).
Taking products of $X_D$ and $Z_D$ generates single qudit generalized Pauli operators, and strings of generalized Pauli operators (multiplied by powers of $\omega$) constitute the generalized $n$-qudit Pauli group.\footnote{A common choice is to let $D$ be a prime number, in which case the relationship between the number of encoded qudits and the number of independent stabilizer generators generalizes directly from the $D = 2$ case, and the qudit code can be put into a standard CSS-like form (see, e.g., \cite{Gheorghiu:2011awf}). This assumption will not be needed for our computations.}
(See, e.g., \cite{Rains:1997uh,Gottesman:1998se,Ashikhmin:2000ylq,eczoo_qudit_stabilizer,Wang:2020ife} for further details.)

Finally, a subsystem QECC is a stabilizer QECC in which the code subspace additionally has a bipartite tensor product structure:
\begin{equation}
    \Hil_\mrm{code} \cong \Hil_A \otimes \Hil_B
\end{equation}
In a subsystem QECC, only one factor of the code subspace (say, $\Hil_A$) is used to store logical information, while the other factor functions as a redundant commodity to be used for error correction.
In the QEC literature, the redundant factor $\Hil_B$ is usually called the ``gauge'' subsystem, although this should not be confused with our use of the word to describe gauge field degrees of freedom in the coming sections.
Since the $\Hil_B$ factor is ultimately irrelevant to the stored logical information, it is perfectly acceptable to finish an encoded computation with an unknown state in $\Hil_B$.
This makes possible various computational simplifications; for example, a given subsystem QECC (such as the 9-qubit Bacon-Shor code \cite{Bacon:2005ckt}) could have stabilizer generators with lower weights than those of the corresponding non-subsystem code (here, the 9-qubit Shor code \cite{Shor:1995hbe}).
Given the logical irrelevance of $\Hil_B$, the KL condition is correspondingly weakened for subsystem QECCs.
In this case, an error set $\mathcal{E} = \{E_\mu\}_{\mu \in \Lambda}$ is correctable if and only if
\begin{equation}\label{eq_KLsubsystem}
    \Pi_\mrm{code} E_\mu^\dagger E_\nu \Pi_\mrm{code} = \mathbbm{1}_A \otimes (\Oh_{\mu\nu})_B \Pi_\mrm{code}
\end{equation}
where $(\Oh_{\mu\nu})_B$ is allowed to be a nontrivial operator on $\Hil_B$.
(See, e.g., \cite{kribs2006, Poulin_2005} for further details.)

\subsection{Abelian Lattice Gauge Theories}
\label{sec:LGT}
Here we review the kinematics of lattice gauge theory \cite{Wilson:1974sk} in the operator formalism \cite{Kogut:1974ag,Kogut:1979wt,Araujo-Regado:2025ejs,Byrnes:2005qx,Zohar:2015hwa,Zohar:2021nyc} for a compact Abelian structure (gauge) group $G$.
Let us consider a spatial lattice without boundaries consisting of a set $\mathcal{V}$ of vertices and a set $\mathcal{L}$ of links, each of which is equipped with an orientation. Henceforth, we denote by $N_V$ the number of vertices in $\mathcal{V}$.

In lattice gauge theory, we place kinematical degrees of freedom corresponding to the gauge field on each link and to matter on each vertex.
The total kinematical Hilbert space is the tensor product of the gauge field and matter Hilbert spaces,
\begin{equation}\label{eq_Hkinbos}
    \Hil_{\rm kin} :=\Hil_{\rm gauge}\otimes\Hil_{\rm matter}\,.
\end{equation}
For the matter, we assume that we have a Hilbert space $\Hil_v$ at each vertex $v\in\mathcal{V}$ carrying \emph{some} (not necessarily irreducible) unitary representation $u_v$ of $G$.
The total matter Hilbert space is then the tensor product\footnote{
Note that for fermionic matter, the algebra on $\Hil_{\rm matter}$ is not simply a tensor product of the local algebras on $\Hil_v$, as reflected in the canonical anti-commutation relation. See App.~\ref{app_fermions} for review of this subtlety. There is no such subtlety for bosonic matter. 
} 
\begin{equation}
    \Hil_{\rm matter} :=\bigotimes_{v\in\mathcal{V}}\Hil_v\,.
\label{eq_Hbosonic_matter}
\end{equation}
The kinematical Hilbert space of the pure gauge sector is 
\begin{equation} \label{eq:Hkingauge}
        \mathcal{H}_{\text{gauge}} 
        := \bigotimes_{\ell \in \mathcal{L}} \mathcal{H}_\ell ,\quad \mathcal{H}_\ell =L^2(G) ,
\end{equation}
where $\mathcal{H}_\ell$ thus carries the regular representation of $G$ and is spanned by the \emph{group basis} $\ket{g}$ associated with group elements $g\in G$.
As we will return to shortly, the $\ket{g}$ basis diagonalizes the link variables that are the lattice counterparts of gauge field operators, more precisely, of Wilson lines. As such, this basis will often also be called the \emph{magnetic basis}, especially when used in the context of Wilson loops.

There is another useful basis, often called the \emph{electric basis}, which diagonalizes the lattice counterparts of electric field operators.
The electric basis is labeled by the Pontryagin dual $\hat{G}$ of $G$, which is the group of continuous homomorphism from $G$ to $\rm{U}(1)$ (i.e., the group of \emph{characters}).\footnote{See \cite[Apps.~C-E]{Carrozza:2024smc} and \cite[App.~J]{Araujo-Regado:2025ejs} for introductions to Pontryagin duality in the context of quantum error correction and lattice gauge theories.}
Particularly important examples for this article are
\begin{\eq}
 \widehat{\rm{U}(1)}=\mathbb{Z} ,\quad    \widehat{\mathbb{Z}_n}=\mathbb{Z}_n ,
\end{\eq}
where the right-hand side physically labels quanta of electric flux on the lattice.
The electric basis $\ket{\chi}$ ($\chi\in\hat{G}$) is related to the group basis $|g\rangle$ ($g\in G$) via the group Fourier transform\footnote{When $G$ is finite group, the symbol $\int_G \dee g $ should be understood as $\frac{1}{|G|} \sum_{g \in G}$. We retain the integral notation below for simplicity. When $G$ is continuous, we absorb the volume of the group inside the normalization of the Haar measure.
Also note that our convention is different from that of Ref.~\cite{Araujo-Regado:2025ejs}.
}
\begin{equation}\label{groupvselectricbases}
    \ket{g} = \sum_{\chi \in \hat{G}}\chi(g)\ket{\chi}, \qquad \ket{\chi} = \int_G \dee g ~ \bar{\chi}(g)\ket{g} ,
\end{equation}
where $\chi (g)\in\rm{U}(1)$ denotes the canonical pairing between a group element $g\in G$ and a dual character $\chi\in\hat G$. Since the two bases are related by a Fourier transform, the Plancherel theorem implies 
\begin{\eq}
\mathcal{H}_\ell = L^2 (G) \simeq L^2(\hat{G}) .
\end{\eq}

On each $\mathcal{H}_\ell$, the Wilson line operator $W^{\chi}$ along the link $\ell$, in the representation $\chi$ of $G$, acts as
\begin{equation}\label{eq_Pdualrep}
        W^{\chi} \ket{g} = \bar{\chi}(g)\ket{g} ,\quad 
        W^{\chi}\ket{\chi'} = \ket{\chi \chi'} ,
        \qquad g\in G ,\ \chi ,\chi'  \in \hat{G}.
\end{equation}
The standard link operator of the lattice gauge theory corresponds to the Wilson line operator in a faithful representation. Associated with $g\in G$, there is an operator $U^g$ on each $\mathcal{H}_\ell$ satisfying the Weyl-type relation\footnote{The representations $U$ of $G$ and $W$ of $\hat{G}$ are then dual to one another \cite[App.~E]{Carrozza:2024smc}.}
\begin{equation} \label{eq:weyl}
        W^{\chi}U^g = \bar{\chi}(g)U^g W^{\chi} 
        \qquad ( g\in G ,\ \chi  \in \hat{G}) .
\end{equation}
This is an extension of the canonical commutation relation from ordinary quantum mechanics, and $U^g$, physically, is the counterpart to the (exponentiated) electric field.\footnote{
Our convention differs from much of the literature on lattice gauge theory, in which $U$ often denotes a link variable. 
}
Indeed, the electric basis $\ket{\chi}$ is an eigenstate of $U^g$, and the group basis $\ket{g}$ transforms as expected under the action of $U^g$:
\begin{equation} \label{eq:Uoperator}
        U^g \ket{\chi} = \chi(g) \ket{\chi} ,\quad
         U^g \ket{g'} = \ket{gg'} .        
\end{equation}
The operators $W^\chi$ and $U^g$ can be expressed as
\begin{equation}
       W^\chi =\int \dee{g}~\bar{\chi}(g)\ket{g}\!\bra{g}, \quad
       U^g = \sum_{\chi\in\hat{G}} \chi (g) \ket{\chi}\bra{\chi},
\end{equation}
and generate the bounded operator algebra on $\mathcal{H}_\ell$:
\begin{equation}
        \mathcal{B}(L^2(G)) = \langle U^g, W^{\chi} ~ \big|  ~ g \in G, \chi \in \hat{G} \rangle .
\end{equation}
 Here, $\langle A,B\rangle$ denotes the unital $*$-algebra generated by operators of the form $A,B$. 

Let us now explicitly label the operators $W^\chi$ and $U^g$ with the link $\ell$ on which they act by writing $W_\ell^\chi$ and $U_\ell^g$, respectively.
The gauge-invariant (perspective-neutral) subspace $\mathcal{H_{\text{pn}}}$ of $\mathcal{H}_{\rm kin}$ consists of those states that satisfy the Gauss law\footnote{
In the context of gauge theory, $\Hil_\mrm{pn}$ is often called the {\it physical} Hilbert space and denoted by $\mathcal{H_{\text{phys}}}$. To avoid potential confusion with the physical layer in a quantum computer, we denote the Hilbert space in question by $\mathcal{H_{\text{pn}}}$, which comes from {\it perspective-neutral}. Indeed, in the context of quantum reference frames (QRFs), to be introduced shortly, the gauge-invariant Hilbert space assumes the role of a QRF-perspective-neutral description of the physics, linking all their internal perspectives  \cite{delaHamette:2021oex,Hoehn:2023ehz,Vanrietvelde:2018pgb,AliAhmad:2021adn}, whence the name.
}
\begin{equation} \label{eq:HpnSec2}
        \mathcal{H_{\text{pn}}} := \{ \ket{\Psi} \in \mathcal{H}_{\text{kin}}  \mid U_v^g \ket{\Psi} = \ket{\Psi}\  \forall\, v \in \mathcal{V}, g \in G  \} .
\end{equation}
Here, $U_v^{g}$ is the gauge transformation at the vertex $v$, defined as
\begin{equation}\label{generalGTs}
U_v^{g} 
:= \left( \bigotimes_{\ell \in \mathcal{L}_{\text{out}}(v)}U_{\ell}^{g}\right) \otimes 
\left( \bigotimes_{\ell \in \mathcal{L}_{\text{in}}(v)}U_{\ell}^{(g^{-1})}    \right)
\otimes u_v^g ,\qquad\qquad g\in G\,,
\end{equation}
where we distinguish the orientation of the links, and $\mathcal{L}_{\text{out}}$ (resp. $\mathcal{L}_{\text{in}}$) is  outgoing (resp.\ incoming) links of $v$.
The operator $u_v^g$ acts on the matter Hilbert space $\Hil_v$ and is the lattice counterpart to the (exponentiated) charge operator, for which we will give explicit expressions in the subsequent constructions.
The first two bracketed terms act on the gauge field and are given by the product of $U_\ell^g$'s on the links that connect with $v$.
Consequently, the Wilson line operator $W_\ell^\chi$ at a link $\ell$ that is oriented from the vertex $v_i$ to the vertex $v_f$ transforms as
\begin{equation}
        W_{\ell}^{\chi} \rightarrow U^{\bm{g}} 
W_{\ell}^{\chi}U^{\bm{g} \dagger} 
= \chi(g_{v_i}g_{v_f}^{-1})W_{\ell}^{\chi} ,
\label{eq:gauge-trans}
\end{equation}
where
\begin{equation}
\label{eq_gaugetrbos}
U^{\bm{g}} := \prod_{v \in \mathcal{V}} U_v^{g_v}         
= U_{\text{gauge}}^{\bm{g}} \otimes U_{\text{matter}}^{\bm{g}}. 
\end{equation}
We will use $\mathcal{G}$ to denote the group of independent gauge transformations on the lattice as the image of the gauge group $\bm{G}$ under the unitary representation $U$. $\mathcal{G}$ is generated by $U^{\bm{g}}$, but not all of the $U^{\bm{g}}$'s are necessarily independent, depending on the boundary conditions.
For instance, with periodic boundary conditions, which we will henceforth adopt, the global  transformation, which acts uniformly across the lattice, acts trivially on the $\Hil_{\rm gauge}$ factor in $\mathcal{H}_{\text{kin}}$. This can be seen from Eq.~\eqref{generalGTs} and the fact that each link connects two vertices at which such transformations act.
It therefore follows that $\bm{G} = G^{\times N_V - 1}$ for periodic boundary conditions when no matter is included. When matter is involved, we instead have $\bm{G} = G^{\times N_V}$.
In terms of the $U^{\bm{g}}$'s, the projector from $\mathcal{H}_{\rm kin}$ onto $\mathcal{H}_{\rm pn}$ is
\begin{equation}\label{Pipn}
        \Pi_{\text{pn}} := \int_{\bm{G}} \dee \bm{g} ~ U^{\bm{g}} ,
\end{equation}
which we will make frequent use of. Besides using it to build gauge-invariant states, we shall also invoke it to build gauge-invariant observables.

Equivalently, one may also fix a gauge by appropriately choosing the gauge transformation in \eqref{eq:gauge-trans}.
Since a gauge transformation at a vertex acts only on the links adjacent to that vertex, gauge-invariant observables remain unchanged if, for a given vertex $v$, we select one adjacent link $\ell$ and project the corresponding Hilbert space $\Hil_\ell$ onto a fixed group basis state $\ket{g_{\ell |v\in \partial \ell}}$, associated with some element $g \in G$ (often taken to be the identity). This will fix the gauge completely at $v$ because the link state has as many parameters as there are independent gauge transformations at $v$. The question is whether or not we can remove all redundancy across the lattice this way. We will return to this in greater detail when discussing quantum reference frames in section~\ref{ssec_QRF}.

\subsection{Quantum reference frames on the lattice}\label{ssec_QRF}

In what follows, we will make heavy use of internal \emph{quantum reference frames} (QRFs), which constitute a universal toolset for dealing with symmetries in quantum systems. QRFs appear in various guises, adapted to different kinds of symmetries (e.g., physical vs.\ gauge), different in the way these symmetries are implemented \cite{Bartlett:2006tzx,Giacomini:2017zju,delaHamette:2021oex,Hoehn:2023ehz,Vanrietvelde:2018pgb,AliAhmad:2021adn,Araujo-Regado:2025ejs,Carrozza:2024smc,Carette:2023wpz,Castro-Ruiz:2025yvi}. The basic idea of a QRF is that it is some internal subsystem of the composite quantum system of interest which is subjected to a symmetry principle. The purpose in life of the QRF is to transform nontrivially under the given symmetry and to parametrize its orbits with its orientations. It may then be used as an internal vantage point from which to describe the remaining degrees of freedom in an invariant manner.

In the present article, we are interested in gauge theories, so we will invoke the formalism of gauge QRFs, which is also known as the \emph{perspective-neutral} formulation \cite{delaHamette:2021oex,Hoehn:2023ehz,Vanrietvelde:2018pgb,AliAhmad:2021adn,Araujo-Regado:2025ejs,Carrozza:2024smc}. In this approach, a QRF $R$ for the gauge group $\bm{G}$ corresponds to a choice of split between redundant ($R$ itself) and non-redundant degrees of freedom, which we collectively call the system $S$. While not strictly necessary, one typically assumes that $R$ and $S$ are jointly described on a kinematical Hilbert space that is a tensor product between Hilbert spaces individually describing $S$ and $R$ and on which $\bm{G}$ acts via a tensor product representation:
    \begin{equation}\label{eq_HkinRSQRF}
        \mathcal{H}_{\text{kin}} = \mathcal{H}_R \otimes \mathcal{H}_S, \qquad U^{\bm{g}} = U_R^{\bm{g}} \otimes U_S^{\bm{g}}\,,\qquad\bm{g}\in\bm{G}\,.
    \end{equation}

We shall henceforth assume that its configurations transform regularly under $\bm{G}$, so that they can be labeled by group elements and used to parametrize $\bm{G}$-orbits. This is analogous to how tetrads take values in the Lorentz group, or how the position of a particle takes value in the translation group.
The configurations $\ket{\bm{g}}_R$, $\bm{g}\in \bm{G}$, of the QRF are henceforth called its \emph{orientations}, and we wish to construct them such that they furnish a generalized (possibly distributional) coherent state system and basis for $\Hil_R$:
\begin{equation}\label{eq_completeQRF}
U_R^{\bm{g}}\ket{\bm{g'}}_R = \ket{\bm{gg'}}_R\,,\qquad\qquad \int_{\bm{G}}{\dee}\bm{g}\,\ket{\bm{g}}\!\bra{\bm{g}}_R=\mathbbm{1}_R\,.
\end{equation}
A frame with these properties is also called \emph{complete} \cite{delaHamette:2021oex,Araujo-Regado:2025ejs}.
When the orientations are perfectly distinguishable, in which case
\begin{equation}\label{eq_ideal}
    \braket{\bm{g}}{\bm{g}'}_R=\delta(\bm{g},\bm{g}')\,,
\end{equation}
we call the QRF \emph{ideal}. 
Here $\delta(\bm{g},\bm{g}')$ denotes the delta function on $\bm{G}$ normalized with respect to the Haar measure ${\rm d}\bm{g}$ (hence a Kronecker delta when $G$ is finite). In what follows, we will only encounter ideal QRFs in lattice gauge theories.

In the perspective-neutral formulation, the aim is to develop a gauge-invariant description of both states and observables. Hence, here as well one invokes the projector $\Pi_{\rm pn}$ given in Eq.~\eqref{Pipn} to proceed to the gauge-invariant (physical) Hilbert space denoted $\Hil_{\rm pn}$. Invoking the QRF $R$, we may now devise a description of $\Hil_{\rm pn}$ and the observable algebra on it that corresponds to a gauge-invariant \emph{relational} description of $S$.

First, we note that the QRF orientation states are useful for gauge fixing. Indeed, we can simply condition on the orientation of $R$ being, say, $\bm{g}$. This is implemented by what is sometimes called a Page-Wootters reduction
\begin{equation}\label{eq_PW}
        \mathcal{R}_R^{\bm{g}} = \left(  \bra{\bm{g}}_R \otimes \mathbbm{1}_S  \right) \Pi_{\text{pn}}: \mathcal{H}_{\text{pn}} \rightarrow \mathcal{H}_{|R} \simeq \mathcal{H}_S\,,
    \end{equation}
where the latter isomorphism holds only for ideal QRFs. This map is a unitary gauge fixing \cite{delaHamette:2021oex}. We can use this map and its adjoint to encode any system observable $f_S\in\mathcal{A}_S=\mathcal{B}(\Hil_S)$ relationally in the gauge-invariant algebra $\mathcal{A}_{\rm pn}=\mathcal{B}(\Hil_{\rm pn})$. More precisely, 
\begin{equation}\label{relobs}
        O_{f_S | R}^{\bm{g}} \coloneqq \mathcal{R}_R^{\bm{g} \dagger} \,f_S \,\mathcal{R}_R^{\bm{g}}=\Pi_{\rm pn}\left(\ket{\bm{g}}\!\bra{\bm{g}}_R\otimes f_S\right)\Pi_{\rm pn}
    \end{equation}
constitutes the relational observable in $\mathcal{A}_{\rm pn}$ measuring $f_S$ conditional on $R$ being in orientation $\bm{g}\in \bm{G}$. When $R$ is complete, every element in $\mathcal{A}_{\rm pn}$ can be written as such a relational observable \cite{delaHamette:2021oex}; indeed, $\mathcal{R}^{\bm{g}}_R$ is unitary on $\Hil_\mrm{pn}$.

In general, for a given composite quantum system, there will be many possible choices of splits between a QRF $R$ and its complement $S$, and one can transform between these choices, yielding QRF transformations between the different relational descriptions \cite{delaHamette:2021oex,Hoehn:2023ehz,Vanrietvelde:2018pgb,AliAhmad:2021adn,Araujo-Regado:2025ejs}.

Let us now turn to the question of how QRFs can be realized in lattice gauge theory. Here, we will follow the exposition in \cite{Araujo-Regado:2025ejs}, where this has been spelled out in depth for general lattice gauge theories (including non-Abelian ones) and to which we refer the reader for further detail. In what follows, we will only use the gauge field to build a QRF for $\bm{G}$, though in principle one could also use appropriately transforming matter degrees of freedom. In line with section~\ref{sec:LGT}, we further restrict to a closed lattice and so must   take into account that global gauge transformations act trivially on the gauge field. This means that a gauge-field-built QRF can only deparametrize $G^{\times N_V-1}$, i.e.\ gauge transformations at each vertex, except one (which suffices when no matter is present). Since $\bm{G}$ acts independently at different vertices, this means that we need to build a lattice field of QRFs $R_v$, one for each  vertex to deparametrize the structure group $G$ there, except at one vertex that henceforth we denote by $v_0$. 

A simple way to build such a frame is to choose some spanning tree rooted at $v_0$.\footnote{
A tree is a connected subgraph without loops, and a spanning tree is one that includes all vertices of the lattice. In lattice gauge theory, this is often referred to as a maximal tree and corresponds to a choice of links on which gauge fixing is imposed.    
} We call the tree $R$ in what follows; it is clear that in general there will be many such choices of tree.
We then obtain a split as in Eq.~\eqref{eq_HkinRSQRF}, where $R$ comprises the subset of links in the tree and $S$ comprises the links in the complement of the tree, as well as any matter. Hence,
\begin{equation} 
\Hil_R=\bigotimes_{\ell\in R} \Hil_\ell\,,\qquad\qquad\Hil_S=\bigotimes_{\ell\in S}\Hil_\ell\otimes\Hil_{\rm matter}\,.
\end{equation}

Our next task is to build suitable orientation states for the tree at each $v\neq v_0$, which must be labeled by group variables. To achieve this, it is convenient to refactorize the QRF Hilbert space $\Hil _R$ as a product over all the Wilson lines in the tree rooted at $v_0$ \cite{Araujo-Regado:2025ejs}.\footnote{Below we take the opposite convention to \cite{Araujo-Regado:2025ejs}, where the Wilson lines in the tree are oriented away from $v_0$, in contrast to here.} That is, we perform a nonlocal unitary change of basis from edge-wise to cumulative node-wise variables in terms of Wilson lines
$\Hil_{R}\simeq\bigotimes_{v\in\mathcal{V}, v\neq v_0}\Hil_{R_v}$, where $\Hil_{R_v}\simeq L^2(G)$ is the Hilbert space associated with the Wilson line \begin{equation}\label{eq_RWilson}
 W_{R_v}^\chi\coloneqq W^\chi_{\gamma_R[v,v_0]}= \bigotimes_{\ell \in \gamma_R[v, v_0]} W_{\ell}^{\chi^{\sigma(\ell, v)}}\,,
 \end{equation}
where $\gamma_R[v,v_0]$ denotes the unique path in the spanning tree from $v$ to the root $v_0$. Here, $\sigma(\ell, v)$ encodes the relative orientation of the path $\gamma_R$ and link $\ell$, so $\sigma(l,v)=1$ if they are equally oriented and $\sigma(l,v)=-1$ otherwise.
 Denoting the remaining vertices by $v_1,\ldots,v_{{N_V}-1}$,
this unitary basis change takes the form
    \begin{equation}
      \ket{\bm{g}}_R :=  \ket{ g_{v_1} , \cdots , g_{v_{{N_V}-1}} }_R \coloneqq \bigotimes_{\ell \in R} \ket{g_{\ell}}, \qquad 
        g_v = \prod_{\ell \in \gamma_R[v, v_0]}g_{\ell}^{\sigma(\ell, v)}\,.
    \end{equation}
Clearly, we have that Eq.~\eqref{eq_ideal} is realized, so the QRF associated with the spanning tree is ideal.

Let us now determine the exponentiated electric fields associated with the Wilson lines. For $v$ either the final vertex in a branch of the tree or a nearest neighbor vertex of $v_0$, we have (suppressing identities on other links for notational simplicity)
 \begin{equation} \label{eq:URv1}
 U_{R_v}^g = U_{\ell_v}^g\,,
 \end{equation}
 with $\ell_v$ the tree link in $\gamma_R[v,v_0]$ incident on $v$, respectively, while for all other vertices we have
 \begin{equation} \label{eq:URv2}
 U_{R_v}^g=\left( \bigotimes_{\ell \in \mathcal{L}_{\text{out}}(v)\cap\gamma_R[v,v_0]}U_{\ell}^{g}\right) \otimes \left( \bigotimes_{\ell \in \mathcal{L}_{\text{in}}(v)\cap\gamma_R[v,v_0]}U_{\ell}^{g^{-1}}    \right)\,.
 \end{equation}
It may be checked that 
 \begin{equation}
 \label{QRFcommutator}
 [U_{R_v}^g,W_{R_{v'}}^\chi] = \delta_{v,v'}(\chi(g)-1)\,W_{R_v}^\chi U_{R_v}^g\,,
 \end{equation}
 so they generate a ``canonical'' algebra, and  generators associated with different Wilson lines in the tree commute. Accordingly, we obtain $\mathcal{B}(\Hil_{R_v})=\langle U^g_{R_v},W_{R_v}^\chi\,|\,g\in G, \chi\in\hat G\rangle$ \cite{Araujo-Regado:2025ejs}, and so
 \begin{equation}
      \mathcal{A}_R=\mathcal{B}(\Hil_R)=\big\langle  U_{R_v}^{g}, W_{R_v}^\chi\,\big|\,g\in G, \chi\in\hat{G}, v\in\mathcal{V}\setminus \{v_0\}\big\rangle\,.
 \end{equation}

This observation also means that gauge transformations 
$U_R^{\bm{g}}=\prod_{v\neq v_0}U^g_v$ with trivial action at $v_0$ act \emph{covariantly} on $R$:
    \begin{equation}\label{eq_QRFgaugetransf}
       U_R^{\bm{g}} \ket{ \bm{g'}}_R 
       = \prod_{v\neq v_0} U_{R_v}^{g_v}\ket{\bm{g}'}_R=\ket{ g_{v_1} g'_{v_1}, \cdots , g_{v_{N_V-1}}g'_{v_{N_V-1}} }_R
       := \ket{\bm{gg'} }_R ,
    \end{equation}
which is node-wise and thus realizes the final crucial property of QRF orientation states in Eq.~\eqref{eq_completeQRF}. Gauge transformations for the gauge field at $v_0$, on the other hand, are not independent of those at the remaining vertices. We can thus express them in terms of those in Eq.~\eqref{eq_QRFgaugetransf}.

We therefore have all the properties we need for a QRF on the lattice. Specifically, we can use the tree orientation states $\ket{\bm{g}}_R$ both to gauge fix and to build gauge-invariant relational observables via Eq.~\eqref{relobs}. The resulting Wilson line dressing formulae have been investigated in detail for pure gauge theories in \cite{Araujo-Regado:2025ejs}. There, it was also explained why it does not matter for constructing gauge-invariant observables that $R$ cannot deparametrize at $v_0$. This is because in closed lattices one always needs ``even-sided'' dressings, i.e.,\ dressings that extend from $v_0$ in an even number of directions, thereby canceling any transformations at $v_0$. Below, we will also discuss relational observables including bosonic and fermionic matter degrees of freedom.

\subsection{The QECC/QRF correspondence}\label{ssec_QECC/QRF}

One of the linchpins of this article is the natural dictionary that exists between QECCs and QRF setups for gauge systems, which we will use to construct codes from compact Abelian LGTs.
We established this dictionary in its general form in \cite{Carrozza:2024smc} and worked out its detailed instantiation for Pauli stabilizer QECCs as a concrete example.
Let us now briefly review the entries of the dictionary that are most relevant for the present work.
We avoid going into too much detail here and refer instead to \cite{Carrozza:2024smc}, in part because that reference contains a general overview section and because some of the results that we quote from it were proven there specifically for Pauli stabilizer codes; we will generalize these results to LGT settings in the following sections. 

The core pillar of the general QECC/QRF dictionary is to identify the code subspace \eqref{eq:sec2Hcode} of a QECC, $\Hil_\mrm{code}$, with the perspective-neutral Hilbert space \eqref{eq:HpnSec2} within a gauge QRF setup, $\Hil_\mrm{pn}$.
The kinematical spaces of both constructs coincide, hence our earlier naming conventions in Sec.~\ref{ssec:stabilizer_codes}.
Already from this identification and before specializing to Pauli stabilizer codes, we obtain structural insights into the deep connection between QECCs and QRFs.
In particular, the Page-Wootters reduction \eqref{eq_PW}, $\mathcal{R}_R^{\bm{g}}$, and its adjoint, $(\mathcal{R}_R^{\bm{g}})^\dagger$, implement covariant decoding and encoding maps, respectively, in the QECC.
Furthermore, the relational observables in Eq.~\eqref{relobs}, which preserve $\Hil_\mrm{pn}$, are precisely the encoded logical operations in the QECC. This has to do with the fact that gauge QRFs define a split between redundant and non-redundant physical data, and thus constitute the gauge theory counterpart to encodings and decoding. While there had been folklore before about a structural link between QECCs and gauge systems, gauge QRFs provided the missing ingredient to make this correspondence precise.

If the identification $\Hil_\mrm{pn} \equiv \Hil_\mrm{code}$ is the structural core of the dictionary, then its operational core is that distinct choices of QRFs are in one-to-one correspondence with distinct equivalence classes of maximal correctable error sets.
More precisely, given a stabilizer QECC and a maximal correctable error set $\mathcal{E}$, it follows that the equivalence class $[\mathcal{E}]$ defines a unique split of $\Hil_\mrm{kin}$ into a frame space, $\Hil_R$, and a system space, $\Hil_S$, such that (the code restriction of) errors drawn from members of $[\mathcal{E}]$ and the code's stabilizers act trivially on $\Hil_S$; that is, $E = E_R \otimes I_S$ and $U = U_R \otimes I_S$ for any $E \in \mathcal{E}$, $U \in \mathcal{G}$.
Conversely, given an ideal QRF $R$ such that the code's stabilizers act trivially on its complementary system $S$, there is a unique equivalence class of maximal correctable error sets such that the action of errors on $S$ is also trivial.
This correspondence was made precise for Pauli stabilizer codes in \cite[Thm.~4.13]{Carrozza:2024smc}, and we generalize it to LGTs throughout Sec.~\ref{sec:general}.

The pairing of a QRF with a maximal correctable error set $\mathcal{E}$ further defines a frame-dependent tensor product factorization\footnote{In Ref.~\cite{Carrozza:2024smc}, what we here call $\Hil_{\tilde{R}}$ is called ``$\Hil_\mrm{gauge}$'', but we avoid this designation for obvious clarity reasons. Later, $\Hil_{\tilde{R}}$ will be isomorphic to the link degrees of freedom, $\Hil_R$, on the spanning tree $R$.}
\begin{equation}
    \Hil_\mrm{kin} \simeq \Hil_\mrm{pn} \otimes_R \Hil_{\tilde{R}}
\end{equation}
with properties that are adapted to $\mathcal{E}$.
First, $\Hil_{\tilde{R}}$ comes with a basis that is labeled by the characters $\bm{\chi}$ of the gauge group's Pontryagin dual, $\hat{\bm{\bm{G}}}$.
Momentarily specializing to Pauli stabilizer codes, if $\mathcal{E}$ consists of Pauli errors, then its elements can also be consistently labeled by $\bm{\chi}$, such that \cite[Thm.~4.17]{Carrozza:2024smc}
\begin{equation}
    E_\mrm{\bm{\chi}} (\ket{\psi}_\mrm{pn} \otimes_R \ket{\bm{1}}_{R}) ~ \propto ~ \ket{\psi}_\mrm{pn} \otimes_R \ket{\bm{\chi}}_{\tilde{R}} \ ,
\end{equation}
where $\bm{1}$ denotes the trivial character (for each tree Wilson line), and the $\ket{\bm{\chi}}_{\tilde R}$ basis is related to the group basis $\ket{\bm{g}}_{\tilde R}$ by a group Fourier transform (see Eq.~\eqref{groupvselectricbases}).
In other words, the unitary Pauli error $E_{\bm \chi}$ maps $\Hil_\mrm{code}$ into the error sector $\Hil_{\bm \chi} = \Hil_\mrm{pn} \otimes_R \ket{\bm{\chi}}_{\tilde{R}}$.
In \cite{Carrozza:2024smc}, we called such Pauli errors ``electric errors'', which could be thought of mapping $\Hil_\mrm{pn}$---the zero-charge sector---to the nontrivial charge sector $\Hil_{\bm{\chi}}$, with $\bm{\chi}$ playing the role of the error's syndrome.
Maximal sets of correctable errors then naturally consist of one error per charge label $\bm{\chi}$.
As we will shortly see, this nomenclature remains appropriate for the LGT codes that we will construct, and analogous intuition will apply.
In these settings, $\bm{\chi}$ will label configurations of charge defects in Gauss law codes and configurations of unfrozen charge in vacuum codes.

For completeness, we note that \cite[Sec.~5]{Carrozza:2024smc} revealed a novel \emph{error duality}. Electric charge excitation errors are dual (in the Pontryagin sense of the Weyl-type relation \eqref{eq:weyl} and \cite[App.~E]{Carrozza:2024smc}) to  errors that correspond to gauge-fixings of the frame orientation, i.e.\ to conditioning on $R$'s orientation via projectors such as $\ket{\bm{g}}\!\bra{\bm{g}}_R$. While these gauge-fixing errors appear projective at first sight, they can be implemented unitarily and can be interpreted as ``magnetic charge'' excitations.
It is perhaps unsurprising that gauge-fixing errors would be correctable; after all, fixing a gauge breaks gauge invariance, but no physical information is lost in the process.
Although such errors would be somewhat natural to consider in an LGT context, for a lack of space, we will not go into further detail here nor in the subsequent sections.
However, we note that it would be straightforward to construct dual magnetic error sets for the LGT codes that we will develop.

\section{Quantum error correcting codes from general Abelian lattice gauge theories}
\label{sec:general}
Let us now explain two ways in which Abelian lattice gauge theories give rise to stabilizer quantum error correcting codes. In the first, we directly identify the kinematical and gauge-invariant Hilbert spaces of the theory with the physical and code spaces of the corresponding quantum code, respectively. This also means that we identify the gauge group with the code's stabilizer group. We shall refer to such codes as \emph{Gauss law codes}, as the Gauss law will constitute the syndrome for the errors of the code. This will extend previous work on aligning gauge-invariant Hilbert spaces with code spaces \cite{Carrozza:2024smc,Stryker:2018efp,Rajput:2021trn,Spagnoli:2024mib,Spagnoli:2026qni,Yao:2025cxs,errorzooGauss}. In particular, errors in such codes map out of the gauge-invariant Hilbert space into other subspaces of the kinematical Hilbert space and are thus not gauge-invariant. Such codes encompass Abelian theories with or without matter and may be useful for quantum simulations in which lattice gauge theories are simulated in quantum systems that do not feature a gauge symmetry \cite{Rajput:2021trn,Spagnoli:2024mib,Spagnoli:2026qni,Yao:2025cxs,Carrozza:2024smc}.

By contrast, the second way to extract quantum error correcting codes from lattice gauge theories is entirely gauge-invariant, including its errors. 
In the formulation that we will develop, it will involve the presence of matter in the theory and identifies the matter vacuum subspace of the gauge-invariant Hilbert space as the code space, and the entire gauge-invariant Hilbert space as the physical space, and matter excitations as correctable errors. 
As such, we shall refer to these codes as \emph{vacuum codes}. Their stabilizer groups are gauge-invariant and constituted by the action of the gauge transformations restricted to the individual matter species only. In this case, the matter charges of each particle type, or equivalently its occupation numbers, comprise the error syndromes. 

While these two types of codes may at first appear as very different codes, we will show that, in the case that $G$ is finite, vacuum codes and pure gauge Gauss law codes are, in fact, unitarily equivalent. 
We further demonstrate that this unitary equivalence carries over to the case that $G$ has continuous components, provided one suitably coarse-grains the vacuum code. 
Accordingly, Gauss law and (possibly coarse-grained) vacuum codes may be viewed as different realizations of the same abstract code with vacuum codes offering a gauge-invariant version of Gauss law codes. 
In section~\ref{ssec_hybrid}, we also briefly consider hybrids of vacuum and Gauss law codes.

In all of the below, we impose periodic boundary conditions for specificity, although a similar discussion follows for open boundary conditions. 

\subsection{Gauss law codes}

The defining feature of a Gauss law code is that we identify $\mathcal{G}$, the group of gauge transformations, with the stabilizer group of the code and the gauge-invariant (perspective-neutral)
space with the code subspace:
\begin{equation}\label{eq_codepn}
    \mathcal{H}_{\text{code}} = \mathcal{H}_{\text{pn}} = \{ \ket{\Psi} \in \mathcal{H}_{\text{kin}} \thinspace \mid \thinspace U_v^g \ket{\Psi} = \ket{\Psi} \thinspace \forall v \in \mathcal{V}, g \in G     \}\,.
\end{equation}
From this perspective, violations of Gauss's law play the role of detectable error syndromes: operators that fail to commute with some $U_v^g$ create  ``charges'' at vertices and are therefore outside the stabilizer-preserving algebra. Logical operators, on the other hand, are those operators that commute with all $U_v^g$
but act nontrivially within $\mathcal{H}_{\text{code}}$. In gauge theory language, these are the gauge-invariant, typically nonlocal operators. 
Both pure lattice gauge theories and those coupled to matter can be written as Gauss law codes. For simplicity, we begin with the pure gauge case, before subsequently adding bosonic and fermionic matter.

\subsubsection{Gauss law codes from the pure gauge sector} \label{sssec_pureGauss}
In the pure gauge case, the kinematical space, interpreted as the physical space of the code, is 
\begin{equation}\label{eq_Hkinpure}
    \mathcal{H}_{\text{kin}} = \bigotimes_{\ell\in\mathcal{L}} \mathcal{H}_{\ell}, \qquad \mathcal{H}_{\ell} = L^2(G)\,,
\end{equation}
while the gauge transformations in the form
\begin{equation}
        \label{eq_Gausslawpure}U^{\bm{g}} \coloneqq \prod_{v \in \mathcal{V}} U_v^{g_v}, \qquad U_v^{g_v} = \left( \bigotimes_{\ell \in \mathcal{L}_{\text{out}}(v)}U_{\ell}^{g_v}\right) \otimes \left( \bigotimes_{\ell \in \mathcal{L}_{\text{in}}(v)}U_{\ell}^{g_v^{-1}}    \right)
    \end{equation}
comprise the stabilizers of the pure-gauge Gauss law code.

The main task in this subsection is to identify the logical operators and to characterize maximal correctable error sets. Invoking the QECC/QRF correspondence \cite{Carrozza:2024smc}, logical operators can be written as relational observables relative to a QRF and correctable errors sets too are closely related to choices of QRFs. Employing the ideal lattice QRFs introduced via spanning trees in section~\ref{ssec_QRF}, let us now elucidate how this works in the context of lattice gauge theories.

The kinematical Hilbert space and gauge transformations then split according to
\begin{equation}
    \mathcal{H}_{\text{kin}} = \mathcal{H}_R \otimes \mathcal{H}_S, \qquad U^{\bm{g}} = U_R^{\bm{g}} \otimes U_S^{\bm{g}}\,,
\end{equation}
where $R$ and $S$ comprise the subsets of links in the tree and the tree complement, respectively, so $\Hil_R=\bigotimes_{\ell\in R} \Hil_\ell$ and $\Hil_S=\bigotimes_{\ell\in S}\Hil_\ell$. 
As explained in section~\ref{ssec_QRF}, we can equivalently refactorize the QRF Hilbert space as a product over all the Wilson lines in the tree rooted at $v_0$, i.e.\ $\Hil_{R}\simeq\bigotimes_{v\in\mathcal{V}, v\neq v_0}\Hil_{R_v}$. The kinematical algebra then factorizes as $\mathcal{A}_{\rm kin}=\mathcal{A}_R\otimes\mathcal{A}_S$, where $\mathcal{A}=\mathcal{B}(\Hil)$ and 
 \begin{equation}
 \label{eq_kinRS}
 \mathcal{A}_R=\big\langle  U_{R_v}^{g}, W_{R_v}^\chi\,\big|\,g\in G, \chi\in\hat{G}, v\in\mathcal{V}\setminus \{v_0\}\big\rangle\,,\qquad\mathcal{A}_S=\big\langle U_\ell^g,W_\ell^\chi\,\big|\, g\in G, \chi\in\hat G, \ell\in S\big\rangle\,.
 \end{equation}
 
\subsubsubsection{Dual code word bases and relational observables as logical operators}\label{sssec_codespacepure}
Given this structure, let us briefly explain how to parametrize the gauge-invariant (perspective-neutral) Hilbert space in terms of the two dual bases. As this will be our code space, they will provide the two dual bases for code words.

To this end, we exploit again the QRF associated with the spanning tree $R$ and invoke its Page-Wootters reduction map {$\mathcal{R}^{\bm{e}}_R$}, whose inverse defines an encoding map \cite{Carrozza:2024smc} (cf.~section~\ref{ssec_QECC/QRF}): $\mathcal{R}^{\bm{e}\dag}_R:\mathcal{H}_S\to\mathcal{H}_{\rm pn}$. 
This permits us to embed the ``magnetic'' (or group) and ``electric'' (or dual group) bases of the system $S$ into the perspective-neutral space:
 \begin{eqnarray}
            \ket{\bm{g}_S}_{\text{pn}} &\coloneqq& \mathcal{R}_R^{\bm{e} \dagger} \ket{\bm{g}_S}_S = \int \dee^{N_V-1} \bm{g'} \ket{\bm{g'}}_R \otimes U_S^{\bm{g'}} \ket{\bm{g}_S}_S\label{magbas}\\ \ket{\bm{\chi}_S}_{\rm pn}&\coloneqq&\mathcal{R}_R^{\bm{e} \dagger} \ket{\bm{\chi}_S}_S=\int \left( \prod_{\ell \in S} \dee g_{\ell} \bar{\chi}_{\ell} (g_{\ell}) \right)\ket{\bm{g}_S}_{\text{pn}}\,,\label{elecbas}
            \end{eqnarray}
where $\bm{g}_S=(g_\ell)_{\ell\in S}$ and $\bm{\chi}_S=(\chi_\ell)_{\ell\in S}$. 
This makes explicit that the tree degrees of freedom serve as a control register that coherently applies all gauge transformations to the subsystem $S$, thereby producing a gauge-invariant encoding of the logical data $\ket{\bm{g}_S}_S$ (and similarly $\ket{\bm{\chi}_S}_S$). 
Eqs.~\eqref{magbas} and~\eqref{elecbas} thus constitute Pontryagin dual magnetic and electric bases of code words, respectively. 

This terminology is justified as these bases diagonalize the magnetic and electric generators in the algebra $\mathcal{A}_{\rm pn}=\mathcal{B}(\Hil_{\rm pn})$ of gauge-invariant observables. 
To see this, we make use of the encoding of logical (i.e.\ $S$) operators as relational observables on $\Hil_{\rm pn}$ in Eq.~\eqref{relobs}. 
 It may be checked that for the elementary link operators on $\ell \in S$ this yields\footnote{See also the derivation of the QRF dressing formulae for lattice gauge theories in \cite[App.~C]{Araujo-Regado:2025ejs} (see especially \cite[Eq.~(C22)]{Araujo-Regado:2025ejs}) which encompass these results when projected with $\Pi_{\rm pn}$.} 
    \begin{equation}
        \label{relobspure}O_{U^h_{\ell}|R}^{\bm{g}} = U_{\ell}^h\Pi_{\text{pn}}, \qquad O_{W_{\ell}^{\chi}|R}^{\bm{g}} = \bar{\chi}_\ell(g_{v_i}g^{-1}_{v_f})\,H_{\ell}^{\chi}\Pi_{\text{pn}}\,,\qquad\qquad\ell\in S\,,
    \end{equation}
where
    \begin{equation}\label{eq_Hdef}
        H_{\ell}^{\chi} := 
        W_{\ell}^{\chi}W_{\gamma_R[v_0, v_i]}^{\chi}W_{\gamma_R[v_0, v_f]}^{\bar{\chi}} 
    \end{equation}
is the holonomy (Wilson loop) resulting from dressing $W_{\ell}^{\chi}$ with the unique Wilson lines in the tree $R$ connecting the initial and final vertices $v_i$, $v_f$ of $\ell$ with $v_0$ (see Fig.~\ref{fig:latticetree}).  

\begin{figure}[h]
\centering
\includegraphics[width=0.55\textwidth]{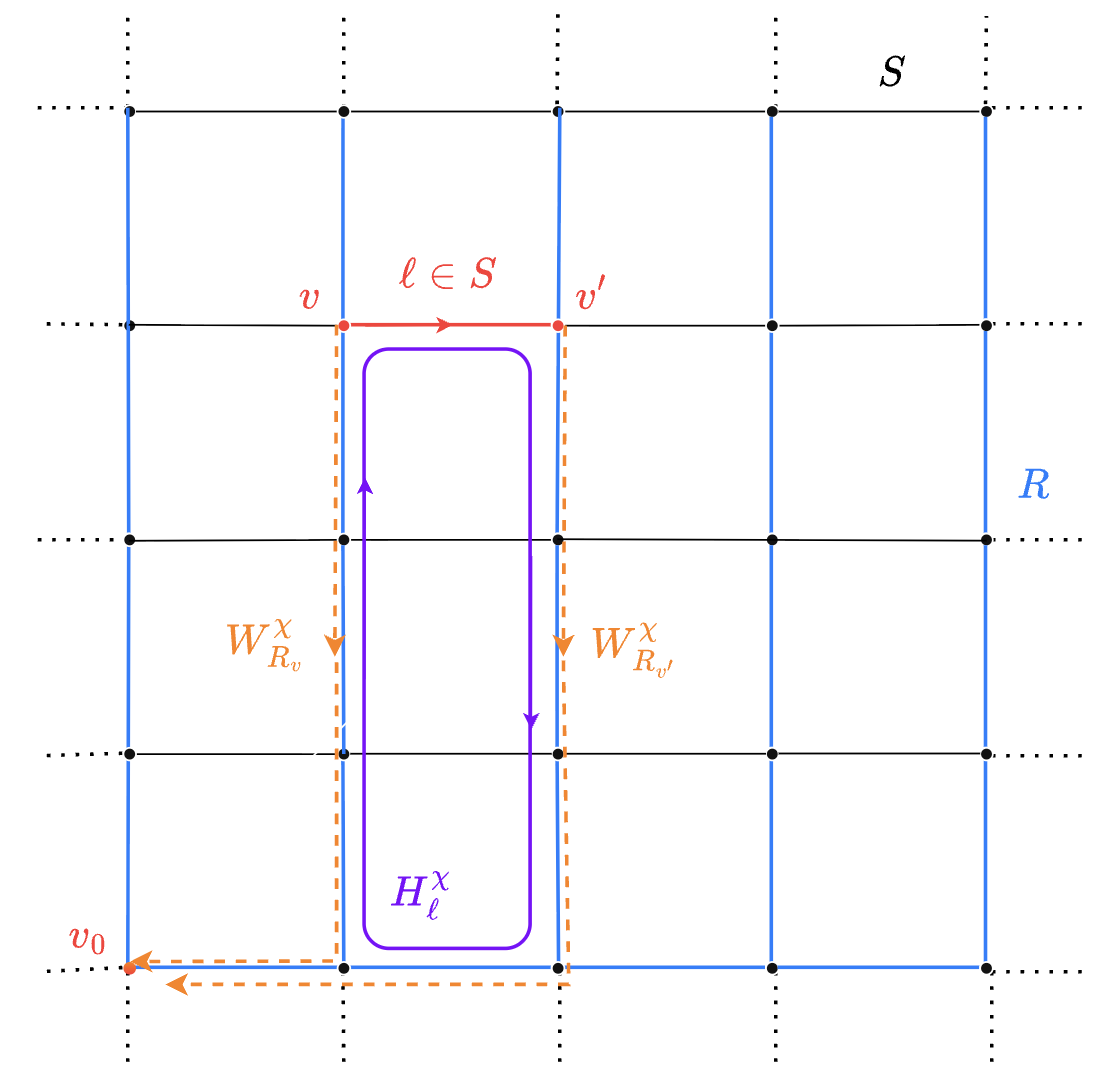}
\caption{Spanning tree $R$ (blue) for a $4 \times 4$ square lattice with periodic boundary conditions. The complement $S$ appears in black. The Wilson line along a system link $\ell \in S$ may be dressed with the tree Wilson lines $W_{R_{v}}^{\bar{\chi}}=(W_{R_v}^{\chi})^{\dagger}$, $W_{R_{v'}}^{\chi}$, connecting its endpoints with the root of the tree, $v_0$, on the bottom left. This results in the Wilson loop $H_{\ell}$. 
}
\label{fig:latticetree}
\end{figure}

Clearly, these relational observables are gauge-invariant. 
It is now straightforward to verify that the magnetic holonomy operators $H^\chi_\ell$ and the electric operators $U_\ell^h$ constituting the relational observables are diagonalized in the magnetic and electric basis of the code space, respectively,
\begin{equation}\label{eq_mageleceigenbasis}
    H_{\ell}^{\chi} \ket{\bm{g}_S}_{\text{pn}} = \bar{\chi}(g_{\ell})\ket{\bm{g}_S}_{\text{pn}}\,,\qquad\qquad U_{\ell}^h\ket{\bm{\chi}_S}_{\text{pn}} = \chi_{\ell}(h) \ket{\bm{\chi}_S}_{\text{pn}}\,,\qquad\qquad\ell\in S\,.
\end{equation}

Furthermore,
\begin{equation}
  H^\chi_\ell U^h_\ell=\bar\chi(h) U^h_\ell H^\chi_\ell\,,
\end{equation}
making the Weyl structure of the logical algebra generated by them manifest. 

This algebra is, in fact, the full algebra of bounded operators on the perspective-neutral space,
    \begin{equation}
       \label{Apnpure} \mathcal{A}_{\rm loops}\coloneqq\mathcal{A}_{\text{pn}}  = \big\langle U_{\ell}^g\Pi_{\text{pn}}, H_{\ell}^{\chi}\Pi_{\text{pn}} \,\big| \, g \in G, \chi \in \hat{G},\ell\in S  \big\rangle\,,
    \end{equation}
which for later purpose we shall also call the (Wilson) loop algebra. 
This follows from the fact that the kinematical observables $U_\ell^h,W^\chi_\ell$, $\ell\in S$, entering the relational observables in Eq.~\eqref{relobspure} generate the entire system algebra $\mathcal{A}_S=\mathcal{B}(\Hil_S)$ and that
 $R$ is a complete QRF for the gauge group $\bm{G} = G^{N_V-1}$ in the pure gauge case; indeed, relational observables associated with generators of the system algebra generate the full perspective-neutral algebra when the QRF is complete \cite{delaHamette:2021oex}. 
 
Thus, Eq.~\eqref{Apnpure} constitutes the logical operator algebra of the Gauss law code in the pure gauge case. 
Note that this algebra includes also other loops than the above $H_\ell^\chi$ with $\ell\in S$. 
Generally, there  exist loops entirely supported in $S$ or supported on more links in $S$ than just a single one. 
However, in the present Abelian case, these can be obtained from products of the generating loops and are thus not independent (see Fig.~\ref{fig:latticeloops}).
These dependence relations are instances of the Mandelstam constraints which encode relations among loops in a general lattice gauge theory \cite{Giles:1981ej,Gambini:1986ew,Loll:1991mh}.
 
\begin{figure}[h]
\centering
\includegraphics[width=0.6\textwidth]{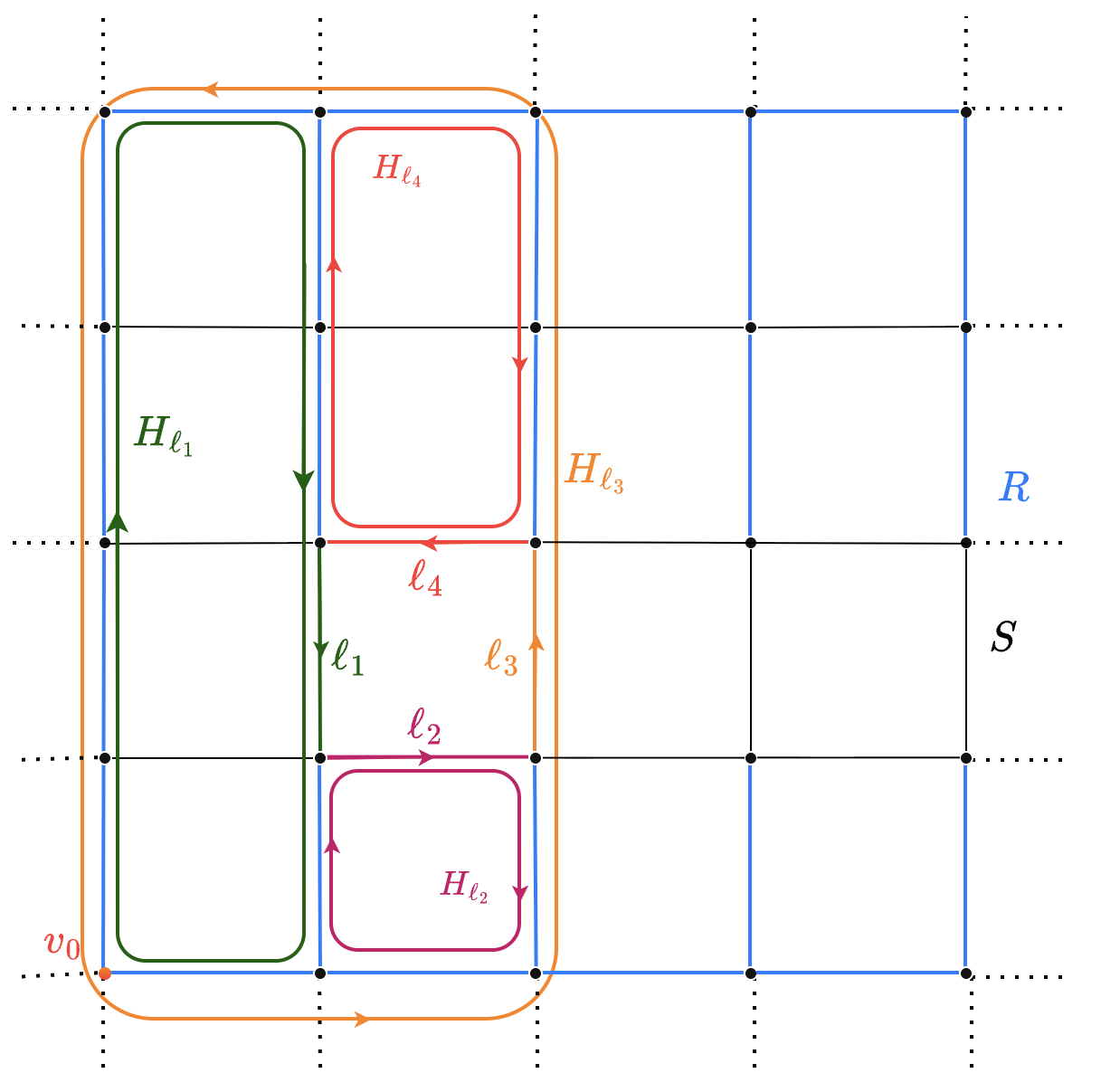}
\caption{System plaquette as a product of generating tree Wilson loops on a $5 \times 5$ square lattice with periodic boundary conditions. The spanning tree $R$ is highlighted in blue, and its complement $S$ appears in black. Multiplying the system holonomies $H_{\ell_4}H_{\ell_3}H_{\ell_2}H_{\ell_1}$, every part outside the plaquette $\{\ell_1, \ell_2, \ell_3, \ell_4 \}$ cancels, as they support either no Wilson lines or two Wilson lines with opposite orientations, yielding a Wilson loop supported only on $S$.}
\label{fig:latticeloops}
\end{figure}

The electric basis offers a convenient way to parametrize code words that will become relevant later when discussing correctable error sets. 
To this end, we introduce what we shall refer to as the \emph{Gauss law map},\footnote{Despite the notation, this is not a boundary map (as used in the context of surface codes in \cite{Carrozza:2024smc}), which would map from power sets of vertices to power sets of edges. Rather, mapping in the opposite direction, it can be regarded, in a sense, as an inverse of a boundary map.}
    \begin{equation}
        \label{boundarymap}    \partial: \hat{G}^{\times N_L} \rightarrow \hat{G}^{\times N_V - 1}, \qquad \partial \bm{\chi} = \left( \prod_{\ell \in \mathcal{L}_{\text{out}}(v)} \chi_{\ell} \prod_{\ell' \in \mathcal{L}_{\text{in}}(v)} \bar{\chi}_{\ell'} \right)_{v \in \mathcal{V} \setminus \{ v_0 \}},
        \end{equation}
which is nothing but the exponential of the discrete divergence of the electric flux on the lattice, as we shall illustrate explicitly in section 4. 
Thus, the Gauss law map maps the collection or electric link data (characters) in the lattice to the collection of electric data that appears in the Gauss law for each vertex (except the root $v_0$). 

Crucially, we can understand the Gauss law map as a discrete\footnote{For a compact Abelian $G$, its Pontryagin dual $\hat G$ is a discrete group. For example, $\widehat{\rm{U}(1)}=\mathbb{Z}$.} fiber bundle with $\partial$ its projection onto the base space $\hat{G}^{\times N_V-1}$. We refer to this as Gauss law bundle. Indeed, each electric link configuration $\bm{\chi}=(\chi_\ell)_{\ell\in\mathcal{L}}$ yields one vertex charge configuration $\partial\bm{\chi}$, but, conversely, for each vertex charge configuration the compatible electric link configurations can be parametrized by the discrete fiber $F:=\hat{G}^{\times(N_L-N_V+1)}$.

This can be identified with the electric configuration space of the system $S$, i.e.\ the complement of the spanning tree QRF $R$. To see this, consider any fixed vertex charge configuration $\partial\bm{\chi}$. 
Now choose the electric data for the links in $S$ arbitrarily and denote by $\chi^v_S$ the total contribution of $S$ at vertex $v\neq v_0$. 
As there are $N_L-N_V+1$ such links, the possible choices are parametrized by $\hat{G}^{\times(N_L-N_V+1)}$. 
Then the total vertex charge configuration $\partial\bm{\chi}$ can be achieved by setting the electric datum of the frame Wilson line hitting $v$ to $\chi_{R_v}=(\partial\chi)_v\bar\chi_S^v$.
Hence, for each electric system and vertex charge configuration, there is a \emph{unique} electric frame configuration compatible with it and so $F$ indeed constitutes a discrete fiber for the Gauss law map.

In particular, gauge-invariant (perspective-neutral) states are characterized by the vertex charge configuration $\partial\bm{\chi}=\bm{1}$ and the fiber $F=\hat{G}^{\times(N_L-N_v+1)}$ above it parametrizes all the possible electric code words compatible with it in terms of the electric system data, as indeed done in Eq.~\eqref{elecbas}. 
Below, we shall use sections of the Gauss law map to characterize maximal sets of correctable errors.

\subsubsubsection{Aligning the tensor product structure with the Gauss law map}\label{ssssec_GausTPS}
For our subsequent error analysis, it will turn out  convenient to first align the kinematical subsystem partition with the Gauss law map. Let us explain this here.

For each assignment of vertex charges $\partial\bm{\chi}  \in \hat{G}^{\times N_V - 1}$ via the Gauss law map, Eq.~\eqref{boundarymap}, consider the charged subspace\footnote{Recall that, for periodic boundary conditions, the vertex charge at the tree root $v_0$ is a product of the charges at all other vertices.}
\begin{equation}
    \mathcal{H}_{\partial\bm{\chi}} 
    = \big\{ \ket{\Psi} \in \mathcal{H}_{\text{kin}} \, \big|\, U^{g}_v \ket{\Psi} = (\partial{\bm{\chi}})_v(g)\ket{\Psi}, \forall \,v\in \mathcal{V}\setminus\{v_0\}, g \in G  \big\}\,.
\end{equation}
This is the simultaneous eigenspace of all Gauss law operators with eigenvalues prescribed by $\partial\bm{\chi}$. In other words, it is the sector carrying fixed vertex charges. The total kinematical space can be decomposed into these charge sectors:
\begin{equation}\label{eq_chargedecomp}
    \Hil_{\rm kin}=\bigoplus_{\partial\bm{\chi}\in\hat{G}^{N_V-1}}\Hil_{\partial\bm{\chi}}\,.
\end{equation}
The direct sum thus runs over the base space of the Gauss law bundle. Clearly, $\Hil_{\partial\bm\chi=\bm1}=\Hil_{\rm pn}$. For later purpose, we emphasize that the decomposition only runs over the gauge charges at vertices \emph{other} than the root $v_0$ (which is determined by those elsewhere).

To describe errors, it will be convenient to invoke the ``dressing frame fields'' associated with the QRF $R$ \cite[Sec.~4.5]{Carrozza:2024smc}, which here have a compelling lattice interpretation. A frame field with charge $\partial\bm{\chi}$ is 
an isometry
\begin{equation}
     R_{\partial\bm{\chi}}: \mathcal{H}_{\text{pn}} \rightarrow \mathcal{H}_{\partial\bm{\chi}
            }\,,
\end{equation}
explicitly given by (cf.~Eq.~\eqref{eq_RWilson})
 \begin{equation}
         \label{eq_framefield}   
         R_{\partial\bm{\chi}} = \prod_{v \in \mathcal{V} \setminus  \{ v_0 \} } W_{R_v}^{\partial\chi_v}\,,
        \end{equation}
and transforms gauge-covariantly: $R_{\partial\bm{\chi}}\to U_v^{g}R_{\partial\bm{\chi}} U_v^{g\dag}=\partial\chi_v(g)R_{\partial\bm{\chi}}$. Each $W_{R_v}^{\partial\chi_v}$ is the Wilson line along the unique tree path from $v$ to the root $v_0$ in representation $\partial\chi_v$. 
Acting with it shifts the Gauss law eigenvalue at $v$ by $\partial\chi_v$ 
and leaves all other vertex charges invariant because it commutes with the Gauss law elsewhere except at $v_0$ (see also the discussion around Eq.~\eqref{QRFcommutator}). 
The frame field is an isometry because each Wilson line is a unitary on $\Hil_{\rm kin}$.  
Since the tree connects every vertex uniquely to the root, the family $\{ R_{\partial\bm{\chi}}\}_{\partial\bm{\chi} \in \hat{G}^{N_V - 1}}$ encompasses all charge sectors of $\Hil_{\rm kin}$. 
In the language of \cite{Carrozza:2024smc}
this family forms a complete set of dressing frame fields: it generates every Gauss law charge sector from the trivial, i.e.\ perspective-neutral one.

For us, the key aspect of the frame fields is that they induce a new $R$-dependent  tensor product structure $\Hil_{\rm kin}=\Hil_{R}\otimes_{R}\Hil_{\rm pn}$ on the kinematical space which enjoys special properties that will simplify the subsequent analysis. Abstractly, it is defined via the bilinear map \cite[Sec.~4.5]{Carrozza:2024smc}
\begin{equation}\label{eq_newTPS}
    \otimes_R: \Hil_R\times\mathcal{H}_{\text{pn}} \rightarrow \mathcal{H}_{\text{kin}}\,,
\end{equation}
which in terms of the electric frame basis and arbitrary perspective-neutral states reads
\begin{equation}
\ket{\partial\bm{\chi}}_R \otimes_R\ket{\psi}_{\text{pn}}\coloneqq R_{\partial\bm{\chi}}\ket{\psi}_{\text{pn}}\,.
    \label{frameTPS}
\end{equation}
One special property of this new factorization is that gauge transformations act trivially on the second factor:
\begin{equation}\label{eq_trivialGauss}
    U^{\bm{g}} = U_R^{\bm{g}}\otimes_R\mathbbm{1}_{\text{pn}}  , \qquad \bm{g}\in G^{\times N_V-1}\,.
\end{equation}
Indeed, for all $v\in\mathcal{V}\setminus\{v_0\}$,
\begin{equation}\label{eq_elecbaserefac}
    U^{g}_v R_{\partial\bm{\chi}}\ket{\psi}_{\rm pn}=\partial\chi_v(g)R_{\partial\bm{\chi}}\ket{\psi}_{\rm pn}=\partial\chi_v(g)\ket{\partial\bm\chi}_R\otimes_R\ket{\psi}_{\rm pn}=U_{R_v}^g\ket{\partial\bm\chi}_R\otimes_R\ket{\psi}_{\rm pn}\,.
\end{equation}
Thus, one tensor factor carries the protected, gauge-invariant information, while the other keeps track of how the state transforms under gauge transformations. This tensor product structure is also naturally associated with the Gauss law bundle: the $R$-factor runs over the base space, while the $\rm{pn}$-factor runs over the fiber.

Let us clarify this further in terms of algebras. Similarly to Eq.~\eqref{eq_elecbaserefac}, for any tree Wilson line $W_{R_v}^\eta$, we have 
\begin{equation}\label{eq_treeWilsonlinenewfac}
    W_{R_v}^\eta R_{\partial\bm\chi}\ket{\psi}_{\rm pn} = W_{R_v}^\eta\ket{\partial\bm\chi}_R\otimes_R\ket{\psi}_{\rm pn}\,.
\end{equation}
Since $U^g_{R_v},W^\eta_{R_v}$ on the r.h.s.\ of the equations generate the frame algebra $\mathcal{A}_R$ (cf.~Eq.~\eqref{eq_kinRS}), this informs us that the generators of the frame algebra $\mathcal{A}_{R'}$ in the \emph{new} tensor product structure can be written in terms of  the \emph{old} kinematical data via the corresponding expressions on the l.h.s., namely as the \emph{full} Gauss law operators (cf.~Eq.~\eqref{eq_Gausslawpure}), and the tree Wilson lines, so that:
\begin{equation}
    \mathcal{A}_{R'}=\big\langle  U_{v}^{g}, W_{R_v}^\chi\,\big|\,g\in G, \chi\in\hat{G}, v\in\mathcal{V}\setminus \{v_0\}\big\rangle\,.
\end{equation}
We equip the $R'$ on the left with a prime to distinguish this representation of the frame algebra from Eq.~\eqref{eq_kinRS}.

To elucidate the algebra generating the new $\rm{pn}$-factor in terms of the old kinematical structure, let us work backwards, i.e.\ from right to left. 
The algebra on $\Hil_{\rm pn}$ is given by the loop algebra in Eq.~\eqref{Apnpure}. 
Now for the electric generators we have
\begin{equation}
    \ket{\partial\bm\chi}_R\otimes_R U_\ell^g\Pi_{\rm pn}\ket{\psi}_{\rm pn}=\ket{\partial\bm\chi}_R\otimes_R U_\ell^g\ket{\psi}_{\rm pn}=R_{\partial\bm\chi}U^g_{\ell}\ket{\psi}_{\rm pn}=U_\ell^g\,R_{\partial\bm\chi}\ket{\psi}_{\rm pn}
\end{equation}
since $[U^g_\ell,R_{\partial\bm\chi}]=0$ for $\ell\in S$. 
The same argument applies to the magnetic generators $H_\ell^\chi\Pi_{\rm pn}$ of the loop algebra. 
Hence, in terms of the old kinematical data we can write the algebra generating the $\rm{pn}$-factor of the new tensor product structure as
\begin{equation}
    \mathcal{A}_{S'}=\big\langle U^g_\ell,H^\chi_\ell\,\big|\,g\in G, \chi\in\hat G,\ell\in S\big\rangle\,,
\end{equation}
i.e.\ essentially as the loop algebra in Eq.~\eqref{Apnpure}, except without the projectors to the perspective-neutral space.

Now clearly, we have that the new frame and system algebras commute, $[\mathcal{A}_{R'},\mathcal{A}_{S'}]=0$, are (Type I) factors,\footnote{A factor is an algebra with a trivial center: the only elements commuting with all other elements in the algebra are multiples of the identity. That both algebras are Type I factors follows from the fact that both are isomorphic to the full algebra of bounded operators on a Hilbert space, $\mathcal{B}(\Hil)$.} and they generate all of $\mathcal{A}_{\rm kin}=\mathcal{B}(\Hil_{\rm kin})$, $\mathcal{A}_{\rm kin}=\mathcal{A}_{R'}\vee\mathcal{A}_{S'}$. This implies that they induce the tensor product structure $\otimes_R$ on $\Hil_{\rm kin}$ \cite{Zanardi:2001zz,Zanardi:2004zz,takesaki2001theory}.\footnote{This provides a lattice gauge theory incarnation of arguments linking abstract \cite[Thm.~4.17]{Carrozza:2024smc} and algebraic results \cite[Thm.~4.13]{Carrozza:2024smc} that establish how any Pauli error set in Pauli stabilizer codes induces a new tensor product structure in which errors and stabilizers act trivially on the encoded logical data.}

In conclusion, every kinematical state can be decomposed uniquely into components with definite vertex charges, and within each such sector the gauge-invariant degrees of freedom are identical. The charge label therefore plays the role of an independent degree of freedom that can be factored out in the new Gauss law aligned tensor product structure.

\subsubsubsection{Maximal correctable error sets as sections of the Gauss law map}\label{sec:pgmaxsets}
Having identified the logical algebra of the Gauss law code in terms of relational observables and aligned the subsystem partition to the Gauss law map, we now turn to the question of error correction. The task is to determine which operators act nontrivially on the kinematical Hilbert space, while remaining correctable on the code subspace, and to characterize maximal sets of such correctable errors.

The Weyl structure of the link algebra provides the guiding intuition. The Wilson line operators $W_{\ell}^{\chi}$ and the electric operators $U_{\ell}^g$ form a generalized Pauli pair: the former play the role of $X$-type operators, while the latter act as generalized  $Z$-type operators. Crucially, the electric operators are gauge-invariant and therefore belong to the logical algebra. They do not generate detectable errors on the code subspace. By contrast, open Wilson line operators are not gauge-invariant, translate between different vertex charge sectors and thus constitute the natural candidates for correctable errors. From the error correction perspective, these are the analogs of Pauli $X$-errors. 

It is therefore natural to restrict attention to error sets generated by general products of Wilson lines:
\begin{equation}\label{prodWilsonlineswithchis}
            W^{\bm{\chi}} = \bigotimes_{\ell \in \mathcal{L}} W_{\ell}^{\chi_{\ell}}.
\end{equation} 
Our goal is to characterize the maximal correctable sets of Wilson line errors. 
Maximal means that no further linearly independent error can be added to the set without violating the Knill-Laflamme condition \eqref{eq_KLstabcode}. 
In our setting this will entail that each maximal set of correctable Wilson line errors will encompass the full spectrum of possible vertex charge configurations, which will also constitute the syndrome of these errors.

The Gauss law aligned tensor product structure makes this analysis straightforward. Indeed, the key observation is that the code projector takes an almost trivial form in it:
\begin{lem}\label{lem_pipn}
   In the Gauss law aligned tensor product structure, we have: $\Pi_{\rm pn}=\ket{\bm{1}}\!\bra{\bm{1}}_R\otimes_R\mathbbm{1}_{\rm pn}$.
\end{lem}
\begin{proof}
    Using Eq.~\eqref{eq_trivialGauss}, we have  
    \begin{equation}
       \Pi_{\rm pn}=\int_{G^{N_V-1}}{\dee}{\bm{g}}\, U^{\bm{g}}=\int_{G^{N_V-1}}\dee {\bm{g}}\,U^{\bm{g}}_R\otimes_R\mathbbm{1}_{\rm pn} \,.
\end{equation} 
Now multiplying 
\begin{equation}
\mathbbm{1}_R=\sum_{\partial\bm\chi\in\hat{G}^{\times N_V-1}}\ket{\partial\bm\chi}\!\bra{\partial\bm\chi}_R
\end{equation}
with $U^{\bm{g}}_R$ yields 
\begin{equation}U^{\bm{g}}_R = \sum_{\partial\bm\chi\in\hat{G}^{\times N_V-1}}\partial\bm\chi(\bm{g})\ket{\partial\bm\chi}\!\bra{\partial\bm\chi}_R\,,
\end{equation}
where $\partial\bm\chi(\bm{g})\coloneqq\prod_{v\in\mathcal{V}\setminus\{v_0\}}\,\partial\chi_v(g_v)$. From the character orthogonality relations, it follows that
\begin{equation}
   \int_{G^{N_V-1}}{\dee}\bm{g} \, \partial\bm{\chi}(\bm{g})=\delta_{\partial\bm\chi,\bm{1}}\,.
\end{equation} 
In conjunction, this proves the result.
\end{proof}

Armed with this result, checking for correctability  of Wilson line error sets is now simple.
\begin{prop}\label{prop_KLpure}
The Knill-Laflamme condition for general Wilson line product errors reads
\begin{equation}\label{eq_KLpure}
  \Pi_{\rm pn}\,W^{\bar{\bm\eta}}\,W^{\bm\chi}\,\Pi_{\rm pn}=H^{\bar{\bm\eta}\bm{\chi}_S}\delta_{\partial\bm\eta,\partial\bm\chi}\,\Pi_{\rm pn}\,,  
\end{equation}
where $H^{\bar{\bm\eta}\bm{\chi}_S}\coloneqq\prod_{\ell\in S} H_{\ell}^{\bar{\eta}_\ell\chi_\ell}$ is a product of system loops, and $\partial\bm\chi$ is the Gauss law map, Eq.~\eqref{boundarymap}.
\end{prop}
\begin{proof}
The set of Wilson line products is closed under multiplication. It thus suffices to evaluate $\Pi_{\rm pn}\,W^{\bm\chi}\,\Pi_{\rm pn}$. The first step is to separate tree and non-tree contributions. 
Writing
\begin{equation}
    W^{\bm{\chi}} = \left(  \bigotimes_{\ell \in R}W_{\ell}^{\chi_{\ell}}  \right) \otimes \left(  \bigotimes_{\ell' \in S} W_{\ell'}^{\chi_{\ell'}}   \right),
\end{equation}
we recall that for every link $\ell \in S$, the combination of $\ell$ with the unique tree path $\gamma_R[v_f, v_i]$, where $v_i,v_f$ are the initial and final vertices of $\ell$, forms a closed loop, so that
\begin{equation}\label{WandH}
    W_{\ell}^{\chi_{\ell}} 
    = 
    H_{\ell}^{\chi_\ell} W_{\gamma_R[v_i, v_f]}^{\chi_{\ell}}\,,\qquad\qquad\ell\in S\,.
\end{equation}
On the other hand, for tree links we have $W_\ell^{\chi_\ell}=W^{\chi_\ell}_{\gamma_R[v_i,v_f]}$, $\ell\in R$, and since loops commute with $\Pi_{\rm pn}$, this yields
\begin{equation}\label{eq_treenontree}
 \Pi_{\rm pn}\,W^{\bm\chi}\,\Pi_{\rm pn}=\prod_{\ell'\in S} H_{\ell'}^{\chi_{\ell'}}\,\Pi_{\rm pn}\,\prod_{\ell\in\mathcal{L}}W^{\chi_\ell}_{\gamma_R[v_i,v_f]}\,\Pi_{\rm pn}=H^{\bm\chi_S}\,\Pi_{\rm pn}\,\prod_{\ell\in\mathcal{L}}W^{\chi_\ell}_{\gamma_R[v_i,v_f]}\,\Pi_{\rm pn}\,.   
\end{equation}
Now tree Wilson lines between $v$ and $v'$ can be decomposed as 
\begin{equation}
    W_{\gamma_R[v, v']}^{\eta} = W_{\gamma_R[v,v_0]}^{\eta} W_{\gamma_R[v', v_0]}^{\bar{\eta}}
    =W^\eta_{R_v}\,W^{\bar\eta}_{R_{v'}}\,. 
\end{equation}
As $W_{R_v}^{\eta}\,R_{\partial\bm\chi}=R_{\partial\bm\chi'}$ with $\partial\bm\chi'=(\partial\chi_{\tilde{v}}\eta^{\delta_{v,\tilde{v}}})_{\tilde{v}\in\mathcal{V}\setminus\{v_0\}}$, Eq.~\eqref{eq_treeWilsonlinenewfac} implies
\begin{equation}\label{Wsonframepuregauge}
\prod_{\ell\in\mathcal{L}}W^{\chi_\ell}_{\gamma_R[v_i,v_f]}=   \sum_{\partial\bm{\eta} \in \hat{G}^{N_V - 1}} \ket{\partial\bm{\chi}\partial\bm\eta}\!\bra{\partial\bm{\eta}}_R \otimes_R \mathbbm{1}_{\text{pn}}\,,
\end{equation}
 where $\partial\bm\chi$ is the Gauss law map, Eq.~\eqref{boundarymap}, and $\partial\bm\chi\partial\bm\eta=\left( \partial\chi_v\partial\eta_v \right)_{v\in\mathcal{V}\setminus\{v_0\}}$. Hence, invoking Lemma~\ref{lem_pipn} gives
\begin{equation}
     \Pi_{\rm pn}\,W^{\bm\chi}\,\Pi_{\rm pn} = H^{\bm\chi_S}\,\delta_{\partial\bm\chi,\bm{1}}\,\Pi_{\rm pn}\,,
 \end{equation} 
which entails the Knill-Laflamme condition in the stated form.
\end{proof}

In other words, a product of Wilson line operators survives conjugation by the code projector if and only if it preserves Gauss's law. In that case, its residual action is purely logical, given by the holonomies $H^{\bm\chi_S}$, which would infringe correctability if they were nontrivial. Thus, two Wilson line product errors are jointly correctable if their Gauss law data $\partial\bm\chi$ differs. 
Different values of $\partial \bm{\chi}$ correspond to orthogonal charge sectors in the redundancy factor and are therefore perfectly distinguishable by the code. 
The structural consequence is immediate: if a set $\mathcal{E}_{\mathcal{I}}=\{W^{\bm\chi}\}_{\bm\chi\in\mathcal{I}\subset\hat{G}^{\times N_L}}$ of Wilson line products is to form a correctable error set, then no two distinct elements in $\mathcal{I}$ may share the same Gauss law data $\partial\bm\chi$. Otherwise, the Kronecker delta in Eq.~\eqref{eq_KLpure} clicks, and one picks up a nontrivial logical action $H^{\bm{\bar{\eta}\chi}_S}$
inside the code, in violation of the Knill-Laflamme condition.\footnote{Of course, the Knill-Laflamme condition would still be obeyed if also $\bar{\bm\eta}\bm{\chi}_S=\bm{1}$ when $\partial\bm\chi=\partial\bm\eta$. However, it follows from the discussion around the Gauss law map in section~\ref{sssec_codespacepure} that under this condition not only the electric data in $S$ coincides, but also on the frame tree Wilson loops. That is, the two errors coincide.} 
In terms of the Gauss law map, $\partial:\hat{G}^{\times N_L}\to\hat{G}^{\times N_V-1}$ in Eq.~\eqref{boundarymap}, this means that $\mathcal{E}_\mathcal{I}$ is correctable if and only if $\mathcal{I}$ contains at most one representative in each fiber $F$ of the Gauss law map.

Maximality of $\mathcal{E}_\mathcal{I}$ enforces the converse requirement: for every Gauss law configuration $\partial\bm{\chi} \in \hat{G}^{N_V - 1}$, $\mathcal{I}$ must contain exactly one representative in each fiber $F$. In other words, a maximal correctable set is obtained by choosing a global section of the Gauss law bundle. 
Concretely, let 
\begin{equation}\label{eq_sectionspureGauss}
    s: \hat{G}^{N_V - 1} \rightarrow \hat{G}^{N_L}, \qquad \partial(s(\partial\bm{\chi})) = \partial\bm{\chi}
\end{equation}
be a global section. The associated maximal set is 
\begin{equation}\label{eq_puremaxset}
    \mathcal{E}_s \coloneqq \{  W^{s(\partial\bm{\chi})} \big| \partial\bm{\chi} \in \hat{G}^{\times N_V - 1}   \}.
\end{equation}

Why are these sets maximal? 
First, we explained above why one cannot add two or more representatives from a given fiber into any $\mathcal{E}_{\mathcal{I}}$ without violating correctability. 
Second, for each $\mathcal{E}_s$, the $W^{s(\partial\bm\chi)}$ exhaust all the possible orthogonal error spaces $\Hil_{\partial\bm\chi}$ in the decomposition in Eq.~\eqref{eq_chargedecomp}. 

An immediate example of a maximal set of correctable Wilson line errors is constituted by the dressing frame fields in Eq.~\eqref{eq_framefield}:
\begin{equation}
    \mathcal{E}_s=\{R_{\partial\bm\chi}\,|\,\partial\bm\chi\in \hat{G}^{\times N_V-1}\}\,.
\end{equation}
This set is distinguished by a trivial action on the tree complement, i.e.\ the system $S$, and so it is the only maximal one that contains no loops. However, there are many other choices. Indeed, our discussion implies that different maximal sets of correctable Wilson line errors are in one-to-one correspondence with different choices of section of $\partial$. Such maximal sets are thus naturally labeled by global sections of the Gauss law map. Changing the section -- and thereby the maximal set -- amounts to multiplying each Wilson line error by an element in the kernel of $\partial$. For example, this could be a Wilson loop, which is a logical operator and yields a fiber translation. 

Lastly, it is also clear what the error syndromes are in this case, namely the Gauss law data $\partial\bm\chi$, i.e.\  the electric vertex charge configurations. Thus, to measure the syndrome, one may measure any set of generators $\{U^g_v\}_{g\in\langle G\rangle}$ of the exponentiated Gauss law at each vertex $v\in\mathcal{V}\setminus\{v_0\}$, where $\langle G\rangle$ denotes any set of group generators of $G$. In summary, in an error context, the input of the Gauss law map are error labels and its output are their syndromes.

\subsubsubsection{Code parameters}\label{ssssec_codeparampure}
To summarize, we have identified the kinematical space and perspective-neutral subspace of the pure gauge theory with the physical space and code subspace of the associated Gauss law code, respectively. On a lattice with $N_L$ links and $N_V$ vertices, the former consists of $N_L$ kinematical $G$-registers living on the links, while the latter arises after imposing the $N_V - 1$ independent Gauss law constraints. This yields $N_L - N_V + 1$ logical $G$-registers, which corresponds to the fibers $F$ in the Gauss law map. These quantities provide the natural generalizations of the code parameters $n$ and $k$ in standard stabilizer codes, corresponding to the number of physical and logical qubits (in our case $G$-registers).

The structure of the stabilizer group $\mathcal{G}$ -- the group of gauge transformations -- is particularly simple: it is generated entirely by $U$-type operators (the $G$-analogs of Pauli $Z$). In particular, there are no $W$-type (Pauli $X$-like) stabilizers. The code is therefore a generalized CSS code \cite{Calderbank:1995dw,Steane:1995vv} with a trivial $W$-stabilizer sector. Equivalently, it behaves as a classical code: only $W$-type errors can be detected, while $U$-type errors act within the code space. This asymmetry is reflected in the code distances. The usual (symmetric) distance $d$, defined as the minimal weight of any nontrivial logical operator, is not a meaningful figure of merit in this setting. Indeed, since a single-link operator $U_\ell^g$ commutes with all Gauss law constraints, it already lies in the logical algebra $\mathcal{A}_{\text{pn}}$ (cf. Eq.~\eqref{Apnpure}). It follows that $d=1$, and more specifically the $U$-distance is $d_U = 1$, reflecting the fact that $U$-type errors cannot be detected or corrected. By contrast, nontrivial logical operators built from $W$-type operators must form closed loops in order to commute with all Gauss law constraints. The $W$-distance $d_W$ is therefore given by the minimal number of links in a non-contractible loop. This depends on the lattice geometry: for instance, a square (cubic) lattice yields $d_W = 4$, while a triangular lattice gives $d_W = 3$. More generally, one always has $d_W \geq 2$. We therefore conclude that the pure gauge sector of an Abelian lattice gauge theory realizes a $G$-generalized $[N_L, N_L - N_V +1, d_W]$ classical bit-flip code, with the code parameters uniquely determined by the lattice parameters and geometry.

\subsubsection{Gauss law codes with bosonic matter as subsystem codes}\label{sssec_bosGauss}
Let us now add bosonic matter on the vertices of the Gauss law code. As reviewed in section \ref{sec:LGT},the Gauss law constraint at a vertex $v$
\begin{equation}
    U^g_v\,\ket{\psi}_{\rm pn}=\ket{\psi}_{\rm pn}
\end{equation}equates the gauge charge flowing out of the vertex from the adjacent links with the charge carried by the matter located at
$v$. The notion of charge is again naturally described by the irreducible unitary representations of $G$, i.e.\ by its characters. In terms of characters (i.e.\ the eigenvalues of electric eigenstates) the Gauss law at $v$ thus reads
\begin{equation}
   \label{eq_bosonicGauss}
   \partial{\chi}_v= 
   \left( \prod_{\ell \in \mathcal{L}_{\text{out}}(v)} \chi_{\ell}   \right) \left( \prod_{\ell \in \mathcal{L}_{\text{in}}(v)} \bar{\chi}_{\ell}     \right) = \bar{\rho}_v\,,
\end{equation}
where the left hand side is the Gauss law map, Eq.~\eqref{boundarymap}, as before, and $\rho_v$ denotes the matter charge. But how do we decompose $\rho_v$ into different matter species contributions? Different species, as well as their particles and antiparticles, will be distinguished by the characters of the group that specify the charge they carry.

\subsubsubsection{Bosonic vertex Hilbert space}\label{ssssec_bosHil}
At this stage, we will deviate from the general discussion and make an assumption that will simplify our subsequent analysis, yet is sufficiently general to encompass a large class of relevant theories, such as, for example, scalar QED. The left hand side of Eq.~\eqref{eq_bosonicGauss} can take arbitrary values in the character group $\hat {G}$. 
To simplify the solution structure of the Gauss law (and thus the permitted code words), we would like there to always be a compensating matter charge configuration on the right hand side. 

This can be achieved by considering a set of $r$ generators $\{\rho_i\}_{i\in\mathbb{Z}_r}$ of the character group $\hat {G}$ and henceforth assuming that we have one bosonic particle type for each generator. 
The type includes not only the species, but also  the distinction between particles and antiparticles of a given species, a split we shall clarify shortly.
Each particle of type $i$ thus contributes a factor of $\rho_i$ to the total (exponentiated) matter charge at $v$
\begin{equation}\label{eq_rhototbos}
    {\rho}_v=\prod_{i=0}^{r-1}{\rho}_i^{n_{v,i}}\,,
\end{equation}
where $n_{v,i}$ denotes the occupation number of type $i$ at $v$. This yields an arbitrary element of $\hat{G}$, provided we permit the occupation number to take any value in $\{0,1,\ldots,D_i-1\}$, where $D_i$ is the order of the corresponding character. 
We distinguish between generators of finite and infinite order. Recall that the order of a character $\rho \in \hat{G}$ is the smallest positive integer $D$ such that $\rho^D = 1$, if such an integer exists; otherwise the character is said to have infinite order.  
Without loss of generality we label the generators such that the first $0\leq p< r$ generators have finite order. Since for a generator $\rho_i$ of finite order the corresponding charge can only take values in the finite cyclic subgroup generated by $\rho_i$, it is indeed natural to represent the associated matter degree of freedom by a $D_i$-level system. Generators of infinite order, on the other hand, are represented by harmonic oscillators.

This permits us to distinguish particle species, as well as particles and antiparticles of a given species. 
A given particle species is characterized by the charge $\rho_i$ it carries, while its antiparticle carries the opposite charge $\bar\rho_i$. 
In our case, since by assumption all $\rho_i$ are generators of $\hat G$, we have that all particle species are nontrivially charged, $\rho_i\neq 1$, $i\in\mathbb{Z}_r$.\footnote{Note that there are ways in which we can think of distinct particle species as having ``the same'' electric charge, despite their character charge thus differing $\rho_i\neq\rho_j$ as elements of $\hat{G}$. For example, when $\hat{G}=\hat{G}_1\times \hat{G}_1\times\cdots$ contains several copies of the same group $\hat{G}_1$, there will exist $i,j\in\mathbb{Z}_r$ such that $\rho_i,\rho_j$ correspond to the same element of $\hat{G}_1$, but different elements of $\hat {G}$. The antiparticle $\bar\rho_j$ therefore cannot annihilate $\rho_i$ (somewhat similarly to how positrons cannot annihilate muons, despite carrying opposite electric charge, because their lepton numbers do not match).} 
This means that each species comes with a pair of nontrivial charges $(\rho_i,\rho_{\bar{i}}\coloneqq\bar\rho_i)$, and so both particle and antiparticle, for otherwise one could not generate the inverse of each group element. 
The question is whether particle and antiparticle correspond to  distinct types $i,\bar i\in\mathbb{Z}_r$. 
This is where a key difference between particle species of finite and infinite character order arises: in the finite case, we have $\bar\rho_i=\rho_i^{D_i-1}$, so that the antiparticle is not an independent type, but corresponds to $D_i-1$ particles (and vice versa). By contrast, in the infinite case, we necessarily have that $\rho_i$ and $\bar\rho_i$ are independent generators, so that particle and antiparticle have to be treated as distinct types. In that case, we henceforth denote the type of the particle by $i$ and the type of the antiparticle by $\bar i$.

As we will illustrate in section~\ref{ssec_scalarQED}, scalar QED obeys these assumptions, where a complex scalar field provides harmonic oscillator modes carrying opposite $\rm{U}(1)$ charges, allowing positive and negative electric charges to be balanced locally. In particular, the character group $\widehat{\rm{U}(1)}=\mathbb{Z}$ indeed has two generators, one for positive and one for negative charge, and they correspond precisely to particle and antiparticle. The construction above simply generalizes this feature to any compact Abelian group. Thus, we can now write the vertex Hilbert space as $\mathcal{H}_v =\bigotimes_{i=0}^{r-1}\Hil_{v,i}$
where $\Hil_{v,i}$ is the Hilbert space of type $i\in\mathbb{Z}_r$ at $v$. The above also implies that to describe the matter states, it is convenient to work in the occupation number basis. This basis is defined by the simultaneous eigenvectors of the number operators $N_{v,i}$, which count the number of quanta of type $i$ at vertex $v$. The vertex Hilbert space may thus also be written as
\begin{equation} \label{eq:bosonic_vertexH}
     \mathcal{H}_v = \text{Span} \left\{ \ket{\bm{n_v}}_v \big| \bm{n_v} \in \mathcal{N}    \right\},
\end{equation}
where the multi-index $\bm{n_v} = (n_{v,0}, \cdots, n_{v, r-1})$ records the occupation numbers of all types. The allowed values are 
 \begin{equation}
     \mathcal{N} = \left\{ \bm{n} = (n_i)_{i=0}^{r-1} \big| n_i \in \mathcal{N}_i  \right\},\qquad\qquad
     \mathcal{N}_i=\begin{cases}
         \mathbb{Z}_{D_i}\,,&D_i<\infty\\
         \mathbb{N}_0\,,&\text{otherwise}\,,
     \end{cases}
 \end{equation}
and the number operators act in the usual way,
  \begin{equation}\label{eq_numop}
     N_{v,i} \ket{\bm{n_v}}_v = n_{v,i} \ket{\bm{n_v}}_v.
 \end{equation}
This puts us finally into the position to specify the unitary representation $u_v$ of gauge transformations on the matter more explicitly. The matter fields carry a representation $\rho_i$ of the gauge group. Eqs.~\eqref{eq_bosonicGauss} and~\eqref{eq_rhototbos} imply 
\begin{equation}\label{eq_gaugegrouprepbos}
     u_v^g = \prod_{i=0}^{r-1} \rho_i^{N_{v,i}}(g)\,,\qquad\qquad \forall\, g\in G\,.
\end{equation}
This expression simply states that each particle of type $i$ contributes a factor $\rho_i(g)$  to the gauge transformation. Consequently, the total phase acquired by a state is determined by the total charge stored at the vertex. This representation naturally decomposes the vertex Hilbert space into charge sectors (irreps). Indeed, the action of $u_v^g$ diagonalizes the space into subspaces labeled by characters $\rho \in \hat{G}$:
  \begin{equation}
     \mathcal{H}_v 
     = \bigoplus_{\rho \in \hat{G}} \mathcal{H}_v^{\rho}.
 \end{equation}
The subspace $\mathcal{H}_v^{\rho}$ consists of all states that transform with character $\rho$, 
\begin{equation}
    \mathcal{H}_v^{\rho} = \left\{  \ket{\varphi}_v \in \mathcal{H}_v  \big| u_v^g \ket{\varphi}_v = {\rho}(g)\ket{\varphi}_v \forall g \in G     \right\}.
\end{equation}
In the occupation basis, this condition has a simple interpretation: the total charge of the configuration must equal $\rho$.  Accordingly,
\begin{equation}  
  \mathcal{H}_v^{\rho}=  \text{Span} \left\{ \ket{\bm{n_v}}_v \big| \prod_{i=0}^{r-1} \rho_i^{n_{v, i}} = \rho  \right\}.
\end{equation}
For example, the right hand side of Eq.~\eqref{eq_bosonicGauss} corresponds to a state from $\Hil_v^{\bar\rho_v}$.

Creation and annihilation operators move states between these charge sectors. Acting on an occupation state, they change the number of quanta of the corresponding species. The annihilation operator for type $i$ acts as
\begin{equation}
    a_{v, i} \ket{\bm{n_v}}_v = C^{-}_i(n_{v,i}) \ket{\bm{n_v} - \bm{e_i}}\,,\qquad\qquad
    (\bm{e_i})_j=\delta_{i,j}\,,
\end{equation}
where the coefficients take the familiar harmonic oscillator form,
\begin{equation}
  C^-_i(n_{v,i})=\sqrt{n_{v,i}}\,,\qquad p<i<r\,,\qquad\qquad  C^{-}_i(n_{v,i}) = \begin{cases}
        \sqrt{n_{v,i}} & 0\leq n_{v,i}\leq D_i-1 \\
        0 & \text{otherwise}
    \end{cases}\,,\qquad 0\leq i\leq p\,,
\end{equation}
except that in the finite order case, the operator is truncated at level $D_i$. 
In that case, it may also be written as
\begin{equation}\label{eq_finiteanihil}
    a_{v,i} = \sqrt{D_i-1}\ket{D_i-2}\!\bra{D_i-1}_{v,i}+\cdots+\sqrt{m}\ket{m-1}\!\bra{m}_{v,i}+\cdots+\ket{0}\!\bra{1}_{v,i}=\sum_{n=0}^{D_i-2}\sqrt{n+1}\ket{n}\!\bra{n+1}_{v,i}\,.
\end{equation}
The creation operators are simply their Hermitian conjugates, acting as
\begin{equation}
    a_{v, i}^{\dagger} \ket{\bm{n_v}}_v = C^{+}_i(n_{v,i}) \ket{\bm{n_v} + \bm{e_i}}\,,\qquad\qquad C_i^+(n_{v,i})\coloneqq C_i^-(n_{v,i}+1)  \,.
\end{equation}
In both the finite and infinite order case, we can then write the number operator in Eq.~\eqref{eq_numop} as usual
\begin{equation}
    N_{v,i}=a_{v,i}^\dag\,a_{v,i}\,,
\end{equation}
but the basic commutation relations differ in the two cases
\begin{equation}\label{eq_aadagcommutator}
    [a_{v,i},a_{v,i}^\dag]=\mathbbm{1}\,,\qquad p\leq i<r\,,\qquad\qquad [a_{v,i},a_{v,i}^\dag]=\mathbbm{1}-(D_i-1)\ket{D_i-1}\!\bra{D_i-1}_{v,i}\,,\qquad 0\leq i< p\,.
\end{equation}
The latter identity is, of course, a manifestation of the absence of canonical commutation relations in finite dimensions. From the right relation, it may seem at first that in the finite order case the algebra generated by $a_{v,i},a_{v,i}^\dag$ does not include the identity. However, noting that 
\begin{equation}\label{eq_highestweight}
   \ket{D_i-1}\!\bra{D_i-1}\propto (a_{v,i}^\dag)^{D_i-1}\,a_{v,i}^{D_i-1}  \,,
\end{equation}
we see that also in the finite case this algebra is unital.

From the definition of $u_v^g$, one immediately obtains how the ladder operators transform. They carry definite gauge charge: 
    \begin{align}
        a_{v, i} &\to U^{\bm{g}}\,a_{v, i}\,U^{\bm{g} \dagger} = \bar{\rho_i}{(g_v)}\,a_{v, i}\nonumber\\
        a_{v, i}^{\dagger} &\to U^{\bm{g}}\,a_{v, i}^{\dagger}\,U^{\bm{g} \dagger} = \rho_i(g_v)\,a_{v, i}^{\dagger}\,.\label{eq_laddertrans}
    \end{align}
Thus, $a_{v,i}$ lowers the $\rho_i$-charge at $v$, while $a_{v,i}^{\dagger}$ raises it, in line with the above discussion.

For the purpose of constructing the gauge-invariant matter algebra later, let us briefly clarify the status of creation and annihilation operators for pairs of particle $i$ and antiparticle $\bar i$. 
Since they carry opposite charge, they come in pairs of oppositely transforming ladder operators 
\begin{align}
    (a_{v,i},a_{v,\bar{i}})&\to \left(\bar{\rho}_i(g_v)\,a_{v,i},\rho_i(g_v)\,a_{v,{\bar i}}\right)\nonumber\\
    (a^\dag_{v,i},a^\dag_{v,\bar{i}})&\to \left({\rho}_i(g_v)\,a^\dag_{v,i},\bar\rho_i(g_v)\,a^\dag_{v,\bar i}\right)\,.\label{eq_partantiparttransform}
\end{align}
As noted above, when the particle species carry character charge of infinite order (as in, e.g., in scalar QED), we have that $\rho_i$ and $\rho_{\bar i}\coloneqq\bar\rho_i$ are independent generators in $\hat{G}$, which also means that $a_{v,i},a_{v,\bar i}$ are independent operators, and likewise for their creation counterparts. 

By contrast, when the character charge is of finite order, due to $\rho_{\bar i}=\bar{\rho}_i = \rho_i^{D_i-1}$, we have that the antiparticle is not independent, but corresponds to $D_i-1$ particles. 
Similarly, due to $\bar\rho_i^{n_{v,\bar i}}=\rho_i^{D_i-n_{v,\bar i}}$, $n_{v,\bar i}$ antiparticles coincide with $D_i-n_{v,\bar i}$ particles (and vice versa). 
Hence, in analogy to Eq.~\eqref{eq_finiteanihil}, the antiparticle annihilation operator for finite order type $i$ reads (in the occupation basis of the \emph{particle} of type $i$)
\begin{equation}
    a_{v,\bar i}=\ket{0}\!\bra{D_i-1}_{v,i}+\sqrt{2}\ket{D_i-1}\!\bra{D_i-2}_{v,i}+\cdots+\sqrt{D_i-1}\ket{2}\!\bra{1}_{v,i}=\sum_{m=1}^{D_i-1}\sqrt{D_i-m}\ket{m+1}\!\bra{m}_{v,i}\,,
\end{equation}
where $\ket{D}_{v,i}\equiv\ket{0}_{v,i}$. 
After a tedious exercise and invoking Eq.~\eqref{eq_highestweight}, it may be checked that
\begin{equation}\label{eq_finiteantianil}
    a_{v,\bar i}=\sum_{k=1}^{D_i-2}\sum_{m=1}^k\sqrt{\frac{D_i-m}{m!(m+1)!}}\frac{(-1)^{k-m}}{(k-m)!} \left(a_{v,i}^\dag\right)^{k+1}\,a^k_{v,i}+\frac{1}{\sqrt{(D_i-1)!}}\,a_{v,i}^{D_i-1}\,.
\end{equation}
That is, the annihilation operator of the antiparticle is a power series of the creation and annihilation operators of the particle, so that it is not independent.
As usual, the creation operator of the antiparticle is simply its Hermitian conjugate $a_{v,\bar i}^\dag$, and we now have $[a_{v,\bar i},a_{v,\bar i}^\dag]=\mathbbm{1}_{v,i}-(D_i-1)\ket{1}\!\bra{1}_{v,i}$. Using a similar argumentation as above, we note that the algebra generated by $a_{v,\bar i},a_{v,\bar i}^\dag$ is also unital.

It is then simple to check that the particle and antiparticle number operators obey 
\begin{equation}
   N_{v,i}+ N_{v,\bar i}=D_i\left(\mathbbm{1}_{v,i}-\ket{0}\!\bra{0}_{v,i}\right)=0\mod D_i\,,\qquad\qquad N_{v,\bar i}\coloneqq a^\dag_{v,\bar i}\,a_{v,\bar i} \,,
\end{equation}
making explicit the statement that $n_{v,\bar i}$ antiparticles correspond to $D_i-n_{v,\bar i}$ particles, and vice versa (when $n_{v,\bar i}>0$). 
In particular, the addition of an antiparticle corresponds to the addition $\!\!\!\mod D_i$ of $D_i-1$ particles: $a^\dag_{v,\bar i}\ket{n_{v,i}}_{v,i}\propto\ket{n_{v,i}+D_i-1\mod D_i}_{v,i}$.
Crucially, for the dressing formulae later, it is easy to see that $a_{v,\bar i}$ in Eq.~\eqref{eq_finiteantianil} (and thus also $a_{v,\bar i}^\dag$) transforms in the correct way given in Eq.~\eqref{eq_partantiparttransform} because each of the summands does.

\subsubsubsection{Code words and logical algebra from relational observables}\label{ssssec_boscodespace}
Let us now specify the structure of the code space, Eq.~\eqref{eq_codepn}, for this bosonic Gauss law code in more detail. We will then complete the discussion with an exposition of the corresponding (encoded) logical algebra, before turning to errors in the subsection thereafter. 

We will make use of several insights from our previous discussion of pure gauge Gauss law codes. For example, we will use the same spanning tree $R$ to define a lattice QRF. We begin by invoking the Gauss law aligned and frame-induced tensor product structure of section~\ref{ssssec_GausTPS} for the $\rm{gauge}$ part of the kinematical Hilbert space in Eq.~\eqref{eq_Hkinbos}
  \begin{equation}\label{eq_Hkinbosfac}
        \mathcal{H}_{\text{kin}} = \mathcal{H}_{\text{loops}} \otimes_R \mathcal{H}_R \otimes \mathcal{H}_{\text{matter}},
    \end{equation}
where $\mathcal{H}_{\text{loops}}$ coincides with the perspective-neutral subspace of the pure gauge theory, except that we now label it by `$\rm{loops}$' for distinction.  This factorization separates the gauge-invariant loop degrees of freedom from the QRF and the matter excitations. 
As we now have gauge-variance only to the right of $\otimes_R$, this already tells us that the new code space will have a subsystem code structure \cite{kribs2006, Poulin_2005} that we further discuss below.

As before, $\Hil_{\rm loops}$ is spanned by a gauge-invariant magnetic and electric basis, Eqs.~\eqref{magbas} and~\eqref{elecbas}, obeying Eq.~\eqref{eq_mageleceigenbasis} with only the label change $\rm{pn}\to\rm{loops}$. Invoking also the electric frame and matter occupation bases, we may expand an arbitrary kinematical state in the basis
\begin{equation}\label{eq_bosbasis}
\ket{\bm{g}_S, \partial\bm{\chi}, {\bf n}} = \ket{\bm{g}_S}_{\text{loops}} \otimes_R \ket{\partial\bm{\chi}}_R \otimes \ket{\bf{n}}_{\text{matter}}\,,\qquad\qquad\ket{\bf{n}}_{\rm matter}=\otimes_{v\in\mathcal{V}}\ket{\bm{n}_v}_v\,.
\end{equation}
Gauge transformations are diagonal in this basis.
For a vertex $v \neq v_0$, the action reads  \begin{equation}
    U_v^g \ket{\bm{g}_S, \partial\bm{\chi}, {\bf n}} =  \partial\chi_v(g)\,\rho_v(g)\,\ket{\bm{g}_S, \partial\bm{\chi}, {\bf n}}\,,\label{gauge_at_root_for_bosons}
\end{equation}
where $\rho_v$ is the total matter charge in Eq.~\eqref{eq_rhototbos}.

The transformation associated with the root $v_0$ requires an additional comment. The reason is simply that, in the presence of matter, the chosen frame $R$ is incomplete: it parametrizes all gauge redundancies except the one at $v_0$. This is not a physical complication but merely a bookkeeping issue. The missing transformation can be expressed in terms of the remaining local transformations together with a global gauge transformation that acts uniformly across the lattice. Concretely, one may write
\begin{equation}
    U_{v_0}^g = U_{\text{global}}^g \left( \prod_{v \in \mathcal{V} \setminus \{ v_0 \} } U_v^{g^{-1}}   \right).
\end{equation}
Since the global transformation acts trivially on the frame degrees of freedom, its action on the basis states can be evaluated straightforwardly:
\begin{eqnarray}\label{bosonicv0transf}
        U_{v_0}^{g} \ket{\bm{g}_S, \partial\bm{\chi}, {\bf n}}=\left( \prod_{v \in \mathcal{V}} {\rho}_v(g)  \right) \!\left(  \prod_{v \in \mathcal{V} \setminus \{ v_0  \}} \partial\bar{\chi}_v(g) \bar{\rho}_v(g)   \right) \ket{\bm{g}_S, \partial\bm{\chi}, {\bf n}}.
\end{eqnarray}
Thus, gauge invariance at $v\neq v_0$ entails the Gauss law in the form Eq.~\eqref{eq_bosonicGauss} with Eq.~\eqref{eq_rhototbos},  while invariance under the transformation at the root imposes an additional global constraint: the total matter charge across the lattice must vanish,
\begin{equation}\label{eq_bostotalchargezero}
    \prod_{v \in \mathcal{V}} \rho_v = 1.
\end{equation}
Since the generators $\rho_i$ are independent, this also means that we have a neutrality condition per species 
\begin{equation}
\sum_{v\in\mathcal{V}}(n_{v,i}-n_{v,\bar i})=0\,,\qquad p\leq i<r\,,\qquad\qquad\sum_{v\in\mathcal{V}} n_{v,i}=0\mod D_i\,,\qquad 0\leq i< p\,.
\end{equation}
Thus, for types of finite character order, this leads to distinct charge sectors labeled by $m$, all of which are consistent with global neutrality:
\begin{equation}\label{eq_chargesector}
    \sum_{v\in\mathcal{V}}n_{v,i}=m D_i\,,\qquad\qquad m=0,1,\ldots,M\,,
\end{equation}
where $M$ depends on the total number of vertices and is finite when that number is finite.

Taken together, these conditions determine the structure of the perspective-neutral (code) subspace:
\begin{equation}
    \mathcal{H}_{\text{pn}} = \text{Span} \left\{ \ket{\bm{g}_S}_{\text{loops}} \otimes_R \ket{\bf{n}}^{\text{dr}}_{\text{matter}} \Big| \bm{g}_S \in G^{N_L - N_V + 1}, {\bf n} \in \mathcal{N}^{N_V}, \prod_{v \in \mathcal{V}} \prod_{i=0}^{r-1} \rho_i^{n_{v,i}} = 1     \right\}.
\end{equation}
The dressed matter states appearing here combine the gauge field variables with the matter occupations in a gauge-invariant way. Concretely, with $\bf{n}$ subject to the global charge neutrality condition,
\begin{equation}\label{partialchin}
    \ket{\bf{n}}^{\text{dr}}_{\text{matter}} \coloneqq \ket{\partial\bm{\chi}^{(\bf{n})}}_R \otimes \ket{\bf{n}}_{\text{matter}}\,,\qquad\qquad{\bf n}\in\widetilde{\mathcal{N}} = \left\{ {\bf n} \in \mathcal{N}^{N_V} \;\middle|\; \prod_{v \in \mathcal{V}} \prod_{i=0}^{r-1} \rho_i^{n_{v,i}} = 1 \right\}\,,
\end{equation}
where the gauge charge $\partial\chi_v^{({\bf n})}$ at $v$ is determined by the matter configuration via Eqs.~\eqref{eq_bosonicGauss} and~\eqref{eq_rhototbos}. Hence, as anticipated, the code space factorizes into a loop sector and a dressed matter sector,
\begin{equation}\label{eq_subsystemcode}
    \mathcal{H}_{\text{pn}} = \mathcal{H}_{\text{loops}} \otimes_R \mathcal{H}^{\text{dr}}_{\text{matter}}\,,
\end{equation}
yielding a subsystem code structure that we explore further in the context of errors below. 

This subsystem code structure is, of course, reflected in the encoded logical algebra, i.e.\ the perspective-neutral one, which too factorizes
\begin{equation}
    \mathcal{A}_{\text{pn}} = \mathcal{A}_{\text{loops}} \otimes_R \mathcal{A}^{\text{dr}}_{\text{matter}}\,,
\end{equation}
where the loop algebra is given in Eq.~\eqref{Apnpure}.\footnote{With the subtle difference that now $\Pi_{\rm pn}=\int_{\mathcal{G}}{\dee}{\bm g}\,U^{\bm g}$ with $U^{\bm g}$ given in Eq.~\eqref{eq_gaugetrbos}, i.e.\ now including the matter.} 
On the other hand, we show in appendix~\ref{app_bosalgebra} that the QRF-dressed bosonic matter algebra is given by
\begin{equation}\label{eq_drmatteralg}
\mathcal{A}^{\text{dr}}_{\text{matter}}=\Big\langle a_{v,\bar{j}}^{\dagger} W_{\gamma_R[v,v']}^{\rho_j} a_{v',{j}}^\dagger\Pi_{\rm pn}\,,a_{v,j} W_{\gamma_R[v,v']}^{\rho_j} a_{v',j}^{\dagger}\Pi_{\rm pn}
\,,
a_{v,\bar{j}}^{\dagger} W_{\gamma_R[v,v']}^{\rho_j} a_{v',\bar{j}}\Pi_{\rm pn}\,,\text{h.c.}\,\Big|\, v,v'\in\mathcal{V}\,,j\in\mathbbm{Z}_r\Big\rangle\,,
\end{equation}
where h.c.\ means that we also include the Hermitian conjugates of the exhibited generators.\footnote{$\mathcal{A}^{\rm dr}_{\rm matter}$, like all our algebras, is an algebra of bounded operators. Some of the generators in Eq.~\eqref{eq_drmatteralg} are not bounded, however, when they include infinite order characters. E.g., the number operator for species described by oscillators is not bounded. In that case, the algebra includes bounded functions of the generators.} 
In words, the dressed matter algebra is generated by all net-zero charge \emph{pairs} of creation and annihilation operators for all species, including their antiparticles, concatenated by the tree Wilson line linking the pertinent vertices in the appropriate representation. 
Their gauge-invariance follows directly from Eq.~\eqref{eq_partantiparttransform}. 
Note that the algebra includes gauge-invariant pairs of creation and annihilation operators at each vertex $v\in\mathcal{V}$, such as, e.g., number operators $N_{v,i}$, by setting $v=v'$ above, which leads to a trivial Wilson line, $W^\eta_{\gamma_R[v,v]}=\mathbbm{1}$. 

Furthermore, also for finite order species, the Wilson line dressed particle-antiparticle pairs are independent of the dressed particle-particle pairs. For example, this can be seen from the fact that all particle-particle pairs involve annihilation operators and so annihilate the vacuum, whereas particle-antiparticle pairs, such as $a^\dag_{v,\bar j}W^{\rho_j}_{\gamma_R[v,v']}a_{v',j}^\dag$ map the vacuum, i.e.\ the $m=0$ charge sector in Eq.~\eqref{eq_chargesector}, to the $m=1$ charge sector of species $j$, while maintaining global net neutrality. When the state of $\mathcal{H}_{v,j}$ is not in the vacuum, the same operator preserves the charge sector, however. Similarly, its conjugate $a_{v,\bar j}W^{\bar\rho_j}_{\gamma_R[v,v']}a_{v',j}$ changes the charge sector $m\to m-1$ when the state in $\Hil_{v,j}$ is the highest occupied particle state $\ket{D_i-1}$ and leaves it invariant otherwise.

\begin{figure}[t]
\centering
\includegraphics[width=0.4\textwidth]{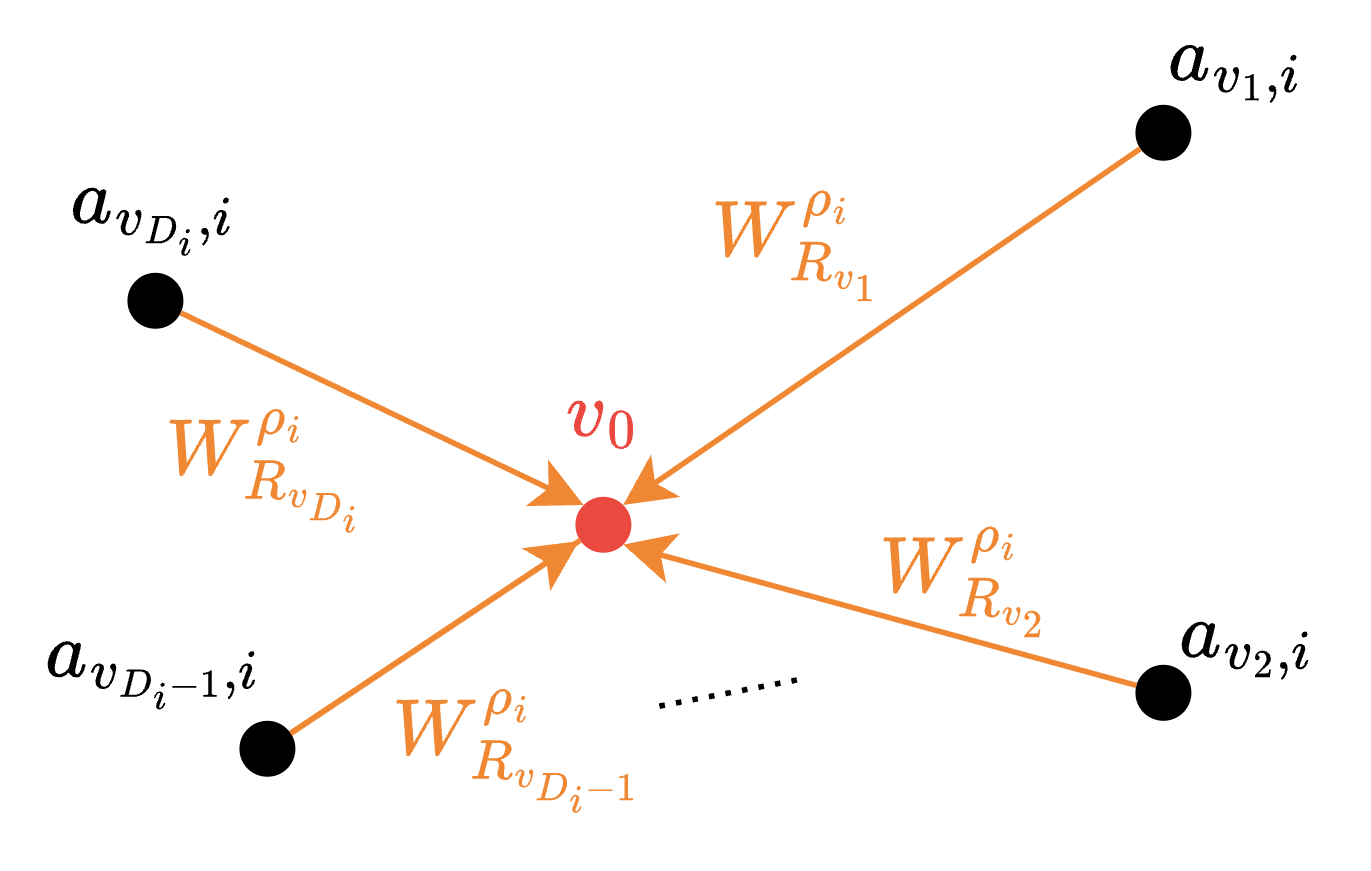}
\caption{Graphical depiction of the gauge invariant star operator $S_{\{ v_k \}}^{(i)}$ in Eq.~\eqref{eq_star} for a particle type $i$ carrying character-valued charge of finite order $D_i$.} 
\label{fig:star}
\end{figure}
It is also worth remarking that, only in the finite order case, one can generate gauge-invariant \emph{star operators} centered at any vertex $v$ from the bipartite generators, such as the following example at $v_0$ (see Fig.~\ref{fig:star}):
\begin{equation}\label{eq_star}
   S_{\{v_k\}}^{(i)} \coloneqq \prod_{k=1}^{D_i} a_{v_k,i}W_{R_{v_k}}^{\rho_i}\,.
\end{equation}
Using $\rho_i^{D_i}=1$ it can be readily checked that this object is invariant. In this example, the $v_k$ may also coincide and any $a_{v_k,i}$ may be replaced by $a_{v_k,\bar i}^\dag$. Similarly, one can connect many of such star operators along the tree graph.

\subsubsubsection{Maximal correctable error sets as sections of the extended Gauss law map}\label{sec:bosmaxsets}
Having identified the structure of the logical operator algebra, we can now turn to the question of errors. In the present setting, there are two natural types to consider: operators built from Wilson lines, which act on the connection degrees of freedom, and undressed operators that act on the matter sector. Understanding how these two classes behave under conjugation with the projectors onto the perspective–neutral subspace will allow us to complete our interpretation of the theory as a quantum error correcting code by identifying the maximal correctable sets.

Of course, our error analysis from the pure gauge Gauss law code remains valid for Wilson line operators. In particular, Proposition~\ref{prop_KLpure}, establishing the Knill-Laflamme condition for Wilson line errors, continues to hold. At this stage, the factorization of the code space in Eq.~\eqref{eq_subsystemcode} becomes important. Since the holonomy operator $H^{\bm{\bar{\chi} \eta }_S}$ appearing on the right hand side of Eq.~\eqref{eq_KLpure} acts entirely within the loops sector, it leaves the dressed matter degrees of freedom unaffected. 
Accordingly,
\begin{equation}
    \Pi_{\text{pn}} W^{\bar{\bm{\chi}}} W^{\bm{\eta}} \Pi_{\text{pn}} \in \left( \mathcal{B}(\mathcal{H}_{\text{loops}}) \otimes_R \mathbbm{1}^{\text{dr}}_{\text{matter}}   \right) \Pi_{\text{pn}}\,,
\end{equation}
which obeys the conditions of a \emph{subsystem stabilizer code} \cite{kribs2006, Poulin_2005}.

Indeed, we may view the loop sector as the ``gauge subsystem'' of the subsystem code.\footnote{We emphasize that the notion of \emph{gauge} subsystem in subsystem codes \cite{kribs2006, Poulin_2005} is not to be confused with the notion of gauge in our lattice gauge theory context.} Wilson line errors $W^{\bm{\chi}}$ act nontrivially on this ``gauge subsystem'', but they act trivially on the dressed matter sector. In the language of subsystem quantum error correction, the dressed matter degrees of freedom are therefore perfectly correctable with respect to Wilson line errors modulo the ``gauge subsystem'' $\mathcal{H}_{\text{loops}}$. This means that we may regard the dressed matter sector as a code on its own and store information in it, ignoring the loops factor as a redundancy. Abelian lattice gauge theories with bosonic matter thus naturally realize the structure of a subsystem stabilizer code.

Both to treat the dressed matter part as a code in its own right and to characterize maximal sets of correctable errors (when we $\Hil_{\rm pn}$ as the full code space), we must enlarge the discussion to include error candidates acting on the matter degrees of freedom. 
Our aim is to describe them by suitably generalized Pauli operators. For bosonic matter types of finite character order, this follows standard qudit stabilizer code constructions, e.g.\ see \cite{Rains:1997uh,Gottesman:1998se,Ashikhmin:2000ylq,eczoo_qudit_stabilizer,Wang:2020ife,Spagnoli:2026qni}. For infinite order type, the situation is slightly more subtle. 

To start, we recall that Pauli errors in ordinary Pauli stabilizer codes produce error syndromes valued in the (Pontryagin dual) character group of the stabilizer group, permitting one to label these errors by dual group elements, e.g.\ see \cite{Carrozza:2024smc}. 
Recalling that in the bosonic Gauss law codes the stabilizer group is given by gauge transformations Eq.~\eqref{eq_gaugetrbos}, this means that, when translated into the matter sector, we look for generalized Pauli errors for any boson of type $i$ at any vertex $v$ that transforms under the gauge transformations acting on $\Hil_{v,i}$ given by $u_{v,i}^g=\rho_i^{N_{v,i}}(g)$ (cf.~Eq.~\eqref{eq_gaugegrouprepbos}). 
Specifically, the generalized Pauli $X$-errors for bosons of type $i$ shall be labeled by elements in the character subgroup generated by $\rho_i,\bar\rho_i$, which also comprise their syndromes. The explicit realization of these operators depends on whether the corresponding charge sector has finite or infinite order. 

For matter species whose charges have finite order $D_i$, $0 \leq i < p$, the local Hilbert space at each vertex is $D_i$-dimensional. In this case the shift and phase operators are given by the generalized Pauli operators \cite{Rains:1997uh,Gottesman:1998se,Ashikhmin:2000ylq,eczoo_qudit_stabilizer,Wang:2020ife,Spagnoli:2026qni} 
\begin{equation}\label{eq_finorderX}
    X_{v,i} \ket{n_{v,i}} = \ket{n_{v,i} + 1}, \qquad Z_{v,i} \ket{n_{v,i}} = e^{\frac{2\pi i}{D_i} n_{v,i}} \ket{n_{v,i}},
\end{equation}
where the addition in the first equation is understood modulo $D_i$. Importantly, these are unitary and $X_{v,i}$ constitutes a representation of the generator $\rho_i$ of the Pontryagin dual $\hat{G}_i\simeq\mathbb{Z}_{D_i}$, while $Z_{v,i}$ is a representation of a generator $g_i$ of the gauge transformations $u_{v,i}^g$ on $\Hil_{v,i}$, which too are isomorphic to $G_i\simeq\mathbb{Z}_{D_i}$. Indeed, $\rho_i(g_i)=e^{\frac{2\pi i}{D_i}}\in\rm{U}(1)$ is the dual pairing between $\rho_i$ and $g_i$, while $n_{v,i}$ is in one-to-one correspondence with the elements of $\hat{G}_i$, thanks to $\rho_i^{n_{v,i}}$, $n_{v,i}\in\{0,1,\ldots,D_i-1\}$, parametrizing that group. 
In particular, the left equation provides a unitary realization of $W_{v,i}^{\rho_i}\ket{\rho_i^{n_{v,i}}}_{v,i}=\ket{\rho_i\rho_i^{n_{v,i}}}_{v,i}$, cf.~Eq.~\eqref{eq_Pdualrep}.

For the remaining species ( $p \leq j < r$), corresponding to characters of infinite order, the spectrum of the charge operator is unbounded and the shift operator is no longer unitary. In this case, the appropriate operators are
\begin{equation}
     X_{v,j} = a_{v,j}^{\dagger}(N_{v,j} + 1)^{-\frac{1}{2}}, \qquad Z_{v,j}(\theta_{v,j}) = e^{i\theta_{v,j}N_{v,j}},\label{eq_inforderX}
\end{equation}
with $\theta_{v,j} \in [0, 2\pi)$. In the occupation number basis, they act as
\begin{equation}
    X_{v,j} \ket{n_{v,j}} = \ket{n_{v,j} + 1}, \qquad Z_{v,j}(\theta_{v,j}) \ket{n_{v,j}} = e^{i \theta_{v,j}n_{v,j}} \ket{n_{v,j}}.
\end{equation}
While $Z_{v,j}$ generates a unitary representation of $G_j=\rm{U}(1)$ on $\Hil_{v,j}$ (to which the action of gauge transformations $G$ on $\Hil_{v,j}$ is isomorphic), $X_{v,j}$ does not by itself generate a representation of its dual $\hat{G}_j=\mathbb{Z}$. This is in line with the fact that, in the infinite order case, $\rho_i$ alone does not generate a subgroup of $\hat{G}$, for which its independent counterpart $\bar\rho_i$ is needed. 

However, there is a sense in which $X_{v,j}$ together with its antiparticle version $X_{v,\bar j}$ generates a unitary representation of $\mathbb{Z}$, albeit on a distinct ``coarse-grained'' Hilbert space of occupation number equivalence classes. For what follows, it is nevertheless instructive to briefly consider the corresponding construction. Defining
\begin{equation}\label{eq_coarsegrain}
 \ket{\Delta_{j\bar j} n}_{v,j\bar j}\coloneqq \big[\ket{n_{v,j},n_{v,\bar j}}_{v,j,\bar j}\in\Hil_{v,j}\otimes\Hil_{v,\bar j}\,\big|\, n_{v,j}-n_{v,\bar j}=\Delta_{j\bar j}n\big] \,,
\end{equation}
where the latter denotes the equivalence class of occupation number eigenstates such that the difference between particle and antiparticle occupation is precisely $\Delta_{j\bar j}n$. These equivalence classes are comprised of the eigenstates of $N_{v,j}-N_{v,\bar j}$ with fixed eigenvalue.  This space of equivalence classes has a linear structure and can be turned into a Hilbert space. Clearly, $\Delta_{j\bar j}n\in\hat{G}_j=\mathbb{Z}$ and, defining,
\begin{equation}
    W^{\rho_j}_{j\bar j}\ket{\Delta_{j\bar j} n}_{v,j\bar j}\coloneqq \big[X_{v,j}\ket{n_{v,j},n_{v,\bar j}}_{v,j,\bar j}\big]=\ket{\Delta_{j\bar j}n+1}_{v,j\bar j}\,,
\end{equation}
and similarly for $W^{\bar\rho_j}_{j\bar j}$ in terms of $X_{v,\bar j}$, which yields a shift $\Delta_{j\bar j}n\to\Delta_{j\bar j}n-1$, we have the ingredients for a unitary representation of the generators $\rho_i,\bar\rho_i$ of $\hat{G}_j=\mathbb{Z}$ (subject to building an appropriate inner product on the space of equivalence classes and completing in norm). This will be dual to a unitary representation of $\rm{U}(1)$ generated by $Z_{v,j\bar j}=e^{i(N_{v,j}-N_{v,\bar j})}=Z_{v,j}\otimes Z_{v,\bar j}$ on this space, provided we set $\theta_{v,j}=-\theta_{v,\bar j}$.

In what follows, we return to our previous ``fine-grained'' Hilbert space, on which we have to contend with the fact that infinite order characters do not admit a unitary representation, and Eq.~\eqref{eq_inforderX} is the best we can achieve. The above digression is useful to understand the structural obstructions to their unitary representation on our Hilbert spaces and will become relevant later. Effectively, our matter Hilbert space keeps track of ``too much'' information for the representation to be unitary.

Both sets of operators, Eqs.~\eqref{eq_finorderX} and~\eqref{eq_inforderX}, satisfy the Weyl relations,
\begin{equation}\label{eq:lattice_weyl_finite}
    Z_{v,i}X_{v,i} = e^{ \frac{2 \pi i}{D_i}}X_{v,i}Z_{v,i}, \qquad 0 \leq i < p\,,
\end{equation}
for the finite order charges, and 
\begin{equation}\label{eq:lattice_weyl_infinite}
    Z_{v,j}(\theta_{v,j})X_{v,j} = e^{i \theta_{v,j}} X_{v,j}Z_{v,j}(\theta_{v,j}), \qquad p \leq j < r
\end{equation}
for the infinite order ones. This relation expresses the dual pairing between translations of the occupation number and the action of gauge transformations on the bosonic matter. 

The operators $X$ implement shifts of the occupation number
$n$, while the operators 
$Z$ realize phases that depend on $n$. 
In both cases, although these phases are defined with respect to the occupation number variable, they will shortly be related to the characters of the gauge group through Gauss’s law, Eqs.~\eqref{eq_bosonicGauss} and~\eqref{eq_rhototbos}, which ties the matter charges to the gauge generators. In this way, the resulting error syndromes become character-valued, as desired. Specifically,  the $Z_{v,i}$ may be used in syndrome measurements.

Let us consider again matter of infinite character order. While we noted that we cannot implement $X_{v,j}$ as a unitary operator, its Hermitian conjugate
\begin{equation}
    X_{v,j}^{\dagger} = (N_{v,j} + 1)^{-\frac{1}{2}} a_{v,j}
\end{equation}
acts as a lowering operator and annihilates the matter vacuum. Nevertheless, $X_{v,j}$ is an isometry,
\begin{equation}
    X_{v,j}^{\dagger}X_{v,j} = \mathbbm{1}_{v,j}, \qquad X_{v,j}X_{v,j}^{\dagger} = \mathbbm{1}_{v,j} - \ket{0}\bra{0}_{v,j}.
\end{equation}
Operators with this property are familiar in bosonic quantum error correction, where photon-loss noise is generated by powers of the annihilation operator $a$ \cite{Michael_2016,Albert_2018}. 

We may then use $X_{v,j}^\dag$ as a stand in for the missing $X_{v,j}^{-1}$. Indeed, in the present setting, the oscillator degrees of freedom may be interpreted as physical systems used to simulate the matter fields of the gauge theory. In the simulation it is not strictly necessary to encode the vacuum of type $j$ of the simulated gauge theory also in terms of the vacuum of the physical oscillator. One may simply encode the vacuum as $\ket{n_0}$ with $n_0>0$ in the physical system. An application of $X^{\dagger}$ on the physical system then lowers the occupation from $n_0$ to $n_0 - 1$, without annihilating the state of the oscillator. Since the kernel of $X^{\dagger}$ occurs only at the vacuum of the oscillator used for the simulation, it lies outside the relevant subspace for the simulation, which is spanned by states with $n\geq n_0$. In this manner, one may view $X^\dag$ legitimately as an error in the simulation, motivating to include both $X_{v,j}$ and $X_{v,j}^{\dagger}$ among the candidate matter errors. 

In conjunction, it is convenient to assemble these candidate matter errors into strings acting on all vertices and matter species,
\begin{equation}
  Z(\bm{z},\bm{\theta}) X^{\bm{x}} = \prod_{v \in \mathcal{V}} \left(   \prod_{0 \leq i < p} Z_{v,i}^{z_{v,i}}X_{v,i}^{x_{v,i}}    \right)
    \left(  \prod_{p \leq j < r} Z_{v,j}(\theta_{v,j}) X_{v,j}^{x_{v,j}}    \right)
\end{equation}
where
\begin{equation}
    \bm{x} \in \mathcal{X} \coloneqq \{ \bm{x} \in \mathbb{Z}^{N_V r} \mid x_{v,i} \in \mathbb{Z}_{D_i}, 0 \leq i < p   \},
\end{equation}
\begin{equation}
    \bm{z} \in \mathcal{Z} \coloneqq \{ \bm{z} \in \mathbb{Z}^{N_V p} \mid z_{v,i} \in \mathbb{Z}_{D_i}    \}, \qquad \bm{\theta} \in [0, 2\pi)^{N_V(r-p)}.
\end{equation}
For clarity, whenever $x_{v,j}$ is negative, we mean $X_{v,j}^{x_{v,j}}=(X_{v,j}^\dag)^{|x_{v,j}|}$.

Together with the Wilson line operators discussed above, these provide the natural candidates for errors in the theory. The most general operator we need to consider therefore takes the form $Z(\bm{z},\bm{\theta}) X^{\bm{x}}\,W^{\bm{\chi}}$. In App.~\ref{app_bosKL}, we show that any pair of such composite errors $E_a = Z(\bm{z}_a,\bm{\theta}_a)X^{\bm{x}_a} W^{\bm{\chi}_a}$, $E_b = Z(\bm{z}_b,\bm{\theta}_b)X^{\bm{x}_b} W^{\bm{\chi}_b}$ satisfies
\begin{equation}\label{bosonicKL}
    \Pi_{\text{pn}}E_a^{\dagger}E_b \Pi_{\text{pn}} \propto Z(\bm{z}_b - \bm{z}_a, \bm{\theta}_b - \bm{\theta}_a) X_{\rm dr}^{\bm{x}_b - \bm{x}_a} \,H^{(\bm{\bar{\chi}_a\chi_b})_S} \delta_{(\partial\bm{\chi}_a)\bm{\rho}[\Delta\bm{x}_a],(\partial\bm{\chi}_b)\bm{\rho}[\Delta\bm{x}_b]}\Pi_{\text{pn}}\,.
\end{equation}
In this expression, we have the extension of the Gauss law map, Eq.~\eqref{boundarymap}, to the inclusion of bosonic matter (and $v_0$)\footnote{Note that, unlike Eq.~\eqref{boundarymap}, the right hand side takes value in $\hat{G}^{N_V}$ since matter charges at $v_0$ are independent of those at other vertices when the state is not gauge-invariant.}
\begin{align}
  &  \partial_B:\hat{G}^{N_L+N_V}\to\hat{G}^{N_V}\nonumber\\
    &\partial_B(\bm{\chi},\bm{\rho})=(\partial\bm\chi)\bm{\rho}=((\partial\chi_v)\rho_v)_{v\in\mathcal{V}}\,,\label{eq_bosGausslawmap}
\end{align}
collecting the Gauss law data, Eq.~\eqref{eq_bosonicGauss}, across the vertices, and
\begin{equation}\label{eq_speciescontrib}
    \bm{\rho}[\Delta\bm{x}]=\left(\prod_{\alpha\in A}\rho_\alpha^{(\Delta_\alpha \bm{x})_v}\right)_{v\in\mathcal{V}}
\end{equation}
is the collection of $X$-error induced matter excitation charges $\rho_v$ in Eq.~\eqref{eq_rhototbos} across the vertices. Here, we take the product over the set of \emph{species}, $A = \{  \alpha = i, (j, \bar{j}) \mid 0 \leq i < p, p \leq j, \bar{j} < r  \}$, combining particle and antiparticle as a joint index for infinite order matter type, and the operator $\Delta$, which already appeared in the equivalence classes in Eq.~\eqref{eq_coarsegrain}, extracts the particle number data relevant for the Gauss law, distinguishing finite and infinite order charges:
 \begin{equation}\label{eq_Delta}
   (\Delta_i \bm{x})_v = x_{v,i} , \qquad 0 \leq i < p\,,\qquad\qquad(\Delta_{j \bar{j}} \bm{x})_v = x_{v,j} - x_{v, \bar{j}}, \qquad p \leq j, \bar{j} < r\,.
\end{equation}
Furthermore, we have the dressed $X$-operators
\begin{equation} \label{eq:fatrho}
    X_{\rm dr}^{\bm{x}_b - \bm{x}_a} \coloneqq \sum_{\bf{n}\in\widetilde{\mathcal{N}}} \ket{{\bf n} + \bm{x}_b - \bm{x}_a} \!\bra{{\bf n}}^{\text{dr}}_{\text{matter}}.
\end{equation}
This provides the  extension of Eq.~\eqref{eq_KLpure} to the inclusion of bosonic matter. With this formula in hand, the structure of correctable error sets can be read off. From  Eq.~\eqref{bosonicKL}, we see that the conjugation with the code projector vanishes unless the two errors induce the same Gauss law violation: 
\begin{equation}\label{eq_sameGaussviobos}
   \partial\bm{\chi}_a\bm{\rho}[\Delta\bm{x}_a]=\partial\bm{\chi}_b\bm{\rho}[\Delta\bm{x}_b]\,.
\end{equation}
To avoid nontrivial logical operations in that case, we would need to have that $(\bar{\bm{\chi}}_a\bm{\chi}_b)_S=1_S$, $\bm{x}_a=\bm{x}_b$, as well as $\bm{z}_a=\bm{z}_b$ and $\bm{\theta}_a=\bm{\theta}_b$. In analogy to the pure gauge theory case, the conjunction of these conditions means, however, that the errors coincide, $E_a=E_b$. 
Thus, the only way for two distinct errors in the set $\{Z(\bm{z},\bm{\theta}) X^{\bm{x}}\,W^{\bm{\chi}}\}$ to obey the Knill-Laflamme condition and so be correctable is if their respective violations of the Gauss laws differ, so that the Kronecker delta vanishes.

Let us consider the structure of correctable error sets in terms of the extended Gauss law map in Eq.~\eqref{eq_bosGausslawmap}. Just like the pure gauge Gauss law map, Eq.~\eqref{boundarymap}, $\partial_B$ defines a discrete fiber bundle with base space $\hat{G}^{N_V}$ and fibers $F\simeq\hat{G}^{N_L}$.\footnote{This can be seen as follows. Consider any fixed electric data configuration $\partial_B(\bm\chi,\bm\rho)=\bm{q}$. At $v\neq v_0$ we can choose the electric data for the links in $S$ and the total matter charge $\rho_v$ arbitrarily. The desired electric configuration $q_v$ at $v$ is then obtained by fixing the frame Wilson line hitting $v$ to the unique solution $\chi_{R_v}=q_v\bar{\bm\chi}_S^v\bar\rho_v$. After this step, the electric link data at $v_0$ is fixed, so that also the matter charge is then fixed to $\rho_{v_0}=q_{v_0}(\partial\bar{\chi})_{v_0}$ to obtain the desired $q_{v_0}$. In conjunction, since there are $N_L-N_V+1$ links in the system and $N_V-1$ vertices where  the matter charge can be freely chosen, we have that for each fixed electric charge configuration $\bm{q}$ the redundancy in obtaining it is parametrized by $F\simeq\hat{G}^{N_L}$.} However, this bundle structure is too coarse to be applied to errors of bosonic Gauss law codes and to parametrize  maximal sets of correctable errors because it does not encode how the total error induced matter charge contribution $\rho_v$ at any vertex decomposes into species, as in Eq.~\eqref{eq_rhototbos}. Let us thus fine grain the bundle by appropriately extending the fibers to include information about $X$- and $Z$-errors
\begin{align}
   & \partial_B':\mathcal{Q}\to\hat{G}^{N_V}\qquad\qquad\qquad\qquad\mathcal{Q}\coloneqq\hat{G}^{N_L}\times\mathcal{Z}\times [0,2\pi)^{N_V(r-p)}\times\mathcal{X}\nonumber\\
   &\partial_B'(\bm\chi,\bm{z},\bm{\theta},\bm{x})\coloneqq(\partial\bm\chi)\bm\rho[\Delta\bm{x}]\,.\label{eq_extbosGaussbundle}
\end{align}
This again defines a fiber bundle with the same base space, where the projection $\partial_B'$ forgets about the $Z$-error data and maintains only the $X$-error data passed through the $\Delta$ operator. The fibers of this bundle are isomorphic to 
\begin{equation} 
F'\simeq \hat{G}^{N_L-N_V}\times\mathcal{Z}\times[0,2\pi)^{N_V(r-p)}\times\mathcal{X}\simeq F\times\mathcal{Z}\times[0,2\pi)^{N_V(r-p)}\times \mathbb{Z}^{N_V(r-p)/2}\,.
\end{equation}
The right hand side clarifies how the fibers $F$ of $\partial_B$ become extended by the entirely redundant $Z$-error data and the $X$-error data that becomes eliminated by $\Delta$. Indeed, for any fixed $v$, when the total matter charge $\rho_v$ is fixed in Eq.~\eqref{eq_speciescontrib}, the contributions from the finite order bosons are uniquely fixed too, while the contribution for each infinite order species is fixed only up to a redundancy parametrized by $\mathbb{Z}$.

This extended Gauss law bundle $\partial_B'$ finally has sufficient structure to parametrize maximal sets of correctable errors. Consider thus an error set $\mathcal{E}_{\mathcal{I}}=\{Z(\bm{z},\bm{\theta}) X^{\bm{x}}\,W^{\bm{\chi}}\}_{(\bm\chi,\bm{z},\bm\theta,\bm{x})\in \mathcal{I}\subset\mathcal{Q}}$. Our discussion below Eq.~\eqref{eq_sameGaussviobos} implies that this set is correctable if and only if $\mathcal{I}$ contains at most one representative in each fiber $F'$ of the extended Gauss law map $\partial_B'$. Similarly, for $\mathcal{E}_{\mathcal{I}}$ to be maximal, $\mathcal{I}$ must contain exactly one representative in each fiber $F'$. Maximal correctable sets are therefore obtained by choosing a global section  of $\partial_B'$, 
\begin{equation}
    s_B': \hat{G}^{N_V}\to \mathcal{Q}
\end{equation}
\begin{equation}
  s_B'(\bm{q}) = (\bm{\chi}_B(\bm{q}),\bm{z}_B(\bm{q}), \bm{\theta}_B(\bm{q}), \bm{x}_B(\bm{q}))\,.
\end{equation}
The corresponding maximal set is 
\begin{equation}
    \mathcal{E}_{s'_B} \coloneqq \{Z( \bm{z}_B(\bm{q}), \bm{\theta}_{B}(\bm{q}))X^{\bm{x}_{B}(\bm{q})}W^{\bm{\chi}_{B}(\bm{q})} \mid \bm{q} \in \hat{G}^{N_V} \}\,.
\end{equation}
Such sets are maximal because they exhaust all possible violations of the Gauss law, Eq.~\eqref{eq_bosonicGauss}, across the lattice, thereby encompassing maps to all possible charge spaces of the gauge transformations, similarly to Eq.~\eqref{eq_chargedecomp}.
Changing the section -- and thus the maximal set -- amounts to multiplying by elements in the kernel of $\partial_B'$, which correspond to logical operators. Thus, in analogy to the pure gauge theory, maximal correctable sets in the bosonic theory are naturally labeled by global sections of the extended Gauss law map $\partial_B'$.

Finally, it is clear that, akin to the pure gauge case, the syndrome of each error is the total (gauge and matter) electric charge at each vertex $v\in\mathcal{V}$, which can be measured through any generating set of the gauge transformations $U^{\bm{g}}$. This information can only distinguish errors associated within a given section.
Thus, we have once more that, in an error context, the input data of the extended Gauss law map are error labels, while its output are the corresponding syndromes.

\subsubsubsection{Code parameters}\label{sssec:boscodepars}
In contrast to the pure gauge case, the bosonic Gauss law code contains two types of information-carrying registers: $G$-registers on links and $\mathcal{N}$-registers on vertices. It is therefore natural to express the code parameters as ordered pairs referring to these two sectors. The total number of physical registers is $n = (N_L, N_V)$. The logical degrees of freedom follow from the code space factorization together with the constraint Eq.~\eqref{eq_bostotalchargezero}. This removes one independent matter degree of freedom, leaving $N_V - 1$ logical $\mathcal{N}$-registers, while the loop sector contributes $N_L - N_V + 1$ logical $G$-registers. Accordingly, $k = (N_L - N_V + 1,\, N_V - 1)$.

We distinguish between $X$-type errors---matter shift operators and bare Wilson lines---and $Z$-type errors---matter phase operators and electric operators $U_\ell^g$, and denote the corresponding distances by $d_{X'}$ and $d_{Z'}$ to emphasize that both matter and gauge contributions are included. Since single-site $Z$-type operators already lie in the logical algebra, one has $d_{Z'} = 1$. By contrast, nontrivial logical $X$-type operators arise either from closed Wilson loops or from dressed $X$-operators in Eq.~\eqref{eq:fatrho}, so that $d_{X'}$ is given by the minimal weight among these operators. This depends on the specific geometry of the lattice, but we always have $d_{X'} = 2$ or $3$. Finally, since the stabilizer is generated entirely by $Z$-type operators, there are no $X$-type stabilizers. The code is therefore a generalized $[(N_L, N_V), (N_L - N_V + 1, N_V - 1), d_{X'}]$ (classical) CSS code, with $d_{Z'} = 1$.

\subsubsection{Gauss law codes with fermionic matter as subsystem codes}\label{sssec_fermGauss}
We now repeat this analysis for a theory with fermionic matter, focusing on the features that differ compared with the bosonic case. Consider again an Abelian lattice gauge theory with compact structure group $G$ in $d+1$ dimensions and periodic boundary conditions, but now introduce a single species of staggered fermions \cite{Kogut:1974ag,Susskind:1976jm} on the lattice.\footnote{Arguments of this section and section~\ref{sec:fermi-vac} can be extended to other types of lattice fermions such as na\"ive fermions, Wilson fermions \cite{Wilson:1974sk} and overlap fermions \cite{Neuberger:1997fp} (See \cite{Hayata:2023zuk,Yamamoto:2023qaa} for construction in operator formalism) up to some technical modifications. Compared with the staggered fermions, the other types of lattice fermions have vanishing $c_v$ defined in \eqref{eq:offset}. For the case of na\"ive fermions and overlap fermions, this point is only the relevant difference from the staggered fermions here while they have different Hamiltonian and symmetries not discussed in this paper. Wilson fermions have a different number of degrees of freedom on the lattice from the staggered fermions depending on dimensions. Note also that continuum limits generally depend on various data such as the type of lattice fermion, Hamiltonian, and the shape of the lattice mainly due to the doubling problem of chiral fermions \cite{Nielsen:1981hk}. For the case of the staggered fermions in $d$-dimensional square lattice, its continuum limit has $2^d$ complex degrees of freedom corresponding to a single Dirac fermion for $d=1$ and two Dirac fermions for $d=2,3$.} This class includes the $\mathbb{Z}_N$ lattice gauge theory coupled to a fermion discussed for prime $N$ in \cite{Spagnoli:2026qni,Spagnoli:2024mib}.

This means that we have one species of fermions and antifermions that sit at alternating lattice sites, which, due to the periodic boundary conditions, requires $N_V$ to be even.\footnote{Note that in the special case $G=\mathbb{Z}_2$ we have that each fermion is its own antifermion.} The paradigmatic example of such a theory is lattice QED whose $1+1$-dimensional case corresponds to the so-called Schwinger model \cite{Schwinger:1962tp} and has been recently studied much in the context of digital quantum simulation (e.g.~\cite{Martinez:2016yna,Muschik:2016tws,Klco:2018kyo,Kokail:2018eiw,Chakraborty:2020uhf}). We will investigate the $2+1$-dimensional lattice QED in more detail in section~\ref{ssec_Schwinger}.

\subsubsubsection{Fermionic vertex Hilbert spaces}
At each vertex $v$ we place a two–dimensional Hilbert space $\mathcal{H}_v = \mathbb{C}^2$. We interpret this space as the occupation space of a single fermionic mode and choose the $Z$-basis as the occupation number basis,
\begin{equation}
    \mathcal{H}_v = \text{Span} \{ \ket{n}_v \mid n=0,1   \}.
\end{equation}
The corresponding number operator is
    \begin{equation}
        N_v = \frac{\mathbbm{1}_v - Z_v}{2}\,,
    \end{equation}
so that the eigenvalues $n=0,1$ represent an empty or occupied site, respectively.

To couple these modes to the gauge field, we follow the usual staggered fermion construction \cite{Kogut:1974ag,Susskind:1976jm} familiar from the lattice Schwinger model \cite{Banks:1975gq,Hamer:1997dx}. 
The gauge group acts locally on each vertex through a unitary representation 
(cf.~Eqs.~\eqref{generalGTs} and~\eqref{eq_gaugegrouprepbos})
    \begin{equation}\label{eq_fermgaugetr}
        u_v^g = {\chi_F}^{N_v - c_v}(g)\,,\qquad\qquad g\in\,G\,,
    \end{equation}
where $\chi_F\in\hat{G}$ is some character in the Pontryagin dual of $G$ determining the electric charge of the fermions. When $G$ is cyclic and the fermion is minimally charged as in the Schwinger model, $\chi_F$ will be a generator of $\hat{G}$, though in general this is not strictly necessary. The offset
\begin{equation}
 c_v = \frac{1 - (-1)^{|v|}}{2}
\label{eq:offset}
\end{equation}
distinguishes even and odd sites of the lattice. Here, $|v|$ denotes the lattice parity of the vertex. In the case of a cubic lattice, writing the position vector $\bm{v}$ of a vertex $v$ in the basis of primitive lattice vectors $\{ \bm{e_{\mu}}  \}_{\mu = 1}^{d}$, we have
\begin{equation}\label{cubicparity}
|v| \coloneqq \sum_{\mu = 1}^d \frac{\bm{v} \cdot \bm{e}_{\mu}}{|\bm{e}_{\mu}|},
\end{equation}
which simply counts how many lattice steps separate the vertex from the origin. Even values correspond to one sublattice and odd values to the other. 

As a result, even vertices host fermions carrying charge $\chi_F$, whereas odd vertices carry antifermions with charge $\bar{\chi_F}$. Local creation and annihilation operators are introduced in the usual way from Pauli operators,
\begin{equation}\label{psiPauli}
    \psi_v = \frac{X_v + i Y_v}{2}, \qquad
    \psi_v^{\dagger} = \frac{X_v -iY_v}{2}.
\end{equation}
These operators move between the two occupation states and therefore change the charge stored at the vertex, and we have $N_v=\psi_v^\dag\,\psi_v$, as usual. 

To obtain the correct fermionic statistics, however, it is important that the vertex Hilbert spaces are not combined with an ordinary tensor product. Instead, the full matter Hilbert space is the $\mathbb{Z}_2$-graded tensor product \cite{Saberi_grading, Bravyi_grading,Szalay:2020xbz}, which we review in appendix~\ref{app_fermions},
\begin{equation}\label{eq_fermHmatter}
    \mathcal{H}_{\text{matter}} = \widehat{\bigotimes_{v \in \mathcal{V}}} \,\mathcal{H}_v.
\end{equation}
The grading keeps track of fermion parity. On each site the parity operator is defined by
\begin{equation}
    P_v = (-1)^{N_v},
\end{equation}
so empty states with even occupation have parity $+1$ and states with odd occupation have parity $-1$. An operator 
$A_v$ acting on $\mathcal{H}_v$
is called even or odd if it transforms under this parity operator with a definite sign, that is
\begin{equation}\label{eq_parity}
    P_vA_vP_v = (-1)^{|A_v|}A_v,
\end{equation}
where $|A_v| = 0,1$ is the fermion parity of the operator. Operators with $|A_v| = 0$ are even and preserve the fermion number $\!\!\!\!\mod 2$, while operators with $|A_v| = 1$ are odd and change it.

The graded tensor product encodes the rule that odd operators anticommute across different sites. Indeed, if $A_v$ and $B_{v'}$ are operators with definite fermion parities $|A_v|$ and $|B_{v'}|$, then for $v \neq v'$,
\begin{equation}\label{eq_oddevencommute}
    A_v B_{v'} = (-1)^{|A_v||B_{v'}|} B_{v'}A_v.
\end{equation}
Since the operators $\psi_v$ and $\psi_{v'}^{\dagger}$ are odd with respect to this grading, the graded tensor product immediately implies the canonical anticommutation relations: 
\begin{equation}\label{eq_fermanticommute}
    \{ \psi_v, \psi^{\dag}_{v'}  \} = \delta_{v,v'}\,, \qquad \{ \psi_v, \psi_{v'}   \} = \{ \psi_v^{\dagger}, \psi_{v'}^{\dagger}    \} = 0\,.
\end{equation}

The action of gauge transformations on these operators is particularly simple. Conjugating with the on-site representation $u_v^g$ in Eq.~\eqref{eq_fermgaugetr} gives 
\begin{equation}
    \psi_v \to u_v^g \psi_v u_v^{g \dagger} = \bar{\chi}_F(g)\psi_v
\end{equation}
\begin{equation}
    \psi_v^{\dagger} \to u_v^g \psi_v^{\dagger} u_v^{g \dagger} = \chi_F(g) \psi_v^{\dagger}\, .
\end{equation}
Thus, the creation operator carries charge $\chi_F$ and the annihilation operator carries the conjugate charge, exactly as expected for fermionic matter fields. In particular, this means that $\psi_v$ and $\psi^\dag_v$ act as annihilation and creation operators, respectively, at even sites, and vice versa at odd sites. A consequence of this feature is that the occupation number does not directly measure the charge at $v$, which instead is measured by $n_v-c_v$, cf.~\eqref{eq_fermgaugetr}. For example, a peculiarity of the staggered model is that the total vacuum with zero charge is given by $n_v=0$ at all even and $n_v=1$ at all odd sites, meaning maximal occupation of antifermions and zero occupation of fermions across the lattice.

\subsubsubsection{Code space and logical algebra from relational observables}
As in the bosonic case, the total kinematical Hilbert space is again of the form in Eq.~\eqref{eq_Hkinbosfac}, except that $\Hil_{\rm matter}$ is now given by the graded Hilbert space in Eq.~\eqref{eq_fermHmatter}. In analogy to Eq.~\eqref{eq_bosbasis}, we may thus work with the product basis
\begin{equation}
    \ket{\bm{g}_S, \partial\bm{\chi}, {\bf n}} = \ket{\bm{g}_S}_{\text{loops}} \otimes_R \ket{\partial\bm{\chi}}_R \otimes \ket{\bf{n}}_{\text{matter}}\,,\qquad\qquad\ket{\bf{n}}_{\rm matter}=\widehat{\otimes}_{v\in\mathcal{V}}\ket{{n}_v}_v\,,\qquad{\bf n}\in\mathbb{Z}^{N_V}_2\,,
\end{equation}
in which gauge transformations are again diagonal. 

Proceeding as in the bosonic case in section~\ref{ssssec_boscodespace}, we find that gauge invariance at $v\neq v_0$ entails
\begin{equation}\label{eq_ferminvatv}
    \partial\chi_v^{(\bf n)}=\bar{\chi}_F^{n_v-c_v}\,,
\end{equation}
so an odd $v$ is charged when the site is unoccupied, and vice versa, while gauge invariance at the root $v_0$ imposes the global neutrality condition on the fermionic matter
\begin{equation}
    \prod_{v \in \mathcal{V}} \chi_F^{n_v - c_v} = 1,
\end{equation}

For periodic boundary conditions one has $\sum_{v \in \mathcal{V}} c_v = N_V/2$, and thus gauge-invariant states satisfy 
\begin{equation}
    \sum_{v \in \mathcal{V}} n_v = \frac{N_V}{2} \mod D_F\,,
    \label{global_neutrality_fermions}
\end{equation}
where $D_F$ is the order of $\chi_F$. Thus, if $\chi_F$ has infinite order, exactly half of the vertices must be occupied. Note that this could be all fermions, all antifermions, or a mix of the two; either way, the total charge vanishes because antifermion occupation is charge-neutral, whereas an empty antifermion site is negatively charged. In particular, as noted above, the total charge vacuum corresponds to full antifermion occupation and empty fermion sites. By contrast, when $D_F<\infty$, we have distinct charge sectors in analogy to Eq.~\eqref{eq_chargesector}
\begin{equation}\label{eq_fermchargesector}
\sum_{v\in\mathcal{V}}=\frac{N_V}{2}+m\,D_F\,,\qquad\qquad m=0,1,\ldots,M\,,
\end{equation} 
with $M$ finite when $|\mathcal{V}|$ is finite. Hence, in the finite order case, the total occupation number can be larger than half of the vertices.

Provided Eqs.~\eqref{eq_ferminvatv} and~\eqref{global_neutrality_fermions} hold, all local Gauss constraints have been solved explicitly, and the physical content reduces to loop degrees of freedom together with dressed matter configurations obeying the global neutrality condition.
It follows that the code space in the guise of the perspective–neutral Hilbert space again admits a subsystem decomposition:
\begin{equation}
    \mathcal{H}_{\text{pn}} = \text{Span} \left\{  \ket{\bm{g}_S}_{\text{loops}} \otimes_R \ket{\bf{n}}^{\text{dr}}_{\text{matter}} \,\Big|\, \bm{g}_S \in G^{ N_L - N_V + 1}, {\bf n} \in \mathbb{Z}_2^{ N_V}, \sum_{v \in \mathcal{V}}n_v = \frac{N_V}{2} \mod D_F \right\}\,,\label{eq_Hpnferm}
\end{equation}
where the dressed matter states are
\begin{equation}
\ket{\bf{n}}^{\text{dr}}_{\text{matter}} = \ket{\partial\bm{\chi}^{({\bf n})}}_R \otimes \ket{\bf{n}}_{\text{matter}}\,,\label{eq_fermdrstates}
\end{equation}
so that 
\begin{equation}
    \mathcal{H}_{\text{pn}} = \mathcal{H}_{\text{loops}} \otimes_R \mathcal{H}^{\text{dr}}_{\text{matter}}\,,
\end{equation}
and we have a subsystem code \cite{kribs2006, Poulin_2005} once more. Similarly to the bosonic case, the loop and dressed fermionic matter degrees of freedom furnish complementary subsystems that can assume the role of ``gauge subsystem'' for one another in the sense of subsystem quantum error correction.

As in the bosonic case, the logical algebra $\mathcal{A}_{\text{pn}}= \mathcal{B}(\mathcal{H}_{\text{pn}})$ thus factorizes as
\begin{equation}
    \mathcal{A}_{\text{pn}} = \mathcal{A}_{\text{loops}} \otimes_R \mathcal{A}_{\text{matter}}^{\text{dr}},
\end{equation}
 with the loop algebra once more given by Eq.~\eqref{Apnpure}. On the other hand, in appendix~\ref{app_fermalg}, we explain that the QRF-dressed fermionic matter algebra differs slightly between  the cases when
 \begin{itemize}
     \item[(i)] $\chi_F$ is of infinite order (and so $G$ is not finite):
\begin{equation}\label{eq_inffermalg}
    \mathcal{A}^{\text{dr}}_{\text{matter}} = \left\langle \psi_v\,{W}_{\gamma_R[v, v']}^{\chi_F} \,\psi^\dag_{v'}\Pi_{\text{pn}}\, \Big|\enspace v, v' \in \mathcal{V}      \right\rangle\,.
\end{equation}
     \item[(ii)] $\chi_F$ is of finite order, $D_F<\infty$:
\begin{equation}\label{eq_finfermalg}
    \mathcal{A}^{\text{dr}}_{\text{matter}} = \left\langle \psi_v\,{W}_{\gamma_R[v, v']}^{\chi_F} \,\psi^\dag_{v'}\Pi_{\text{pn}}\,,S_{\{v_k\}}\,\Pi_{\rm pn}\,,S^\dag_{\{v_k\}}\,\Pi_{\rm pn}\,\Big|\enspace v, v' \in \mathcal{V}      \right\rangle\,,
\end{equation}
where $S_{\{v_k\}}$ is any fixed star operator centered at $v_0$ (cf.~Eq.~\eqref{eq_star} for the bosonic counterpart)
\begin{equation}\label{eq_fermstar}
    S_{\{v_k\}}=\prod_{k=1}^{D_F}\,\psi_{v_k}\,W^{\chi_F}_{R_{v_k}}\,,
\end{equation}
and the $\{v_k\}$ are \emph{any} $D_F$ distinct vertices.\footnote{Due to $\psi_v^2=0$ the $v_k$ must be distinct, though one of them may coincide with $v_0$.} Thanks to $\chi_F^{D_F}=1$, this operator is gauge-invariant.  The star operator assumes the role that the particle-antiparticle combination $a_{v,\bar{j}}^{\dagger} W_{\gamma_R[v,v']}^{\rho_j} a_{v',{j}}^\dagger$ had for finite order species in the bosonic algebra~\eqref{eq_drmatteralg}. Indeed, without the star operator, $\mathcal{A}^{\rm dr}_{\rm matter}$ would be superselected with respect to the charge sectors in Eq.~\eqref{eq_fermchargesector} because the other generators $\psi_vW^{\chi_F}_{\gamma_R[v,v']}\psi^\dag_{v'}$ leave $\sum_v n_v$ invariant. By contrast, $S_{\{v_k\}}^\dag$ changes the total occupation number by $+D_F$ when the sites $\{v_k\}$ are empty, $n_{v_k}=0$, and $S_{\{v_k\}}$ lowers it by $D_F$ in the opposite case. 

As we explain in appendix~\ref{app_fermalg}, using elementary moves involving the bipartite  $\psi_v W^{\chi_F}_{\gamma_R[v,v']}\psi_{v'}^\dag$, one can map any such star operator into any other one, including those centered at other vertices, and that different star operators can also be connected. For this reason, the inclusion of a single star operator (and its conjugate) among the generators of Eq.~\eqref{eq_finfermalg} suffices.
\end{itemize}

\subsubsubsection{Maximal correctable error sets and syndromes} \label{sec:fermimaxsets}

We now turn to correctable errors. As before, Wilson line products are correctable by the same arguments as in the pure gauge case, so Proposition~\ref{prop_KLpure} continues to hold. In addition, we must consider matter excitations. Although the fermionic operators $\psi_v$ and $\psi_v^{\dagger}$ are not unitary, they are linear combinations of Pauli operators,
\begin{equation}
    X_v = \psi_v + \psi_v^{\dagger},
\end{equation}
\begin{equation}
    Y_v = i(\psi_v^{\dagger} - \psi_v).
\end{equation}

Acting on a single vertex, $X_v$ and $Y_v$ are individually detectable, but they are not simultaneously correctable. Indeed, their product satisfies $X_vY_v \propto Z_v$, and $Z_v$ is gauge invariant, implementing a logical phase flip on the dressed matter subsystem. Consequently, if two errors differ only by exchanging $X_v$ and $Y_v$ at a single site while agreeing everywhere else, their product acts as a logical operator and the Knill--Laflamme condition cannot be satisfied.

Before proceeding, it is important to note that the operators $X_v$ do not transform with a fixed character under gauge transformations. Acting with $X_v$ switches the local fermion occupation number, so the charge created at $v$ depends on the initial state: if the vertex is initially unoccupied the excitation carries charge $\chi_F$, while if it is initially occupied the resulting hole carries charge $\bar{\chi}_F$. Thus, the charge associated with an $X_v$ excitation is state dependent,
\begin{equation}\label{eq_fermX}
    U^{g}_v X_v U^{g \dagger}_v = \chi_F^{1 - 2N_v}(g)X_v\,.
\end{equation}
Only when $\chi_F$ has order two, so that $\chi_F=\bar{\chi}_F$, does $X_v$ create a definite charge independent of the state.

It is convenient to organize matter operators directly in terms of Pauli strings
\begin{equation}
    Z^{\bm{z}}X^{\bm{x}} \coloneqq 
    \widehat{\bigotimes_{v \in \mathcal{V}}} 
    Z_v^{z_v} X_v^{x_v},
    \qquad 
    \bm{x}, \bm{z} \in \mathbb{Z}_2^{N_V}.
\end{equation}
Following the steps for bosonic Gauss law codes, we now consider two composite operators
\begin{equation}\label{eq_fermerrors}
    E_a = Z^{\bm{z}_a}X^{\bm{x}_a}W^{\bm{\chi}_a},
    \qquad 
    E_b = Z^{\bm{z}_b}X^{\bm{x}_b}W^{\bm{\chi}_b}.
\end{equation}
Inserting them into the Knill--Laflamme condition gives
\begin{equation}\label{eq_fermionicKL}
    \Pi_{\text{pn}}E_a^{\dagger}E_b \Pi_{\text{pn}} \propto 
    Z^{\bm{z}_a \oplus \bm{z}_b}
    X_{\rm dr}^{\bm{x}_a \oplus \bm{x}_b}
    H^{(\bar{\bm{\chi}}_a\bm{\chi}_b)_S}
    \left( 
    \prod_{v \in \mathcal{V}} 
    \delta_{(\partial \chi_a)_v\chi_F^{ x_{a,v}(1-2N_v)},\, (\partial \chi_b)_v\chi_F^{ x_{b,v}(1-2N_v)}} 
    \right)
    \Pi_{\text{pn}},
\end{equation}
where $\oplus$ denotes bit-wise addition modulo $2$ and we have the dressed $X$-operator
\begin{equation}
    X_{\rm dr}^{\bm{x}_a \oplus \bm{x}_b} 
    \coloneqq 
    \sum_{{\bf n} \in \mathbb{Z}_2^{N_V}\cap\,\text{Eq.~\eqref{global_neutrality_fermions}}} 
    \!\ket{{\bf n} \oplus \bm{x}_a \oplus \bm{x}_b }
    \!\bra{\bf{n}}^{\text{dr}}_{\text{matter}}.
\end{equation}

The arguments of the Kronecker delta contain the (gauge-invariant) number operator $N_v$, reflecting the fact that the charge created by $X_v$ depends on the initial occupation of the vertex. Consequently, the delta is generally an operator that should be understood via spectral decomposition. The Knill-Laflamme condition for the \emph{full} code space can therefore only be satisfied if the delta evaluates to the same value for every occupation configuration. 
The only way for the delta to ``click'' on the \emph{entire} code space is if
\begin{equation}
    \partial\bm\chi_a=\partial\bm\chi_b\,,\qquad\qquad\text{and}\qquad\qquad \bm{x}_a=\bm{x}_b\,.
\end{equation}
In that case, the only way to avoid a nontrivial logical operation on the right hand side of Eq.~\eqref{eq_fermionicKL} is if also $({\bm\chi}_a)_S=({\bm{\chi}_b})_S$ and $\bm{z}_a=\bm{z}_b$. But this again entails that $E_a=E_b$. Hence, the only way for the Knill-Laflamme condition for distinct errors to be obeyed on the entire code space  is if the delta vanishes everywhere on it. This, in turn, is only possible if at any $v\in\mathcal{V}$ either
\begin{align}
   & x_{a,v}=x_{v,b}\qquad\qquad\text{and}\qquad\qquad (\partial{\chi}_a\partial{\bar{\chi}}_b)_v\neq1\,,\nonumber\\
   \text{or}\qquad\qquad
&x_{a,v}\neq x_{b,v}\qquad\qquad\text{and}\qquad\qquad(\partial{\chi}_a\partial{\bar{\chi}}_b)_v\neq\chi_F,\bar{\chi}_F\,.\label{eq_fermGaussdisterrors}
\end{align}
Note that in the latter case, $(\partial\chi_a)_v=(\partial\chi_b)_v$ is permitted, so the errors would only differ by whether or not an $X$-error occurred at site $v$.

This renders the error analysis more subtle than in all previous cases. As we shall now explain, he notion of maximal correctable subsets among those parametrized by Eq.~\eqref{eq_fermerrors} for fermionic Gauss law codes depends on whether and where we admit $X$-type errors to be included, though the cardinality of these sets will remain the same in some cases.  Similarly, the permissible syndrome measurements will now depend on which vertices are permitted to feature $X$-errors. 

Let us begin with the case that we admit \emph{no} $X$-errors at all, so $\bm{x}=0$ for all errors considered. Clearly, this case is covered by the pure gauge Gauss law map~\eqref{boundarymap}  (trivially extended by the $Z$-error data) and the error discussion surrounding Eq.~\eqref{eq_puremaxset}; every section of the pure gauge Gauss law map corresponds to a maximal set of correctable errors also for the fermionic Gauss law code.\footnote{For each choice of that section, we can equip each error inside it with any fixed choice of $\bm{z}$, thereby trivially increasing the choice of different maximal sets.}  In particular, the frame fields in Eq.~\eqref{eq_framefield} furnish such a set, with the corresponding syndromes exhausting all possible Gauss law violations at every vertex as they appear in the error space decomposition in Eq.~\eqref{eq_chargedecomp}.
The syndrome measurement is given by measuring generators of the gauge transformations, as before.

Now consider some vertex $v$ and suppose we admit $X_v$ into the error fray. In that case, to detect whether or not $X_v$ has happened, we are not permitted to simply measure the total charge at $v$ via a set of generators of $\{U_v^g\,|\,g\in G\}$. Thanks to Eq.~\eqref{eq_fermX}, the charge will depend on the code words, about which we are not allowed to glean information. Thus, the charge cannot be the syndrome and we must seek for a different quantity. Consider two errors $E_a$ and $E_b$, which will lead to vertex charge excitations reading 
\begin{itemize}
    \item[(i)] (a) $(\partial\chi_a)_v\chi_F^{x_{a,v}}$ and (b) $(\partial\chi_b)_v\chi_F^{x_{b,v}}$ when $n_v=0$, or
    \item[(ii)] (a) $(\partial\chi_a)_v\bar{\chi}_F^{x_{a,v}}$ and  (b) $(\partial\chi_b)_v\bar{\chi}_F^{x_{b,v}}$ when $n_v=1$,
\end{itemize} 
respectively. Code word protection demands (i) and (ii) to be indistinguishable, while error corrections requires (a) and (b) to be distinguishable via the conditions in Eq.~\eqref{eq_fermGaussdisterrors}.

Indistinguishability of (i) and (ii) entails that the charge measurement at $v$ is permitted to distinguish $(\partial\chi)_v\chi_F$ and $(\partial\chi)_v\bar{\chi}_F$ from $(\partial\chi)_v$, but not from one another. Let us therefore group charge measurement outcomes at $v$ into bins 
\begin{equation}\label{eq_bins}
   \big[(\partial\chi)_v\big]_x\coloneqq\begin{cases}
       (\partial\chi)_v&\text{if }x=0\\
       \{(\partial\chi)_v\chi_F,(\partial\chi)_v\bar{\chi}_F\}&\text{if } x=1\,.
   \end{cases}
\end{equation}
which correspond to the projective measurements
\begin{equation}\label{eq_coarsemeasure}
    \Pi^v_{\bm{x},\bm{\chi}}\coloneqq \Pi^v_{[(\partial\chi)_v]_x}=\Pi^v_{(\partial\chi)_v\chi_F^{x_{v}}}\ket{0}\!\bra{0}_v+\Pi^v_{(\partial\chi)_v\bar{\chi}_F^{x_{v}}}\ket{1}\!\bra{1}_v\,.
\end{equation}
Here $\Pi^v_{\eta}$ is the projector onto the joint eigenspace of $U^g_v$, $g\in G$, with eigenvalues $\eta(g)$. Note that $\Pi^v_\eta$ is diagonal in the occupation basis and so commutes with $\ket{n}\!\bra{n}_v$. 

Of course, the bins in Eq.~\eqref{eq_bins}, or, equivalently, the projectors in Eq.~\eqref{eq_coarsemeasure}, overlap nontrivially for different $\partial\chi$ and $x$. Indeed, projective measurements $\Pi^v_{\bm{x}_a,\bm{\chi}_a},\Pi^v_{\bm{x}_b,\bm{\chi}_b}$ in Eq.~\eqref{eq_coarsemeasure} are mutually exclusive -- and so perfectly distinguishable -- precisely if the conditions in Eq.~\eqref{eq_fermGaussdisterrors} are obeyed at $v$, taking care of the distinguishability of (a) and (b).
Furthermore, note that the $\Pi^v_{\bm{x}_a,\bm{\chi}_a},\Pi_{\rm{x}_b,\bm{\chi}_b}^{v'}$ commute -- and so can be jointly measured -- since the number operator at each vertex is even under fermion parity. Accordingly, a set $\{\Pi_{\bm{x},\bm{\chi}}^v\,|\,v\in\mathcal{V}\}$ obeying Eq.~\eqref{eq_fermGaussdisterrors} constitutes a valid set of syndrome measurements across the lattice. In particular, $\prod_{v\in\mathcal{V}}\Pi^v_{\bm{0},\bm{1}}=\Pi_{\rm pn}$ coincides with the code projector. These syndromes do not measure the charge directly when $x_v=1$, but rather conditional on the occupation number.

It is important to observe that there can be quite different maximal sets of syndrome measurements $\{\Pi^v_{\bm{x}_a,\bm{\chi}_a}\}$ obeying Eq.~\eqref{eq_fermGaussdisterrors} at a given vertex (and therefore also for the entire lattice). Here, we call such a set of projectors maximal if no further projector can be added without violating Eq.~\eqref{eq_fermGaussdisterrors}. For example, we have the two extreme cases that $x_{a,v}=0$ for all $a$, or that $x_{a,v}=1$ for all $a$, both of which comprise maximal sets, and the first condition in Eq.~\eqref{eq_fermGaussdisterrors} tells us that for either there are as many admissible bins at $v$ as there are possible Gauss law violations. By contrast, in the more generic case when $x_{a,v}=0$ for some $a$ and $x_{a,v}=1$ for the rest, we have that the number of admissible bins is one less than in the extreme cases when the group is finite because the second condition in Eq.~\eqref{eq_fermGaussdisterrors} implies that $x=0$ and $x=1$ bins are not permitted to overlap. For example, suppose $G=\mathbb{Z}_3\simeq\hat{G}$, so that there are three $(\partial\chi)_v$ levels. One may check that $[1_v]_0,[1_v]_1$ comprise a maximal set of bins~\eqref{eq_bins} for this case and that the second condition in Eq.~\eqref{eq_fermGaussdisterrors} forbids the addition of another one. 

Ignoring $Z$-errors for now, it is clear that maximal sets of correctable errors among those parametrized by Eq.~\eqref{eq_fermerrors} are in one-to-one correspondence with maximal sets of syndromes. Both are parametrized by maximal sets of solutions to Eq.~\eqref{eq_fermGaussdisterrors}, and we essentially have adapted the syndromes to the error sets. It only trivially becomes a many-to-one relation between maximal error sets and syndromes when we include $Z$-error data, which does not feature in the syndromes. Clearly, per syndrome we can choose only a single $Z$-type error to be included in the error set. 

This leads to a striking difference with the pure gauge and bosonic Gauss law codes which we considered previously. There, we had a single maximal syndrome set for all error sets under consideration. This permitted us to identify the base space in the pertinent Gauss law bundle with the set of syndromes and parametrize maximal correctable error sets via global sections of that bundle. The nontrivial fibers parametrized the multiplicity in error-set/syndrome combinations. The same picture is no longer useful for fermionic Gauss law codes, given that we now have a one-to-one correspondence between maximal error sets and syndromes when we ignore $Z$-errors. If anything, we have a separate Gauss law bundle per maximal syndrome set, except that in that case the fibers would be somewhat trivial, being parametrized exclusively by the $Z$-error data $\bm{z}$, which does not affect the syndromes. 

\subsubsubsection{Code parameters}\label{ssssec_fermcodeparam}
The kinematical degrees of freedom consist of $N_L$ $G$-registers on the links and $N_V$ qubits on the vertices. Similarly to the bosonic case, the loop sector contributes $N_L - N_V + 1$ logical $G$-registers. However, due to its fermionic nature, the matter sector requires a different treatment. Imposing the global neutrality condition Eq.~\eqref{global_neutrality_fermions} restricts the allowed configurations to a family of charge sectors labeled by $m=0,1,\cdots, M$, with total occupation number $N_V/2 + mD_F$, as discussed around Eq.~\eqref{eq_fermchargesector}, when $D_F<\infty$. The dressed matter space therefore has dimension
\begin{equation}
    \mathcal{D}^{\text{dr}}_{\text{matter}} \coloneqq \begin{cases}
    \sum_{m=0}^M {{N_V}\choose{N_V/2 + mD_F}}, \qquad D_F < \infty, \\
    {{N_V}\choose{N_V/2}}, \qquad D_F = \infty
    \end{cases}.
\end{equation}
In general, this space does not admit a tensor product decomposition into qubits. One exception occurs for $D_F = 2$, where the constraint Eq.~\eqref{global_neutrality_fermions} fixes the parity of the total occupation number. In that case, exactly half of the Hilbert space is retained, so that $\mathcal{D}^{\text{dr}}_{\text{matter}} = 2^{N_V - 1}$ and the dressed matter sector can be identified with $N_V - 1$ logical qubits. For general $D_F < \infty$, however, $\mathcal{D}^{\text{dr}}_{\text{matter}}$ is not necessarily a power of two, and such an identification is not exact. Nevertheless, the dimension $\mathcal{D}^{\text{dr}}_{\text{matter}}$ provides a natural notion of an effective number of logical qubits, given by $\log_2 \mathcal{D}^{\text{dr}}_{\text{matter}}$, which we use below as a convenient bookkeeping device to compare with standard $[n,k]$ qubit codes.

The distances follow as in the bosonic case. We distinguish between $X$-type errors---consisting of matter $X$-shifts and Wilson lines---and $Z$-type operators---consisting of matter $Z$'s and electric operators $U_{\ell}^{g}$. Since all the stabilizers are $Z$-type, any $Z$-type operator commutes with the stabilizers and is therefore undetectable, implying $d_{Z'} = 1$, where the prime indicates that both matter and connection $Z$-type errors are included. On the other hand, the $X'$ distance coincides with the smallest weight of closed Wilson loops and dressed $X$-operators. Similarly to the bosonic case, the exact value depends on the geometry of the lattice, but we can guarantee $d_{X'} = 2$ or $3$.

Since the stabilizer group is generated entirely by $Z$-type operators, the code is a generalized CSS one with a trivial $X$-stabilizer sector. Equivalently, it is a generalized classical bit-flip code, as only $X$-type operators can be detected. Altogether, in the notation of Sec.~\ref{sssec:boscodepars}, the fermionic Gauss law realizes a generalized $[(N_L, N_V), (N_L - N_V + 1, \log_2 \mathcal{D}^{\text{dr}}_{\text{matter}}), d_{X'}]$-code.

\subsection{Vacuum Codes}
In the Gauss law codes that we considered so far, we took the entire gauge theory, including its kinematical data, and interpreted it as a quantum error correcting code, comprising error spaces, code space, and ``physical qubits''. Correctable errors then violate the Gauss law and thereby gauge invariance. 
We anticipate that, besides the structural insights, the main relevance of such codes pertains to quantum \emph{simulations} of lattice gauge theories, where 
a gauge theory is simulated on a subspace of a quantum system that is not itself a gauge theory \cite{Georgescu_2014,Daley:2022eja,Rajput:2021trn,Spagnoli:2024mib,Spagnoli:2026qni,Yao:2025cxs,Carrozza:2024smc}.

However, such errors cannot be expected to occur in an actual physical system described by a gauge theory, where there is no evidence of the literal breaking of bona fide gauge invariance.\footnote{Spontaneous symmetry breaking is, of course, crucial in gauge theories but does not break local gauge symmetry.} In such systems, nothing but the trivial syndrome is physically realizable, rendering Gauss law codes somewhat unphysical beyond simulations.

Accordingly, one may wonder whether one can turn Gauss law codes including errors fully gauge-invariant in an appropriate manner, thereby providing a physical realization of them that does not invoke a qualitatively different physical system (like the quantum simulator) to make sense of. This effort will culminate in what we shall refer to as \emph{vacuum} codes, which under certain conditions are unitarily equivalent to Gauss law codes. 

\subsubsection{Basic ideas underlying vacuum codes}
\label{sssec_basicvacuum}
The basic idea is to include sufficiently many additional matter degrees of freedom on the vertices of the lattice to compensate the Gauss law violations. 
Placing them in their vacuum state\footnote{Here ``the vacuum state'' means the superselection sector on which $U_{\rm matter}^{\bm{g}}$ for any $\bm{g}$ acts as identity and physically corresponds to states with vanishing charge densities. It does not necessarily mean a ground state of a particular Hamiltonian.
}, errors will now amount to gauge-invariant matter excitations above the vacuum. 
That is, rather than having that the violation of a particular form of the Gauss law constitute a breaking of gauge invariance, it will now merely indicate the excitation of previously ``frozen'' degrees of freedom (see also the discussion in \cite[Sec.~2.3]{Carrozza:2024smc}).\footnote{While Kitaev's surface codes \cite{Kitaev:1997wr,Bravyi:1998sy} certainly have a flavor of vacuum codes, we explain in section~\ref{ssec_hybrid} why they are better regarded as hybrid Gauss-law-vacuum codes.} 
For example, the analog of this in continuum electrodynamics would be a transition from the matter vacuum form of the Gauss law to one including a charge density $\rho$,
\begin{equation}
   C= \nabla\cdot\mbf{E}=0\qquad\qquad\Rightarrow\qquad\qquad C'=\nabla\cdot\mbf{E}-\rho=0\,,
\end{equation}
thus changing the form of the gauge constraint, as opposed to violating gauge invariance.

In the case of Gauss law codes, we identified the full kinematical Hilbert space $\Hil_{\rm kin}$ in Eq.~\eqref{eq_Hkinbos} with the physical space of the code, while we took its gauge-invariant (perspective-neutral) subspace $\Hil_{\rm pn}$ in Eq.~\eqref{eq_codepn} as the code space. Accordingly, gauge transformations constituted the stabilizers and error syndromes were given by Gauss law violations.  

In contradistinction,  for vacuum codes,  we identify the full gauge-invariant (perspective-neutral) Hilbert space $\Hil_{\rm pn}$ as the physical space of the code and the pertinent dressed matter vacuum as the code subspace:
    \begin{equation}\label{eq_vacuumcodespace}
        \mathcal{H}_{\text{code}} = \mathcal{H}_{\text{vac}} = \{ \ket{\Psi}_{\text{pn}} \in \mathcal{H}_{\text{pn}}    \mid u^g_{v,i} \ket{\Psi}_{\text{pn}} = \ket{\Psi}_{\text{pn}}, \enspace \forall\, v \in \mathcal{V},\,g\in G\,,\, i\in\{\text{matter types}\}  \}\,,
   \end{equation}
where the $u^g_{v,i}=\exp\left(i\alpha(g)(N_{v,i}-c_{v,i})\right)$, for some $\alpha(g)\in\mathbb{R}$ and $c_{v,i}\in\{0,1\}$ ($c_{v,i}=0$ for bosons), are the gauge transformations at vertex $v$ on matter type $i$, cf.~Eqs.~\eqref{eq_gaugegrouprepbos} and~\eqref{eq_fermgaugetr} for our previous theories. 
In other words, stabilizer transformations are now defined in terms of the number operators. Minding the subsystem factorization of the perspective-neutral space, 
 \begin{equation}
        \mathcal{H}_{\text{pn}} = \mathcal{H}_{\text{loops}} \otimes_R \mathcal{H}^{\text{dr}}_{\text{matter}}\,,
        \label{pn_bosonic}
    \end{equation}
where $\otimes_R$ is defined in \eqref{frameTPS} and which holds more generally than in the cases we have discussed so far, the corresponding code projector reads
    \begin{equation}\label{eq_vacuumcodeproj}
        \Pi_{\text{vac}} =\prod_{v\in\mathcal{V}}\,\prod_i\,\int_G{\dee g} \, u^g_{v,i}= \mathbbm{1}_{\text{loops}} \otimes_R \ket{0}\!\bra{0}^{\text{dr}}_{\text{matter}}\,.
    \end{equation}
The last equality yielding the projector onto the dressed matter vacuum assumes that no matter type is uncharged, as in all the previous theories that we considered. Indeed, the only state of a charged matter type $i$ that is gauge-invariant on its own is its vacuum $\ket{n_{v,i}-c_{v,i}=0}_{v,i}$. Recall that in staggered fermion theories the vacuum is the charge-zero state, which for antiparticles means the occupied state.

The code space is therefore one-dimensional in the matter factor and retains the full loop degrees of freedom. Logical operators are those operators on $\mathcal{H}_{\text{pn}}$ that preserve the vacuum in the dressed matter sector. Since the code projector acts trivially on $\mathcal{H}_{\text{loops}}$, the logical algebra coincides with the loop algebra, and we have
\begin{equation}\label{eq_vacuumloops}
  \Hil_{\rm code}\simeq\Hil_{\rm loops}\,,\qquad\qquad  \mathcal{A}_{\text{code}} \simeq \mathcal{A}_{\text{loops}}.
\end{equation}

All relevant structure for the error analysis therefore arises from gauge-invariant operators that move states out of the vacuum in the dressed matter sector. 
In other words, correctable errors now have to be extracted from the dressed matter algebra $\mathcal{A}^{\rm dr}_{\rm matter}$. Since the dressed matter vacuum is defined in terms of the number operators, the error syndromes will be given by the electric charges for each matter type $i$ at each vertex. Thus, a syndrome measurement can be carried out by measuring any generating set of the stabilizers, $u^g_{v,i}$.

Regarding the relationship between Gauss law codes and vacuum codes, Eq.~\eqref{eq_vacuumloops} says that the code space of vacuum codes coincides with the code space of pure gauge Gauss law codes discussed in section~\ref{sssec_pureGauss}. 
This does not necessarily mean, however, that the full quantum error correcting codes are unitarily equivalent. 
For that to be the case, we would have to have that also the error spaces inside the perspective-neutral space $\Hil_{\rm pn}$ of the vacuum code are naturally isomorphic\footnote{When both are infinite-dimensional, there will of course exist a continuum of isomorphisms. However, the nontrivial question is whether any of them is physically natural.} to the error spaces $\Hil_{\partial\bm\chi}$ in the decomposition~\eqref{eq_chargedecomp} of the kinematical space $\Hil_{\rm kin}$ of the pure gauge Gauss law code. 
Indeed, in the pure gauge Gauss law code, errors induce some vertex charge $(\partial\chi)_v$. When there are multiple matter species, this charge may be compensated in multiple ways by matter charges. Hence, several matter excitations in vacuum codes may correspond to the same error in the pure gauge Gauss law code. 

A strict unitary equivalence between the pertinent Gauss law and vacuum codes will then generally not exist. Indeed, for bosonic codes subject to a $G$ with continuous component, we will find such a unitary equivalence only upon coarse-graining the vacuum codes  in terms of the distinguishability of their matter charge excitations; i.e.\ we impose the operational restriction that one cannot distinguish the electric charges of the individual matter types, but only the total matter charges at a vertex. By contrast, for finite $G$, in fact, we will find an exact unitary equivalence.
There is only one special class of bosonic codes encompassing cases with continuous $G$, discussed briefly in section~\ref{ssssec_Gausseqvacuum}, where no such coarse-graining is needed. The coarse-grained vacuum codes then constitute a unitarily equivalent physical implementation of pure gauge Gauss law codes and can abstractly be viewed as the same codes. In the fermionic vacuum code case, the situation will be the opposite to the bosonic case: it is now the pure gauge Gauss law code that features a richer choice of error sets than the vacuum code, though on their joint domain, the codes are again unitarily equivalent.

In what follows, we will focus our attention on bosonic and fermionic vacuum codes that are related to the pure gauge Gauss law codes of section~\ref{sssec_pureGauss}. However, it is clear that one can extend the construction of vacuum codes to encompass our bosonic and fermionic Gauss law codes discussed in sections~\ref{sssec_bosGauss} and~\ref{sssec_fermGauss}, as well as more general matter theories.

\subsubsection{Bosonic vacuum codes}
Let us begin with bosonic matter vacuum codes.
\subsubsubsection{Bosonic vacuum codes equivalent to pure gauge Gauss law codes }\label{ssssec_Gausseqvacuum}
There is a special class of bosonic theories, which for $G$ with continuous component is not entirely encompassed by our previous assumptions on bosonic matter of section~\ref{ssssec_bosHil},\footnote{In section~\ref{ssssec_bosHil}, we assumed that for continuous $G$, each generator of $\hat{G}$ is modeled by an oscillator. This does not match the $L^2(G)$ Hilbert spaces of this section; e.g.\ when $G=\rm{U}(1)$, we would have $L^2(\rm{U}(1))\simeq L^2(\mathbb{Z})$, which does not naturally support two oscillators, as needed for the two generators of $\hat{G}=\mathbb{Z}$.} which  yields vacuum codes that are unitarily equivalent to the pure gauge Gauss law codes of section~\ref{sssec_pureGauss}. 

When $G$ is continuous, this theory features a few physical oddities and does not encompass scalar QED, however, for the sake of argument, let us discuss it regardless. 
This is given by choosing a bosonic matter species such that the vertex Hilbert space in Eq.~\eqref{eq_Hbosonic_matter} is simply\footnote{For $G=\rm{U}(1)$, this is a system of a rotor matter coupled to a $\rm{U}(1)$ gauge field which appears as a low energy effective theory of $\rm{U}(1)$ gauge theory with charge 1 complex scalar and Higgs potential. The case for $G=\mathbb{Z}_N$ corresponds to a system of a clock matter to a $\mathbb{Z}_N$ gauge field. This system appears as a deep IR energy effective theory of $\rm{U}(1)$ gauge theory with charge $N$ Higgs scalar.}

\begin{equation}\label{eq_l2g}
    \Hil_v=L^2(G)\,,
\end{equation}
thus isomorphic to a link Hilbert space. 
The stabilizers $u^g_v$ are then given by the convention of choosing either the left or right regular representation of $G$ on $\Hil_v$. 
By the Peter-Weyl theorem, $\Hil_v$ contains every irrep of $G$ once and only once since $G$ is Abelian and all of these irrep spaces are one-dimensional. The trivial representation $\Hil_{v,\chi=1}$ will then define the vacuum at any vertex, $\ket{0}_v\coloneqq\ket{\chi=1}_v$, while all $\Hil_{v,\chi\neq1}$ correspond to matter excitations. 

When the group decomposes into a product $G=G_1\times G_2\times G_3\times\cdots$ we also have that the vertex Hilbert space factorizes accordingly, $\Hil_v=L^2(G_1)\otimes L^2(G_2)\otimes L^2(G_3)\otimes\cdots$, and we would be entitled to associate the different factors with distinct bosonic matter species. In what follows, let us assume, for simplicity, that $G$ is indecomposable. The decomposable case proceeds in the same manner, except that the following argumentation would have to be made for each species separately. The overall conclusion would be the same.

Treating nontrivial characters as excitations, the number operator reads $N_v=\sum_{\chi\in\hat{G}}n_\chi\ket{\chi}\!\bra{\chi}_v$ for some ordering of the $\chi\neq1$,  and, since $u^g_v$ is diagonal in $\ket{\chi}_v$, we can write $u^g_v=\exp(i\alpha(g) N_v)$ with $\alpha$ so that $\chi(g) = \exp(i\alpha(g)n_\chi)$ with $\alpha\in\mathbb{R}$. 
Since we must have that $\bar{\chi}(g)\chi(g)=1$ for all $g\in G$ and both $\chi,\bar\chi\in\hat{G}$, we see that this requires $n_\chi=-n_{\bar\chi}$ when $\hat{G}$ is of infinite cardinality (so $G$ continuous), so that the occupation number $n_\chi$ 
must be allowed to take \emph{negative} values in that case. Indeed, when $G=\rm{U}(1)$, we have that $\hat{G}=\mathbb{Z}$ and we can simply choose $n_\chi\equiv z\in\mathbb{Z}$. 
The interpretation of this is that we treat the positive occupation numbers as particles and the negative ones as antiparticles, reminiscent of the Dirac sea. This is consistent with the fact that these will correspond to positive and negative charges. When $G$ is finite, by contrast, we only need that $n_\chi+n_{\bar\chi}=0\mod|G|$.

It is straightforward to check that the code space~\eqref{eq_vacuumcodespace} is
\begin{equation}
    \Hil_{\rm code}=\Hil_{\rm vac} = \Hil_{\rm loops}\otimes_R\bigotimes_{v\in\mathcal{V}}\ket{0}_v\,
\end{equation}
and thus isomorphic to the one of the corresponding pure gauge Gauss law code in section~\ref{sssec_codespacepure}.
Crucially, we now also have that the dressed matter Hilbert space has the form
\begin{equation}
    \Hil_{\rm matter}^{\rm dr} \simeq L^2(G)^{\otimes N_V-1}\,,
\end{equation}
corresponding to $N_V-1$ copies of the vertex Hilbert spaces. The ``lost'' degree of freedom corresponds to $\Hil_{v_0}$.
This can be seen by following the same steps of section~\ref{ssssec_boscodespace} for the present bosonic theory. 
Indeed, the analog of Eq.~\eqref{eq_bostotalchargezero} implies $\sum_v n_v=0$ in the infinite order case and $\sum_v n_v=0\mod|G|$ in the finite group order case. 
In both cases, this means that the state in $\Hil_{v_0}$ is fixed through gauge invariance by the states elsewhere and so we have the basis states
\begin{equation}\label{eq_chin}
    \ket{\bf{n}}^{\rm dr}_{\rm matter}=\ket{\partial\bm\eta^{({\bf n})}}_R\otimes\ket{\bf n}_{\rm matter}\,,\qquad\qquad {\bf n}=(n_v)_{v\neq v_0}\,,
\end{equation}
spanning $N_V-1$ copies of $L^2(G)$. Here, $(\partial\eta^{({\bf n})})_v(g)=\bar{\chi}_{n_v}(g)$, where $\chi_{n_v}(g)=\exp(i\alpha(g) n_v)$ labels the matter charge corresponding to occupation $n_v$ at $v$.

This also means that the $\ket{\bf n}^{\rm dr}_{\rm matter}$ with ${\bf n}\neq{\bf 0}$ span the error spaces of the vacuum code, the syndrome being ${\bf n}$, which is equivalent to the matter charges $\bm{\chi}_{\bf n}=(\chi_{n_v})_{v\neq v_0}$. Hence, the error spaces of the vacuum code are unitarily equivalent to the error spaces $\Hil_{\partial\bm\chi}$ in Eq.~\eqref{eq_chargedecomp} of the pure gauge Gauss law code, which too only run over the gauge charge configurations at $v\neq v_0$. 

Let us now consider the errors that lead to the transitions $\ket{\bf 0}^{\rm dr}_{\rm matter}\to\ket{\bf n}^{\rm dr}_{\rm matter}$, beginning with the case that $G$ is finite. In this case, we can build particle and antiparticle creation and annihilation operators in the same way as for the finite order boson types in section~\ref{ssssec_bosHil}. Minding that we have only a single species in the present case, we denote the particle and antiparticle ladder operators as $a_v,a_v^\dag$ and $\bar{a}_v,\bar{a}_v^\dag$, respectively, for \emph{any} $v\in\mathcal{V}$. They transform as 
\begin{align}
    U^{\bm g}a_v\,U^{\bm{g}\dag}&=\bar{\chi}_1(g_v)\,a_v\,,\qquad\qquad U^{\bm g}a^\dag_v\,U^{\bm{g}\dag}={\chi}_1(g_v)\,a^\dag_v\nonumber\\
    U^{\bm g}\bar{a}_v\,U^{\bm{g}\dag}&={\chi}_1(g_v)\,\bar{a}_v\,,\qquad\qquad U^{\bm g}\bar{a}^\dag_v\,U^{\bm{g}\dag}=\bar{\chi}_1(g_v)\,a^\dag_v\,, \label{eq_ladderbos}
\end{align}
where $\chi_1$ was defined below Eq.~\eqref{eq_chin}. 
We can then build gauge-invariant bipartite operators as in Eq.~\eqref{eq_drmatteralg}. The only ones that do not annihilate the dressed matter vacuum are the tree Wilson line dressed particle-antiparticle combinations of creation operators, 
\begin{equation}
    \widetilde{W}^{\chi_1}_{\gamma_R[v,v']}\coloneqq\bar{a}_v^\dag\,W^{\chi_1}_{\gamma_R[v,v']}\,a_{v'}^\dag\,,
\end{equation}
and they generate the elementary moves
\begin{align}
\widetilde{W}^{\chi_1}_{\gamma_R[v,v']}\,\ket{\bf 0}_{\rm matter}^{\rm dr}&\propto\ket{0,\ldots,0,n_v=|G|-1,0,\ldots,0,n_{v'}=1,0,\ldots}^{\rm dr}_{\rm matter}\,,\qquad v\neq v_0\,,\nonumber\\
\widetilde{W}^{\chi_1}_{\gamma_R[v_0,v']}\,\ket{\bf 0}_{\rm matter}^{\rm dr}&\propto\ket{0,\ldots,0,n_{v'}=1,0,\ldots}^{\rm dr}_{\rm matter}\,,\label{eq_aadag0}
\end{align}
in the latter case leaving the change in the dependent variable $n_{v_0}$ implicit, cf.~Eq.~\eqref{eq_chin}.
Via Eq.~\eqref{eq_chin}, this leads to the corresponding gauge charge transition $\partial\bm\eta^{({\bf 0})}\to\partial\bm\eta^{({\bf n})}$ in the electric frame data. By combinations of such bipartite operators one can create any transition $\ket{\bf 0}^{\rm dr}_{\rm matter}\to\ket{\bf n}^{\rm dr}_{\rm matter}$ and thus any gauge charge transition $\partial\bm\eta^{({\bf 0})}\to\partial\bm\eta^{({\bf n})}$, exhausting the $\partial\bm\chi$-spectrum in Eq.~\eqref{eq_chargedecomp}. This is especially clear from the second line in Eq.~\eqref{eq_aadag0}. 

Notably, the combination of the tree Wilson lines between the ladder operators constitutes the non-loop contribution $\prod_{\ell\in\mathcal{L}}W^{\chi_\ell}_{\gamma_R[v_i,v_f]}$ inside $W^{\bm\chi({\bf n})}$ in Eq.~\eqref{eq_treenontree}. Of course, to also encompass the loop contributions $H^{\bm\chi_S}$, one could simply multiply them in from the loop algebra $\mathcal{A}_{\rm loops}$. In this way, we obtain a one-to-one correspondence between the gauge-invariant matter excitation errors $\widetilde{W}^{\bm\chi({\bf n})}$ of the vacuum code and the errors $W^{\bm\chi}$ of the pure gauge Gauss law code of section~\ref{sec:pgmaxsets}. Specifically, also the discussion of maximal sets  carries over to the vacuum code. Indeed, there are many maximal sets of correctable errors; e.g., it is easy to see that $\widetilde{W}^{\chi_1}_{\gamma_r[v,v_0]}\widetilde{W}^{\chi_1}_{\gamma_R[v_0,v']}$ and $\widetilde{W}^{\chi_1}_{\gamma_R[v,v']}$ lead to the same syndrome in terms of occupation numbers at $v\neq v_0$, although the former has $a_{v_0}^\dag\bar{a}_{v_0}^\dag$ at the root.

When $G$ is continuous, the conclusion is the same. In this case, we can again build two pairs of particle and antiparticle creation and annihilation operators, $a_v,a_v^\dag$ and $\bar{a}_v,\bar{a}_v^\dag$, respectively, where the former cover $0$ and the positive eigenvalues of $N_v$ in the usual harmonic oscillator manner, while the latter cover $0$ and the negative eigenvalues with $\bar{a}_v^\dag$ lowering and $\bar{a}_v$ raising them. Their transformation is then again given by Eq.~\eqref{eq_ladderbos}, and we can proceed as in the finite $G$ case, except that the elementary bipartite particle-antiparticle operators raise and lower occupation numbers by $\pm1$ in the analog of Eq.~\eqref{eq_aadag0}. 
The remainder of the argumentation is, however, completely analogous.

In conclusion, vacuum codes with  bosonic matter of the special form in Eq.~\eqref{eq_l2g} are unitarily equivalent to the pure gauge Gauss law codes of section~\ref{sssec_pureGauss}.

\subsubsubsection{More general bosonic vacuum codes}
Next, let us consider the more general class of bosonic matter theories described in section~\ref{ssssec_bosHil}, which encompasses scalar QED. We note that for finite $G$, this class encompasses the class of theories described in the previous subsection~\ref{ssssec_Gausseqvacuum}, in which case it will once more yield an exact unitary equivalence with pure gauge Gauss law codes. The main difference arises for continuous $G$, where such an equivalence is only obtained upon a coarse-graining of the vacuum code.

We noted below Eq.~\eqref{eq_vacuumloops} that correctable errors must be extracted from the dressed matter algebra, which in the present case is given by Eq.~\eqref{eq_drmatteralg}. Since all the particle-particle or antiparticle-antiparticle dressed Wilson lines annihilate the vacuum $\ket{\bf 0}^{\rm dr}_{\rm matter}$, we again restrict our analysis to particle-antiparticle type dressed Wilson lines errors. It is convenient to organize arbitrary products of such operators into the form 
\begin{equation}
    \widetilde{W}^{\bm{\chi}(\bm{k})} \coloneqq \prod_{\ell \in \mathcal{L}} \left( \prod_{\alpha \in A} (a_{v_i(\ell), \bar{\alpha}}^{\dagger})^{k_{\ell, \alpha}} (a_{v_i(\ell), \alpha}^{\dagger})^{k_{\ell,\bar{\alpha}}} W_{\ell}^{\rho_{\alpha}^{\Delta_{\alpha} \bm{k}_{\ell} }} (a_{v_f(\ell), \bar{\alpha}}^{\dagger})^{k_{\ell, \bar{\alpha}}}(a_{v_f(\ell),\alpha}^{\dagger})^{k_{\ell, \alpha}}   \right)\,,
\end{equation}
where we recall that $\alpha$ runs over species (cf.~the discussion below Eq.~\eqref{eq_speciescontrib}), $k_{\ell,i}, k_{\ell, \bar{i}} \in \mathbb{Z}_{D_i}$, $k_{\ell, j} \in \mathbb{N}_0$ for each $\ell \in \mathcal{L}, 0 \leq i < p \leq j < r $, and the operator $\Delta_\alpha\bm{k}$ is defined  in Eq.~\eqref{eq_Delta}.
The integer data $\bm{k}$
therefore specifies a flux configuration on the links, expressed in the chosen generating set $\{  \rho_i \}_{i=0}^{r-1}$ of $\hat{G}$. 

To evaluate the Knill-Laflamme condition, it may be checked that, similarly to Eq.~\eqref{eq_aadag0},\footnote{Note that, in contrast to Eq.~\eqref{eq_chin} of the previous subsection, ${\bf n}$ encompasses all $v\in\mathcal{V}$ in the present case, in line with the formulation around Eq.~\eqref{partialchin}.} 
\begin{align}
   \widetilde{W}^{\bm{\chi}(\bm{k})} \Pi_{\text{vac}} &\propto H^{\bm{\chi}_S(\bm{k})}\prod_{v \in \mathcal{V}} \left( \bigotimes_{i = 0}^{p-1} \ket{ D_i - \nabla( \bm{k}_i - \bm{k}_{\bar{i}} )_v \mod D_i } \!\bra{0}_{v,i}^{\text{dr}}  \right)\nonumber\\
   &\qquad\qquad\,\,\,\,\,\otimes \left( \bigotimes_{j=p}^{r-1} \ket{\sum_{\ell \in \mathcal{L}_{\text{out}}(v) } k_{\ell, \bar{j}} + \sum_{\ell' \in \mathcal{L}_{\text{in}}(v)} k_{\ell', j}} \!\bra{0}^{\text{dr}}_{v,j}  \right)\,,\label{eq_Wpivac}
\end{align}
where $\nabla$   denotes the discrete divergence across $v$,
\begin{equation}
    \nabla( \bm{k}_i)_{v} = \sum_{\ell \in \mathcal{L}_{\text{out}}(v)} {k}_{\ell,i} - \sum_{\ell' \in \mathcal{L}_{\text{in}}(v)} {k}_{\ell',i}\,,
\end{equation}
$\Pi_{\rm vac}$ is given in Eq.~\eqref{eq_vacuumcodeproj}, and
where the proportionality reflects overall normalization and phase factors originating from the creation and annihilation operators entering the dressed Wilson line. Inserting this expression into the Knill-Laflamme condition yields 
\begin{equation}
\begin{split}
\Pi_{\text{vac}}\widetilde{W}^{\bm{\chi}(\bm{k}_a) \dagger} \widetilde{W}^{\bm{\chi}(\bm{k}_b)}\,\Pi_{\text{vac}}
&\propto H^{\bar{\bm{\chi}}_S(\bm{k}_a)\bm{\chi}_S(\bm{k}_b)} \prod_{v \in \mathcal{V}} \Bigg(
\prod_{i = 0}^{p-1} 
\delta_{\nabla (\bm{k}_{a,i} - \bm{k}_{a, \bar{i}})_v, \nabla(\bm{k}_{b,i} - \bm{k}_{b,\bar{i}})_v} \\
& \prod_{j=p}^{r-1} 
\delta\!\left(
\sum_{\ell \in \mathcal{L}_{\text{out}}(v)} k_{a,\ell, \bar{j}} 
+ \sum_{\ell' \in \mathcal{L}_{\text{in}}(v)} k_{a, \ell', j},
\sum_{\ell \in \mathcal{L}_{\text{out}}(v)} k_{b,\ell, \bar{j}} 
+ \sum_{\ell' \in \mathcal{L}_{\text{in}}(v)} k_{b, \ell', j}
\right)
\Bigg)\Pi_{\text{vac}}\,.\label{eq_KLbosvacuum}
\end{split}
\end{equation}
Therefore, similarly to the pure gauge Gauss law code, any two such dressed Wilson line errors are simultaneously correctable as long as they are not related only by multiplication with Wilson loops.

Let us consider the syndromes of the vacuum code. As noted in section~\ref{sssec_basicvacuum}, syndrome measurements correspond to measuring any generating set of the $u^g_{v,i}$, or, equivalently, the matter vertex charge \emph{for each type $i$}. The integer link data $\bm{k}$ associated with a product of dressed
Wilson line operators $\widetilde{W}^{\bm{\chi}(\bm{k})}$ determines the matter vertex charge that is detected by such a
measurement (cf.~Eq.~\eqref{eq_Wpivac}):
\begin{equation}
    Q_{v,i}(\bm{k}) \coloneqq  \bar{\rho}_i^{\nabla(\bm{k}_i - \bm{k}_{\bar{i}})_v}\,,\qquad0\leq i<p\,,\qquad\qquad Q_{v,j}(\bm{k})\coloneqq\rho_j^{\sum_{\ell\in\mathcal{L}_{\rm out}(v)}k_{\ell,\bar{j}}+\sum_{\ell'\in\mathcal{L}_{\rm in}(v)}k_{\ell',j}}\,,\qquad p\leq j<r\,,\label{eq_bosvacuumsyndrome}
\end{equation}
where $\rho_i\in\hat{G}$ denotes the charge carried by particle type $i$.
The syndrome of $\widetilde{W}^{\bm{\chi}(\bm{k})}$ is therefore
the character-valued charge configuration
$\bm{Q}_i(\bm{k})=(Q_{v,i}(\bm{k}))_{v\in\mathcal{V}}$ for each type $i\in\mathbb{Z}_r$. 

Two errors $\widetilde{W}^{\bm{\chi}(\bm{k}_a)},\widetilde{W}^{\bm{\chi}(\bm{k}_b)}$ that yield the same charge profile $Q_{v,i}(\bm{k}_a)=Q_{v,i}(\bm{k}_b)$ for all $i\in\mathbb{Z}_r$ and $v\in\mathcal{V}$ are thus indistinguishable, and this corresponds precisely to the case that the Kronecker deltas in the Knill-Laflamme condition~\eqref{eq_KLbosvacuum} ``click''. This is the case when the two errors differ by a logical operation, such as a product of Wilson loops or number operators applied to any vertex. 

To entertain maximal sets of correctable errors, consider the map
\begin{align} 
   & \partial_{B,\rm{vac}}:\bigtimes_{0\leq i<p}\,\mathbb{Z}_{D_i}\times\bigtimes_{p\leq j<r}\mathbb{N}_0\to\hat{G}^{rN_V}\nonumber\\
 &   \partial_{B,\rm{vac}}(\bm{k}) = \left(Q_{v,i}(\bm{k})\right)_{i\in\mathbb{Z}_r,v\in\mathcal{V}}\,, \label{eq_vacbundle}
\end{align} 
which, in analogy to the Gauss law codes, constitutes a discrete fiber bundle. As in the case of bosonic Gauss law codes, a maximal set of correctable errors is parametrized by any section of this bundle.

\subsubsubsection{Unitary equivalence with pure gauge Gauss law codes upon coarse-graining}\label{ssssec: unitaryequivvaccodes}
Let us now investigate the relation of these bosonic vacuum codes with the pure gauge Gauss law codes of section~\ref{sssec_pureGauss}. As noted around Eq.~\eqref{eq_bosvacuumsyndrome}, the syndromes of vacuum code errors distinguishes not only the total charge at each vertex, as the syndromes of Gauss law codes, but even the individual contributions of the different boson types $i$. \emph{A priori} this appears like a finer syndrome with more information. However, this is only the case when $G$ is continuous, as we will now explain.

To relate vacuum codes to pure gauge Gauss law codes, we thus need to introduce the operational coarse-graining condition that we are henceforth no longer entitled to distinguish individual type charges at a vertex, but only the total matter charge.
Hence, rather than Eq.~\eqref{eq_bosvacuumsyndrome}, we restrict to the coarser total charge measurement
\begin{equation}\label{eq_vactotalcharge}
    Q_v(\bm{k}) \coloneqq \rho_v(\bm{k})=\prod_{0 \leq j,\bar{j} < r} {\bar{\rho}}_j^{\nabla(\bm{k}_j - \bm{k}_{\bar{j}})_v}\,.
\end{equation}
The coarse-grained syndrome of an error $\widetilde{W}^{\bm{\chi}(\bm{k})}$ is therefore
the total charge configuration
$\bm{Q}(\bm{k})=(Q_v(\bm{k}))_{v\in\mathcal{V}}$. Gauss's law~\eqref{eq_bosonicGauss} implies that this charge configuration coincides with the (exponentiated) discrete
divergence of the flux pattern induced on the links,
\begin{equation}\label{eq_Gaussvacuumcode}
    \partial\bm{\chi}(\bm{k})=\bar{\bm{Q}}(\bm{k})\,.
\end{equation}
The left hand side coincides with the syndromes of the pure gauge Gauss law codes. Note that $\bm{k}$ induces an occupation number ${\bf n}(\bm{k})$ when the corresponding $\widetilde{W}^{\bm{\chi}(\bm{k})}$ acts on $\ket{\bf 0}^{\rm dr}_{\rm matter}$, which is many-to-one, i.e.\ many $\bm{k}$ configurations yield the same ${\bf n}(\bm{k})$. For example, two errors $\widetilde{W}^{\bm{\chi}(\bm{k}_a)}, \widetilde{W}^{\bm{\chi}(\bm{k}_b)}$ may differ only by Wilson loops, which means they still have $\bm{k}_a\neq\bm{k}_b$, without generating distinct occupation number configurations. What about the relation between the total charge configuration $\bm{Q}(\bm{k})$ and ${\bf n}$?

In the case that $G$ is finite -- and so are all the generators $\rho_i$ of the different boson types -- we have that the total vertex charge $Q_v(\bm{k})$ uniquely specifies the occupation number of each type via Eq.~\eqref{eq_rhototbos} and so also the vertex charge contribution of each $i$. Thus, in this case, the finer syndrome of the vacuum code does not, in fact, encode more information than the coarser syndrome of the corresponding pure gauge Gauss law code.

By contrast, when $\hat{G}$ has infinite order generators $\rho_j$, this one-to-one correspondence is broken. This can be seen from the fact that for infinite order species $\alpha = j,\bar{j}$, the charge contribution is $\bar{\rho}_j^{\nabla(\bm{k}_j-\bm{k}_{\bar j})_v}=Q_{v,j}(\bm{k})Q_{v,\bar{j}}(\bm{k})$, which depends on the \emph{difference} of the occupation numbers $n_{v,j}-n_{v,\bar{ j}}$. As these are independent in the infinite order case, there is an infinite degeneracy with multiple distinct ${\bf n}$ configurations yielding the same $Q_{v}(\bm{k})$. In this case, the finer syndrome of the vacuum code \emph{does} contain more information about which error occurred.

Let us now clarify what we mean by ``coarse-graining the bosonic vacuum code into a pure gauge Gauss law code''. We have seen in section~\ref{sssec_basicvacuum} that the code spaces of the two codes are naturally isomorphic. The discussion of the syndromes, on the other hand, tells us that the error spaces of the two types of codes are only naturally isomorphic when $G$ is finite, and otherwise the charge coarse-graining provides a means to relate the error spaces. We shall now specify more precisely what this means technically.

The perspective-neutral Hilbert space of the vacuum code factorizes as in Eq.~(\ref{eq_subsystemcode}), where the dressed matter factor is spanned by the states in Eq.~\eqref{partialchin}. Recalling that $\partial\bm\chi^{({\bf n})}$ is uniquely determined by ${\bf n}$ via Gauss's law, the dressing construction therefore induces a linear map\footnote{This is essentially a Page-Wootters conditioning~\eqref{eq_PW} on the non-invariant occupation state $\bra{\bf n}_{\rm matter}$.}
\begin{equation}
\Pi_R:\mathcal{H}^{\text{dr}}_{\text{matter}} \to \mathcal{H}_R , \qquad 
\ket{\partial\bm{\chi^{({\bf n})}}}_R \otimes \ket{\bf{n}}_{\text{matter}} \mapsto \ket{\partial\bm{\chi^{({\bf n})}}}_R \,,
\end{equation}
which associates to each dressed matter configuration the flux profile that it induces on the 
links. The left hand side corresponds to the states spanning the error spaces of the vacuum code, the right hand side are the states which span the error spaces of the pure gauge Gauss law code via Eqs.~\eqref{eq_chargedecomp} and~\eqref{frameTPS}. The map thus provides the link between the error spaces. It corresponds to discarding the microscopic particle content and 
retaining only the induced gauge flux. It is surjective, since every flux configuration arises from some occupation assignment. 
However, as we noted above, it is not injective when $G$ has a continuous component, in which case distinct matter configurations ${\bf n, n'}$ induce the same flux
\begin{equation}
\partial\bm{\chi^{({\bf n})}}=\partial\bm{\chi^{({\bf n'})}},
\end{equation}
and the corresponding dressed states differ by a vector in the kernel of $\Pi_R$. In fact, the kernel is precisely the span of such differences. Passing to the quotient by this kernel removes the microscopic degeneracy. By the first isomorphism theorem, the induced map
\begin{equation}\label{eq_tildePiR}
\tilde{\Pi}_R :
\mathcal{H}^{\text{dr}}_{\text{matter}}/
\ker(\Pi_R)
\;\longrightarrow\;
\mathcal{H}_R ,
\qquad
\tilde{\Pi}_R\big([\ket{\bf{n}}^{\text{dr}}_{\text{matter}}]\big)
=
\Pi_R \ket{\bf{n}}^{\text{dr}}_{\text{matter}}
\end{equation}
is a linear isomorphism. Extending trivially on the loop sector yields the canonical identification
\begin{equation}
\mathcal{H}_{\text{pn}}/
\ker(\mathbbm{1}_{\text{loops}}\otimes\Pi_R)
\;\simeq\;
\mathcal{H}_{\text{loops}} \otimes \mathcal{H}_R ,
\end{equation}
between the coarse-grained physical space of the vacuum code, on the left hand side, and the kinematical Hilbert space of the pure gauge theory, on the right hand side.

The quotient becomes unnecessary when all matter species have finite-dimensional Hilbert 
spaces. In that case the decomposition of characters in terms of the generators $\rho_i$ 
is unique, so distinct occupation configurations necessarily induce distinct flux 
configurations. The map $\Pi_R$ is therefore already injective and the kernel is trivial. 
The quotient construction nevertheless remains formally valid, but reduces to the identity 
in this case.

Under the quotient identification, a dressed Wilson line product operator $\widetilde{W}^{\bm{\chi}(\bm{k})}$ descends to the bare Wilson line operator $W^{\bm{\chi}(\bm{k})}$ in Eq.~\eqref{prodWilsonlineswithchis}:
\begin{equation}
    \Theta(\widetilde{W}^{\bm{\chi}(\bm{k})}) \coloneqq \tilde\Pi_R\,\widetilde{W}^{\bm\chi(\bm{k})}\,\tilde\Pi_R^\dag=W^{\bm{\chi}(\bm{k})}\,.
\end{equation}
Thanks to the Gauss law in Eq.~\eqref{eq_Gaussvacuumcode}, the induced operator has the same character-valued divergence as the matter charge created by $\widetilde{W}^{\bm{\chi}(\bm{k})}$. Thus, $\Theta$ preserves the coarse-grained syndrome and maps vacuum code errors to pure gauge Gauss law code errors with identical charge profile. 

Finally, let us come to the discussion of maximal sets of correctable errors. To this end, note that the total charge map in Eq.~\eqref{eq_vactotalcharge} once more defines a discrete fiber bundle
\begin{align}
    {\bm Q}:\mathbb{Z}^{r N_L}\to \hat{G}^{N_V-1}\,,\qquad \bm{k}\mapsto \bm{Q}(\bm{k})\,.
\end{align}
since, as noted earlier, many $\bm{k}$ configurations yield the same total charge configuration. Now, let 
\begin{equation}
    s_{\text{vac}}: \hat{G}^{N_V -1} \to \{ \bm{k} \in \mathbb{Z}^{N_Lr} | k_{\ell, i}  \in \mathbb{Z}_{D_i} \enspace \forall \ell \in \mathcal{L}, i \in \mathbb{Z}_r  \}, \qquad \bm{Q}(s_{\text{vac}}(\bm{q})) = \bm{q}, \enspace \forall\, \bm{q} \in \hat{G}^{N_V - 1}
\end{equation}
be one of its sections, indexing a maximal correctable set in the coarse-grained vacuum code,
\begin{equation}
    \mathcal{E}_{s_{\text{vac}}}^{\text{vac}} = \{ \widetilde{W}^{s_{\text{vac}}(\bm{q})}    \mid \bm{q} \in \hat{G}^{N_V - 1}  \} \,,
\end{equation}
and define:
\begin{equation}
    s_{\text{GL}}: \hat{G}^{N_V - 1} \to \hat{G}^{N_L}, \qquad s_{GL}(\bm{q}) = {\bm{\chi}}(s_{\text{vac}}(\bm{q}))\,.
\end{equation}
Then, via Eqs.~\eqref{eq_Gaussvacuumcode} and the pure gauge Gauss law map~\eqref{boundarymap}, we have
\begin{equation}
    \partial(s_{\text{GL}}(\bm{q})) = \bar{\bm{q}}\,,
\end{equation}
so $s_{\text{GL}}$ maps the coarse-grained vacuum code syndrome to the pure gauge Gauss law code syndrome. The corresponding maximal correctable set~\eqref{eq_puremaxset} in the pure gauge Gauss law code is therefore
\begin{equation}
    \mathcal{E}_{s_{\text{GL}}} = \{ W^{s_{\text{GL}}(\bm{q})} \mid \bm{q} \in \hat{G}^{N_V - 1}    \},
\end{equation}
and by construction
\begin{equation}
    \Theta(\mathcal{E}_{s_{\text{vac}}}^{\text{vac}}) = \mathcal{E}_{s_{\text{GL}}}.
\end{equation}
Thus, every maximal correctable set in the vacuum code induces a maximal correctable set in the pure gauge Gauss law code. This relation is many-to-one: two vacuum code sections $s_{\text{vac}}, s_{\text{vac}}'$ give rise to the same Gauss law code section if and only if they induce the same gauge flux on the vertices
\begin{equation}
    \partial\bm{\chi}^{({\bf n}(s_{\text{vac}}(\bm{q})))} = \partial\bm{\chi}^{({\bf n}(s_{\text{vac}}'(\bm{q})))} \qquad \forall \,\bm{q} \in \hat{G}^{N_V - 1}.
\end{equation}
In terms of maximal sets, we therefore declare 
\begin{equation}
    \mathcal{E}_{s_{\text{vac}}}^{\text{vac}} \sim \mathcal{E}_{s_{\text{vac}}'}^{\text{vac}} \iff  \Theta(\mathcal{E}^{\text{vac}}_{s_{\text{vac}}})=\Theta(\mathcal{E}^{\text{vac}}_{s_{\text{vac}}'}).
\end{equation}
With this equivalence relation, the induced map $\Theta$ becomes bijective. Therefore, maximal correctable sets in the Gauss law code are in one-to-one correspondence with equivalence classes of maximal correctable sets in the vacuum code. The equivalence identifies different lifts that differ only by occupation number degeneracies invisible to the pure gauge theory data, whenever $G$ has a continuous component. 

This equivalence relation implements the operational restriction that we may distinguish total vertex charges, however, not the individual species contributions to it. Under this restriction, the vacuum code is unitarily equivalent to the pure gauge Gauss law code also when $G$ has a continuous component, and the map $\Theta$ establishes a correspondence between their maximal correctable sets after coarse-graining. Whenever $G$ is finite, the two codes are unitarily equivalent without coarse-graining.

\subsubsection{Fermionic vacuum codes}
\label{sec:fermi-vac}
Lastly, we turn to  codes based on fermionic matter residing in a gauge-invariant vacuum state. To describe the fermionic degrees of freedom, we once more invoke the staggered construction of section~\ref{sssec_fermGauss}, except that we now use the dressed matter vacuum as the code space and the perspective-neutral space~\eqref{pn_bosonic}, described in Eq.~\eqref{eq_Hpnferm} as the code's physical space.
As noted in section~\ref{sssec_basicvacuum}, correctable errors must be extracted from the dressed fermionic matter algebra, which is given in Eqs.~\eqref{eq_inffermalg} and~\eqref{eq_finfermalg} for the cases of infinite and finite cardinality of $G$, respectively. Here, we will only focus on the differences with the bosonic case.

The main modification arises from the staggered realization of the matter gauge transformations in Eq.~\eqref{eq_fermgaugetr}, which constitute the stabilizers of the fermionic vacuum code, cf.~\eqref{eq_vacuumcodespace}.

In the vacuum, one has $n_v = c_v$, so the state is gauge-invariant. The code space is thus of the form (cf.~\eqref{eq_fermdrstates})
\begin{equation}
    \Hil_{\rm vac}=\Hil_{\rm loops}\otimes_R\ket{\bf c}^{\rm dr}_{\rm matter}\,,
\end{equation}
where
\begin{equation}
\ket{\bf c}^{\rm dr}_{\rm matter}=\ket{\bf 1}_R\otimes\ket{{\bf n=c}}_{\rm matter}\,,\qquad\qquad \ket{\bf n=c}_{\rm matter}=\widehat{\otimes}_{v\in\mathcal{V}}\ket{n_v=c_v}_v\,.
\end{equation}

Errors are given by combinations of the dressed Wilson lines of the form $\psi_v\,W^{\chi_F}_{\gamma_R[v,v']}\,\psi^\dag_{v'}$ and, only in the $D_F<\infty$ case, the star operators in Eq.~\eqref{eq_fermstar}. These errors must shift the occupation numbers away from their vacuum values. Because of the staggered filling, the allowed shifts are fixed,
\begin{equation}
    \Delta N_v = \begin{cases}
        + 1, & v \text{ even} \\
        -1, & v \text{ odd}
    \end{cases}\,,
\end{equation}
so only odd sites can be hit by a $\psi_v$ and only even ones by a $\psi_v^\dag$. Accordingly, a dressed Wilson line acts nontrivially on the vacuum only if the shift in vertex charge obeys
\begin{equation}\label{eq_fermvacadmiss}
    (\partial \bm{\chi})_v \in \{1, \bar{\chi}_F^{\Delta N_v} \}
\end{equation}
with $\Delta N_v$ given above. If this condition fails at some vertex, the operator necessarily annihilates the vacuum and therefore cannot constitute a correctable error. 

Let $\mathcal{F}_L \subset \hat{G}^{N_L}$ denote the set of character-valued link data $\bm{\chi}$ satisfying this condition and $\mathcal{F}_{V} \coloneqq \partial(\mathcal{F}_L) \subset \hat{G}^{N_V-1}$ its image under the Gauss law map.\footnote{As we remain gauge-invariant in vacuum codes, it suffices to check the syndrome at only $N_V-1$ vertices due to the global neutrality constraint Eq.~\eqref{global_neutrality_fermions}.} The syndrome of a dressed Wilson line error $\widetilde{W}^{\bm{\chi}}$ with $\bm{\chi} \in \mathcal{F}_L$ is given by $\partial \bm{\chi} \in \mathcal{F}_V$. Accordingly, the structure of maximal sets of correctable dressed Wilson line errors is encoded in the map
\begin{equation}
    \partial_{F, \text{vac}}: \mathcal{F}_L \to \mathcal{F}_V, \qquad \partial_{F, \text{vac}}(\bm{\chi}) \coloneqq \partial \bm{\chi},
\end{equation}
that is, the Gauss law map Eq.~\eqref{boundarymap} extended to all vertices and restricted to the admissible flux configurations on the links. In analogy with the pure gauge Gauss law code, this map constitutes a discrete fiber bundle, and maximal correctable sets of dressed Wilson lines are parametrized by its sections.

Comparing with the pure gauge Gauss law code, whose maximal correctable sets of Wilson lines are parametrized by sections of $\partial: \hat{G}^{N_L} \to \hat{G}^{N_V-1}$, restriction to $\mathcal{F}_L$ identifies \emph{a priori} different sections of $\partial$. Consequently, there are multiple maximal sets in the pure gauge Gauss law code that correspond to the same maximal set in the fermionic vacuum code, whereas going in the opposite direction, we obtain submaximal correctable sets from every maximal set. At the level of individual errors, however, the correspondence becomes one-to-one on $\mathcal{F}_L$: for every $\bm{\chi} \in \mathcal{F}_L$, the Wilson line $W^{\bm{\chi}}$ admits a unique admissible dressing $\widetilde{W}^{\bm{\chi}}$ that does not annihilate the vacuum, and one can similarly undress every such admissible $\widetilde{W}^{\bm{\chi}}$ by simply removing the matter operators $\psi_v$ and $\psi_{v}^{\dagger}$.

This structural equivalence should be compared with the bosonic case. There, unitary equivalence with the pure gauge Gauss law code is obtained only after coarse-graining the syndrome, reflecting the fact that distinct matter configurations may induce the same charge profile. In the fermionic case, the charge is uniquely determined by the local occupation, so this ambiguity does not arise and no coarse-graining is required. Instead, the admissibility condition Eq.~\eqref{eq_fermvacadmiss} restricts the set of allowed Wilson lines and consequently, the correspondence does not preserve maximal sets.

\subsection{Hybrid Gauss-law-vacuum codes}\label{ssec_hybrid}
So far, we have only considered codes that are either Gauss law or vacuum codes. However, it is clear that one can also mix the structures, i.e.\ have codes which are hybrids of the two types. There are many ways in which this may be done. 

For example, one may place \emph{some}, but not all matter degrees of freedom in a dressed vacuum state. For the vacuum part of the code one uses gauge-invariant matter excitations as errors, for the remainder Gauss law violations. 
Another option is to extend the error sets in our vacuum codes above by still taking the dressed matter vacuum as the code space, but the total kinematical Hilbert space as the physical space of the code, as in the case of Gauss law codes. Also in that case, there will be errors that are gauge-invariant and ones that violate Gauss's law.

In fact, the paradigmatic example of a hybrid Gauss-law-vacuum code is the toric code, and more generally the surface codes \cite{Kitaev:1997wr,Bravyi:1998sy}. Their code space is defined by two independent sets of stabilizers, the plaquette and star operators. In the standard gauge theory interpretation of these codes, the former correspond to Gauss's law and thus impose gauge invariance, while the former implement a flatness condition.\footnote{In the context of the QECC/QRF correspondence, reviewed in section~\ref{ssec_QECC/QRF}, surface codes have also been explored from the perspective of treating \emph{both} sets of stabilizers as gauge \cite[Sec.~6]{Carrozza:2024smc}.} Furthermore, both sets of operators sit inside a Hamiltonian, and the code spaces corresponds exactly to its ground  or vacuum state space.
Hence, errors that violate the flatness condition, but leave the Gauss law invariant, are gauge-invariant and may be viewed as vacuum code errors, while errors that violate the star operators are therefore Gauss law code errors. It should be noted, however, that the excitations beyond the vacuum correspond to emergent quasi-particles and not standard matter as in the codes discussed here. Surface codes are also often viewed as emergent gauge theories, where the gauge symmetry is not fundamental and only arises in a subset of physical states, in contrast to the bona fide gauge theories explored here. This highlights that the notion of hybrid code is a more general one and does not necessitate the inclusion of matter, nor a genuine gauge theory.

\section{Examples}
\label{sec:examples}
In this section, we illustrate the general constructions of Section~\ref{sec:general} in a series of concrete
settings, focusing on specific lattice gauge theories in which the underlying error-correction structure can be made fully explicit. We begin from the pure $\mathbb{Z}_2$ gauge theory on a triangular lattice, which already realizes a Gauss law code in a minimal setting and reduces to the three-qubit repetition code, so that the bit flip errors that would naturally arise in a simulation of the gauge dynamics align with the set of errors that the code corrects. Introducing bosonic matter in this system then leads to a corresponding vacuum code, which we relate explicitly to its Gauss law counterpart in the pure gauge sector. The same structure can be developed for scalar QED, which provides a continuum-inspired realization of Gauss law and vacuum codes. Finally, we  elaborate QED with staggered fermions on a two-dimensional square lattice as a further example within the same framework.

\subsection{Gauss law code from the pure gauge sector of a 1+1 dimensional $\mathbb{Z}_2$ theory}
\label{ssec:3qubit_puregauge}
This first example illustrates the general construction of Section~\ref{sssec_pureGauss} with a concrete and well-known code in QEC: the three-qubit repetition code. We consider a $\mathbb{Z}_2$ lattice gauge theory on a one-dimensional periodic lattice with three vertices (namely, $G=\mathbb{Z}_2 =\hat{G}$ and $N_V =3$). 
We label the vertices
\begin{equation}
\mathcal{V} = \{0,1,2\},
\end{equation}
and the oriented links
\begin{equation}
\mathcal{L} = \{[0,1],[1,2],[2,0]\},
\end{equation}
where the link $[i,i+1]$ is oriented from the vertex $i$ to the vertex $i+1$ (modulo $3$), as shown in Fig.~\ref{fig:triangle}.
\begin{figure}[h]
\centering    \includegraphics[width=0.3\textwidth]{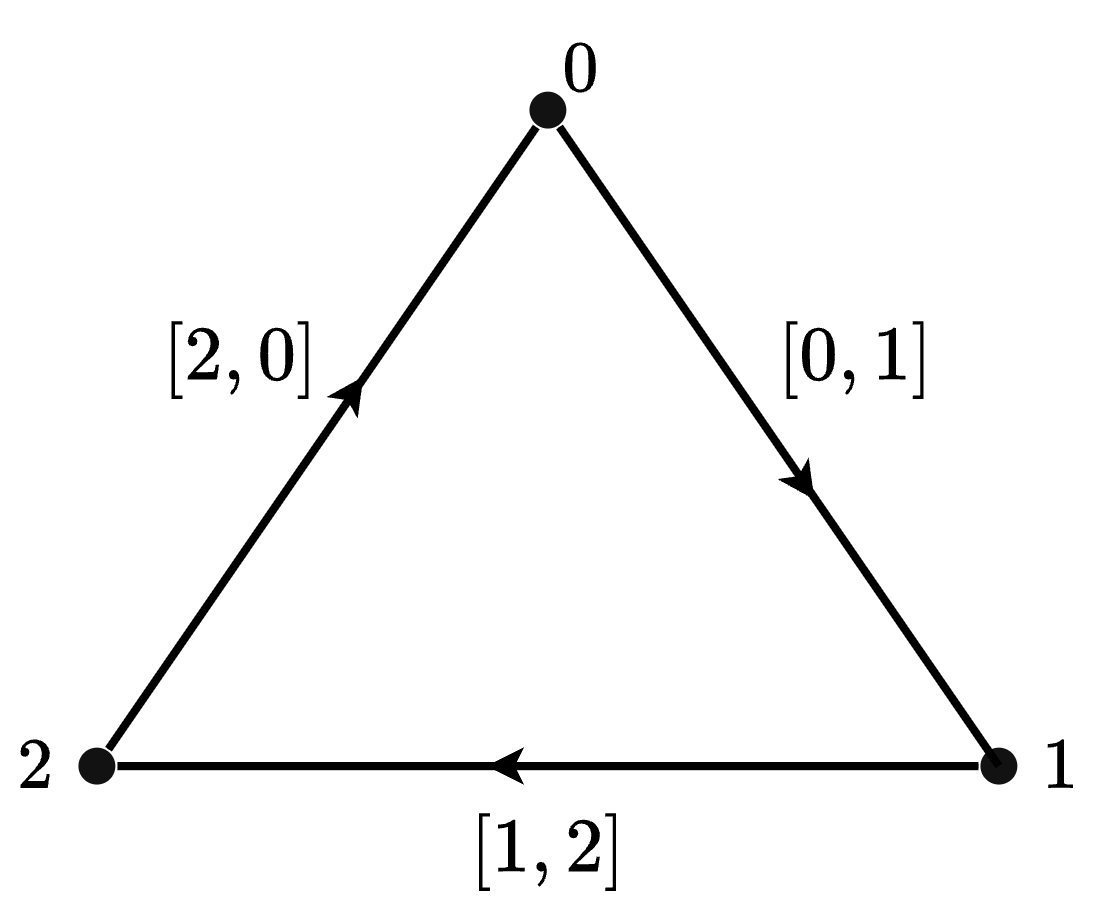}
\caption{$\mathbb{Z}_2$ gauge theory on a triangular lattice}
\label{fig:triangle}
\end{figure}

We first restrict attention to the pure-gauge sector, which we use to illustrate the Gauss law code construction, before turning to codes with bosonic matter in the next subsections.
The kinematical Hilbert space in Eq.~\eqref{eq_Hkinpure} specializes to 
\begin{equation}\label{eq_Z2Hgauge}
\mathcal{H}_{\text{gauge}}
=
\mathcal{H}_{[0,1]}
\otimes
\mathcal{H}_{[1,2]}
\otimes
\mathcal{H}_{[2,0]},
\end{equation}
where each link Hilbert space satisfies $\mathcal{H}_\ell \simeq \mathbb{C}^2$ and where we have identified $\Hil_\mrm{kin}$ with $\Hil_\mrm{gauge}$. The link algebra is therefore generated by Pauli operators
\begin{equation}
\mathcal{A}_\ell
=
\mathcal{B}(\mathcal{H}_\ell)
=
\langle X_\ell , Z_\ell \rangle,
\qquad
\ell\in\mathcal{L}.
\end{equation}
The eigenbases of $X_\ell$ and $Z_\ell$ correspond to the dual  bases for $\Hil_\ell$ from Section~\ref{sssec_codespacepure}. The magnetic (group) basis consists of states $\ket{\pm}_{\ell}$ labeled by the elements of the multiplicative group $G = \{+,-\} \simeq \mathbb{Z}_2$, with
\begin{equation}
X_\ell \ket{\pm}_\ell = \pm \ket{\pm}_\ell, \qquad Z_{\ell}\ket{\pm}_{\ell} = \ket{\mp}_{\ell}, 
\end{equation}
so that $Z_{\ell}$ acts as the group shift. The electric (dual) basis, by contrast, diagonalizes $Z_{\ell}$ and shifts under $X_{\ell}$,
\begin{equation}
Z_\ell \ket{\varepsilon}_\ell
=
(-1)^{\varepsilon}\ket{\varepsilon}_\ell,
\qquad
X_\ell \ket{\varepsilon}_\ell
=
\ket{\varepsilon \oplus 1}_\ell,
\qquad 
\varepsilon \in \{0,1\},
\end{equation}
where $\oplus$ denotes addition modulo $2$. These are precisely the group and (Pontryagin) dual bases of Eq.~\eqref{groupvselectricbases}, and the Pauli operators $(Z_{\ell}, X_{\ell})$ realize the Weyl pair $(U_{\ell}^-, W_{\ell}^{\chi_1} )$ of Eq.~\eqref{eq_kinRS}, with $\chi_1$ the nontrivial character of $\mathbb{Z}_2$, defined as $\chi_1(\pm) = \mp$.
(The remaining operators $U_\ell^+$ and $W_\ell^{\chi_0}$ are both equal to the identity operator on link $\ell$.)

\subsubsection{Stabilizers, code subspace and logical operators}
The structure group being $\mathbb{Z}_2$, gauge transformations at vertex $i$ (cf.~Eq.~\eqref{eq_Gausslawpure}) furnish a faithful representation of $\mathbb{Z}_2$ on $\mathcal{H}_{\text{gauge}}$:
\begin{equation}
    U_i^+ = \mathbbm{1}_{\text{gauge}}, \qquad U_i^- = Z_{[i-1,i]}Z_{[i, i+1]}.
\end{equation}
Gauss's law then demands that physical states be invariant under all gauge transformations,
\begin{equation}
    U_i^{\pm} \ket{\Psi}_{\text{pn}}= \ket{\Psi}_{\text{pn}}, \qquad \forall ~ i = 0,1,2,
\end{equation}
which is nontrivially imposed by the $U_i^-$. The operators $\{ U_i^- \}_{i=0}^2$ thus generate the stabilizer group of this Gauss law code,
\begin{equation}
    \mathcal{G} = \left\langle  Z_{[01]}Z_{[12]}, Z_{[12]}Z_{[20]}  \right\rangle \simeq \mathbb{Z}_2 \times \mathbb{Z}_2,
\end{equation}
where the third generator $U_2^- = Z_{[20]}Z_{[01]}$ is redundant due to the periodic boundary conditions,
\begin{equation}
    \prod_{i=0}^2 U_i^- = \prod_{i=0}^2 Z_{[i-1,i]}Z_{i,i+1} = \mathbbm{1}_{\text{gauge}}.
\end{equation}

The Gauss law code construction identifies $\mathcal{H}_{\text{gauge}}$ as the physical Hilbert space of the code and its perspective-neutral subspace,
\begin{equation}\label{eq_Z2PGHpn}
    \mathcal{H}_{\text{pn}} \coloneqq \{ \ket{\Psi} \in \mathcal{H}_{\text{gauge}} ~ | ~ U_i^- \ket{\Psi} = \ket{\Psi}, \quad \forall i = 0,1,2  \},
\end{equation}
as the code subspace. In particular, we recognize $\mathcal{G}$ as the usual stabilizer group of a three-qubit repetition code, and the code space is spanned by the logical codewords
\begin{equation}\label{eq_Z2PGCodewords}
\ket{\bar{0}}
=
\ket{0}_{[01]}
\otimes
\ket{0}_{[12]}
\otimes
\ket{0}_{[20]},
\end{equation}
\begin{equation}
\ket{\bar{1}}
=
\ket{1}_{[01]}
\otimes
\ket{1}_{[12]}
\otimes
\ket{1}_{[20]} .
\end{equation}
 
Next, we examine the logical operators of the code.
Since both codewords have uniform electric-field configurations, any single-link operator $Z_\ell$ has the same action on the code space. For instance,
\begin{equation}
Z_{[01]} \ket{\bar{0}} = \ket{\bar{0}},
\qquad
Z_{[01]} \ket{\bar{1}} = -\ket{\bar{1}},
\end{equation}
and similarly for the other links. We therefore identify the logical operator
\begin{equation}
\bar{Z} = Z_{\ell} \Pi_{\text{pn}},
\end{equation}
where the code projector $\Pi_{\text{pn}}$ was defined in Eq.(\ref{Pipn}) and may be explicitly rewritten in terms of the stabilizer generators as 
\begin{equation}
\Pi_{\text{pn}}
=\prod_{i=0}^2 \frac{\mathbbm{1}_{\text{gauge}}  +Z_{[i-1,i]} Z_{[i,i+1]}}{2} .
\end{equation}
The Wilson loop operator $X_{[01]} X_{[12]} X_{[20]}$ flips all three links and exchanges the two codewords,
\begin{equation}
X_{[01]}X_{[12]}X_{[20]}\ket{\bar{0}}=\ket{\bar{1}},
\qquad
X_{[01]}X_{[12]}X_{[20]}\ket{\bar{1}}=\ket{\bar{0}} .
\end{equation}
This operator commutes with every stabilizer but is not itself a stabilizer, and therefore acts as the logical bit-flip operator
\begin{equation}\label{eq_Z2logX}
\bar{X}
=
X_{[01]}X_{[12]}X_{[20]} \Pi_{\text{pn}} .
\end{equation}

\subsubsection{Gauss law map, error syndromes and maximal correctable sets }
We now examine how Pauli errors are detected by the stabilizers. Up to an overall phase, a general Pauli error (whether correctable or not) can be written as
\begin{equation}
E(\bm{k},\bm{z})
=
X^{\bm{k}} Z^{\bm{z}}
=
\bigotimes_{\ell\in\mathcal{L}}
X_\ell^{k_\ell} Z_\ell^{z_\ell},
\end{equation}
with $\bm{k},\bm{z}\in(\mathbb{Z}_2)^3$. The syndrome of such an error is determined by its commutation relations with the nontrivial stabilizers $U_i^-$. Since these operators are $Z$-type, the syndrome depends only on the $X$-part, $\bm{k}$. One finds
\begin{equation} \label{eq:3qubitsyndrome}
U_i^- E(\bm{k},\bm{z})
=
(-1)^{(\nabla \bm{k})_i}
E(\bm{k},\bm{z}) U_i^-,
\end{equation}
where
\begin{equation}\label{eq_nabla}
    (\nabla \bm{k})_i = k_{[i-i, i]} \oplus k_{[i,i+1]}
\end{equation}
gives the discrete divergence of $\bm{k}$ across vertex $i$. This expression may be understood as a concrete realization of the Gauss law map 
\begin{equation}
\partial: \hat{G}^{\times N_L}=(\mathbb{Z}_2)^3 \rightarrow \hat{G}^{\times N_V - 1} =(\mathbb{Z}_2)^2
\end{equation}
in Eq.~\eqref{boundarymap} extended to every vertex. 
In the present setting, characters are labeled by $\bm{k} \in \{0,1\}^3 \simeq (\mathbb{Z}_2)^3$ and act on group elements $\bm{g}~\in~\{\pm\}^3~\simeq~(\mathbb{Z}_2)^3$ according to 
\begin{equation}\label{eq_coordtochar}
    \bm{\chi}[\bm{k}](\bm{g}) = \prod_{\ell} g_{\ell}^{k_{\ell}}.
\end{equation}
Here, $\bm{\chi}[\bm{k}]$ should be understood as assigning an irreducible representation of $\mathbb{Z}_2$ to each link of the lattice, rather than as a character of the gauge group acting at the vertices. 
The Gauss law map $\partial$ then combines this link data into a character of the gauge group acting on the vertices. Concretely, one finds
\begin{equation}\label{eq_partialtonablapure}
    \partial \bm{\chi}[\bm{k}] = \bm{\rho}[\nabla \bm{k}] \qquad \text{with} \qquad
    \begin{aligned}
    (\bm{\rho}[\nabla \bm{k}])_i ~ : ~ \pm ~ \mapsto ~ (\pm 1)^{(\nabla \bm{k})_i},
    \end{aligned}
\end{equation}
where $\bm{\rho}[\nabla \bm{k}]$ is the vertex-supported character obtained from the link representation data $\bm{\chi}$ via the Gauss law map.
Thus, in this example, $\partial$ is  effectively replaced by the discrete divergence $\nabla$, and characters $\bm{\chi}[\bm{k}]$, by elements of the additive group $\bm{k} \in \{0,1\}^3$.

Per Eq.~\eqref{eq:3qubitsyndrome}, the quantity $\nabla \bm{k}$ therefore records the pattern of Gauss-law violations created by the error---in other words, the error syndrome. In gauge-theory language, it specifies the $\mathbb{Z}_2$ charge configuration that would need to be inserted at the vertices to restore Gauss’s law. Because the lattice is periodic,
\begin{equation}\label{eq_Z2neutral}
\sum_{i=0}^2 (\nabla \bm{k})_i = 0 \mod 2,
\end{equation}
so charges can only be created in pairs. The image of the discrete divergence is therefore
\begin{equation} \label{eq:range_nabla}
\nabla((\mathbb{Z}_2)^3)
=
\{
(0,0,0),
(1,1,0),
(1,0,1),
(0,1,1)
\},
\end{equation}
corresponding to the four possible syndrome sectors. Two errors share the same syndrome precisely when their $X$ components differ by an element of the kernel of $\nabla$. A direct computation shows
\begin{equation}\label{eq_kernabla}
\ker(\nabla)
=
\{(0,0,0),(1,1,1)\}.
\end{equation}
The single nontrivial kernel element corresponds to the Wilson loop operator $\bar{X}$: a logical operator and a fiber translation in the Gauss law bundle with the base space $\hat{G}^{\times N_V -1}=(\mathbb{Z}_2 )^2$ and the fiber $\hat{G}^{\times N_L-N_V +1}=\mathbb{Z}_2 $.
This structure directly instantiates Proposition~\ref{prop_KLpure} and the discussion above Eq.~\eqref{eq_sectionspureGauss}: two Wilson line errors $X^{\bm{k}_a}$ and $X^{\bm{k}_b}$ are jointly correctable if $\nabla \bm{k}_a \neq \nabla \bm{k}_b$, i.e.\ they lie in distinct fibers of $\nabla$.
A maximal correctable set must contain exactly one representative from each fiber $\nabla^{-1}(\bm{q})$ with $\bm{q}\in\hat{G}^{\times N_V}=(\mathbb{Z}_2 )^3$. 
Choosing such representatives is equivalent to specifying a section of the divergence map,
\begin{equation} \label{eq:sGLZ2}
s_{\nabla} :
\nabla((\mathbb{Z}_2)^3)
\to
(\mathbb{Z}_2)^3,
\qquad
\nabla(s_{\nabla}(\bm{q}))=\bm{q},
\end{equation}
which induces a maximal correctable set (c.f.~Eq.\eqref{eq_puremaxset})\footnote{Note that, since the $Z$-operators do not affect the syndrome, each non-identity, $X$-error in $\mathcal{E}_{s_{\nabla}}$ may be multiplied with any $Z$-type error without affecting the syndrome, trivially augmenting the number of these possible maximal sets.}
\begin{equation}\label{eq_Z2puremaxset}
    \mathcal{E}_{s_{\nabla}}^{\text{GL}} \coloneqq \{ X^{s_{\nabla}(\bm{q})} \mid \bm{q} \in \nabla((\mathbb{Z}_2)^3)    \}.
\end{equation}
For example, choosing
\begin{equation}
    \begin{array}{c|c|c}
        \bm{q} & s_\nabla(\bm{q}) & X^{s_\nabla(\bm{q})} \\
        \hline
        (0,0,0) & (0,0,0) & \mathbbm{1} \\
        (1,1,0) & (1,0,0) & X_1 \\
        (0,1,1) & (0,1,0) & X_2 \\
        (1,0,1) & (0,0,1) & X_3
    \end{array}
\end{equation}
reproduces the standard correctable error set for a three-qubit repetition code.
From the perspective of a syndrome $\bm{q}$ that signals ``missing charge'' at vertices $v_i$ and $v_{i+1}$, this choice amounts to dressing the charge with an erroneous $X$ operator on the link $[i,i+1]$.
However, one could alternatively choose as the section
\begin{equation}\label{eq_blablabla}
    \begin{array}{c|c|c}
        \bm{q} & s_\nabla(\bm{q}) & X^{s_\nabla(\bm{q})} \\
        \hline
        (0,0,0) & (0,0,0) & \mathbbm{1} \\
        (1,1,0) & (0,1,1) & X_2 X_3 \\
        (0,1,1) & (1,0,1) & X_1 X_3 \\
        (1,0,1) & (1,1,0) & X_1 X_2
    \end{array}
\end{equation}
This choice also yields a correctable, albeit nonstandard error set consisting of double bit flips.
Here, charge at vertices $v_i$ and $v_{i+1}$ gets dressed with $X$ operators on the links $[i-1,i]$ and $[i+1,i+2]$. 

Finally, each section determines a recovery channel acting on the gauge Hilbert space. Let $P_{\bm q}^{\text{GL}}$ denote the projector onto the stabilizer eigenspace with syndrome $\bm q$. The recovery channel associated with the section $s_{\nabla}$ is
\begin{equation}\label{eq_Z2purerec}
\mathcal{R}_{s_{\nabla}}^{\text{GL}}(\rho)
=
\sum_{\bm q \in \nabla((\mathbb{Z}_2)^3)}
X^{s_{\nabla}(\bm q)}
\,P_{\bm q}^{\text{GL}}\,
\rho\,
P_{\bm q}^{\text{GL}}\,
X^{s_{\nabla}(\bm q)} .
\end{equation}
This channel first resolves the syndrome sector and then applies the correction operator associated with the representative chosen by the section.
Different sections therefore correspond to different choices of correctable error set, each of which entails its own recovery operation.

\subsubsection{Code parameters}
The code is classical, since it cannot correct $Z$-type errors. More precisely, it realizes a classical bit-flip code, namely the three-qubit repetition code. It has parameters $[n,k,d_X] = [3,1,3]_X$, with $n=3$ physical qubits (links) and $k=1$ encoded logical degree of freedom. The $X$-distance $d_X$ is the minimal weight of a nontrivial logical $X$ operator, and the $Z$-distance $d_Z$ is the minimal weight of a nontrivial logical $Z$ operator. The logical operator $\bar{X}$ has weight three, so $d_X=3$, whereas  any single-link $Z_\ell \Pi_{\mrm{pn}}$ already implements a logical $\bar{Z}$, implying $d_Z=1$.

\subsection{Bosonic Gauss law code in a $1+1$ dimensional $\mathbb{Z}_2$ theory}
\label{ssec:3qubit_gauss}
We now extend the previous construction by including bosonic matter degrees of freedom on the lattice's vertices. To that end, let each vertex of the lattice in Fig.~\ref{fig:triangle} support a two-dimensional Hilbert space,
\begin{equation}
    \mathcal{H}_{\text{matter}} = \mathcal{H}_0 \otimes \mathcal{H}_1 \otimes \mathcal{H}_2 ,
\end{equation}
equipped with the corresponding Pauli operators $X_i$ and $Z_i$.
We let the $Z$-basis label the occupation number of the site: the state $\ket{0}$ corresponds to an empty vertex while $\ket{1}$ represents the presence of a particle. In terms of the general language introduced in Sec.~\ref{sssec_bosGauss}, $Z_i$ measures the $\mathbb{Z}_2$ charge located at vertex $i$. The kinematical Hilbert space of the gauge theory is therefore
\begin{equation}
    \mathcal{H}_{\text{kin}} =
    \mathcal{H}_{\text{gauge}} \otimes \mathcal{H}_{\text{matter}}
    \simeq (\mathbb{C}^2)^{\otimes 6}.
\end{equation}
\subsubsection{Stabilizers, code subspace and logical operators}
Here, we identify $\Hil_\mrm{kin}$ with the physical (i.e., computational) space of the Gauss law code, while we declare the code space to be its perspective-neutral (gauge-invariant) subspace,
\begin{equation}\label{eq_Hpn_bosQED}
   \mathcal{H}_{\text{pn}} = \{ \ket{\Psi} \in \mathcal{H}_{\text{kin}} \mid U_i^{-} \ket{\Psi} = \ket{\Psi}   \},
\end{equation}
where the nontrivial gauge transformations
\begin{equation}
    U_i^{-} \coloneqq Z_{[i-1, i]}Z_{[i, i+1]}Z_i
\end{equation}
now extend to the matter degrees of freedom (cf. Eq.~\eqref{generalGTs}).
Also note that we omit the trivial gauge transformations $U_i^+ = \mathbbm{1}$ from Eq.~\eqref{eq_Hpn_bosQED}.
These operators generate the stabilizer group of the code,
\begin{equation}
    \mathcal{G} = \left\langle   U_0^-, U_1^-, U_2^- \right\rangle \simeq \mathbb{Z}_2 \times \mathbb{Z}_2 \times \mathbb{Z}_2,
\end{equation}
and the code space in Eq.~\eqref{eq_Hpn_bosQED} is their joint $+1$ eigenspace. In other words, states in $\mathcal{H}_{\text{pn}}$ satisfy Gauss's law at every vertex, which relates the $\mathbb{Z}_2$ flux on adjacent links to the matter charge located at the common vertex. In particular, if the flux values on the two links adjacent to a vertex differ, the vertex must carry a particle so that Gauss's law is satisfied. A useful consequence of this constraint is that once the flux configuration $\bm{\varepsilon} = (\varepsilon_{[i,i+1]})_{i=0}^2$ on the links is specified, the matter configuration is uniquely fixed. Indeed, one may verify that the states
\begin{equation}
    \ket{\bm{\varepsilon}}_{\text{pn}}
    \coloneqq
    \ket{\bm{\varepsilon}}_{\text{gauge}}
    \otimes
    \ket{\nabla \bm{\varepsilon}}_{\text{matter}},
    \label{Z2_basis}
\end{equation}
with $\nabla$ defined in Eq.~\eqref{eq_nabla}, form an orthonormal basis of the full gauge-invariant Hilbert space $\mathcal{H}_{\text{pn}}$. Thus, the gauge-invariant Hilbert space is isomorphic to that of three qubits,
$\mathcal{H}_{\text{pn}} \simeq (\mathbb{C}^2)^{\otimes 3}$.

This reduction in dimension is expected: Gauss's law imposes three independent constraints on the six qubits of the kinematical Hilbert space. The structure of the physical operator algebra can be read off directly from this description. It is generated by the flux operators $Z_\ell$ on the links together with the dressed Wilson line operators
\begin{equation}
    \tilde{X}_{[i,i+1]} = X_i X_{[i,i+1]} X_{i+1},
\label{eq:dressed_Z2boson}
\end{equation}
which create a unit of flux on the link while simultaneously producing a pair of matter excitations at its endpoints so that Gauss's law remains satisfied. These operators therefore represent the gauge-invariant Pauli operators of the theory, and so the logical algebra of the code is generated by
\begin{equation}
    \mathcal{A}_{\text{pn}}
    =
    \langle
    \tilde{X}_{[i,i+1]} \Pi_{\text{pn}},
    Z_{[i,i+1]} \Pi_{\text{pn}} \mid i = 0,1,2
    \rangle ,
\end{equation}
where $\Pi_{\text{pn}}$ denotes the projector onto $\mathcal{H}_{\text{pn}}$, now including the matter part of the gauge transformations. In terms of stabilizer generators, it is explicitly given by
\begin{equation}
\Pi_{\text{pn}}
=\prod_{i=0}^2 \frac{\mathbbm{1}  +Z_{[i-1,i]} Z_{[i,i+1]}Z_i }{2} .
\end{equation}
Although the operators $Z_i$ acting on the matter sites are themselves gauge invariant, they are not independent within the physical subspace. Indeed, Gauss's law implies the operator identity
\begin{equation}\label{eq_Z2bosGausslaw}
    Z_{[i-1,i]} Z_{[i,i+1]} \Pi_{\text{pn}}
    =
    Z_i \Pi_{\text{pn}},
\end{equation}
so the matter charge at each vertex is fully determined by the adjacent link fluxes and, in particular $Z_i \Pi_{\text{pn}} \in \mathcal{A}_{\text{pn}}$.

\subsubsection{Code space refactorization and subsystem structure}
Before turning to more general errors, it is useful to isolate the structure already present at the level of bare Wilson line errors, in other words, $X$-type errors on the links alone. This will allow us to make contact with the subsystem code structure outlined in Sec~\ref{sec:bosmaxsets}.

Eq.~\eqref{Z2_basis} provides an orthonormal basis for $\mathcal{H}_{\text{pn}}$ labeled by flux configurations $\bm{\varepsilon}$, with the matter occupation numbers fixed by Gauss's law as $\nabla \bm{\varepsilon}$. The discrete divergence does not distinguish configurations that differ by the constant shift $(1,1,1)$, so for each value of $\nabla \bm{\varepsilon}$ there are exactly two compatible flux configurations. Writing
\begin{equation}
    \omega(\bm{\varepsilon}) = \varepsilon_{[01]} \oplus \varepsilon_{[12]} \oplus \varepsilon_{[20]},
\end{equation}
the pair $(\omega, \nabla \bm{\varepsilon})$ uniquely specifies $\bm{\varepsilon}$. It follows that we have an isomorphism $\mathcal{H}_{\text{pn}} \simeq \mathcal{H}_{\text{loop}} \otimes_R \mathcal{H}^{\text{dr}}_{\text{matter}}$, where $\mathcal{H}_{\text{loop}} \simeq \mathbb{C}^2$ and $\mathcal{H}^{\text{dr}}_{\text{matter}} \simeq (\mathbb{C}^2)^2$ are spanned by orthonormal states labeled by the possible values of $\omega$ and $\nabla \bm{\varepsilon}$, respectively.\footnote{Recall that the range of $\nabla$ is only 4-dimensional here; see Eq.~\eqref{eq:range_nabla}.}
Concretely, the isomorphism maps
\begin{equation} \label{eq:3qubit_otimesR}
\ket{\bm{\varepsilon}}_{\text{pn}} \to \ket{\omega(\bm{\varepsilon})}_{\text{loop}} \otimes_R \ket{\nabla \bm{\varepsilon}}^{\text{dr}}_{\text{matter}}.
\end{equation}
We use the symbol $\otimes_R$ to emphasize that this tensor product is different from the local tensor product structure intrinsic to $\mathcal{H}_{\text{pn}}$.

We now examine the Knill-Laflamme condition Eq.~\eqref{eq_KLstabcode} for the bare Wilson line operators $X^{\bm{k}} = \otimes_{\ell \in \mathcal{L}} X_\ell^{k_\ell}$.
These operators are ``bare'' in the sense that they are not accompanied by any charge excitations at the vertices, and thus map codewords out of the code subspace when $\bm{k} \notin \{(0,0,0),(1,1,1)\}$.
Since 
\begin{equation}
    X^{\bm{k}} \ket{\bm{\varepsilon}}_{\text{gauge}} = \ket{\bm{\varepsilon} \oplus \bm{k}}_{\text{gauge}},
\end{equation}
one finds
\begin{equation}
    \Pi_{\text{pn}} X^{\bm{k}_a}X^{\bm{k}_b} \Pi_{\text{pn}} = \sum_{\bm{\varepsilon} \in (\mathbb{Z}_2)^3} \ket{\bm{\varepsilon} \oplus \bm{k}_a \oplus \bm{k}_b} \bra{\bm{\varepsilon}}_{\text{pn}} \delta_{\nabla \bm{k}_a, \nabla \bm{k}_b} \Pi_{\text{pn}}.
\end{equation}
When $\nabla \bm{k}_a \neq \nabla \bm{k}_b$, the right-hand side vanishes and the Knill-Laflamme condition is satisfied for collections of such errors.
If instead $\nabla \bm{k}_a = \nabla \bm{k}_b$, then $\bm{k}_a \oplus \bm{k}_b \in \ker (\nabla) = \{(0,0,0), (1,1,1)   \}$.
For the nontrivial case in which $\bm{k}_a \neq \bm{k}_b$ and hence $\bm{k}_a \oplus \bm{k}_b = (1,1,1)$, then $X^{\bm{k}_a} X^{\bm{k}_b} \Pi_{\mrm{pn}}$ is equal to the Wilson loop $(X_{[0,1]}X_{[1,2]}X_{[2,0]}) \Pi_{\mrm{pn}}$, resulting in a violation of the Knill-Laflamme condition.
However, since $\nabla(\bm{k}_a \oplus \bm{k}_b) = 0,$ this operator leaves $\nabla \bm{\varepsilon}$ invariant and therefore acts trivially on the dressed matter factor.
Explicitly, with respect to the $\otimes_R$ factorization of Eq.~\eqref{eq:3qubit_otimesR}, we may write
\begin{equation}
    \Pi_\mrm{pn} (X_{[0,1]}X_{[1,2]}X_{[2,0]}) \Pi_{\mrm{pn}} = \bar{X}_\mrm{loop} \otimes_R \mathbbm{1}_\mrm{matter}^\mrm{dr} \; \Pi_\mrm{pn} \, .
\end{equation}
Hence, the Knill-Laflamme condition \emph{is} satisfied in the subsystem code sense of Eq.~\eqref{eq_KLsubsystem}.
In particular, the dressed-matter factor is perfectly correctable with respect to bare Wilson line errors, and the loop sector plays the role of the ``gauge subsystem'' \footnote{Not to be confused with ``gauge'' in the context of gauge theories.} in the language of subsystem codes \cite{aly2006subsystemcodes}.

\subsubsection{Maximal correctable sets as sections of the extended Gauss law map}
Next, we identify the maximal correctable error sets that consist of Pauli errors on both the links and vertices. We begin by computing the syndrome introduced by a composite error consisting of both vertex and link Pauli $X$ operators,
\begin{equation}\label{eq_Z2mattererrorsnoZ}
    E(\bm{k}, \bm{x}) \coloneqq \prod_{i=0}^2 (X_{[i,i+1]})^{k_{[i,i+1]}}(X_{i})^{x_i},
\end{equation}
ignoring $Z$ operators, as they commute with the stabilizers and thus do not contribute to the syndrome. Acting with this operator on the basis states from Eq.~\eqref{Z2_basis} gives
\begin{equation}
    E(\bm{k}, \bm{x}) \ket{\bm{\varepsilon}}_{\text{pn}} = \ket{\bm{\varepsilon} \oplus \bm{k}}_{\text{gauge}} \otimes \ket{ \nabla \bm{\varepsilon} \oplus \bm{x}}_{\text{matter}} ,
\end{equation}
and a subsequent gauge transformation at vertex $i$ yields
\begin{eqnarray}\label{eq_Z2bossyndromes}
    U_i^{-} E(\bm{k}, \bm{x})\ket{\bm{\varepsilon}}_{\text{pn}} &=& (-1)^{ (\nabla \bm{k})_i \oplus x_i} E(\bm{k}, \bm{x}) \ket{\bm{\varepsilon}}_{\text{pn}}.
\end{eqnarray}
Hence, the error $E(\bm{k}, \bm{x})$ produces a syndrome $\{ (-1)^{(\nabla \bm{k})_i \oplus x_i} \}_{i=0}^2$, which is controlled by the quantity $\nabla \bm{k} \oplus \bm{x}$.

With this motivation, we introduce the map
\begin{equation}
    \nabla_B: (\mathbb{Z}_2)^3 \times (\mathbb{Z}_2)^3 \to (\mathbb{Z}_2)^3,
\end{equation}
\begin{equation}
    \nabla_B(\bm{k}, \bm{x}) \coloneqq \nabla \bm{k} \oplus \bm{x},
\end{equation}
which returns the error syndrome corresponding to the error determined by $\bm{k}, \bm{x}$ in the parametrization in Eq.~\eqref{eq_Z2mattererrorsnoZ}. This is the coordinate realization of the bosonic Gauss law map $\partial_B$ from Eq.~\eqref{eq_bosGausslawmap}. To see this, recall that $\partial_B$ acts on pairs of characters $\bm{\chi} \in \hat{G}^{N_L}$, $\bm{\rho} \in \hat{G}^{N_V}$. In the present $\mathbb{Z}_2$ model, both the (matter) characters with support on the vertices and the (gauge) characters with support on the links are labeled by elements of $\{ 0, 1  \}^3$. Thus, we may replace $\bm{\chi}$ and $\bm{\rho}$ in the arguments of $\partial_B$ with $\bm{\chi}[\bm{k}]$ and $\bm{\rho}[\bm{x}]$, as defined in Eq.~\eqref{eq_coordtochar} and Eq.~\eqref{eq_partialtonablapure}, respectively. A direct evaluation then yields
\begin{equation} \label{eq:3qubit_gmap_boson}
    \partial_B(\bm{\chi}[\bm{k}], \bm{\rho}[\bm{x}]) = \bm{\rho}[\nabla \bm{k} \oplus \bm{x}] = \bm{\rho}[\nabla_B(\bm{k}, \bm{x})].
\end{equation}
Similarly to the pure gauge case, the periodic boundary conditions impose the constraint
\begin{equation} \label{eq:fluxconstraint}
    \sum_{i=0}^2 (\nabla \bm{k})_i = 0 \mod 2.
\end{equation}
In particular, the image of $\nabla$ is a strict subspace of $(\mathbb{Z}_2)^3$, and violations of Gauss's law arising from $\bm{k}$ alone occur in pairs. The inclusion of the matter contribution $\bm{x}$ extends the set of attainable syndromes to the full space $(\mathbb{Z}_2)^3$. This extension is captured by $\nabla_B$, which provides an explicit realization of the Gauss law bundle with bosonic matter described below Eq.~\eqref{eq_sameGaussviobos} (without the $Z$-errors that we are currently ignoring).

In this formulation, errors are organized into fibers labeled by charge configurations. In our example, these fibers are given by 
\begin{equation}
    \nabla_B^{-1}(\bm{q}) = \{ (\bm{k}, \bm{x}) \mid \nabla \bm{k} \oplus \bm{x} = \bm{q}    \}, \qquad \bm{q} \in (\mathbb{Z}_2)^3.
\end{equation}
The maximal sets of correctable errors are those that contain exactly one representative of each fiber. This is already apparent from the fact that two elements of the same fiber are related by elements in the kernel of $\nabla_B$,
\begin{equation}
    \ker(\nabla_B) = \{ (\bm{k}, \bm{x}) \mid \nabla \bm{k} = \bm{x} \} \, .
\end{equation}
For a pair  $(\bm{k}_a, \bm{x}_a), (\bm{k}_b, \bm{x}_b) \in \nabla^{-1}_B(\bm{q})$ that belongs to the same fiber, the errors that they parametrize rotate the basis states in Eq.~\eqref{Z2_basis} into one another,
\begin{equation}
    E(\bm{k}_a, \bm{x}_a)^{\dagger} E(\bm{k}_b, \bm{x}_b) \ket{\bm{\varepsilon}}_{\text{pn}} = \ket{\bm{\varepsilon} \oplus \bm{k}_a \oplus \bm{k}_b}_{\text{pn}},
\end{equation}
so they will not satisfy the Knill-Laflamme condition Eq.~\eqref{eq_KLstabcode}. Conversely, if $(\bm{k}_a, \bm{x}_a), (\bm{k}_b, \bm{x}_b)$ lie in different fibers of $\nabla$, the combination $E(\bm{k}_a, \bm{x}_a)^{\dagger}E(\bm{k}_b, \bm{x}_b)$ produces a nontrivial syndrome per Eq.~\eqref{eq_Z2bossyndromes}, and thus vanishes upon conjugation by the code projector. Indeed, a simple calculation verifies
\begin{equation}\label{eq_Z2bosKL}
    \Pi_{\text{pn}}E(\bm{k}_a, \bm{x}_a)^{\dagger}E(\bm{k}_b, \bm{x}_b) \Pi_{\text{pn}} = \sum_{\bm{\varepsilon} \in (\mathbb{Z}_2)^3} \ket{\bm{\varepsilon} \oplus \bm{k}_a \oplus \bm{k}_b}\bra{\bm{\varepsilon}}_{\text{pn}} \delta_{\nabla_B(\bm{k}_a, \bm{x}_a), \nabla_B(\bm{k}_b, \bm{x}_b)} \Pi_{\text{pn}},
\end{equation}
so there are only two ways in which two such errors may satisfy the Knill-Laflamme condition: either the Kronecker delta vanishes, which in turn implies that $(\bm{k}_a, \bm{x}_a)$, $(\bm{k}_b, \bm{x}_b)$ belong to different fibers of $\nabla_B$; or the delta ``clicks'' and $(\bm{k}_a, \bm{x}_a) = (\bm{k}_b, \bm{x}_b)$, for otherwise there would be a nontrivial logical operator on the right hand side of Eq.~\eqref{eq_Z2bosKL}. Hence, by the same reasoning above Eq.~\eqref{eq:sGLZ2}, maximal sets of correctable errors are in one-to-one correspondence with sections of $\nabla_B$. Concretely, let $s_{\nabla_B}$ be a section of $\nabla_B$ i.e., a map
\begin{equation}
    s_{\nabla_B}: (\mathbb{Z}_2)^3 \to (\mathbb{Z}_2)^2 \times (\mathbb{Z}_2)^3,
\end{equation}
\begin{equation}
    s_{\nabla_B}(\bm{q}) = (\bm{k}_{\nabla_B}(\bm{q}), \bm{x}_{\nabla_B}(\bm{q}) ) 
\end{equation}
satisfying
\begin{equation}
    \nabla_B(s_{\nabla_B} (\bm{q}) ) = \bm{q}, \qquad \forall \bm{q} \in (\mathbb{Z}_2)^3.
\end{equation}
In other words, $s_{\nabla_B}$ assigns to each possible syndrome $\bm{q}$ a pair of error parameters $(\bm{k}_{\nabla_B}(\bm{q}), \bm{x}_{\nabla_B}(\bm{q}))$ such that $E(\bm{k}_{\nabla_B}(\bm{q}), \bm{x}_{\nabla_B}(\bm{q}))$ in Eq.~\eqref{eq_Z2mattererrorsnoZ} produces that syndrome. The associated maximal correctable set is therefore
\begin{equation}
    \mathcal{E}_{s_{\nabla_B}} \coloneqq \{ E(\bm{k}_{\nabla_B}(\bm{q}), \bm{x}_{\nabla_B}(\bm{q})) \mid \bm{q} \in (\mathbb{Z}_2)^3 \}.
\end{equation}
As in the pure gauge case, the physical interpretation of a correctable error is still a violation of the Gauss law.
In contrast to the pure gauge case, however, here a violation of the Gauss law is not a discontinuity in the electric flux across a vertex, but rather a mismatch between the discontinuity in electric flux and the amount of charge located at the vertex.

Each choice of section (and hence correctable error set) also determines a recovery operation. Let $P_{\bm{q}}^B$ denote the projector onto the stabilizer eigenspace with syndrome $\bm{q}$. The recovery channel associated with the section $s_{\nabla_B}$ is then given by
\begin{equation}
    \mathcal{R}_{s_{\nabla_B}}(\rho) = \sum_{\bm{q} \in (\mathbb{Z}_2)^3} E(\bm{k}_{\nabla_B} (\bm{q}), \bm{x}_{\nabla_B}(\bm{q}))^{\dagger} P_{\bm{q}}^B \rho P_{\bm{q}}^B E(\bm{k}_{\nabla_B}(\bm{q}), \bm{x}_{\nabla_B}(\bm{q})).
\end{equation}
As in the pure gauge case, one can of course append $Z$-type operators to the non-identity correctable errors $E(\bm{k},\bm{x})$ within a correctable set of errors.
Physically, however, all this amounts to is asserting that whenever a given bit flip error occurs, a corresponding, fixed, and known $Z$-type operator is also applied.

\subsubsection{Code parameters}
Of course, the code still cannot correct isolated phase flips, and is therefore classical. It involves $n=6$ physical (computational) qubits: three associated with the links and three with the vertices. The three independent stabilizer generators reduce this to $k=3$ encoded logical degrees of freedom, corresponding to the gauge-invariant sector of the theory. To determine the $X$-distance, we examine the logical $X$ operators. Both the Wilson loop operator $\bar{X}$ and the dressed Wilson line operators $\tilde{X}_{[i,i+1]} = X_i X_{[i,i+1]} X_{i+1}$ act nontrivially on the code space and have weight three. Since no nontrivial logical $X$ operator has smaller weight, it follows that $d_X=3$. By contrast, $Z$-type errors commute with all stabilizers and cannot be detected, so $d_Z=1$. The code therefore has parameters $[n,k,d_X] = [6,3,3]$.

\subsection{Bosonic vacuum code in a $1+1$ dimensional $\mathbb{Z}_2$ theory}
\label{ssec:3qubit_vacuum}

We now focus on the vacuum sector of the theory, defined as the subspace of physical states in which no matter excitations are present. Concretely, this corresponds to the states for which every vertex lies in the $+1$ eigenspace of the vertex operators $Z_i$,
\begin{equation}
    \mathcal{H}_{\text{vac}}
    =
    \{
    \ket{\Psi} \in \mathcal{H}_{\text{pn}}
    \mid
    Z_i \ket{\Psi} = \ket{\Psi}
    \ \forall i
    \} ,
\end{equation}
for the same $\Hil_\mrm{pn}$ as defined in Eq.~\eqref{eq_Hpn_bosQED}.
To obtain a vacuum code, we let $\Hil_\mrm{pn}$ play the role of the physical (computational) Hilbert space, and we identify $\Hil_\mrm{vac}$ as the code subspace.

\subsubsection{Stabilizers, code subspace and logical operators}
Using the Gauss law relation in Eq.~\eqref{eq_Z2bosGausslaw}, the condition $Z_i \ket{\Psi} = \ket{\Psi}$ can equivalently be expressed purely in terms of link operators.
The vacuum sector is therefore, equivalently, the $+1$ eigenspace of the stabilizer group
\begin{equation}
    \mathcal{S}
    =
    \langle
    Z_{[01]} Z_{[12]} \Pi_{\text{pn}},
    Z_{[12]} Z_{[20]} \Pi_{\text{pn}}
    \rangle
    \simeq
    \mathbb{Z}_2 \times \mathbb{Z}_2 ,
\end{equation}
and thus naturally constitutes a stabilizer code.
In fact, it is once again immediately recognized as the familiar three-qubit repetition code. 
In the basis \eqref{Z2_basis} the vacuum projector takes the simple form (cf.~\eqref{eq_vacuumcodeproj})
\begin{equation}
    \Pi_{\text{vac}}
    =
    \ket{000}\bra{000}_{\text{pn}}
    +
    \ket{111}\bra{111}_{\text{pn}},
\end{equation}
so the code space is two-dimensional, as expected since the stabilizer imposes two independent constraints. A convenient choice of logical basis is therefore
\begin{equation}
    \ket{0}_{\text{vac}} = \ket{000}_{\text{pn}},
    \qquad
    \ket{1}_{\text{vac}} = \ket{111}_{\text{pn}} .
\end{equation}

\subsubsection{Maximal correctable sets as sections of the discrete divergence}
The logical algebra of the code consists of those gauge-invariant operators that preserve the vacuum sector, or equivalently, those operators in $\mathcal{A}_{\text{pn}}$ that commute with all the $Z_i$ when restricted to $\mathcal{H}_{\text{pn}}$. 
Detectable errors in this setting correspond to gauge-invariant processes that create matter excitations. These may be constructed from products of dressed Wilson lines $\tilde{X}_{\ell}$ defined in Eq.~\eqref{eq:dressed_Z2boson}. For any $\bm{k} \in (\mathbb{Z}_2)^3$, we write 
\begin{equation}
    \tilde{X}^{\bm{k}}
    \coloneqq
    \left(
    \prod_{\ell \in \mathcal{L}}
    X_{v_i(\ell)}^{k_\ell}
    X_{\ell}^{k_\ell}
    X_{v_f(\ell)}^{k_\ell}
    \right),
\end{equation}
where $v_i(\ell)$ and $v_f(\ell)$ denote the initial and final vertices of the link $\ell$. Similarly to the bosonic Gauss law code in the previous section, $Z$ operators commute with the vacuum projector, so we may ignore them for the present error analysis. In contrast, the $X$ operators create matter excitations and thus annihilate the vacuum sector unless they combine either into the identity or into a closed Wilson loop around the triangle. Explicitly, 
\begin{eqnarray}
    \Pi_{\text{vac}} \tilde{X}^{\bm{k}} \Pi_{\text{vac}}
    &=&
    \delta_{\bm{k},0} \Pi_{\text{vac}}
    +
    \left(
    \prod_{\ell \in \mathcal{L}} \delta_{k_\ell,1}
    \right)
    X_{[01]} X_{[12]} X_{[20]} \Pi_{\text{vac}} \\
    &=&
    \bar{X}^{k_{[01]}}
    \delta_{\nabla \bm{k},0}
    \Pi_{\text{vac}},
\end{eqnarray}
with $\bar{X}$ defined in Eq.~\eqref{eq_Z2logX}. One may then directly evaluate the Knill-Laflamme condition for two dressed Wilson lines $\widetilde{X}^{\bm{k}_a}$, $\widetilde{X}^{\bm{k}_b}$, obtaining (cf.~\eqref{eq_KLbosvacuum})
\begin{equation}
    \Pi_{\text{vac}} \widetilde{X}^{\bm{k}_a \dagger} \widetilde{X}^{\bm{k}_b} \Pi_{\text{vac}}
    =
    \bar{X}^{k_{a,[01]} \oplus k_{b,[01]}}
    \delta_{\nabla \bm{k}_a, \nabla \bm{k}_b}
    \Pi_{\text{vac}} .
\end{equation}
We see that two errors violate the Knill-Laflamme condition whenever they generate the same charge pattern,
$\nabla \bm{k}_a = \nabla \bm{k}_b$,
while differing in their flux components. This observation shows that correctable sets of errors are naturally organized by the map that assigns to every operator the charge pattern it produces. In the present case, this is once again given the discrete divergence $\nabla$ defined in Eq.~\eqref{eq_nabla}. The set of attainable charges is $\nabla((\mathbb{Z}_2)^3) \subset (\mathbb{Z}_2)^3$, as determined in Eq.~\eqref{eq:range_nabla}. From the perspective of error correction, the role of $\nabla$ is to assign a syndrome to each error operator: two errors are indistinguishable within the vacuum sector precisely when they produce the same charge pattern $\bm{q} = \nabla \bm{k}$. A correctable set must therefore contain at most one representative error for each such syndrome.

Maximal correctable sets are obtained by making a consistent choice of such representatives for all possible syndromes. Mathematically, this corresponds to choosing a section $s_{\nabla}$ of the discrete divergence, exactly as in Eq.~\eqref{eq:sGLZ2} from the 3-qubit Gauss law code. The corresponding maximal correctable set is then \footnote{Note that, once again, the representatives in $\mathcal{E}_{s_{\nabla}}^{\text{vac}}$ may be multiplied by any string of $Z$-operators, as this doesn't affect the syndrome. This yields an immediate enlargement of the number of maximal correctable sets, which factorizes into a choice of section $s_{\nabla}$ for the $X$-component and an unconstrained $Z$-component.}
\begin{equation}\label{eq_Z2vacmaxset}
    \mathcal{E}_{s_\nabla}^{\text{vac}}
    =
    \{
    \widetilde{X}^{s_\nabla(\bm{q})}
    \mid
    \bm{q} \in \nabla((\mathbb{Z}_2)^3)
    \}.
\end{equation}

The section $s_\nabla$ selects, for every admissible charge configuration $\bm{q}$, a canonical, gauge-invariant error operator that produces it. Physically, this amounts to choosing a representative operator for each possible pattern of matter excitations created from the vacuum. Any other error that generates the same charge configuration differs from this representative by an operator that acts nontrivially within the vacuum sector and therefore cannot belong to the same correctable set. As we already saw with the 3-qubit Gauss law codes, different choices of section correspond to different ways of selecting these representatives, and hence to different maximal correctable sets of errors. 

Given such a section, one may construct a recovery operation that corrects all errors in the corresponding set $\mathcal{E}_{s_\nabla}^{\text{vac}}$. The recovery begins by measuring the charge configuration created by the error. For each admissible pattern $\bm q \in \nabla((\mathbb{Z}_2)^3)$, let $P_{\bm{q}}^{\text{vac}} = \widetilde{X}^{s_\nabla(\bm{q})}\Pi_{\text{vac}}\widetilde{X}^{s_{\nabla}(\bm{q})}$
denote the projector onto the error space $\mathcal{H}_{\bm{q}} = \widetilde{X}^{s_{\nabla}(\bm{q})}(\mathcal{H}_{\text{vac}})$. Operationally, a measurement of these projectors determines the syndrome associated with the error. Once the charge pattern $\bm q$ has been identified, the recovery applies the inverse of the representative error chosen by the section $s_\nabla$. The resulting recovery channel can therefore be written as
\begin{equation}\label{eq_Z2vacrec}
\mathcal{R}_{s_\nabla}^{\text{vac}}(\rho)
=
\sum_{\bm q \in \nabla( (\mathbb{Z}_2)^3 ) }
\widetilde{X}^{s_{\nabla}(\bm{q}) \dagger}
P_{\bm q}^{\text{vac}}\,\rho\,P_{\bm q}^{\text{vac}}\widetilde{X}^{s_{\nabla}(\bm{q})}.
\end{equation}

\subsubsection{Exact unitary equivalence with pure gauge Gauss law code}
Finally, we identify a unitary equivalence with the Gauss law code from the pure gauge sector in Section~\ref{ssec:3qubit_puregauge}. The computational space of that Gauss law code is the kinematical space $\mathcal{H}_{\text{gauge}}$ in Eq.~\eqref{eq_Z2Hgauge}, while the computational space of the present vacuum code is the full perspective-neutral space $\mathcal{H}_{\text{pn}}$ of the bosonic theory. As shown in Eq.~\eqref{Z2_basis}, $\mathcal{H}_{\text{pn}}$ admits a basis $\{ \ket{\bm{\varepsilon}}_{\text{pn}} \}$ labeled by flux configurations $\bm{\varepsilon} \in (\mathbb{Z}_2)^3$. This labeling coincides with that of the electric basis of $\mathcal{H}_{\text{gauge}}$ and therefore induces a canonical identification between the two computational spaces. Concretely, we define a unitary map
\begin{equation}
    \mathcal{T}: \mathcal{H}_{\text{pn}} \to \mathcal{H}_{\text{gauge}},
\end{equation}
by
\begin{equation}
    \mathcal{T} \ket{\bm{\varepsilon}}_{\text{pn}} = \ket{\bm{\varepsilon}}_{\text{gauge}}.
\end{equation}
By construction, $\mathcal{T}$ implements an isomorphism between the computational spaces. Moreover, it preserves the action of gauge invariant operators. In particular, dressed Wilson lines are mapped to their pure gauge counterparts,
\begin{equation}
    \mathcal{T}\widetilde{X}^{\bm{k}} \mathcal{T}^{\dagger} = X^{\bm{k}}.
\end{equation}
The code spaces are obtained by restricting to configurations with $\nabla \bm{\varepsilon} = 0$. Since $\mathcal{T}$ maps each configuration $\bm{\varepsilon}$ in $\mathcal{H}_{\text{pn}}$ to the same configuration in $\mathcal{H}_{\text{gauge}}$, it restricts to an isomorphism between the two code subspaces. Under this map, the maximal correctable sets associated with a section $s_{\nabla}$ (c.f.~Eq.~\eqref{eq_Z2puremaxset} and Eq.~\eqref{eq_Z2vacmaxset}) are related by
\begin{equation}
    \mathcal{E}_{s_{\nabla}}^{\text{GL}} = \mathcal{T} \mathcal{E}_{s_{\nabla}}^{\text{vac}} \mathcal{T}^{\dagger},
\end{equation}
and similarly for the corresponding recovery maps in Eq.~\eqref{eq_Z2purerec} and Eq.~\eqref{eq_Z2vacrec},
\begin{equation}
\mathcal{R}_{s_{\nabla}}^{\text{GL}} ( \mathcal{T} \rho \mathcal{T}^{\dagger}) = \mathcal{T}\mathcal{R}_{s_{\nabla}}^{\text{vac}}(\rho)\mathcal{T}^{\dagger}, \qquad \rho \in \mathcal{S}(\mathcal{H}_{\text{vac}}).
\end{equation}
This equivalence shows that the error correction structure of the vacuum code is entirely inherited from the pure gauge sector. In particular, the classification of maximal correctable sets in terms of sections of the discrete divergence, as well as the associated recovery procedures, are preserved under $\mathcal{T}$. Consequently, the vacuum code may be viewed as a physical embedding of the pure gauge code, with the unitary map $\mathcal{T}$ making this identification explicit.

\subsection{$1+1$ dimensional scalar QED as a vacuum code}\label{ssec_scalarQED}
As our next example, we consider the vacuum code corresponding to quantum electrodynamics on a 1D periodic lattice with $N$ sites and a single species of scalar matter of charge $\pm 1$. The structure group is $G = U(1)$, and the degrees of freedom of the theory consist of gauge variables associated with the links and matter fields located on the vertices of the lattice. We begin by introducing the gauge sector and the algebra it generates on the lattice, and
then subsequently incorporate the matter degrees of freedom. The first part of the discussion does not depend on the number of spatial dimensions; we specialize to one spatial dimension from Sec.~\ref{sssec:1DU1-pn} onwards.

\subsubsection{Scalar QED on a periodic lattice }\label{sssec:latticeQED}
In lattice gauge theory, the fundamental observable associated with a link is the parallel transporter along that link, i.e., the Wilson line operator. For $G = U(1)$ in the continuum, the Wilson line along a path $\ell$ is given by the path-ordered exponential of the vector potential,
\begin{equation}
W_\ell
=
\exp\!\left(i\int_\ell \bm{A}(\bm{x}) \cdot \bm{\dee x}\right).
\end{equation}
Physically, this operator measures the phase acquired by a charged particle transported along the link.
Going to the lattice, we let $\ell$ once again label one of its oriented edges.
For a sufficiently small lattice spacing $a$, the gauge field varies little along a single link and the integral may be approximated by the value of the vector potential at the midpoint of the link. Consequently, the Wilson line may be written in the form
\begin{equation}
W_\ell = e^{i\theta_\ell},
\end{equation}
where the lattice variable $\theta_\ell \approx a\,\bm{A}(x) \cdot \bm{e}_{\mu(\ell)}$ represents the discretized vector potential along the link, which runs in the direction $\bm{e}_{\mu(\ell)}$.

The gauge degree of freedom on each link is therefore characterized by a group element of $U(1)$. As usual, the Hilbert space describing this degree of freedom is the space of square–integrable functions on the group, $\mathcal{H}_\ell \simeq L^2(U(1))$. The group basis on each link consists of generalized eigenstates $\ket{\theta}_\ell$, labeled by the group element $\theta\in[0,2\pi)$, which satisfy
\begin{equation} \label{eq:lightWilsonG}
W_\ell \ket{\theta}_\ell = e^{i\theta}\ket{\theta}_\ell
\end{equation}
and have continuum normalization $\braket{\theta}{\theta'}_\ell = \delta(\theta-\theta')$ (cf. Eq.~\eqref{eq_Pdualrep}). The full kinematical Hilbert space of the gauge field, $\Hil_\mrm{gauge}$ in Eq.~\eqref{eq:Hkingauge}, is then obtained by taking the tensor product over all links.

Besides the Wilson line, the gauge algebra also contains the operator conjugate to the group coordinate. This is exactly the electric field operator $E_\ell$, which generates translations of the group coordinate on the link. Concretely, the translation operators $U_\ell^\alpha = e^{i\alpha E_\ell}$ act on the group basis according to (cf. Eq.~\eqref{eq:Uoperator})
\begin{equation}
U_\ell^\alpha \ket{\theta}_\ell
=
\ket{\theta+\alpha}_\ell .
\end{equation}

A complementary description of the same Hilbert space is obtained by diagonalizing the electric field. This leads to the electric basis, consisting of eigenstates of $E_\ell$,
\begin{equation} \label{eq:Ebasis}
E_\ell \ket{\varepsilon}_\ell
=
\varepsilon \ket{\varepsilon}_\ell,
\qquad
\varepsilon \in \mathbb{Z}.
\end{equation}
The integer $\varepsilon$ labels the quantized electric flux carried by the link. The group and electric bases therefore provide dual representations of the same gauge degree of freedom, related by a Fourier transform (cf. Eq.~\eqref{groupvselectricbases}). In particular,
\begin{equation}
\braket{\theta}{\varepsilon}_\ell
=
\frac{1}{\sqrt{2\pi}}e^{i\varepsilon\theta},
\end{equation}
and, in the electric basis, the Wilson line acts as a raising operator
\begin{equation} \label{eq:lightWilsonE}
W_\ell \ket{\varepsilon}_\ell
=
\ket{\varepsilon+1}_\ell ,
\end{equation}
It also follows that $U_\ell^\alpha$ acts as
\begin{equation}
    U_\ell^\alpha \ket{\epsilon}_\ell = e^{-i\alpha \epsilon} \ket{\epsilon}_\ell \, .
\end{equation}

To make complete contact with Sec.~\ref{sec:LGT}, note that we have been suppressing the character label $\chi \in \hat{G} $ when writing the Wilson line $W_{\ell}$. For $\theta \in G = U(1)$, $\hat{G} \simeq \mathbb{Z}$, the characters are labeled by integers. In detail, denoting $\chi[k]$ the character corresponding to $k \in \mathbb{Z}$, it acts as $\chi[k](\theta) = e^{-ik \theta}$. In other words, here we abbreviate $W_{\ell}^{\chi[1]} \equiv W_{\ell}$;  Eqs.~\eqref{eq:lightWilsonG} and \eqref{eq:lightWilsonE} are thus just the specialization of Eq.~\eqref{eq_Pdualrep} to the present example, and Eq.~\eqref{eq:Ebasis} is the character basis.
Happily, it furthermore follows that $W_\ell^{\chi[k]} = (W_\ell)^k$.

Having reviewed the gauge sector, we now introduce matter degrees of freedom on the vertices. At each vertex $v$, we place two harmonic oscillators describing particle and antiparticle modes of charge $+1$ and $-1$, respectively. The matter Hilbert space \eqref{eq_Hbosonic_matter} therefore takes the form
\begin{equation}
    \mathcal{H}_{\text{matter}} = \bigotimes_{v \in \mathcal{V}} \mathcal{H}_{v, +} \otimes \mathcal{H}_{v,-} \; .
\end{equation}
A convenient description of the local Hilbert space is provided by the occupation number basis (cf. Eq.~\eqref{eq:bosonic_vertexH}),
\begin{equation}
    \mathcal{H}_{v, +} \otimes \mathcal{H}_{v,-} = \text{Span} \{  \ket{\bm{n_v}} \mid n_{v,+}, n_{v,-} \in \mathbb{N}_0   \},
\end{equation}
where $\bm{n_v} = (n_{v,+},n_{v,-})$ labels the joint eigenstates of the number operators, $N_{v, +}$ and $N_{v,-}$, associated with the particle and antiparticle modes, respectively. We denote the corresponding annihilation operators $a_v$ and $b_v$. They satisfy the canonical commutation relations
\begin{equation}
[a_v,a_{v'}^\dagger]=\delta_{vv'}, \qquad
[b_v,b_{v'}^\dagger]=\delta_{vv'},
\end{equation}
with all remaining commutators vanishing. The corresponding number operators are then
\begin{equation}
N_{v,+}=a_v^\dagger a_v, \qquad
N_{v,-}=b_v^\dagger b_v \, ,
\end{equation}
and they count the number of particle and antiparticle excitations residing at vertex $v$. From these operators we define the local matter charge (cf.~Eq.~\eqref{eq_rhototbos})
\begin{equation}
Q_v := N_{v,+}-N_{v,-},
\end{equation}
which measures the net electric charge carried by the matter degrees of freedom at the vertex.

We now describe how gauge transformations act on the combined gauge and matter sectors, i.e., on $\Hil_\mrm{kin} = \Hil_\mrm{gauge} \otimes \Hil_\mrm{matter}$. Local gauge transformations act independently with respect to each vertex and are parametrized by angles $\bm{\alpha}=\{\alpha_v\}_{v\in\mathcal V} \in [0, 2\pi)^{N_V}$. The action of a local gauge transformation factorizes into a part acting on the gauge field and a part acting on the matter sector,
\begin{equation}
    U^{\bm{\alpha}} = U^{\bm{\alpha}}_{\text{gauge}} \otimes u^{\bm{\alpha}}_{\text{matter}}.
\end{equation}
On the gauge sector, they are generated by the discrete divergence of the electric field at every vertex (cf. Eq.~\eqref{eq_Gausslawpure}),
\begin{eqnarray}
    U^{\bm{\alpha}}_{\text{gauge}} &=& \prod_{v \in \mathcal{V}} \left(  \bigotimes_{\ell \in \mathcal{L}_{\text{out}}(v)}U_{\ell}^{\alpha_v} \right) \otimes \left(  \bigotimes_{\ell' \in \mathcal{L}_{\text{in}}(v)} U_{\ell'}^{\alpha_{v} \dagger}   \right) \\
    &=& e^{-i \bm{\alpha} \cdot \nabla \bm{E}},
\end{eqnarray}
with
\begin{equation}
    (\nabla \bm{E})_v \coloneqq \sum_{\ell \in \mathcal{L}_{\text{out}}(v)} E_{\ell} - \sum_{\ell' \in { \mathcal{L}_{\text{in}}(v)}} E_{\ell'}.
\end{equation}
On the matter sector, the gauge transformations are generated by the charge operators, $Q_v$, and take the form (cf.~Eq.~\eqref{eq_gaugegrouprepbos})
\begin{equation}
u_{\text{matter}}^{\bm{\alpha}}
=
\exp \left(i\sum_{v\in\mathcal V}\alpha_v Q_v\right).
\label{transformations_matter_scalar_QED}
\end{equation}
Their adjoint action rotates the phases of the matter modes (cf.~Eq.~\eqref{eq_partantiparttransform}) ,
\begin{align}\label{eq_aGTQED}
u_{\text{matter}}^{\bm{\alpha}} a_v (u_{\text{matter}}^{\bm{\alpha}})^\dagger &= e^{-i\alpha_v} a_v,\\
u_{\text{matter}}^{\bm{\alpha}} b_v (u_{\text{matter}}^{\bm{\alpha}})^\dagger &= e^{i\alpha_v} b_v .
\end{align}
The two oscillators therefore transform with opposite phases under a local gauge rotation and can be interpreted as particle and antiparticle modes carrying charges $+1$ and $-1$. The gauge field on the links transforms consistently with these local phase rotations. If a link $\ell$ is oriented from vertex $v_i$ to $v_f$, the Wilson line operator transforms according to
\begin{equation}
U^{\bm{\alpha}} W_\ell (U^{\bm{\alpha}})^\dagger
=
e^{i(\alpha_{v_i}-\alpha_{v_f})} W_\ell .
\end{equation}
Therefore, the parallel transporter picks up the relative phase between the endpoints of the link.
On a periodic lattice, a gauge transformation with constant parameter $\alpha_v=\alpha$ leaves all Wilson line operators invariant. The gauge-field sector therefore transforms trivially under such global transformations. On the matter sector, however, the same transformation acts as
\begin{equation}
u^{\alpha}_{\text{global}}=
\exp\!\left(i\alpha\sum_{v\in\mathcal V}Q_v\right).
\end{equation}
Gauge-invariant states must be invariant under this transformation as well, which implies the constraint (cf.~Eq.~\eqref{eq_bostotalchargezero})
\begin{equation}
\sum_{v\in\mathcal V} Q_v = 0 .
\end{equation}
Hence, physical states must carry vanishing total matter charge.

The total local gauge transformations are generated by operators that combine the electric flux degrees of freedom on neighboring links with the matter charge localized at a given vertex. These operators constitute the Gauss law constraints and, at a vertex $v$, they assume the form 
\begin{equation}\label{QEDGauss}
G_v =
\sum_{\ell\in\mathcal{L}_{\text{out}}(v)}E_\ell
-
\sum_{\ell'\in\mathcal{L}_{\text{in}}(v)}E_{\ell'}
-
Q_v .
\end{equation}
$G_v$ measures the mismatch between the divergence of the electric field and the charge located at vertex $v$; states invariant under all gauge transformations satisfy Gauss's law,
\begin{equation}\label{eq_QEDGausslaw}
G_v\ket{\psi}=0 \qquad \forall v \in \mathcal{V}  .
\end{equation}
Exponentiating these generators gives a gauge transformation everywhere on the lattice (cf. Eq.~\eqref{generalGTs}),
\begin{equation}
    U^{\bm{a}} = \prod_{v \in \mathcal{V}} U_v^{\alpha_v} = \prod_{v \in \mathcal{V}} e^{i \alpha_v G_v},
\end{equation}
whence the perspective-neutral space \eqref{eq:HpnSec2} is
\begin{equation}
    \mathcal H_{\text{pn}}
=
\{\ket{\psi}\in\mathcal H_{\text{gauge}}\otimes\mathcal H_{\text{matter}}
\mid
U_v^\alpha \ket{\psi}=0\ \forall v \in \mathcal{V}, \alpha \in U(1) \}.
\end{equation}

Equivalently, $\Hil_\mrm{pn}$ is also the kernel of the constraints $G_v$:
\begin{equation}
\mathcal H_{\text{pn}}
=
\{\ket{\psi}\in\mathcal H_{\text{gauge}}\otimes\mathcal H_{\text{matter}}
\mid
G_v\ket{\psi}=0\ \forall v \in \mathcal{V} \}.
\end{equation}
This space constitutes the physical (computational) Hilbert space of the vacuum code. To further characterize this space and the corresponding algebra, it is convenient to make a choice of QRF.

\subsubsection{Perspective-neutral subspace and physical operator algebra}\label{sssec:1DU1-pn}

We now specialize to a one-dimensional periodic lattice and solve the Gauss constraints by expressing the gauge field degrees of freedom relative to a chosen quantum reference frame. We do this by selecting a spanning tree within the lattice, in line with the discussion in Sec.~\ref{sec:general}. Let us number the vertices and links as
\begin{equation}
    \mathcal{V} = \{ v_0, v_1, \cdots, v_{N-1}       \},
\end{equation}
\begin{equation}
    \mathcal{L} = \{ \ell_0, \ell_1, \cdots, \ell_{N-1}     \} 
\end{equation}
and assume an orientation where link $\ell_i$ is oriented from $v_i$ to $v_i + 1$.
In the current setting, a spanning tree is simply a subset of $N-1$ adjacent links in $\mathcal{L}$. Without loss of generality, we pick the tree $R$ consisting of the first $N-1$ links and denote its complement by $S$:
\begin{equation}
    R = \{   \ell_0, \ell_1, \cdots, \ell_{N-2}   \},
\end{equation}
\begin{equation}
    S = \{ \ell_{N-1}   \}. 
\end{equation}

In Sec.~\ref{ssssec_GausTPS}, we saw that such a choice of QRF induces a refactorization of the kinematical space of the gauge field into a Wilson loop sector and a frame sector,
\begin{equation}
    \mathcal{H}_{\text{gauge}} \simeq \mathcal{H}_{\text{loops}} \otimes \mathcal{H}_{R}.
\end{equation}
On a 1D periodic lattice, there is precisely one Wilson loop---the loop that goes around the full lattice itself,
\begin{equation}
    H = \bigotimes_{i = 0}^{N-1}W_{\ell_{i}}.
\end{equation}
We thus expect that $\mathcal{H}_{\text{loops}} \simeq L^2(U(1))$, and this is indeed that case:
in general, we defined the loop Hilbert space as the $+1$ eigenspace of the restriction of the gauge transformations to the gauge field sector, i.e. the perspective-neutral subspace of the corresponding pure-gauge theory (see discussion below Eq.~\eqref{eq_Hkinbosfac}).
Here, this restriction imposes $N-1$ independent constraints on the $N$ group variables $\theta_\ell$. In terms of the group coordinate operators $\hat{\theta}_{\ell}$, the Wilson loop $H$ reads
\begin{equation}
    H = \exp \left( i \sum_{i = 0 }^{N-1} \hat{\theta}_{\ell_i}  \right),
\end{equation}
which makes it explicit that $H$ measures the total holonomy accumulated along the chain.
Its eigenstates provide a convenient basis for the loop Hilbert space,
\begin{equation}
    H \ket{\phi}_{\text{loops}} = e^{i \phi} \ket{\phi}_{\text{loops}}, \qquad \phi \in [0, 2\pi)
\end{equation}
which makes the isomorphism $\mathcal{H}_{\text{loops}} \simeq L^2(U(1))$ explicit. 

At a more physical level, this refactorization admits a simple interpretation. In one spatial dimension, imposing Gauss' law removes all local gauge-redundant degrees of freedom from the gauge sector, leaving behind a single global degree of freedom together with relational information along the chain. The loop sector $\mathcal{H}_{\text{loops}}$ captures the total holonomy around the periodic lattice, which constitutes the single global gauge-invariant degree of freedom of the gauge field. In the conjugate electric basis, this degree of freedom corresponds to the total electric flux that is threaded through the chain, so it may be thought of as setting a uniform background field. The frame sector $\mathcal{H}_R$ then describes the remaining gauge field degrees of freedom relative to this global quantity.

In the group representation, $R$ parametrizes the phases accumulated along Wilson lines connecting each vertex to the root, thereby encoding how the gauge field is distributed along the lattice with respect to the chosen reference frame. We can make this even more explicit by passing from the canonical description in terms of link variables $\theta_{\ell_0}, \cdots, \theta_{\ell_{N-2}}$ to a collection of \emph{vertex} phases, defined by the variables
\begin{equation}
    \Theta_{v_i} = -\sum_{j = 0}^{i-1} \theta_{\ell_j} \mod 2\pi,
\end{equation}
which keeps track of the change of phase that a charged particle undergoes when transported from $v_i$ to $v_0$ along the path $\gamma_R[v_i,v_0] = \ell_{0}^{-1} \circ \ell_{1}^{-1} \circ \cdots \circ \ell_{i-1}^{-1}$ (see Fig.~\ref{fig:SQED}). We then refactorize $\mathcal{H}_R$ with respect to these non-local degrees of freedom by mapping $\ket{\bm{\theta}}_R \to \ket{\bm{\Theta}}_R$ where $\bm{\theta} = (\theta_{\ell_i})_{i=0}^{N-2}$ and $\bm{\Theta} = (\Theta_{v_i})_{i=1}^{N-1}$ are each a collection of $N-1$ phases, but defined on links and vertices, respectively.

\begin{figure}[h]
\centering
\includegraphics[width=0.95\textwidth]{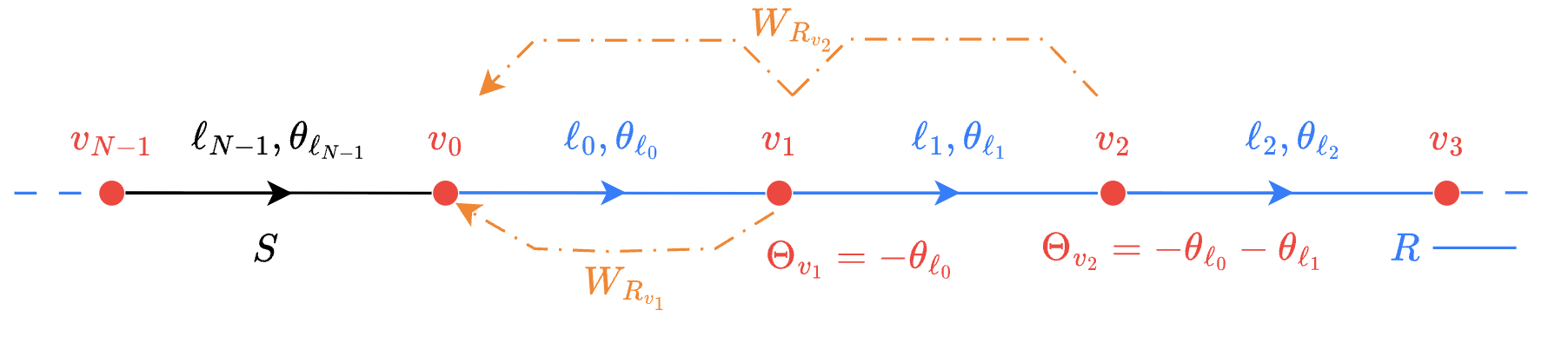}
\caption{Collective vertex variables (red) in relation to link $U(1)$ variables. $\Theta_v$ keeps track of the phase acquired by a charged particle when transported along the tree Wilson lines $W_{R_v}$ (orange). The choice of spanning tree $R$ is shown in blue and its complement, $S$, in black. }
\label{fig:SQED}
\end{figure}

In this way, a gauge transformation at vertex $v$ with parameter $\alpha_v$ shifts the vertex variable at $v$ as $\Theta_v \to \Theta_v + \alpha$, which is the usual action of the symmetry group on QRF orientation variables. In contrast, tree Wilson lines connecting each vertex to the root act diagonally, producing phases at the nontrivial endpoints,
\begin{equation}
    W_{{R_v}} \ket{\bm{\Theta}}_R = e^{i \Theta_v} \ket{\bm{\Theta}}_R,
\end{equation}
where $W_{R_v} = W_{\gamma_R[v,v_0]}$ (c.f. Eq.\eqref{eq_RWilson}). In this particular theory, taking into account the labeling of the vertices, tree Wilson lines take the form 
\begin{equation}
    W_{R_{v_i}} = \bigotimes_{j = 0}^{i-1}W_{\ell}^{\dagger}, \qquad 0 \leq i <  N,
\end{equation}
where we use $W_{\ell}^{\dagger}$ instead of $W_{\ell}$ because the lattice links' orientations are opposite to the directions of the paths $R_{v_i}$.

A complementary interpretation emerges in the electric basis, obtained via a Fourier transform of the group (magnetic) variables. In this representation, the frame degrees of freedom describe how the electric flux departs from the uniform background fixed by the loop sector, with the vertex labels tracking the location and magnitude of flux discontinuities along the chain:
\begin{equation} \label{eq:U1nonlocalEbasis}
    \ket{\bm{q}}_R = \int \frac{\dee ^{N-1} \bm{\Theta}}{(2 \pi )^{\frac{N-1}{2}}} e^{i \bm{q} \cdot \bm{\Theta} } \ket{\bm{\Theta}}_R, \qquad q_v \in \mathbb{Z}.
\end{equation}
In this basis, Wilson lines act by translation of the conjugate coordinate,
\begin{eqnarray}
    W_{R_{v_i}} \ket{\bm{q}}_R &=&  \int \frac{\dee^{N-1} \bm{\Theta}}{(2 \pi)^{\frac{N-1}{2}}} e^{i \bm{q} \cdot \bm{\Theta}} e^{i \Theta_{v_i}} \ket{\bm{\Theta}}_R \\
    &=& \int \frac{\dee ^{N-1}\bm{\Theta}}{(2 \pi)^{\frac{N-1}{2}}} e^{i (\bm{q} + \bm{e}_{v_i} ) \cdot \bm{\Theta} } \ket{\bm{\Theta}}_R \\
    &=& \ket{\bm{q} + \bm{e}_{v_i} }_R,
\end{eqnarray}
where $(\bm{e}_{v_i})_{v_j}=\delta_{ij}$. In contrast, gauge transformations at $v_i \neq v_0$ (Eqs.~\eqref{eq:URv1} and \eqref{eq:URv2}) act diagonally, multiplying the electric basis states by a phase
\begin{equation}
    U_{R_{v_i}}^{\alpha_{v_i}} \ket{\bm{q}}_R = e^{i \alpha_{v_i}q_{v_i}} \ket{\bm{q}}_R.
\end{equation}
Gauge transformations at $v_0$, the root of the spanning tree, must be treated differently because they are not covered by the frame orientations, but they may be expressed in terms of the remaining independent transformations by exploiting the invariance of the gauge sector under a global transformation:
\begin{eqnarray}
    U_{R_{v_0}}^{\alpha_{v_0}} = \left(  \prod_{v' \in \mathcal{V}}U_{R_{v'}}^{\alpha_{v_0}} \right) \left( \prod_{v \in \mathcal{V} \setminus \{ v_0  \} } U_{R_v}^{\alpha_{v_0} \dagger}  \right)
\end{eqnarray}

To recap, starting from the gauge field degrees of freedom on the lattice, we let the $N^\text{th}$ link constitute the system, $S$, for a frame $R$ that consists of the first $N-1$ links.
We then passed to a collective magnetic basis for $R$, whose labels consist of a phase $\Theta_{v_i}$ for each of the last $N-1$ vertices $v_i$.
It now follows that the corresponding collective electric basis for $R$, whose labels consist of an integer $q_{v_i}$ for each vertex $v_i$, transforms in nearly the same way as the occupation number eigenstates for the bosonic \emph{matter} living at the vertices.

Compare the action of the gauge transformation
\begin{equation}
    U^{\bm{\alpha}}_\mrm{gauge} = \prod_{i = 0}^{N-1} U_{R_{v_i}}^{\alpha_{v_i}} \, ,
\end{equation}
expressed with respect to vertex variables, to that of the matter-supported gauge transformation $u_\mrm{matter}^{\bm{\alpha}}$ in Eq.\eqref{transformations_matter_scalar_QED}.
According to the analysis in Sec.~\ref{ssssec_boscodespace}, it follows that the gauge invariant subspace of the joint frame-matter sector is generated by states $\ket{-\bm{q}(\bf{n})}_R \otimes \ket{\mathbf{n}}_{\text{matter}}$,  where $q_v(\mathbf{n}) \coloneqq  n_{v, +} - n_{v, -}$ and the occupation numbers satisfy the global neutrality condition $\sum_{v \in \mathcal{V}} n_{v,+} - n_{v,-} = 0$. For clarity we denote the set of such occupation numbers by $\tilde{\mathcal{N}}$,
\begin{equation}
    \tilde{\mathcal{N}} \coloneqq \left\{ \vphantom{\sum_{v \in \mathcal{V}}} \mathbf{n} \in (\mathbb{N}_0)^{2N} ~ \right| \left. ~ \sum_{v \in \mathcal{V}} n_{v, +} - n_{v, -} = 0     \right\}.
\end{equation}
Therefore, the perspective-neutral space has a subsystem structure $\mathcal{H}_{\text{pn}} \simeq \mathcal{H}_{\text{loops}} \otimes \mathcal{H}^{\text{dr}}_{\text{matter}}$, with
\begin{equation}\label{drmatterSQED}    \mathcal{H}^{\text{dr}}_{\text{matter}} = \text{Span} \{ \ket{-\bm{q}(\mathbf{n})}_R \otimes \ket{\mathbf{n}}_{\text{matter}} \in \mathcal{H}_R \otimes \mathcal{H}_{\text{matter}}  \mid \mathbf{n} \in \tilde{\mathcal{N}}  \}.
\end{equation}

Having obtained this decomposition, the structure of the physical algebra follows directly: It factorizes into loop and dressed-matter factors, $\mathcal{A}_{\text{pn}} = \mathcal{A}_{\text{loops}} \otimes \mathcal{A}^{\text{dr}}_{\text{matter}}$. Let us focus first on the loop sector. By construction, $\mathcal{H}_{\text{loops}}$ is isomorphic to the perspective-neutral space of Maxwell theory. As argued in Sec.~\ref{sssec_codespacepure}, its corresponding algebra is the unital $*$-algebra generated by the relational observables describing the system $S$ relative to the frame $R$. In Eq.~\eqref{relobspure}, we explicitly showed how these operators may be constructed from $S$-observables through conjugation by the Page-Wootters reduction map Eq.~\eqref{eq_PW}. For QED, this yields the algebra
\begin{equation}
    \mathcal{A}_{\text{loops}} = \left\langle H\Pi_{\text{pn}}, U_{\ell_{N-1}}^{\alpha}\Pi_{\text{pn}} \mid \alpha \in [0, 2\pi)     \right\rangle,
\end{equation}
where the perspective-neutral projector averages uniformly over gauge transformations involving both the electromagnetic field and matter sectors,
\begin{equation}
    \Pi_{\text{pn}} \coloneqq \int \frac{\dee ^{N-1} \bm{\alpha}}{(2 \pi)^{\frac{N-1}{2}}} U^{\bm{\alpha}}.
\end{equation}
Similarly, in App.~\ref{app_bosalgebra}, we follow the same procedure to characterize the dressed-matter algebra in terms of dressed Wilson lines along tree paths $\gamma_R[v,v']= \gamma_R[v',v_0]^{-1} \circ \gamma_R[v, v_0]$, with creation and annihilation operators at their endpoints. In the present discussion, there are three possible ways that one could carry out such a dressing with the available matter fields, giving rise to operators $b^{\dagger}_{v} W_{\gamma_R[v,v']} a_{v'}^{\dagger}$, $a_v W_{\gamma_R[v,v']} a_{v'}^{\dagger}$, $b_{v}^{\dagger} W_{\gamma_R[v,v']} b_{v'}$ and their Hermitian conjugates, which are manifestly gauge invariant and generate the complete dressed matter algebra 
\begin{equation}
    \mathcal{A}^{\text{dr}}_{\text{matter}} = \left\langle a_vW_{\gamma_R[v,v']}a_{v'}^{\dagger} \Pi_{\text{pn}}, \thinspace b_{v}^{\dagger}W_{\gamma_R[v,v']}b_{v'} \Pi_{\text{pn}} , \thinspace  b_v^{\dagger}W_{\gamma_R[v,v']}a_{v'}^{\dagger} \Pi_{\text{pn}}, \thinspace \text{h.c.} \mid v, v' \in \mathcal{V}     \right\rangle.
    \label{drmatalgQED}
\end{equation}

\subsubsection{Code subspace, correctable errors and recovery}

We now turn to the identification of the code subspace within $\mathcal{H}_{\text{pn}}$. In line with the discussion above, this is obtained by restricting to the sector of trivial matter content---the matter vacuum---and regarding the remaining gauge-invariant degrees of freedom as carrying the logical information. More precisely, we define
\begin{equation}
    \mathcal{H}_{\text{code}} = \mathcal{H}_{\text{vac}} \coloneqq \{ \ket{\Psi} \in \mathcal{H}_{\text{pn}} \mid N_{v, +} \ket{\Psi} = N_{v, -} \ket{\Psi} = 0  \}.
\end{equation}
This coincides with the $+1$ eigenspace of the group generated by the exponentials of the local occupation number operators,
\begin{equation}\label{sgenvacQED}
    \mathcal{S} \coloneqq \left\langle e^{i\alpha N_{v, +}} \Pi_{\text{pn}}, \thinspace e^{-i \alpha N_{v,-}} \Pi_{\text{pn}}  \mid v \in \mathcal{V}, \alpha \in [0, 2\pi)  \right\rangle,
\end{equation}
so the vacuum sector indeed has the structure of a $U(1)$-stabilizer code. The code projector takes the simple form
\begin{equation}
\Pi_{\text{vac}} = \mathbbm{1}_{\text{loops}} \otimes_R \ket{0}\bra{0}^{\text{dr}}_{\text{matter}},
\label{Pi_vac_QED}
\end{equation}
which is formally identical to the perspective-neutral projector of the pure-gauge theory in Lemma 3.1 upon replacing the \emph{dressed-matter} labels with a $R$ labels.
It thus follows that the present code space is isomorphic to the code subspace of the Gauss-law code obtained from the pure gauge sector. This will shortly allow us to relate maximal sets of correctable errors for both codes as we did in Sec.~\ref{ssssec: unitaryequivvaccodes}; however, first we construct the algebra of logical operators.
These are operators that preserve the vacuum, and from Eq.(\ref{Pi_vac_QED}), it is clear that this is simply $\mathcal{A}_{\text{vac}} = \mathcal{A}_{\text{loops}} \otimes_R \mathbbm{1}^{\text{dr}}_{\text{matter}} \simeq \mathcal{A}_{\text{loops}}$. Therefore, the logical operators are exactly the Wilson loops across the lattice.

Prior to performing the error analysis, we first isolate the subset of dressed matter operators that take states out of the vacuum sector and hence give rise to detectable errors. This identification is straightforward: among the dressed Wilson line operators appearing in Eq.~(\ref{drmatalgQED}), the only ones that fail to annihilate the vacuum are those involving both $a^{\dagger}$ and $b^{\dagger}$. For subsequent analysis, it is convenient to arrange any product of such operators into the form
\begin{eqnarray}
    \widetilde{W}^{\bm{k}} &\coloneqq& \prod_{\ell \in \mathcal{L}} \left( b_{v_i(\ell)}^{\dagger} W_{\ell}a_{v_f(\ell)}^{\dagger}     \right)^{k_{\ell, +}} \left( a_{v_i(\ell)}^{\dagger}W_{\ell}^{\dagger} b_{v_f(\ell)}^{\dagger}  \right)^{k_{\ell, -}} \\
    &=& \prod_{i = 0}^{N-1} \left( b^{\dagger}_{v_i}W_{\ell_i}a_{v_{i+1}}^{\dagger}   \right)^{k_{\ell_i, +}} \left( a_{v_i}^{\dagger} W_{\ell_i}^{\dagger} b_{v_{i+1}}^{\dagger}   \right)^{k_{\ell_{i}, -}}
\end{eqnarray}
indexed by integer link data $\bm{k} \in (\mathbb{N}_0)^{2N}$.
While the notation is delicate, the intuition is plain. In 1D, one can add a unit of flux to the electric field at the link $\ell_i$ by creating a negatively-charged antiparticle at $v_i$ and a positively-charged particle at $v_{i+1}$ $(b^{\dagger}_{v_i}W_{\ell_i}a_{v_{i+1}}^{\dagger})$.
Similarly, one can subtract a unit of flux by first creating a particle at $v_i$, followed by an antiparticle at $v_{i+1}$ $(a_{v_i}^{\dagger} W_{\ell_i}^{\dagger} b_{v_{i+1}}^{\dagger})$.

This parametrization encompasses the complete set of operators capable of generating matter excitations from the vacuum while preserving gauge invariance. When acting on the vacuum, $\widetilde{W}^{\bm{k}}$ yields a specific occupation-number configuration, fully determined by the number of particle and antiparticle creation operators associated with each vertex. We have
\begin{eqnarray}
\widetilde{W}^{\bm{k}}\Pi_{\text{vac}} ~ \propto ~ H^{(k_{\ell_{N-1}, +}) - (k_{\ell_{N-1}, -})} \prod_{i = 0}^{N-1} \left(  \ket{n_{v_i, +} (\bm{k})} \bra{0}_{v_i, +} \otimes \ket{ n_{v_i, -}(\bm{k})   } \bra{0}_{v_i, -}  \right) \Pi_{\text{pn}},
\end{eqnarray}
where
\begin{equation}
    n_{v_i, +} (\bm{k}) = k_{\ell_{i-1}, +} + k_{\ell_{i}, -}, \qquad n_{v_i, -}(\bm{k}) =  k_{\ell_{i}, +} + k_{\ell_{i-1}, -},
\end{equation}
and the subscript $\pm$ specifies whether the occupation number in the ket (bra) refers to a particle ($+$) or to an antiparticle ($-$). Inserting this expression into the Knill-Laflamme condition for two Wilson lines $\widetilde{W}^{\bm{k}}$, $\widetilde{W}^{\bm{k}'}$ yields 
\begin{equation}\label{KLvacQED}
    \Pi_{\text{vac}}\widetilde{W}^{\bm{k} \dagger}\widetilde{W}^{\bm{k}'} \Pi_{\text{vac}} \propto \delta_{\mbf{n}(\bm{k}), \mbf{n}(\bm{k}')} H^{(k'_{\ell_{N-1}, +}) - (k'_{\ell_{N-1},-}) - (k_{\ell_{N-1}, +}) + (k_{\ell_{N-1}, -})} \Pi_{\text{vac}}  .
\end{equation}
In other words, any two such Wilson line errors are simultaneously correctable as long as they are not related by multiplication with Wilson loops alone.

The natural syndrome is given by the individual particle and antiparticle number, as expected from the form of the stabilizer generators in Eq.(\ref{sgenvacQED}):
\begin{equation}
    e^{\pm i (N_{v_i, \pm})} \widetilde{W}^{\bm{k}} \ket{0}^{\text{dr}}_{\text{matter}} = n_{v_i, \pm} (\bm{k}) \widetilde{W}^{\bm{k}} \ket{0}^{\text{dr}}_{\text{matter}}.
\end{equation}
This defines a map (cf.~Eq.~\eqref{eq_vacbundle})
\begin{equation}
    \partial_{B,\text{vac}}: (\mathbb{N}_0)^{2N} \to \tilde{\mathcal{N}}, \quad \partial_{B, \text{vac}}(\bm{k}) = \mbf{n}({\bm{k}}),
\end{equation}
which endows the space of error parameters with the structure of a discrete fiber bundle over the set of admissible occupation number configurations, with fibers given by all $\bm{k} \in (\mathbb{N}_0)^{2N}$ that produce the same syndrome. Eq.~\eqref{KLvacQED} shows that two distinct dressed Wilson line errors can only satisfy the Knill-Laflamme condition Eq.~\eqref{eq_KLstabcode} if they lie in distinct fibers. Indeed, if $\mathbf{n}(\bm{k}) \neq \mathbf{n}(\bm{k}')$, the expression vanishes, as required. By contrast, if $\mathbf{n}(\bm{k}) = \mathbf{n}(\bm{k}')$, the right hand side of Eq.~\eqref{KLvacQED} is non-vanishing and acquires a nontrivial loop contribution unless $\bm{k} = \bm{k}'$. It follows that a correctable set can contain at most one representative from each fiber of $\partial_{B,\text{vac}}$. Maximal correctable sets are therefore obtained by selecting exactly one parameter $\bm{k}$ in each fiber. Equivalently, they are parametrized by sections $s_{\text{vac}}$ of $\partial_{B, \text{vac}}$ and take the form
\begin{equation}
    \mathcal{E}_{s_{\text{vac}}} = \{ \widetilde{W}^{s_{\text{vac}}(\mathbf{n})} \mid \mathbf{n} \in \tilde{\mathcal{N}}   \}.
\end{equation}

A recovery operation for such a set is obtained by measuring the syndrome $\mathbf{n}$, which identifies the unique representative $\widetilde{W}^{s_{\text{vac}}(\mathbf{n})}$, and then undoing it. This leads to the channel
\begin{equation}
\mathcal{R}_{s_{\text{vac}}}(\rho)
=
\sum_{\mathbf{n} \in \tilde{\mathcal{N}}}
\Pi_{\text{vac}}
\widetilde{W}^{s_{\text{vac}}(\mathbf{n})\dagger}
P_{\mathbf{n}}^{\text{vac}} \rho  P_{\mathbf{n}}^{\text{vac}}
\widetilde{W}^{s_{\text{vac}}(\mathbf{n})}
\Pi_{\text{vac}},
\end{equation}
where the projectors onto the syndrome sectors are given by
\begin{equation}
P_{\mathbf{n}}^{\text{vac}}
\propto
\widetilde{W}^{s_{\mathrm{vac}}(\mathbf{n})}
\Pi_{\text{vac}}
\widetilde{W}^{s_{\text{vac}}(\mathbf{n})\dagger}.
\end{equation}
In particular, each syndrome sector is obtained by acting with the corresponding representative error on the vacuum subspace. 

\subsubsection{Unitary equivalence with pure gauge Gauss law code upon coarse-graining }
To make contact with the discussion on Gauss law codes, let us introduce an additional restriction and consider a scenario in which our syndrome measurements are limited to the net matter charge $Q_v = N_{v, +} - N_{v,-}$ at every vertex. This amounts to a coarse-graining of the syndrome, where the individual particle and antiparticle numbers are no longer resolved separately. Concretely, we introduce an equivalence relation on the vertex Hilbert space $\mathcal{H}_v = \mathcal{H}_{v, +} \otimes \mathcal{H}_{v, -}$, and identify occupation states whenever they carry the same matter charge:
\begin{equation}
   \ket{n_{v,+}, n_{v, -}}_v \sim \ket{m_{v,+}, m_{v, -}}_v \quad \iff \quad n_{v, +} - n_{v,-} = m_{v,+} - m_{v, -}.
\end{equation}
We denote the equivalence classes by 
\begin{equation}
    \ket{q_v}_v = \left[ \ket{n_{v, +}, n_{v,-}}_v \in \mathcal{H}_{v, +} \otimes \mathcal{H}_{v,-} \mid n_{v, +} - n_{v,-} = q_v    \right],
\end{equation}
which comprise eigenstates of $Q_v$ with eigenvalue $q_v$. The set of equivalence classes, $\mathcal{H}_v / \sim$ has a linear structure and admits a well-defined inner product $\braket{q_v}{q'_v}_v = \delta_{q_v q'_v}$, so we may view it as a Hilbert space.

We denote the tensor product over all coarse-grained vertex spaces by $\mathcal{H}_{\text{charge}}$ and write $\ket{\bm{q}}_{\text{charge}} = \bigotimes_{v \in \mathcal{V}} \ket{q_v}_v$.  As we briefly mentioned in Sec.~\ref{sec:bosmaxsets} (see the discussion around Eq.~\eqref{eq_coarsegrain}), this space carries a unitary representation of both the gauge group $U(1)^{N}$ and its Pontryagin dual $\widehat{U(1)}{}^N = \mathbb{Z}^N$:
\begin{equation}
    U_{\text{charge}}^{\bm{\alpha}} \ket{\bm{q}}_{\text{charge}} = e^{i\bm{\alpha} \cdot \bm{q}} \ket{\bm{q}}_{\text{charge}}.
\end{equation}
\begin{equation}
    {\mathcal{W}}^{\bm{k}} \ket{\bm{q}}_{\text{charge}} = \ket{\bm{q} - \bm{k}}_{\text{charge}},
\end{equation}
This is the key ingredient that will shortly allow us to establish a unitary relation with the Gauss law code that arises in the pure gauge sector.

Extending this coarse-graining to the full perspective-neutral space of the bosonic theory produces the space
\begin{equation}
    \mathcal{H}_{\text{pn}}^{\text{CG}} = \mathcal{H}_{\text{loops}} \otimes_R \mathcal{H}^{\text{dr}}_{\text{matter}},
\end{equation}
where $\mathcal{H}^{\text{dr}}_{\text{charge}}$ consists of joint frame and matter charge configurations where the flux on the links compensates the matter charge in accordance with Gauss's law. These are spanned by orthonormal states $\ket{\bm{q}}^{\text{dr}}_{\text{charge}}=\ket{-\bm{q}}_R \otimes \ket{\bm{q}}_{\text{charge}}$, which are manifestly invariant under the action of the gauge group and carry definite matter charges. At this point, it becomes clear that we have an isomorphism $\mathcal{T}: \mathcal{H}_{\text{pn}}^{\text{CG}} \to \mathcal{H}_{\text{gauge}}$ from the coarse-grained space of the bosonic theory to the kinematical space of the pure-gauge sector \footnote{In terms of the map in Eq.~\eqref{eq_tildePiR}, we have $\mathcal{T} = \mathbbm{1}_{\text{loops}} \otimes_R \tilde{\Pi}_R$.}. Explicitly,
\begin{equation}\label{eq_QEDT}
    \mathcal{T} ( \ket{\theta}_{\text{loop}} \otimes_R \ket{\bm{q}}^{\text{dr}}_{\text{charge}} ) = \ket{\theta}_{\text{loop}} \otimes_R \ket{-\bm{q}}_R,
\end{equation}
where we have exploited the fact that states in $\mathcal{H}^{\text{dr}}_{\text{charge}}$ are in total charge-neutral to label frame configurations by $\bm{q} \in \nabla (\mathbb{Z}^N) = \{ \bm{q'} \in \mathbb{Z}^N \mid \sum_{v \in \mathcal{V}} q_v' = 0 \}$, even though $q_{v_0}$ is redundant.

Under the above coarse-graining, the code space is promoted from the matter vacuum of $\mathcal{H}_{\text{pn}}$ to the charge zero sector of $\mathcal{H}_{\text{pn}}^{\text{CG}}$:
\begin{equation}
    \mathcal{H}_{\text{code}}^{\text{CG}} = \{ \ket{\Psi}^{\text{CG}} \in \mathcal{H}_{\text{pn}}^{\text{CG}} \mid Q_v \ket{\Psi}^{\text{CG}} = 0 \enspace \forall v \in \mathcal{V}   \}.
\end{equation}
Upon noting that $\mathcal{H}_{\text{code}}^{\text{CG}}$ coincides with the +1 eigenspace within $\mathcal{H}_{\text{pn}}^{\text{CG}}$ of the stabilizer group
\begin{equation}
    \mathcal{S}^{\text{CG}} = \left\langle e^{-i \alpha_v Q_v} \Pi_{\text{pn}}^{\text{CG}} \mid v \in \mathcal{V} \setminus \{ v_0  \}, \alpha_v \in [0, 2\pi)  \right\rangle \simeq U(1)^{N-1},
\end{equation}
with
\begin{equation}
    \Pi_{\text{pn}}^{\text{CG}} \coloneqq \int \frac{\dee ^{N-1} \bm{\alpha}}{(2 \pi)^{\frac{N-1}{2}}} \left(  \mathbbm{1}_{\text{loops}} \otimes_R  U_R^{\bm{\alpha}} \otimes U_{\text{charge}}^{\bm{\alpha}}    \right),
\end{equation}
we may interpret the coarse-grained construction as a $U(1)$-stabilizer code. Since states in $\mathcal{H}_{\text{pn}}^{\text{CG}}$ satisfy Gauss's law, Eqs.~\eqref{QEDGauss}-\eqref{eq_QEDGausslaw}, and $Q_v = (\nabla \bm{E})_v$ at every vertex, the above stabilizer group coincides with the group of gauge transformations $\{ e^{-i \alpha_v  (\nabla \bm{E})_v} \mid  v \in \mathcal{V} \setminus \{ v_0  \}, \alpha_v \in [0, 2\pi) \}$ in Maxwell theory, which stabilizes the corresponding Gauss law code. Furthermore, the code spaces of both constructions are related through the unitary map $\mathcal{T}$ in Eq.~\eqref{eq_QEDT}:
\begin{eqnarray}
    \mathcal{T} (\mathcal{H}_{\text{code}}^{\text{CG}}) &=& \{ \ket{\Psi} \in \mathcal{H}_{\text{gauge}} \mid e^{-i \alpha_v (\nabla \bm{E})_v} \ket{\Psi} = \ket{\Psi}, \quad \forall v \in \mathcal{V} \setminus \{ v_0  \}, \alpha_v \in [0, 2\pi)  \} \\
    &=& \mathcal{H}_{\text{code}}^{\text{GL}}
\end{eqnarray}
The Gauss law code from Maxwell theory and the coarse-grained vacuum code from scalar QED may therefore be regarded as distinct realizations of the same abstract stabilizer code.

This correspondence is not limited to the code spaces, but also extends to correctable errors.
In Sec.~\ref{sssec_codespacepure}, we saw that maximal sets of correctable errors in the pure gauge Gauss law code were indexed by sections of the Gauss law bundle 
$\partial: \hat{G}^{N_L} \to \hat{G}^{N_V - 1}$, which assign to each Wilson line product $W^{\bm{\chi}}$ the character-valued gauge charge configuration $\partial \bm{\chi}$ that it excites on the frame degrees of freedom located at the vertices.
In the present discussion, $\hat{G} = \mathbb{Z}$, so the projection map of the Gauss-law bundle is simply the (exponential of the) discrete divergence on the lattice.
For the $U(1)$ example at hand, $\nabla : \mathbb{Z}^N \rightarrow \mathbb{Z}^N$ acts as
\begin{equation}
    (\nabla\bm{k})_i = k_{\ell_i} - k_{\ell_{i-1}}
\end{equation}
and its range is $\nabla(\mathbb{Z}^N) = \{ \bm{q} \in \mathbb{Z}^{N} \mid \sum_{v \in V} q_v = 0 \}$.
This operator assigns a syndrome $\nabla \bm{k}$ to $W^{\bm{k}} = \prod_{\ell \in \mathcal{L}} W_\ell^{k_\ell}$,
which we can also rewrite in the basis \eqref{eq:U1nonlocalEbasis} as
\begin{equation}
    W^{\bm{k}} = H^{k_{\ell_{N-1}}} \sum_{\bm{q} \in \mathcal{Q}} \ket{\bm{q} + \nabla \bm{k}}\bra{\bm{q}}_R, \qquad H^{k_{\ell_{N-1}}} = \prod_{\ell \in \mathcal{L}}W_{\ell}^{k_{\ell_{N-1}}}.
\end{equation}
(Note also the congruence with Eqs.~\eqref{WandH} and \eqref{Wsonframepuregauge}.)
Conversely, sections $s_{\nabla}: \nabla(\mathbb{Z}^N) \to \mathbb{Z}^N $ of $\nabla$ take us in the opposite direction, and, for each charge configuration $\bm{q} \in \nabla (\mathbb{Z}^N)$, they return a Wilson line $W^{s_{\nabla}(\bm{q})}$ that shifts the frame orientations by $ \bm{q}$.
This determines the error set
\begin{equation}
    \mathcal{E}_{s_{\nabla}}^{\text{GL}} = \{ W^{s_{\nabla}(\bm{q})}  \mid \bm{q} \in \nabla(\mathbb{Z}^N)   \}.
\end{equation}
To interpret this construction, we note that the action of an error $W^{\bm{k}}$ on the code space $\mathcal{T}(\mathcal{H}_{\text{code}}^{\text{CG}})$ is completely determined by its syndrome $\nabla \bm{k}$.
In particular, if $\nabla \bm{k} = \nabla \bm{k}'$, then $W^{\bm{k}}$ and $W^{\bm{k}'}$ are related by a holonomy, which is a logical operator in the Gauss law code, so they are not simultaneously correctable.
The set $\mathcal{E}_{s_{\nabla}}^{\text{GL}}$ is correctable because it contains at most one representative of each syndrome, and it is maximal because it contains exactly one.

We now arrive at the precise point of contact between the two constructions.
The operators $\mathbb{W}^{\bm{k}} \coloneqq W^{\bm{k}}\mathcal{W}^{\nabla \bm{k}}$ implement, within the coarse-grained vacuum code, the same organization of errors by charge configuration that characterizes the Gauss law code.
Their action on the dressed charge sector is completely determined by the induced charge $\nabla \bm{k}$, so that whenever $\nabla \bm{k} = \nabla \bm{k}'$, $\mathbb{W}^{\bm{k}}$ and $\mathbb{W}^{\bm{k}'}$ differ by a Wilson loop.
As a result, selecting a section $s_{\nabla}: \nabla (\mathbb{Z}^N) \to \mathbb{Z}^N $ determines a maximal correctable set 
\begin{equation}
    \mathcal{E}_{s_{\nabla}}^{\text{CG}} \coloneqq \{ \mathbb{W}^{s_{\nabla}(\bm{q})}  \mid \bm{q} \in \nabla (\mathbb{Z}^N) \},
\end{equation}
in direct parallel with the construction of $\mathcal{E}_{s_{\nabla}}^{\text{GL}}$.
The choice of section $s_{\nabla}$ also canonically determines recovery operations for both codes.
In each case, recovery proceeds by measuring the syndrome and applying the inverse of the chosen representative within the corresponding fiber.
For the Gauss law code, writing $P_{\bm{q}}^{\text{GL}}$ for the projector onto the subspace with charge configuration $\bm{q}$ in $\mathcal{H}_{\text{gauge}}$, the recovery channel is
\begin{equation}
\mathcal{R}_{s_{\nabla}}^{\text{GL}}(\rho)
=
\sum_{\bm{q} \in \mathcal{Q}}
W^{s_{\nabla}(\bm{q})\dagger}
P_{\bm{q}}^{\text{GL}} \rho P_{\bm{q}}^{\text{GL}} 
W^{s_{\nabla}(\bm{q})}.
\end{equation}
Similarly, for the coarse-grained vacuum code, writing $P_{\bm{q}}^{\text{CG}}$ for the projector onto the dressed charge sector with configuration $\bm{q}$, we define
\begin{equation}
\mathcal{R}_{s_{\nabla}}^{\text{CG}}(\rho) =
\sum_{\bm{q} \in \nabla(\mathbb{Z}^N)}
\mathbb{W}^{s_{\nabla}(\bm{q})\dagger}
 P_{\bm{q}}^{\text{CG}} \rho
 P_{\bm{q}}^{\text{CG}} 
\mathbb{W}^{s_{\nabla}(\bm{q})}.
\end{equation}

The unitary map $\mathcal{T}$ identifies these two pictures by mapping dressed charge configurations to frame configurations and relates both the errors and their organization into correctable sets,
\begin{equation}
    \mathcal{T}\mathbb{W}^{\bm{k}} \mathcal{T}^{\dagger} = W^{\bm{k}}, 
    \qquad 
    \mathcal{T} \mathcal{E}_{s_{\nabla}}^{\text{CG}} \mathcal{T}^{\dagger} = \mathcal{E}_{s_{\nabla}}^{\text{GL}},
\end{equation}
as well as the associated recovery operations,
\begin{equation}
\mathcal{R}_{s_{\nabla}}^{\text{GL}}(\mathcal{T} \rho \mathcal{T}^{\dagger}) = \mathcal{T} \mathcal{R}_{s_{\nabla}}^{\text{CG}}(\rho) \mathcal{T}^{\dagger}, \qquad \rho \in \mathcal{S}(\mathcal{H}_{\text{code}}^{\text{CG}}).
\end{equation}
Thus, the classification of errors by charge configurations, the role of sections of $\nabla$, and the resulting maximal correctable sets coincide in the two constructions. The vacuum code after coarse-graining, and the Gauss law code for the pure gauge sector therefore furnish equivalent realizations of the same stabilizer structure, differing only in their physical interpretation.

\subsection{$2+1$-dimensional fermionic QED as a Gauss law code}
\label{ssec_Schwinger}
We now specialize the general construction of Gauss law codes with fermionic matter that we developed in Sec~\ref{sssec_fermGauss} to quantum electrodynamics with one species of staggered fermion (in $1+1$ dimensions, this corresponds to the lattice Schwinger model \cite{Kogut:1974ag}). 
Building directly on our discussion in Sec~\ref{sssec:latticeQED}, we consider a $U(1)$
lattice gauge theory in 2+1 dimensions with periodic boundary conditions. As usual, the gauge degrees of freedom are associated with links, and described by Wilson line operators $W_{\ell} = e^{i \hat{\theta}_{\ell}}$ together with electric field operators $E_{\ell}$. 

\subsubsection{Code structure and logical operators}
To incorporate fermionic matter, we adopt the staggered fermion formulation. At each vertex, we place a single fermionic mode with local Hilbert space $\mathcal{H}_v \simeq \mathbb{C}^2$, spanned by occupation number states $\ket{n_v}_v$, $n_v = 0,1$. The corresponding number operator is $N_v = \psi_v^{\dagger}\psi_v$, where an explicit expression for the creation and annihilation operators $\psi_v$, $\psi_v^{\dagger}$ in terms of $X$ and $Y$ Pauli operators may be found in Eq.(\ref{psiPauli}). This immediately guarantees the correct anticommutation relations at single vertices but, in order to satisfy anticommutation at different vertices, it is necessary to define the matter Hilbert space via a graded tensor product across the vertices:
\begin{equation}
    \mathcal{H}_{\text{matter}} = \widehat{\bigotimes_{v \in \mathcal{V}}} \mathcal{H}_v.
\end{equation}
The graded structure, which we cover in detail in App~\ref{app_fermions}, ensures that operators with odd fermion parity, like $\psi_v$ and $\psi_v^{\dagger}$, anticommute when acting on different tensor factors, so we recover the canonical anticommutation relations:
\begin{equation}
    \{ \psi_v, \psi_{v'}^{\dagger} \} = \delta_{vv'},
\end{equation}
with all remaining anticommutators vanishing.

The staggered structure is implemented by assigning a value $|v| = 0,1$ to each vertex $v$, which we call its parity, in such a way that its nearest neighbors have opposite parity.
Eq.(\ref{cubicparity}) defines $|v|$ on a general cubic lattice; on the two-dimensional lattice in this example, this assignment partitions the lattice into two sublattices arranged in a checkerboard pattern (see Fig~\ref{fig:staggered}), consisting of what we shall refer to as even ($|v| = 0$), and odd ($|v| = 1$) vertices.
\begin{figure}[h]
\centering
\includegraphics[width=0.4\textwidth]{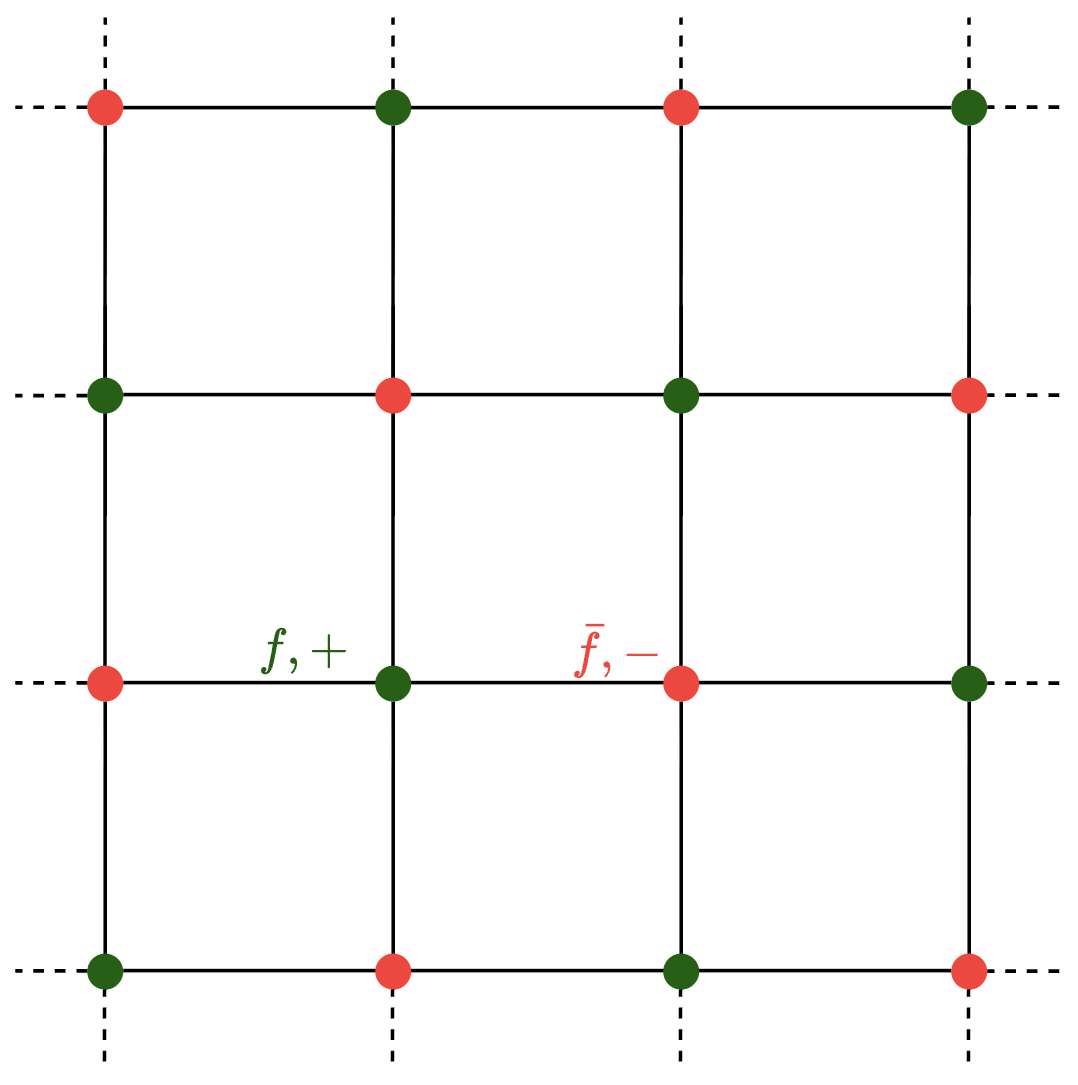}
\caption{Staggered fermions on a square lattice. Even sites (green) support positively charged fermions, whereas odd sites (red) support negatively charged antifermions.}
\label{fig:staggered}
\end{figure}
The matter charge operator at vertex $v$ is then given by
\begin{equation}
    Q_v = N_v - c_v, \qquad c_v = \frac{1 - (-1)^{|v|}}{2}.
\end{equation}
With this assignment, even sites ($c_v = 0$) carry fermions of charge $+1$, while odd sites ($c_v = 1$) correspond to antifermionic degrees of freedom in which an empty site carries charge $-1$ and an occupied site is neutral. The Gauss constraints therefore take the form
\begin{equation}
    G_v = (\nabla \cdot \bm{E})_v - Q_v,
\end{equation}
so that the discrete divergence of the electric field is locally balanced by the staggered charge. These operators generate gauge transformations at each vertex according to 
\begin{equation}
    U_v^{\alpha} = e^{-i \alpha G_v}, \qquad \alpha \in [0, 2\pi),
\end{equation}
under which the annihilation operator $\psi_v$ transforms with charge $-1$, in direct analogy with the bosonic case in Eq.~\eqref{eq_aGTQED}.

As in the scalar case, Gauss's law defines a perspective-neutral subspace which admits a subsystem decomposition into loop and dressed matter degrees of freedom that is analogous to the one above Eq.(\ref{drmatterSQED}) but with the frame part of dressed occupation states defined in terms of the staggered charge:
\begin{equation}
\ket{\mathbf{n}}^{\text{dr}}_{\text{matter}} \coloneqq \ket{\mathbf{c} - \mathbf{n}}_R \otimes \ket{\mathbf{n}}_{\text{matter}},
\end{equation}
where $\bm{c} = (c_v)_{v \in \mathcal{V}}$ parametrizes the staggering and $R$ is now a spanning tree on the two-dimensional lattice, $\mathcal{L}$. As always for Gauss law codes, we interpret the invariant subspace $\Hil_\mrm{pn}$ as the code space, and similarly, we identify the logical algebra of the code with the corresponding algebra of bounded operators. 
This algebra inherits the subsystem structure and factorizes as 
\(\mathcal{A}_{\text{pn}} = \mathcal{A}_{\text{loops}} \otimes \mathcal{A}^{\text{dr}}_{\text{matter}}\).
The loop factor, given in Eq.~\eqref{Apnpure}, is generated by elementary Wilson loops
\begin{equation}
    H_{\ell} \coloneqq W_{\ell}W_{R_{v_i(\ell)}} W_{R_{v_f(\ell)}}^{\dagger}, \qquad \ell \in S,
\end{equation}
where $W_{R_v}$ denotes the unique tree Wilson line from $v_0$ to $v$, and $v_i(\ell)$, $v_f(\ell)$ are the initial and final vertices of $\ell$;
together with link operators \(U_{\ell}^g\) on the complement of the tree \(S = \mathcal{L} \setminus R\). The dressed matter factor, on the other hand, is generated by matter-dressed tree Wilson lines 
$\psi_v W_{\gamma_R[v,v']} \psi_{v'}^{\dagger}$, as we show in App~\ref{app_fermalg}.

All of these operators are gauge invariant and therefore map codewords to codewords. In contrast, bare Wilson lines along the links and Pauli $X$ and $Y$ operations on the vertices produce Gauss law defects, and so carry the interpretation of detectable errors.

\subsubsection{Distinguishability of errors}
We now turn to the characterization of maximal correctable sets of composite errors arising from products of these elementary errors. This is a simple task if we only allow products of Wilson lines to be included in the error set. In that case, the matter parts of the gauge transformations leave the Knill-Laflamme condition unaffected and Prop.~\ref{prop_KLpure} continues to hold in the form 
\begin{equation}
    \Pi_{\text{pn}}W^{\bm{k_a} \dagger}W^{\bm{k_b}} \Pi_{\text{pn}} = H^{\bm{k}_{S,b} - \bm{k}_{S,a} } \delta_{\nabla \bm{k}_a, \nabla \bm{k}_b} \Pi_{\text{pn}},
\end{equation}
where $\bm{k}_S$ denotes the restriction of the link data $\bm{k}$ to $S$ and we have replaced the Gauss law map $\partial$ with the discrete divergence $\nabla$, reflecting the fact that the characters of $U(1)$ are labeled by integers. Comparing with Eq.~\eqref{eq_KLsubsystem} and noting $H^{\bm{k}_{S,a} - \bm{k}_{S,b}} \in \mathcal{A}_{\text{loops}} \otimes \mathbbm{1}^{\text{dr}}_{\text{matter}}$, it becomes clear that the perspective-neutral subspace of this theory has the structure of a subsystem code \cite{kribs2006} in which the dressed-matter sector is perfectly correctable with respect to Wilson line products.

The analysis becomes more subtle once we allow Pauli errors on the matter sector. The new difficulty arises from the fact that the charges excited by these operators are no longer codeword-independent. Indeed, a gauge transformation introduces a phase that depends on the local occupation number,
\begin{equation}
    U_v^{\alpha}X_vU_v^{\alpha \dagger} = e^{i \alpha (1 - 2N_v)} X_v,
\end{equation}
so simply measuring the vertex charge through a generating set $\{ U_v^{\alpha_v}| \alpha_v \in [0, 2\pi), v \in \mathcal{V} \}$, as one usually does for stabilizer codes, no longer constitutes an admissible syndrome measurement because it reveals information about the protected codewords. Instead, as we explained in Section~\ref{sec:fermimaxsets}, we need to adapt the syndrome measurements to the subset of vertices that may support $X$-excitations in our error model. To this end, consider the composite operators in Eq.~\eqref{eq_fermerrors}, which in QED read
\begin{equation}\label{eq_QED_composite}
    E_a = Z^{\bm{z}_a}X^{\bm{x_a}}W^{\bm{k}_a}, \qquad E_b = Z^{\bm{z}_b} X^{\bm{x}_b} W^{\bm{k}_b},
\end{equation}
with $\bm{x}_a, \bm{x}_b, \bm{z}_a, \bm{z}_b \in \mathbb{Z}_2^{N_V}$ and $\bm{k}_a, \bm{k}_b \in \mathbb{Z}^{N_L}$. Inserting them into the Knill-Laflamme condition yields
\begin{equation}
    \Pi_{\text{pn}}E_a^{\dagger}E_b \Pi_{\text{pn}} \propto Z^{\bm{z}_a \oplus \bm{z}_b}X_{\text{dr}}^{\bm{x}_a \oplus \bm{x}_b} H^{\bm{k}_{S,b} - \bm{k}_{S,a} } \left( \prod_{v \in \mathcal{V}} \delta_{(\nabla \bm{k}_a)_v + x_{a,v}(1-2N_v), (\nabla \bm{k}_b)_v + x_{b,v}(1 - 2N_v)}   \right) \Pi_{\text{pn}},
\end{equation}
which is nothing but Eq.~\eqref{eq_fermionicKL} specialized to the present model of QED. The Kronecker delta that we used throughout Sections~\ref{sec:pgmaxsets} and ~\ref{sec:bosmaxsets} to identify the bundle structure of the code now depends on the occupation numbers and has therefore been promoted to a logical operator. It follows that the only way for $E_a$ and $E_b$ to satisfy the Knill-Laflamme condition is if the delta evaluates to the same number on the whole code space: there are two possibilities. On the one hand, we may have that the delta clicks on the entire code space. This is only possible if the two terms involving $N_v$ are equal at every vertex, i.e. $\bm{x}_a = \bm{x}_b$ and $\nabla \bm{k}_a = \nabla \bm{k}_b$. In this case, the Knill-Laflamme condition requires $\bm{k}_{S,a} = \bm{k}_{S,b}$, $\bm{z}_a = \bm{z}_b$, and these conditions together imply $E_a = E_b$. Hence, the only way that two distinct errors parametrized by Eq.~\eqref{eq_QED_composite} may be simultaneously correctable is if the delta vanishes everywhere on the code space. In terms of the error parameter, this condition requires that at every vertex one of the following is true (cf.~Eq.~\eqref{eq_fermGaussdisterrors}):
\begin{align}
& x_{a,v}=x_{b,v}\qquad\qquad\text{and}\qquad\qquad (\nabla \bm{k}_a)_v \neq (\nabla \bm{k}_b)_v\,,\nonumber\\
   \text{or}\qquad\qquad
&x_{a,v}\neq x_{b,v}\qquad\qquad\text{and}\qquad\qquad (\nabla \bm{k}_a)_v \neq (\nabla \bm{k}_b)_v \pm 1\,.\label{eq_fermQEDerrors}
\end{align}
This condition admits a direct interpretation in terms of the electric flux violations at each vertex. The quantity $(\nabla \bm{k})_v$ encodes the charge created by Wilson line errors, while an $X_v$ operation shifts this charge by $\pm 1$, depending on the local occupation. The second line of Eq.~\eqref{eq_fermQEDerrors} therefore expresses that this intrinsic $\pm 1$ ambiguity must be accounted for when distinguishing errors. 

\subsubsection{Syndrome measurement and recovery}
To proceed, we now construct syndrome measurements adapted to this structure. At each vertex $v$, rather than resolving the Gauss law violation $G_v$ exactly, we group together values that differ only by the shift induced by an $X$-error. Concretely, we introduce the bins
\begin{equation}
\big[(\nabla \bm{k})_v\big]_{x_v} \coloneqq
\begin{cases}
(\nabla \bm{k})_v & \text{if } x_v = 0,\\
\{(\nabla \bm{k})_v + 1, (\nabla \bm{k})_v - 1 \} & \text{if } x_v = 1,
\end{cases}
\end{equation}
which identify charge configurations that cannot be distinguished without accessing the occupation number at $v$. These bins define projective measurements built from the projectors onto the eigenspaces of $G_v$. Denoting by $\Pi^v_q$ the projector onto the eigenspace of $G_v$ with eigenvalue $q \in \mathbb{Z}$, we set 
\begin{eqnarray}
    \Pi^v_{\bm{x}, \bm{k}} &\coloneqq& \Pi^v_{[(\nabla \bm{k})_v]_{x_v}} \\
    &=& \Pi^v_{(\nabla \bm{k})_v + x_v} \ket{0}\bra{0}_v + \Pi^v_{(\nabla \bm{k})_v - x_v} \ket{1}\bra{1}_v,
\end{eqnarray}
which is the QED specialization of Eq.~\eqref{eq_coarsemeasure}.
In particular, these projectors act conditionally on the local occupation number, are diagonal in the occupation basis, and therefore commute with the logical algebra. Two such projectors at a given vertex are mutually exclusive precisely when the conditions in Eq.~\eqref{eq_fermQEDerrors} are satisfied. Moreover, since the number operators commute across vertices, the family $\{\Pi^v_{\bm{x},\bm{k}} \mid v \in \mathcal{V}, \bm{x}, \bm{k} \in \text{Eq.~\eqref{eq_fermQEDerrors}} \}$ can be measured jointly, yielding a well-defined syndrome across the lattice, just as discussed below Eq.~\eqref{eq_coarsemeasure}. 

It is now clear that maximal sets of correctable errors are again in one-to-one correspondence with maximal sets of admissible bins. In particular, we have the two extreme cases that $x_{a,v} = 0$ for all $a$, or that $x_{a,v} = 1$ for all $a$. In both cases, Eq.~\eqref{eq_fermQEDerrors} ensures that no additional bins can be included without violating the distinguishability condition, so these choices define maximal sets. Moreover, in either case, there are as many admissible bins at $v$ as there are possible values of $(\nabla \bm{k})_v$, since for fixed $x_v$ each such value defines a distinct bin and Eq.~\eqref{eq_fermQEDerrors} guarantees that these bins are mutually distinguishable. Ignoring $Z$-errors, which do not enter the syndrome, we therefore obtain a one-to-one correspondence between maximal sets of correctable errors and maximal sets of admissible bins, specified by solutions to Eq.~\eqref{eq_fermQEDerrors}. Including $Z$-error data only introduces a trivial multiplicity, since for each admissible syndrome one may choose a single representative $Z$-configuration.
This leads to the same qualitative picture as in the general fermionic construction in Section~\ref{sec:fermimaxsets}. In contrast to the pure gauge and bosonic cases discussed in Sections~\ref{sec:pgmaxsets} and~\ref{sec:bosmaxsets}, there is no longer a single universal set of syndromes. Instead, the admissible syndrome measurements depend on the choice of error set, and maximal sets of correctable errors are in one-to-one correspondence with maximal sets of such measurements, as characterized by Eq.~\eqref{eq_fermQEDerrors}. As a result, the bundle picture discussed previously no longer applies in a uniform way: rather than a single structure capturing all maximal sets, one effectively obtains a distinct structure for each choice of admissible syndromes, with the remaining freedom residing only in the $Z$-error data. 

In particular, fixing a maximal set of admissible bins determines both a set of correctable errors and a compatible syndrome measurement, and therefore specifies a recovery procedure. Let $\mathcal{E}'$ denote a maximal set of parameters satisfying Eq.~\eqref{eq_fermQEDerrors}. For each $(\bm{x},\bm{k}) \in \mathcal{E}'$, we define the global syndrome projector
\begin{equation}
    \Pi_{\bm{x},\bm{k}} \coloneqq \prod_{v \in \mathcal{V}} \Pi^v_{\bm{x},\bm{k}},
\end{equation}
so that the family $\{ \Pi_{\bm{x},\bm{k}} \}_{(\bm{x},\bm{k}) \in \mathcal{E}'}$ implements the corresponding coarse-grained syndrome measurement. For each outcome, we fix a representative error operator $E_{\bm{x},\bm{k}} = Z^{\bm{z}} X^{\bm{x}} W^{\bm{k}}$,
where $\bm{z}$ is arbitrary, since $Z$-errors do not affect the syndrome. The recovery map is then given by
\begin{equation}
    \mathcal{R}(\rho) = \sum_{(\bm{x},\bm{k}) \in \mathcal{E}'} E_{\bm{x},\bm{k}}^{\dagger} \Pi_{\bm{x},\bm{k}} \rho  \Pi_{\bm{x},\bm{k}} E_{\bm{x},\bm{k}}.
\end{equation}
For any fixed maximal set $\mathcal{E}'$, the map above implements the corresponding recovery procedure by first resolving the coarse-grained syndrome and then applying the inverse of a fixed operator $E_{\bm{x},\bm{k}}$ for each admissible outcome $(\bm{x},\bm{k})$. The nontrivial feature of the fermionic setting is thus entirely encoded in the structure of the admissible pairs $(\bm{x},\bm{k})$ satisfying Eq.~\eqref{eq_fermQEDerrors}, which determine both the allowed error sets and the compatible syndrome measurements. Once this structure has been fixed, however, the recovery itself proceeds exactly as in the usual setting.

\section{Connections with Quantum Simulations} \label{sec:disc}

In this section, we briefly discuss possible implications of our results in Secs.~\ref{sec:general} and \ref{sec:examples} for practical aspects of quantum computation.

In Sec.~\ref{sec:general}, we used generic compact Abelian lattice gauge theories with matter to construct quantum error correcting codes, and in Sec.~\ref{sec:examples}, we constructed codes for specific lattice gauge theory models. Since we effectively showed how to interpret a broad class of physical systems as quantum error correcting codes, our results in principle provide a pathway for experimentally realizing error correction using such (analog) systems. For this, vacuum codes are more relevant, since a vacuum code's physical space aligns with the space of physical states of the system.
 
On the other hand, we may also ask how our results might be useful for digital simulation using quantum computers with error correction.
We begin with a general comment regarding the practical relevance of the errors we have described in LGTs. While differing between Gauss law and vacuum codes, the correctable error sets that we explored in either case are natural to consider from a gauge theory point of view. A separate question is whether in an actual physical realization of the code, such as a quantum simulation of the pertinent LGT, the actual errors occurring in the simulation align with our errors in the LGT. If so, our error analysis would be of immediate practical relevance. While we do not answer this question conclusively at this time, the evidence suggests that the gauge theory errors may well be relevant in the physical implementation of  the respective LGT code. For example, we have seen between Eqs.~\eqref{eq_Z2puremaxset} and~\eqref{eq_blablabla} that natural gauge theory errors align with the bit flip errors in the three-qubit repetition code interpreted as a LGT. We suspect that such an error alignment between implemented code and gauge theory occurs more generally, though we leave an exploration of this question to future work.

For digital simulation, Gauss law codes may be more relevant for simulating a gauge system on one that does not itself feature a gauge symmetry.
In the spirit of previous work on lattice gauge theory and quantum error correction \cite{Rajput:2021trn,Carena:2024dzu,Spagnoli:2024mib,Yao:2025cxs,Spagnoli:2026qni}, our codes can be expected to reduce the computation overhead required for error correction when compared to a non-bespoke code, due to combining Gauss-law checks within the simulation with error correction steps.\footnote{While Ref.~\cite{Rajput:2021trn} also discussed the $\mathbb{Z}_2$ gauge theories with and without a fermion, the codes studied there are different from those here in the sense that Ref.~\cite{Rajput:2021trn} first doubled a part of links in lattice and studied properties of the system as repetition codes. On the other hand, we have studied properties of original lattice gauge theories as codes.
}
For example, Ref.~\cite{Spagnoli:2024mib} argued that $\mathbb{Z}_2$ lattice gauge theory with a fermion (or its generalization to a $\mathbb{Z}_N$ gauge group \cite{Spagnoli:2026qni}) can be interpreted as a bit-flip code, and that its concatenation with a phase-flip code is able to achieve fault-tolerant Hamiltonian simulation. This last observation raises a crucial point, however, which is that one must not only take care to protect states from errors, but also be able to implement logical gates with error correction, fault tolerantly.

In kind, our results in Sec.~\ref{sec:general} amount to showing how to protect states from errors in a general Abelian lattice gauge theory, while we have not so far discussed the implementation of quantum algorithms with error correction.
Let us discuss the prospects for such computational aspects next. In order to perform Hamiltonian simulation, one must rewrite the Hamiltonian in terms of logical operators  \cite{Spagnoli:2024mib,Spagnoli:2026qni}. When the simulated Hamiltonian is any gauge-invariant Hamiltonian involving only the pure gauge sector (e.g., the Maxwell theory on a lattice), we may use gauge sector Gauss law codes and vacuum codes. It is always possible to rewrite a pure gauge Hamiltonian in terms of logical operators for either code, although the kinematical Hilbert spaces are different.
For the Gauss law code, this is simply because the Hamiltonian is always a sum of gauge-invariant operators, and gauge-invariant operators are logical operators.
For the vacuum codes, while gauge-invariant operators are not necessarily logical operators,\footnote{
For instance, the standard kinetic terms for matter with the link variable are gauge-invariant but regarded as errors in vacuum codes.
} those without support on matter degrees of freedom are always logical operators.

When the simulated Hamiltonian also includes matter (in a gauge-invariant way), we may use Gauss law codes with the same field content and hybrid codes whose code subspaces are the same as the physical Hilbert space of the system to be simulated.\footnote{In a vacuum code, any matter species in excess of those being simulated are fixed to their vacuum sector when constructing the code subspace.}

Next, we consider the implementation of logical gates required for Hamiltonian simulation.\footnote{
See also \cite{Haruna:2025piy} for detailed constructions of logical gates in $\mathbb{Z}_2$ lattice gauge theories.
} In \cite{Spagnoli:2024mib,Spagnoli:2026qni}, it was shown how to implement time evolution for the Gauss law code of $\mathbb{Z}_N$ gauge theory with fermions using both Quantum Signal Processing (QSP) \cite{Low:2016sck,Low:2016znh} and Trotterization.
To implement both algorithms fault-tolerantly, all non-Clifford operations\footnote{The QSP case uses the linear combination of unitaries (LCU) algorithm \cite{Childs:2012gwh}, which requires Toffoli gates in the so-called SELECT operation. The Trotterization applies exponentiations of Pauli operators in the Hamiltonian, which are typically approximated using many $T$ gates.
} can be performed on ancilla qubits encoded in a 7-qubit Steane code \cite{Steane:1996va}, so that one must apply only Clifford operations to the qubits of the lattice gauge theory. Moreover, these Clifford operations can be implemented transversally by exploiting the fact that the code is a special case of a Calderbank–Shor–Steane (CSS) code \cite{Calderbank:1995dw,Steane:1995vv}, in which each stabilizer consists exclusively of either $X$ or $Z$ operators.

A natural question is whether this structure extends to our codes. To address this, we must determine which of our codes can be identified as CSS qudit codes.
In fact, as noted in sections~\ref{ssssec_codeparampure}, \ref{sssec:boscodepars}, and~\ref{ssssec_fermcodeparam}, both the Gauss law codes and vacuum codes can be regarded as CSS codes when their kinematical Hilbert spaces $\mathcal{H}_{\rm kin}$ are finite.
This is because, for finite-dimensional $\mathcal{H}_{\rm kin}$, there exist bases such that all the operators appearing in the Gauss law are either $X$'s or only $Z$'s. Thus, in principle, we can perform Clifford operations transversally for our codes as well for finite-dimensional $\mathcal{H}_{\rm kin}$. We leave the infinite-dimensional case, as well as further connections to quantum simulation, to future work.

\section{Conclusion} \label{sec:conc}

Exploiting tools from quantum reference frames (QRFs), we showed in this work that Abelian lattice gauge theories give rise to two natural families of quantum error correcting codes, namely \emph{Gauss law codes} and \emph{vacuum codes}. These results expand on the general correspondence between QECCs and gauge theories formulated in \cite{Carrozza:2024smc} using QRFs as a mediator. Both constructions are rooted in the same underlying Gauss law constraint, which simultaneously enforces gauge invariance, defines the stabilizer structure of Gauss law codes, and selects the computational space of vacuum codes. Despite this common origin, the two families differ in their error models and physical interpretations, though admit a clear physical relation. In either case, QRFs make a systematic discussion of both the logical and error structure of the codes particularly straightforward, underscoring the power of the QRF toolkit also in the context of quantum error correction.

Gauss law codes are defined on the kinematical space $\mathcal{H}_{\text{kin}}$, and encompass fermionic and bosonic matter, with the code subspace given by the perspective-neutral Hilbert space $\mathcal{H}_{\text{pn}}$. Their error structure is organized by the Gauss law map Eq.~\eqref{boundarymap}, and---except for the fermionic case described below---maximal correctable sets can be parametrized systematically with global sections of an associated discrete fiber bundle. This provides a unified and detailed understanding of error sets across a wide class of Abelian lattice gauge theories (we only demand that $G$ is compact), significantly extending and generalizing previous constructions in the literature \cite{Spagnoli:2026qni, Rajput:2021trn, Carena:2024dzu}. From a practical perspective, Gauss law codes are particularly relevant for quantum simulations of gauge theories, where noise processes may violate a fictitious gauge symmetry. 

Vacuum codes, by contrast, are defined on $\mathcal{H}_{\text{pn}}$ and their code subspace coincides with the matter vacuum sector $\mathcal{H}_{\text{vac}}$, again for bosonic and fermionic matter. In this setting, errors correspond to genuine gauge-invariant matter excitations. While both code constructions are unitarily equivalent---either exactly or upon an appropriate coarse-graining depending on the matter content---the vacuum codes retain additional microscopic information lacking in the Gauss law code, reflecting the fact that matter configurations generally carry more information than local total charge alone.

A central outcome of our analysis is that, in all cases except for the fermionic Gauss law code, the Gauss law constraint determines a unique syndrome structure given by the generators of the group of gauge transformations $\{ U^g_v  \}_{v \in \mathcal{V}}$, and maximal sets of correctable errors are parametrized accordingly. The fermionic Gauss law code provides a notable exception: here, the stabilizer generators do not fix a unique syndrome measurement, which should instead be adapted to a particular choice of admissible errors. This contrasts with the standard paradigm of quantum error correction, where the stabilizer structure determines the syndrome measurement independently of the error model.

Throughout, we have focused on natural generalizations of Pauli errors to the (possibly infinite-dimensional) lattice gauge theory setting, including Wilson line and matter $X$-shift operators, as well as their electric and $Z$-type counterparts. All resulting codes are CSS and, in fact, classical in the sense that they protect against bit-flip-like errors associated with charge and flux excitations. This makes them particularly well suited for mitigating physically relevant errors in quantum simulations, where such excitations arise naturally, or in situations where device noise is heavily biased toward $X$-type or $Z$-type errors. At the same time, this also points to a clear direction for future work: how should phase-flip protection be incorporated in this setting?

A straightforward approach would be to concatenate these constructions with a phase-flip code as originally suggested in \cite{Rajput:2021trn}. In a similar vein, since any code other than a lattice $\mathbb{Z}_2$ code requires working with qudits, optimizing how one encodes qudits into hardware qubits is another opportunity for improved error correction.
In principle, both concatenating with phase-flip codes and invoking small local qudit codes need not be homogeneous over the lattice, and so one could envision adapting the local error correction to device noise. Nevertheless, it remains to be understood whether conferring phase-flip protection can be done in a more informed way that exploits the underlying structure of the gauge theory.

Finally, in the context of quantum simulations of LGTs, our analysis effectively assumes a direct geometric identification between the lattice of the gauge theory and the underlying hardware layout. This assumption is not essential. One may instead consider implementations in which the two are not aligned, thereby gaining additional freedom in how the degrees of freedom are arranged and how the gauge constraints are enforced. Understanding how our characterization of correctable errors extends to such settings is an interesting direction for future work.


\section*{Acknowledgements}
\noindent
A.~C.-D. thanks Stephon Alexander and ChunJun (Charles) Cao for helpful discussions during the preparation of this manuscript.
M.~H.~would like to thank Christian Bauer, Junichi Haruna, Yoshimasa Hidaka, Yutaro Iiyama, Neel Modi, Lento Nagano, Koji Terashi and Nobuyuki Yoshioka for discussions on related works.
M.~H.~is supported by JST CREST Grant Number JPMJCR24I3, JSPS Grant-in-Aid for Transformative Research Areas (A) ``Extreme Universe'' JP21H05190 [D01], JSPS KAKENHI Grant Number JP22H01222 and the Royal Society grants ICA/R2/242058 and IEC/R3/243103.
This work was made possible through the support of the ID\# 62312 grant from the John Templeton Foundation, as part of the project \href{https://www.qiss.fr/}{``The Quantum Information Structure of Spacetime'' (QISS)} and through the \href{https://withoutspacetime.org}{``WithOut SpaceTime project'' (WOST)}, led by the Center for Spacetime and the Quantum (CSTQ), and funded by Grant ID\# 63683 from the John Templeton Foundation. The opinions expressed in this work are those of the authors and do not necessarily reflect the views of the John Templeton Foundation.

\appendix
\section{On the tensor product structure of fermionic systems}\label{app_fermions}
In this appendix, we review the tensor product structure of fermionic systems based on \cite{Szalay:2020xbz}.
For the usual tensor product $\mathcal{H}_A \otimes \mathcal{H}_B$ for Hilbert spaces $\mathcal{H}_A ,\mathcal{H}_B $, operators satisfy the relation
\begin{\eq}
\left( \mathcal{O}_A \otimes \mathcal{O}_B \right) \left( \mathcal{O}_{A'} \otimes \mathcal{O}_{B'} \right)
= \mathcal{O}_A \mathcal{O}_{A'} \otimes \mathcal{O}_B \mathcal{O}_{B'} ,
\end{\eq}
where $\mathcal{O}_{A,A'}$ ($\mathcal{O}_{B,B'}$) are operators on $\mathcal{H}_A$ ($\mathcal{H}_B$).
However, for fermionic systems, the canonical anti-commutation relation for creation and annihilation operators implies that fermionic operators at different points do not commute with each other. Thus, operators are not simply given by tensor products of operators with support in the relevant subsystem and the identity operator in its complement generically.

\subsection{Mapping between qubits and fermionic oscillators}
Let us consider an $L$-qubit system.
For each qubit labeled by the index $j=1,\cdots ,L$, we have a two dimensional Hilbert space $\mathcal{H}_j$:
\begin{\eq}
\mathcal{H}_j := {\rm Span}_{\mathbb{C}} \left\{ |\nu_j \rangle \ \vert \ \nu\in \{ 0,1\} \right\} ,
\end{\eq}
where the bases satisfy the orthonormal relation
\begin{\eq}
\langle \mu_j |\nu_j \rangle = \delta^{\mu ,\nu} .
\end{\eq}
We also have the algebra $\mathcal{A}_j$ on $\mathcal{H}_j$ with the standard basis:
\begin{\eq}
\left\{ E_j^{\nu ,\nu'} := |\nu_j \rangle\langle \nu_j' | \in \mathcal{A}_j \ \vert \  \nu ,\nu' \in \{ 0,1\}  \right\} ,
\end{\eq}
which is orthonormal: 
\begin{\eq}
( E_j^{\mu ,\mu'} |  E_j^{\nu ,\nu'} ) = \delta^{\mu ,\mu'} \delta^{\nu ,\nu'} ,
\end{\eq}
with 
the Hilbert-Schmidt inner product
\begin{\eq}
(A|B) := {\rm Tr} [ A^\dag B ] .
\end{\eq}
Let us also introduce the creation and annihilation operators as
\begin{\eq}
\psi_j^\dag := |1_j \rangle \langle 0_j | = \frac{X_j -iY_j}{2} ,\quad  
\psi_j := |0_j \rangle \langle 1_j | =\frac{X_j +iY_j}{2} .
\end{\eq}

Now let us consider a subsystem of the $L$ qubit system with the Hilbert space
\begin{\eq}
\mathcal{H}_Y := \bigotimes_{j\in Y} \mathcal{H}_j \quad  Y\subseteq \{ 1,\cdots , L\} ,
\end{\eq}
which has the standard basis
\begin{\eq}
\left\{ | \vec{\nu}_Y \rangle := \bigotimes_{j\in Y} |\nu_j \rangle \right\} ,
\end{\eq}
with 
\begin{\eq}
\langle \vec{\mu}_Y  | \vec{\nu}_Y  \rangle =  \delta^{\vec{\mu}  ,\vec{\nu} } .
\end{\eq}
We would like to extend the operators in $\mathcal{A}_j$ to the algebra ($ \simeq \bigotimes_{j\in Y} \mathcal{A}_j $) on $Y$.
One might consider the following na\"ive map:
\begin{\eq}
\Gamma_Y : E_j^{\nu ,\nu'} \ \mapsto\ 
\left( \bigotimes_{k\in Y , k<j} I_k \right) \otimes E_j^{\nu ,\nu'}\otimes \left( \bigotimes_{k\in Y , k>j} I_k \right) ,
\end{\eq}
which is simply the inclusion map.
However, the operators given by this map of the creation and annihilation operators $(\psi_j^\dag ,\psi_j)$ do not satisfy the canonical anti-commutation relation for fermions.\footnote{Rather they satisfy the commutation relations for so-called hardcore bosons.}
The latter implies that in fermionic systems, operators in disjoint subsystems generically do not commute with each other. 
Therefore, the operator algebra of fermionic systems is not simply a usual tensor product of those of subsystems. 

Instead let us consider the following map:
\begin{\eq}
\tilde{\Gamma}_Y : E_j^{\nu ,\nu'} \ \mapsto\ 
\left( \bigotimes_{k\in Y , k<j} (Z_k )^{\nu +\nu'} \right) \otimes E_j^{\nu ,\nu'}\otimes \left( \bigotimes_{k\in Y , k>j} I_k \right) ,
\end{\eq}
where
\begin{\eq}
Z_k :=   |0_k \rangle \langle 0_k | - |1_k \rangle \langle 1_k | .
\end{\eq}
Indeed, the creation and annihilation operators $(\tilde{\psi}_{j,Y}^\dag ,\tilde{\psi}_{j,Y} )$ defined as
\begin{\eq}
\tilde{\psi}_{j,Y}^\dag := \tilde{\Gamma}_Y (\psi_j^\dag ) ,\quad \tilde{\psi}_{j,Y} := \tilde{\Gamma}_Y (\psi_j ) ,
\end{\eq}
satisfy the canonical anti-commutation relation for fermions:
\begin{\eq}
\{ \tilde{\psi}_{j,Y} , \tilde{\psi}_{k,Y} \} = \{ \tilde{\psi}_{j,Y}^\dag , \tilde{\psi}_{k,Y}^\dag \} =0 ,\quad
\{ \tilde{\psi}_{j,Y} , \tilde{\psi}_{k,Y}^\dag \} = \delta_{j,k} .
\end{\eq}
The operators $(\tilde{\psi}_{j,Y}^\dag ,\tilde{\psi}_{j,Y} )$ are the Jordan-Wigner representations of the fermionic creation and annihilation operators \cite{Jordan:1928wi}.

\subsection{Fermionic tensor product}
A natural question is how the operator algebra of a composite fermionic system is related to those of subsystems.
To see this, let us define the following two bases of $ \bigotimes_{j\in Y} \mathcal{A}_j $ using $\Gamma_Y$ and $\tilde{\Gamma}_Y$: 
\begin{\eq}
\left\{ E_Y^{\vec{\nu} ,\vec{\nu}'} :=  \prod_{j\in Y} \Gamma_Y (E_j^{\nu_j ,\nu_j'} ) \in  \bigotimes_{j\in Y} \mathcal{A}_j  \right\}
\end{\eq}
and
\begin{\eq}
\left\{ \tilde{E}_Y^{\vec{\nu} ,\vec{\nu}'} :=  \prod_{j\in Y}^\rightarrow \tilde{\Gamma}_Y (E_j^{\nu_j ,\nu_j'} ) \in  \bigotimes_{j\in Y} \mathcal{A}_j  \right\} ,
\end{\eq}
which are orthonormal under the Hilbert-Shchmidt inner product. 
Their explicit expressions are given by
\begin{\eq}
E_Y^{\vec{\nu} ,\vec{\nu}'} = \bigotimes_{j\in Y} E_j^{\nu_j ,\nu_j'} ,
\end{\eq}
and
\begin{\eq}
\tilde{E}_Y^{\vec{\nu} ,\vec{\nu}'} 
= \prod_{j\in Y}^\rightarrow \left\{ \begin{matrix} 
\tilde{\psi}_{j,Y} \tilde{\psi}_{j,Y}^\dag & {\rm for}\ (\nu_j ,\nu_j' )=(0,0) \cr
\tilde{\psi}_{j,Y}                         & {\rm for}\ (\nu_j ,\nu_j' )=(0,1) \cr
             \tilde{\psi}_{j,Y}^\dag     & {\rm for}\ (\nu_j ,\nu_j' )=(1,0) \cr
\tilde{\psi}_{j,Y}^\dag \tilde{\psi}_{j,Y} & {\rm for}\ (\nu_j ,\nu_j' )=(1,1) 
\end{matrix} \right\} .
\end{\eq}
Then let us consider the two algebras $\mathcal{A}_Y$ and $\tilde{\mathcal{A}}_Y$ whose canonical bases are $E_Y^{\vec{\nu} ,\vec{\nu}'} $ and $\tilde{E}_Y^{\vec{\nu} ,\vec{\nu}'}$, respectively.
To study a relation between them, we consider a linear map $\Phi_Y : \mathcal{A}_Y \rightarrow \tilde{A}_Y$ defined by
\begin{\eq}
\Phi_Y :  E_Y^{\vec{\nu} ,\vec{\nu}'}  \mapsto \tilde{E}_Y^{\vec{\nu} ,\vec{\nu}'}.
\end{\eq}
A little computation shows
\begin{\eq}
\Phi_Y (E_Y^{\vec{\nu} ,\vec{\nu}'} ) =\tilde{E}_Y^{\vec{\nu} ,\vec{\nu}'} =f_Y^{\vec{\nu},\vec{\nu}' } E_Y^{\vec{\nu} ,\vec{\nu}'} ,
\end{\eq}
where $f_Y^{\vec{\nu},\vec{\nu}' }$ is a phase factor given by
\begin{\eq}
f_Y^{\vec{\nu},\vec{\nu}' } := (-1)^{\sum_{j\in  Y} \nu_j' \sum_{k\in Y ,j<k} (\nu_k +\nu_k' ) } .
\end{\eq}
The linear map $\Phi_Y$ is unitary\footnote{
But it is not a $\ast$-isomorphism meaning $\Phi_Y (A_Y^\dag )\neq \Phi_Y (A_Y )^\dag$ generically.
} and it connects $\Gamma_Y$ and  $\tilde{\Gamma}_Y$ as
\begin{\eq}
\tilde{\Gamma}_Y = \Phi_Y \circ \Gamma_Y .
\end{\eq}
In particular, the identity operator is unchanged under $\Phi_Y$:
\begin{\eq}
\tilde{I}_Y := \Phi_Y (I_Y ) = I_Y .
\end{\eq}
Similarly, given the usual trace map ${\rm Tr}: \mathcal{A}_Y \rightarrow \mathbb{C}$, we define its counter part $\widetilde{{\rm Tr}}: \tilde{\mathcal{A}}_Y \rightarrow \mathbb{C}$ as
\begin{\eq}
\widetilde{\rm Tr} = \Phi_\emptyset \circ {\rm Tr}\circ \Phi_Y^{-1} ,
\end{\eq}
which satisfies 
\begin{\eq}
\widetilde{\rm Tr} [\tilde{\Gamma}_Y (E_j^{\nu ,\nu'} ) ]
={\rm Tr} [\Gamma_Y (E_j^{\nu ,\nu'} ) ]=2^{|Y|-1}\delta^{\nu ,\nu'}.
\end{\eq}

In terms of $\Phi_Y$, we define the fermionic tensor product as follows.
Let us partition $Y$ into $\xi =\{ X_1 ,X_2 ,\cdots \}$:
\begin{\eq}
\bigcup_{X\in \xi} X =Y\,,
\end{\eq}
where $X \in\xi$ is a nonempty and disjoint subsystem of $Y$.
The canonical basis $E_Y^{\vec{\nu},\vec{\nu}'}$ for $\mathcal{A}_Y$ has the usual tensor product structure
\begin{\eq}
E_Y^{\vec{\nu},\vec{\nu}'} = \bigotimes_{X\in \xi} E_X^{\vec{\nu}_X ,\vec{\nu}_X'} .
\end{\eq}
As we will see soon, its counter part for $\tilde{E}_Y^{\vec{\nu},\vec{\nu}'}$ is more complicated.
To see this, let us introduce the fermionic tensor product as
\begin{\eq}
 \widetilde{\bigotimes_{X\in\xi}} \tilde{A}_X  
 := \tilde{\Psi}_\xi \left( \bigotimes_{X\in\xi} \tilde{A}_X  \right) ,\quad
 \tilde{\Psi}_\xi := \Phi_Y \circ \bigotimes_{X\in \xi} \Phi_X^{-1} \qquad
(\tilde{A}_X \in\mathcal{\tilde{A}}_X ). 
\end{\eq}
One can show that $\tilde{E}_Y^{\vec{\nu},\vec{\nu}'}$ satisfies 
\begin{\eq}
\tilde{E}_Y^{\vec{\nu},\vec{\nu}'} = \widetilde{\bigotimes_{X\in\xi}} \tilde{E}_X^{\vec{\nu}_X ,\vec{\nu}_X'} ,
\end{\eq}
because
\begin{\eqa}
\tilde{E}_Y^{\vec{\nu},\vec{\nu}'} 
=\Phi_Y \left( \bigotimes_{X\in\xi}  E_X^{\vec{\nu}_X ,\vec{\nu}_X'} \right)
=\Phi_Y \left( \bigotimes_{X\in\xi} \Phi_X^{-1} \left( \tilde{E}_X^{\vec{\nu}_X ,\vec{\nu}_X'} \right) \right)
&=& \left( \Phi_Y \circ \bigotimes_{X\in\xi} \Phi_X^{-1} \right) \left( \bigotimes_{X\in\xi}  \left( \tilde{E}_X^{\vec{\nu}_X ,\vec{\nu}_X'} \right) \right) \NN\\
&=& \tilde{\Psi}_\xi  \left( \bigotimes_{X\in\xi}  \left( \tilde{E}_X^{\vec{\nu}_X ,\vec{\nu}_X'} \right) \right) .
\end{\eqa}
Then the algebra $\tilde{\mathcal{A}}_Y$ can be written as
\begin{\eq}
\tilde{\mathcal{A}}_Y = \widetilde{\bigotimes_{X\in\xi}} \tilde{\mathcal{A}}_X  .
\end{\eq}
While the fermionic tensor product is convenient, it does not fully obey the $\ast$-algebraic structure.

\subsection{Fermionic canonical embedding}
Let us consider the embedding of the algebras in $X$ into the ones in $Y$.
For the usual qubits case, we have the canonical embedding $\iota_{X,Y}: \mathcal{A}_X \rightarrow \mathcal{A}_Y$ defined by
\begin{\eq}
\iota_{X,Y}: A_X \ \mapsto\  A_X \otimes I_{\bar{X}} \qquad  (\bar{X}= Y \setminus X) .
\end{\eq}
For the fermionic counterpart, we introduce the fermionic canonical embedding $\tilde{\iota}_{X,Y}: \tilde{\mathcal{A}}_X \rightarrow \tilde{\mathcal{A}}_Y$ defined by
\begin{\eq}
\tilde{\iota}_{X,Y}: \tilde{A}_X \ \mapsto\  \tilde{A}_X \tilde{\otimes} \tilde{I}_{\bar{X}} \qquad  (\bar{X}= Y \setminus X) .
\end{\eq}
It is known that the two canonical embeddings are unitarily equivalent:
\begin{\eq}
\tilde{A}_X \tilde{\otimes} \tilde{I}_{\bar{X}}  = U_{X\bar{X}} ( \tilde{A}_X \otimes \tilde{I}_{\bar{X}} ) U_{X\bar{X}}^\dag
\qquad (U_{X\bar{X}} \in U(\mathcal{H}_Y) ) ,
\end{\eq}
and $\ast$-homomorphic.
For the special case of $X=\{ j \}$, the fermionic canonical embedding is nothing but $\tilde{\Gamma}_Y$:
\begin{\eq}
\tilde{\Gamma}_Y (A_j ) 
= \tilde{A}_{\{ j \} } \tilde{\otimes} \tilde{I}_{Y\setminus \{ j \} }
= \tilde{\iota}_{\{ j \} ,Y} (\tilde{A}_{\{ j \}} ) ,
\end{\eq}
and therefore it is the extension of $\tilde{\Gamma}_Y$ to general $X$.
In terms of the fermionic canonical embedding, it is convenient to introduce 
\begin{\eq}
\widehat{\bigotimes_{X\in\vec{\xi} }}\tilde{A}_X 
:= \prod_{X\in\vec{\xi}}^\rightarrow \tilde{\iota}_{X,Y} (\tilde{A}_X) ,
\end{\eq}
where $\vec{\xi}$ denotes an ordered partition of $Y$.
This satisfies
\begin{\eq}
\widehat{\bigotimes_{X\in\vec{\xi} }}\tilde{A}_X 
=\widetilde{\bigotimes_{X\in \xi }}\tilde{A}_X .
\end{\eq}

The above objects lead to fermionic counterparts of the partial trace.
For the usual qubits, we have the partial trace ${\rm Tr}_{Y,X}: \mathcal{A}_Y \rightarrow \mathcal{A}_X$ for $X\subseteq Y$ defined as
\begin{\eq}
{\rm Tr}_{Y,X}: A_X \otimes B_{\bar{X}}\ \mapsto\  A_X  {\rm Tr}[ B_{\bar{X}}] .
\end{\eq}
Similarly, we introduce the fermionic partial trace $\widetilde{{\rm Tr}}_{Y,X}: \tilde{\mathcal{A}}_Y \rightarrow \tilde{\mathcal{A}}_X$ defined as
\begin{\eq}
\widetilde{\rm Tr}_{Y,X} = \Phi_X  \circ {\rm Tr}\circ \Phi_Y^{-1} \,,
\end{\eq}
which satisfies
\begin{\eq}
\widetilde{{\rm Tr}}_{Y,X}: \tilde{A}_X \tilde{\otimes} \tilde{B}_{\bar{X}}\ \mapsto\  \tilde{A}_X  \widetilde{{\rm Tr}}[ \tilde{B}_{\bar{X}}] .
\end{\eq}
These trace maps are related to the canonical embedding maps as
\begin{\eq}
{\rm Tr}_{Y,X} = \iota_{X,Y}^\dag ,\quad \widetilde{{\rm Tr}}_{Y,X} = \tilde{\iota}_{X,Y}^\dag .
\end{\eq}
It is also known that the fermionic partial trace preserves positivity:
\begin{\eq}
\tilde{\rho}_Y \geq 0 \quad \Rightarrow \quad \widetilde{\rm Tr}_{Y,X}[ \tilde{\rho}_Y ] \geq 0 .
\end{\eq}
While the fermionic canonical embedding is useful, we generically have the property
\begin{\eq}
[\tilde{\iota}_{X,Y} (\tilde{A}_X) , \tilde{\iota}_{\bar{X},Y} (\tilde{B}_{\bar{X}}) ] 
= [\tilde{A}_X \tilde{\otimes} \tilde{I}_{\bar{X}} , \tilde{I}_X \tilde{\otimes} \tilde{B}_{\bar{X}} ] \neq 0 \qquad 
(\tilde{A}_X \in \tilde{\mathcal{A}}_X , \tilde{B}_{\bar{X}}\in \tilde{\mathcal{A}}_{\bar{X}}) ,
\label{eq:fermi-issue}
\end{\eq}
in contrast to the usual qubit case:
\begin{\eq}
[A_X \otimes I_{\bar{X}} , I_X \otimes B_{\bar{X}} ] = 0 \qquad (A_X \in \mathcal{A}_X , B_{\bar{X}}\in\mathcal{A}_{\bar{X}}) . 
\end{\eq}

\subsection{Superselection rule for fermion parity}
The above issue \eqref{eq:fermi-issue} does not occur if we impose a superselection rule for the fermionic parity.
Let us consider the fermion parity operator $\tilde{F}_Y$ on $Y$:
\begin{\eq}
\tilde{F}_Y 
:= \prod_{j\in Y}^\rightarrow \tilde{\Gamma}_Y (Z_j )
=\widehat{\bigotimes_{j\in Y}} Z_{\{ j\} }
=\widetilde{\bigotimes_{j\in Y}} Z_{\{ j\} }
=\bigotimes_{j\in Y} Z_{j } .
\end{\eq}
Then we can divide $\mathcal{H}_Y$ into the two subspaces $\mathcal{H}_Y^\pm$ by the eigenvalues of $\tilde{F}_Y $:
\begin{\eq}
\mathcal{H}_Y^\pm
:= \left\{ |\psi_Y \rangle \in \mathcal{H}_Y \ \vert\ \tilde{F}_Y |\psi_Y \rangle = \pm |\psi_Y \rangle \right\} ,
\end{\eq}
or, equivalently,
\begin{\eq}
\mathcal{H}_Y^\pm
= {\rm Span} \left\{ |\vec{\nu}_Y \rangle \ \vert\  (-1)^{\sum_{j\in Y} \nu_j} = \pm 1 \right\} .
\end{\eq}
Similarly, given the partition $\xi =\{ X_1 ,X_2 ,\cdots \} $ of $Y$ and local fermion parity $\epsilon_X$ on $X\in \xi$, we can consider the subspace
\begin{\eq}
\mathcal{H}_\xi^{\vec{\epsilon}}
:= \left\{ |\psi_Y \rangle \in \mathcal{H}_Y \ \vert\ \forall X\in\xi :
\tilde{\iota}_{X,Y}( \tilde{F}_X ) |\psi_Y \rangle = \epsilon_X |\psi_Y \rangle \right\} .
\end{\eq}

We can also decompose the algebra $\tilde{\mathcal{A}}_Y$ into two subspaces (cf.~Eq.~\eqref{eq_parity}):
\begin{\eq}
\tilde{\mathcal{A}}_Y^\pm
:= \left\{ \tilde{A}_Y \in \tilde{\mathcal{A}}_Y \ \vert\ \tilde{F}_Y \tilde{A}_Y \tilde{F}_Y^{-1} = \pm \tilde{A}_Y \right\} ,
\end{\eq}
or equivalently
\begin{\eq}
\tilde{\mathcal{A}}_Y^\pm
= {\rm Span} \left\{ \tilde{E}_Y^{\vec{\nu} ,\vec{\nu}' }  \ \vert\  (-1)^{\sum_{j\in Y} ( \nu_j +\nu_j' )} = \pm 1 \right\} .
\end{\eq}
Similarly, the algebra for the $\xi$-local subspace is decomposed into
\begin{\eq}
\tilde{\mathcal{A}}_\xi^{\vec{\epsilon}}
:= \left\{ \tilde{A}_Y \in \tilde{\mathcal{A}}_Y \ \vert\  \forall X\in\xi :
\tilde{\iota}_{X,Y} ( \tilde{F}_X ) \tilde{A}_Y \tilde{\iota}_{X,Y} ( \tilde{F}_X^{-1} ) = \epsilon_X \tilde{A}_Y \right\} .
\end{\eq}
The fermionic canonical embedding and the femionic partial trace preserve the parity:
\begin{\eq}
\tilde{\iota}_{X,Y} : \tilde{A}_X^\pm \ \rightarrow\ \tilde{A}_Y^\pm ,\quad
\widetilde{\rm Tr}_{Y,X} : \tilde{A}_Y^\pm \ \rightarrow\ \tilde{A}_X^\pm .
\end{\eq}
For $\tilde{A}_X \in \tilde{\mathcal{A}}_X^{\epsilon_X}$ and  $\tilde{B}_{\bar{X}} \in \tilde{\mathcal{A}}_{\bar{X}}^{\epsilon_{\bar{X}}}$, one can show (cf.~Eq.~\eqref{eq_oddevencommute})
\begin{\eq}
\tilde{\iota}_{X,Y} (\tilde{A}_X) \tilde{\iota}_{\bar{X},Y} (\tilde{B}_{\bar{X}})
=\pm \tilde{\iota}_{\bar{X},Y} (\tilde{B}_{\bar{X}}) \tilde{\iota}_{X,Y} (\tilde{A}_X)  ,
\end{\eq}
where the lower sign case is for $\epsilon_X = \epsilon_{\bar{X}}=-1$.
This implies the commutativity of the local operators of even parity:
\begin{\eq}
[\tilde{\iota}_{X,Y} (\tilde{A}_X) , \tilde{\iota}_{\bar{X},Y} (\tilde{B}_{\bar{X}}) ] 
=[\tilde{A}_X \tilde{\otimes} \tilde{I}_{\bar{X}} , \tilde{I}_X \tilde{\otimes} \tilde{B}_{\bar{X}} ] = 0 \qquad 
(\tilde{A}_X \in \tilde{\mathcal{A}}_X^+ , \tilde{B}_{\bar{X}}\in \tilde{\mathcal{A}}_{\bar{X}}^+ ) .
\end{\eq}

\section{Derivation of the dressed bosonic matter algebra}\label{app_bosalgebra}

Here, we expound on the shape of the dressed matter algebra in Eq.~\eqref{eq_drmatteralg}. Let us start by considering the factor of the kinematical Hilbert space in Eq.~\eqref{eq_Hkinbosfac} involving only frame and matter degrees of freedom,
\begin{equation}
    \mathcal{H}_{\text{kin}}' \coloneqq \mathcal{H}_R \otimes \mathcal{H}_{\text{matter}}.
\end{equation}
This space carries a faithful unitary representation of the gauge group $G^{N_V}$,
\begin{equation}
    U_{Rm}^{(\bm{g}, g_{v_0})} \coloneqq U_R^{(\bm{g}, g_{v_0})} \otimes U_{\text{matter}}^{(\bm{g}, g_{v_0})}, \qquad \bm{g} \in G^{N_{V}-1}, g_{v_0} \in G,
\end{equation}
where for convenience below we split off gauge transformations at the root $v_0$, and the dressed matter Hilbert space is the subspace of $\mathcal{H}_{\text{kin}}'$ invariant under this representation,
\begin{equation}
    \mathcal{H}^{\text{dr}}_{\text{matter}} = \{ \ket{\Psi}' \in \mathcal{H}_{\text{kin}}' \mid U_{Rm}^{(\bm{g}, g_{v_0})} \ket{\Psi}' = \ket{\Psi}' \enspace \forall \,(\bm{g}, g_{v_0}) \in G^{N_V}  \}  .
\end{equation}

The orientations of $R$ are perfectly distinguishable: it comprises an ideal QRF for the full group of gauge transformations $\mathcal{G}=\{U_{Rm}^{(\bm{g}, g_{v_0})} \mid (\bm{g}, g_{v_0}) \in G^{N_V}  \}$, as we saw in section~\ref{ssec_QRF}. However, its orientations take values in $G^{N_V - 1}$ and therefore cannot parametrize gauge transformations at $v_0$, i.e.\ it is an incomplete QRF  for $\mathcal{G}$ \cite{delaHamette:2021oex,Araujo-Regado:2025ejs}. This means that the tree QRF has an isotropy group isomorphic to $G$ that leaves its orientations invariant (and corresponds to the gauge transformations at $v_0$). 

While one can still describe the entire algebra on $\Hil_{\rm matter}^{\rm dr}$ as relational observables of matter degrees of freedom relative to the QRF, there is the additional subtlety, compared to complete frames, that there is an equivalence relation on system observables (here the matter observables), namely two such observables are indistinguishable relative to the QRF if they differ by an action of its isotropy group \cite{delaHamette:2021oex}. This means that the Page-Wootters relationalization map in Eq.~\eqref{relobs} implements an average of matter observables over the isotropy group of the frame, see \cite[App.~E]{Araujo-Regado:2025ejs} for a careful discussion of this in lattice gauge theory. 

In what follows, we will take care of this in a two-step process:
\begin{itemize}
    \item[(i)] First ignore gauge transformations at $v_0$, treating the tree QRF $R$ as ideal \emph{and} complete for gauge transformations everywhere else.
    \item[(ii)] Thereafter, project the relational algebra of step (i) onto its subalgebra that is invariant at $v_0$ also.
\end{itemize}

Let us thus begin with step (i), taking $R$ as  complete and ideal for the group
\begin{equation}
    \widetilde{\mathcal{G}} \coloneqq \{ U_{R_v}^{g} \otimes U_v^g \mid g \in G, v \in \mathcal{V} \setminus \{ v_0 \}   \},
\end{equation}
which motivates the introduction of the Page-Wootters reduction map $ \widetilde{\mathcal{R}}_R^{\bm{g}}: \mathcal{H}_{\widetilde{\mathcal{G}}} \to \mathcal{H}_{\text{matter}} $ from the $\widetilde{\mathcal{G}}$-invariant subspace of $\mathcal{H}_{\text{kin}}'$ to the matter Hilbert space:
\begin{equation}
    \widetilde{\mathcal{R}}_R^{\bm{g}} \coloneqq \left( \bra{\bm{g}}_R \otimes \mathbbm{1}_{\text{matter}}   \right) \Pi_{\widetilde{\mathcal{G}}}.
\end{equation}
Here, $\Pi_{\widetilde{\mathcal{G}}}$ is the projector onto $\mathcal{H}_{\widetilde{\mathcal{G}}}$,
\begin{equation}
    \Pi_{\widetilde{\mathcal{G}}} \coloneqq \int \dee^{N_V -1} \bm{h} \left( \prod_{v \in \mathcal{V} \setminus \{ v_0 \}  }  U_{R_v}^{h_v} \otimes u_v^{h_v}   \right).
\end{equation}
In this setting, relational observables constructed through $\widetilde{\mathcal{R}}_R^{\bm{g}}$ from the matter sector generate the algebra of bounded operators on $\mathcal{H}_{\widetilde{\mathcal{G}}}$, $\mathcal{A}_{\widetilde{\mathcal{G}}} = \mathcal{B}(\mathcal{H}_{\widetilde{\mathcal{G}}})$ because $\tilde{\mathcal{R}}^{\bm{g}}_R$ is unitary and so we have $\mathcal{A}_{\tilde{\mathcal{G}}}\simeq\mathcal{A}_{\rm matter}$ \cite{delaHamette:2021oex}. 

We recall that the matter algebra is generated by creation and annihilation operators for the various particles types $i\in\mathbb{Z}_r$, c.f.~the discussion in section~\ref{ssssec_bosHil}. Moreover, since for complete and ideal QRFs (here for $\widetilde{\mathcal{G}}$) the relational observable map, i.e.\ the dressing, is a unital $*$-homomorphism \cite{delaHamette:2021oex}, their frame-dressed counterparts generate $\mathcal{A}_{\widetilde{\mathcal{G}}}$:
\begin{equation}
    \mathcal{A}_{\widetilde{\mathcal{G}}} = \left\langle  \widetilde{O}_{a_{v,i}^{\dagger}|R}^{\bm{g}},     \widetilde{O}_{a_{v,i}|R}^{\bm{g}} \Big| v \in \mathcal{V},  i\in\mathbb{Z}_r   \right\rangle,
\end{equation}
where 
\begin{eqnarray}
\widetilde{O}_{a_{v,i}|R}^{\bm{g}} &\coloneqq&(\widetilde{\mathcal{R}}_R^{\bm{g}})^{\dagger} a_{v,i} \widetilde{\mathcal{R}}_R^{\bm{g}} \\
    &=& \Pi_{\widetilde{\mathcal{G}}} \left( \ket{\bm{g}}\!\bra{\bm{g}}_R \otimes \mathbbm{1}_{\text{matter}}  \right)a_{v,i} \Pi_{\widetilde{\mathcal{G}}} \\
    &=& \int \dee^{N_V - 1} \bm{h} \left( \prod_{v' \in \mathcal{V} \setminus \{ v_0 \} } U_{R_{v'}}^{h_{v'}} \otimes u_{v'}^{h_{v'}}   \right) \ket{\bm{g}}\!\bra{\bm{g}}_R a_{v,i} \left( \prod_{v' \in \mathcal{V} \setminus \{ v_0 \}} U_{R_{v'}}^{h_{v'} \dagger} \otimes u_{v'}^{h_{v'} \dagger}   \right) \Pi_{\widetilde{\mathcal{G}}}. ~ ~ ~ ~
\end{eqnarray}
At this point, we need to distinguish $v \neq v_0$ from $v=v_0$. In the former case, the gauge transformations act nontrivially on $a_{v,i}$, so we get a factor $\bar \rho_i(h_v)$ contributing to the $h_v$ integral, while the integrals over the remaining gauge parameters resolve the identity at every vertex:
\begin{eqnarray}
\widetilde{O}_{a_{v,i}|R}^{\bm{g}} &=& \left( \int \dee h_v \bar \rho_i(h_v) \ket{h_vg_v} \bra{h_vg_v}_{R_v} \right) a_{v,i} \Pi_{\widetilde{\mathcal{G}}} \\
&\propto& \left( \int \dee h_v \bar \rho_i(h_v) \ket{h_v}\bra{h_v}_R \right) a_{v,i}\Pi_{\widetilde{\mathcal{G}}},
\end{eqnarray}
where we recognize the term in brackets as $W_{R_v}^{{\rho_i}}$, the tree Wilson line connecting $v$ to $v_0$ in representation ${\rho_i}$. On the other hand, if $v=v_0$, the gauge transformations in the integral leave it invariant and we instead get
\begin{equation}
\widetilde{O}^{\bm{g}}_{a_{v_0,i} | R} \propto a_{v,i} \Pi_{\widetilde{\mathcal{G}}}.
\end{equation}

Hence, we arrive at 
\begin{equation}
    \mathcal{A}_{\widetilde{\mathcal{G}}} = \left\langle a_{v,i} W_{R_v}^{{\rho_i}} \Pi_{\widetilde{\mathcal{G}}}, \thinspace a_{v_0, i} \Pi_{\widetilde{\mathcal{G}}}, \thinspace h.c. \Big| v \in \mathcal{V} \setminus \{ v_0  \}, i\in\mathbb{Z}_r  \right\rangle.
    \label{alg_tilde_G}
\end{equation}

Of course, this algebra still contains operators transforming nontrivially at $v_0$, so let us now implement step (ii). Here, we make use of the fact that the $G$-twirl (or incoherent group average) $\mathcal{G}_{v_0}:\mathcal{A}_{\rm kin}'\to \mathcal{A}_{\rm kin}'^{G_{v_0}}$, $\mathcal{G}_{v_0}(\bullet)=\int{\dee}g_{v_0}U_{Rm}^{g_{v_0}}(\bullet)U_{Rm}^{g_{v_0}\dag}$, is the orthogonal (with respect to the Hilbert-Schmidt inner product) projector onto the subalgebra invariant under gauge transformations at $v_0$ when $G$ is compact. We further note that in $\mathcal{H}_{\widetilde{\mathcal{G}}}$ global gauge transformations are equivalent to gauge transformations at $v_0$ and that $\mathcal{G}_{v_0}(\bullet)\Pi_{v_0} = \Pi_{v_0}(\bullet)\Pi_{v_0}$, where $\Pi_{v_0}$ is the coherent group average over gauge transformations at $v_0$:
\begin{eqnarray}
    \mathcal{A}^{\text{dr}}_{\text{matter}} &=& \mathcal{G}_{v_0} (\mathcal{A}_{\widetilde{\mathcal{G}}}) \Pi_{v_0}\label{eq_drmatalg3} \\
    &=& \Big\langle \int \dee h \dee h'  U^{h}_{\text{global}} O_{\widetilde{\mathcal{G}}}U_{\text{global}}^{h' \dagger}  \Big| O_{\widetilde{\mathcal{G}}} \in \mathcal{A}_{\widetilde{\mathcal{G}}} \Big\rangle\,.\label{eq_drmatalg2}
\end{eqnarray}

It is clear that the generators in Eq.~(\ref{alg_tilde_G}) vanish upon integration, as they are charged in nontrivial irreps $\rho_i$ or $\bar\rho_i$ at $v_0$ (and so have zero support on the trivial irrep onto which $\mathcal{G}_{v_0}$ projects), $\mathcal{G}_{v_0}(\tilde{O}^{\bm{g}}_{a_{v,i}|R})=0$. However, the $G$-twirl is not a homomorphism, $\mathcal{G}_{v_0}(AB)\neq\mathcal{G}_{v_0}(A)\mathcal{G}_{v_0}(B)$, and, of course, the projection of invariant combinations of the generators of $\mathcal{A}_{\tilde{\mathcal{G}}}$ does not vanish in general. Indeed, we may combine the half open Wilson lines $a_{v,i}W_{R_v}^{{\rho_i}}$ into fully gauge invariant objects by appending either $a_{v', i}^{\dagger} W_{R_{v'}}^{\bar \rho_i}$ or $a_{v', \bar{i}}W_{R_{v'}}^{ \bar \rho_i}$ for any vertex $v'$, including $v$ and $v_0$. For $v' \neq v, v_0 $, this gives rise to two different types of dressed tree Wilson lines for each particle type $i$ and antiparticle type $\bar i$:
\begin{equation}
     a_{v,i} W_{\gamma_R[v,v']}^{ \rho_i}a_{v',i}^{\dagger}\,, \qquad a_{v,\bar i}^\dag W^{\rho_i}_{\gamma_R[v,v']}a_{v',\bar i}\,,\qquad\text{and} \quad a_{v,\bar{i}}^{\dagger}W_{\gamma_R[v,v']}^{ \rho_ i}a_{v', i}^{\dagger},
    \label{transport_and_pairing}
\end{equation}
together with their Hermitian conjugates, whereas taking $v'=v$ or $v=v_0$, we recover the number operators $a_{v,i}^{\dagger}a_{v,i}$, the gauge-invariant products $a_{v,i}^{\dagger}a_{v',\bar i}^{\dagger}$, etc. Since $\mathcal{G}_{v_0}$ leaves them invariant, inserting them for $O_{\tilde{\mathcal{G}}}$ into Eq.~\eqref{eq_drmatalg2} is equivalent to simply multiplying them with $\Pi_{\rm pn}$, yielding the generators of the dressed matter algebra, Eq.~\eqref{eq_drmatteralg}, given in the main text.

It thus remains to argue that these are indeed the generators. Now, when the species $i$ is of infinite order, the only way to neutralize a charge is with an anti-charge, so that charge-zero observables require an equal number of charges and anti-charges, which always yield zero charge in pairs. Thus, any such zero-charge observable will be a function of the bipartite ones given in Eq.~\eqref{transport_and_pairing}, including their Hermitian conjugates, since these encompass all such bipartite ones. 

When the species $i$ is of finite order, due to $\rho_i^{D_i}=1$ and neutrality only being defined $\!\!\!\mod D_i$, the situation is slightly more subtle. Indeed, we may now have gauge-invariant operators of the kind discussed around the star operator in Eq.~\eqref{eq_star}, which involve an odd number of creation and/or annihilation operators when $D_i$ is odd. However, one may convince oneself that, again, any such net-zero-charge observable can again be written as an algebraic combination of the bipartite ones given in Eq.~\eqref{transport_and_pairing}. The key to this observation is the fact that, thanks to Eqs.~\eqref{eq_aadagcommutator} and~\eqref{eq_highestweight}, the identity operator can always be written as 
\begin{equation}\label{eq_identity}
    \mathbbm{1}_{v,i}=[a_{v,i},a_{v,i}^\dag]+ c\,\left(a_{v,i}^\dag\right)^{D_i-1}\,a_{v,i}^{D_i-1}
\end{equation}
for some constant $c>0$, involving only an equal number of creation and annihilation operators (and similarly for the antiparticle $\bar i$).\footnote{For example, one may replace the identity $\mathbbm{1}_{v_0,i}$ in the center of the star in Eq.~\eqref{eq_star} with the right hand side in Eq.~\eqref{eq_identity}, upon which one will find that the star operator can be written as an algebraic combination of the bipartite operators in Eq.~\eqref{transport_and_pairing}.} Moreover, Eq.~\eqref{eq_finiteantianil}, including for exchanging $i\leftrightarrow \bar i$ and taking their Hermitian conjugates, permits one to replace $D_i-1$ powers of any of the particle or antiparticle creation and annihilation operators in terms of a sum of terms that only involve odd numbers of such operators. Combined in the appropriate way, this permits one to write any gauge-invariant object involving an odd number of creation and/or annihilation operators in terms of an expression that involves an even number of them, such that charge-anti-charge pairs arise. 
The main insight is that since $D_i-1$ particles of type $i$ correspond to one anti-particle of the same species, one can always obtain  operators such as Eq.~\eqref{eq_star}, even when the $D_i-1$ particles reside on distinct vertices, from a single bipartite particle-antiparticle observable and acting on it with other bipartite zero-charge observables that spread the $D_i-1$ particles out across different vertices. 

Thus, also for finite order, every zero-charge observable can be written as an operator containing an equal number of charges and anti-charges that neutralize in pairs, so that the bipartite operators in Eq.~\eqref{transport_and_pairing} and their Hermitian conjugates indeed generate the full algebra, yielding Eq.~\eqref{eq_drmatteralg}.

\section{Knill-Laflamme condition for bosonic Gauss law codes}\label{app_bosKL}
Here we provide the missing details in our derivation of Eq.~(\ref{bosonicKL}), which allowed us to characterize maximal sets of correctable errors in the bosonic Gauss law code in terms of sections of the extended Gauss law bundle $\partial_B'$ in Eq.~\eqref{eq_extbosGaussbundle}. The goal is to show that two distinct errors $E_a = Z(\bm{z}_a, \bm{\theta}_b)X^{\bm{x}_a}W^{\bm{\chi}_a}$, $E_b = Z(\bm{z}_b, \bm{\theta}_b)X^{\bm{x}_b}W^{\bm{\chi}_b}$ satisfy the Knill-Laflamme condition if and only if they violate Gauss's law by the same amount, i.e.\  if $\partial'_B(\bm{\chi}_a,\bm{z}_a,\bm{\theta}_a,\bm{x}_a) = \partial'_B(\bm{\chi}_b,\bm{z}_b,\bm{\theta}_b,\bm{x}_b)$. 

We begin by invoking Eqs.~(\ref{WandH}) and~(\ref{Wsonframepuregauge}) to write
\begin{equation}
    W^{\bm\chi}=H^{\bm\chi_S}\sum_{\partial\bm\eta\in\hat{G}^{N_V-1}}\ket{\partial\bm\chi\partial\bm\eta}\!\bra{\partial\bm\eta}_R\,.
\end{equation}
Hence, the combination $X^{\bm{x}}W^{\bm{\chi}}$ acts on the dressed matter states in Eq.~\eqref{partialchin} as 
\begin{equation}
    X^{\bm{x}}W^{\bm{\chi}} \ket{\psi}_{\rm loops}\otimes_R\ket{\bf{n}}^{\text{dr}}_{\text{matter}} = H^{\bm{\chi}_S}\ket{\psi}_{\rm loops}\otimes_R\ket{\partial\bm\chi\partial \bm{\eta}^{ ( \bf{n} )}}_R \otimes \ket{{\bf{n}} + \bm{x}}_{\text{matter}}
\end{equation}
for arbitrary $\ket{\psi}_{\rm loops}\in\Hil_{\rm loops}$.
This lies outside of the code space and so produces a detectable syndrome unless $\bm{x}$ and $\bm{\chi}$ are such that $\partial\bm\chi\partial \bm{\eta}^{ ( \bf{n} )}= \partial \bm{\eta}^{(\bm{x} + \bf{n})}=\partial\bm\eta^{(\bm{x})}\partial\bm\eta^{({\bf n})}$. Since $\partial \bm{\eta}^{(\bm{x})}$ only depends on $\bm{x}$ through $\Delta \bf{x}$ as defined in Eq.~\eqref{eq_Delta} (recall that $\partial\bm\eta^{(\bm{x})}$ is defined via Eqs.~\eqref{eq_bosonicGauss} and~\eqref{eq_rhototbos}), this requirement is equivalent to $(\partial\bm\chi)\bm\rho[\Delta\bm{x}]=\bm{1}$, with $\bm\rho[\Delta\bm{x}]$ given in Eq.~\eqref{eq_speciescontrib}. Projecting onto $\mathcal{H}_{\text{pn}}$ therefore yields
\begin{equation}\label{XWondressedstates}
    \Pi_{\text{pn}}X^{\bm{x}}W^{\bm{\chi}} \ket{\psi}_{\rm loops}\otimes_R\ket{\bf{n}}^{\text{dr}}_{\text{matter}} = \delta_{(\partial\bm\chi)\bm\rho[\Delta\bm{x}], \bm{1}}\,H^{\bm{\chi}_S}\ket{\psi}_{\rm loops}\otimes_R  \ket{{\bf n} + \bm{x}}^{\text{dr}}_{\text{matter}},
\end{equation}
which essentially constitutes the part of Eq.~(\ref{bosonicKL}) without the $Z$-errors when inserted in the Knill-Laflamme condition. 

In more detail, including the $Z$-errors, which commute with $\Pi_{\rm pn}$, and making use of the fact that the set of errors that we consider is closed under multiplication (up to phase), we find
\begin{eqnarray}
    \Pi_{\text{pn}} E_a^{\dagger}E_b \Pi_{\text{pn}} &\propto& Z(\bm{z}_b - \bm{z}_a, \bm{\theta}_b - \bm{\theta}_a) \Pi_{\text{pn}}  X^{\bm{x}_b - \bm{x}_a} W^{\bar{\bm{\chi}}_a\bm{\chi}_b} \Pi_{\text{pn}} \\
    &\propto& Z(\bm{z}_b - \bm{z}_a, \bm{\theta}_b - \bm{\theta}_a) \sum_{\bf{n} \in \widetilde{\mathcal{N}}} \Pi_{\text{pn}}X^{\bm{x}_b - \bm{x}_a}W^{\bar{\bm{\chi}}_a\bm{\chi}_b} \left(\mathbbm{1}_{\rm loops}\otimes_R\ket{\bf{n}} \!\bra{\bf{n}}^{\text{dr}}_{\text{matter}}\right)\,,
\end{eqnarray}
and invoking Eq.~(\ref{XWondressedstates}) we recover Eq.~(\ref{bosonicKL}) in full:
\begin{equation}
    \Pi_{\text{pn}}E_a^{\dagger}E_b\Pi_{\text{pn}} \propto Z(\bm{z}_b - \bm{z}_a, \bm{\theta}_b - \bm{\theta}_a) X_{\rm dr}^{\bm{x}_b - \bm{x}_a} H^{(\bm{\bar{\chi}}_a\bm{\chi}_b)_S} \delta_{(\partial\bm{\chi}_a)\bm{\rho}[\Delta\bm{x}_a],(\partial\bm{\chi}_b)\bm{\rho}[\Delta\bm{x}_b]}\Pi_{\text{pn}},
\end{equation}
where we have the dressed $X$-operator
\begin{equation}
    X_{\rm dr}^{\bm{x}_b - \bm{x}_a} \coloneqq \sum_{\bf{n} \in \tilde{\mathcal{N}}} \ket{{\bf n} + \bm{x}_b - \bm{x}_a} \!\bra{\bf{n}}^{\text{dr}}_{\text{matter}}.
\end{equation}

\section{Dressed fermionic matter algebra}\label{app_fermalg}

Let us briefly explain the shapes of the dressed fermionic algebras in Eqs.~\eqref{eq_inffermalg} and~\eqref{eq_finfermalg}. Noting that $u_v^g$ in Eq.~\eqref{eq_fermgaugetr} is an even operator, the derivation proceeds along the exact same QRF-dressing steps as in the bosonic case, leading up to Eq.~\eqref{eq_drmatalg2}, except that we replace every instance of $a_{v,i},a_{v,j}^\dag$
with $\psi_v,\psi_v^\dag$, respectively, and that at each site we either only have fermions or antifermions in line with the staggered nature. This immediately yields the gauge-invariant bipartite combinations
\begin{equation}
\psi_v\,W^{\chi_F}_{\gamma_R[v,v']}\,\psi_{v'}^\dag\,\Pi_{\rm pn}\,,\qquad\qquad v,v'\in\mathcal{V}\,,
\end{equation}
as algebra generator candidates. 

For $\chi_F$ of infinite order, these also constitute a complete set of generators, for the same reason as in the bosonic case, resulting in Eq.~\eqref{eq_inffermalg}.

A few additional words are in place for the opposite case that $D_F<\infty$. As explained below Eq.~\eqref{eq_finfermalg}, the star operators in Eq.~\eqref{eq_fermstar} are independent in that case and at least one of them must be added to the generators. 

Let us consider any one of them, centered at the root $v_0$, and argue that we can obtain any other from it through elementary moves involving the bipartite generators:
\begin{equation}\label{eq_fermalgA}
    S_{\{v_k\}} =\prod_{k=1}^{D_F}\,\psi_{v_k}\,W^{\chi_F}_{R_{v_k}}=(-1)^{s_1}\psi_{v_1}\,W^{\chi_F}_{R_{v_1}}\,\prod_{k=2}^{D_F}\,\psi_{v_k}\,W^{\chi_F}_{R_{v_k}}\,.
\end{equation}
On the right hand side, we have factored out the ``ray'' connecting $v_1$ with the root $v_0$, which may lead to an overall sign $(-1)^{s_1}$, depending on where in the subsystem ordering of the graded tensor product $v_1$ sits (thus taking the anti-commutation in Eq.~\eqref{eq_fermanticommute} into account). Let us assume that $v_k\neq v_0$, $k=1,\ldots,D_F$ for definiteness. Now using $\mathbbm{1}_{v_1}=\psi_{v_1}\psi_{v_1}^\dag+\psi_{v_1}^\dag\psi_{v_1}$ and keeping track of all the signs from anti-commutations, we find, using one of the bipartite generators,
\begin{equation}
    (-1)^{s_1}\,S_{\{v_k\}} \psi_{v_1}^\dag W^{\bar{\chi}_F}_{R_{v_1}}\,\psi_{v_0}+(-1)^{s_1+1}\psi_{v_1}^\dag W^{\bar{\chi}_F}_{R_{v_1}}\,\psi_{v_0}\,S_{\{v_k\}} = \psi_{v_0}\,\prod_{k=2}^{D_F}\,\psi_{v_k}\,W^{\chi_F}_{R_{v_k}}=(-1)^{s_0}\,S_{\{v_k'\}}\,,
\end{equation}
where $s_0$ parametrizes a second sign that arises upon anti-commuting $\psi_{v_0}$ into its right place in the subsystem ordering in the graded tensor product structure, and $\{v_k'\}$ agrees with $\{v_k\}$, except for the replacement $v_1\to v_0$. In other words, our procedure has cut of the ray $\psi_{v_1}\,W^{\chi_F}_{R_{v_1}}$ linking $v_1$ and $v_0$ from the star and instead placed a $\psi_{v_0}$ at the root, resulting in a new star operator.

Using such elementary moves, we can move a $\psi_{v_k}$ from any ray tip to any other $\psi$-unoccupied vertex in the star, or even build new rays, as well as move the star center $v_0$ to any other vertex in the lattice. Similarly, with such elementary moves one can ``glue'' distinct stars together. 

In this manner, one can obtain all possible net-zero charge  combinations of the ($U_{v_0}^g$-ignoring) relational observable generators 
\begin{equation}
    \widetilde{O}^{\bm{g}}_{\psi_v|R}\propto \psi_v W^{\chi_F}_{R_{v_k}}
\end{equation}
that one obtains in analogy to Eq.~\eqref{alg_tilde_G} for the algebra that is invariant under gauge transformations everywhere except at $v_0$. Indeed, all of these net-zero charge combinations correspond to graphs in the tree $R$ equipped with certain $\psi,\psi^\dag$ combinations in appropriate vertices. These then exhaust the image of $\mathcal{G}_{v_0}$ (cf.~Eq.~\eqref{eq_drmatalg3}). 

In conclusion, it suffices to add a single star operator and its conjugate to the bipartite generators to obtain a complete set for $\mathcal{A}_{\rm matter}^{\rm dr}$, yielding Eq.~\eqref{eq_finfermalg}.

\bibliographystyle{utphys-modified}
\bibliography{refs.bib}

\end{document}